\documentclass[a4paper,11pt]{article}
\pdfoutput=1
\usepackage[margin=2.5cm]{geometry}
\usepackage[utf8]{inputenc}
\usepackage[T1]{fontenc}
\usepackage{pdflscape}
\usepackage{enumitem}
\usepackage{hyperref}
\usepackage{tikz}
\usepackage{ifthen}
\usepackage{xspace}
\usepackage[sort&compress,numbers,square,comma,longnamesfirst]{natbib}
\usepackage{subcaption}
\usepackage{authblk}

\DeclareFontShape{T1}{cmr}{m}{scit}{<->ssub*cmr/m/sc}{}%Avoid some font warning

\usetikzlibrary{patterns,arrows,decorations.pathreplacing,decorations.pathmorphing,calc,arrows.meta,shapes,backgrounds}

\definecolor[named]{urlblue}{cmyk}{1,0.58,0,0.21}

\hypersetup{
    breaklinks=true,
    colorlinks=true,
    citecolor=purple!70!blue!60!black,
    linkcolor=purple!70!blue!90!black,
    urlcolor=urlblue,
    pdflang={en},
    pdfborder={0 0 0},
    pdftitle={Tight Complexity Bounds for Counting Generalized Dominating Sets in Bounded-Treewidth Graphs},
    pdfauthor={Jacob Focke, Dániel Marx, Fionn {Mc Inerney}, Daniel Neuen, Govind S. {Sankar}, Philipp Schepper, and Philip Wellnitz},
    pdfcreator={},
}

\usepackage{amsmath, amssymb, amstext, amsthm, mathtools}
\usepackage{thm-restate}

\makeatletter
\DeclareRobustCommand{\crefnosort}[1]{%
	\begingroup\@cref@sortfalse\cref{#1}\endgroup
}
\makeatother

\newcommand{\CC}{{\mathcal C}}

\newcommand{\CG}{{\mathcal G}}

\newcommand{\CJ}{{\mathcal J}}

\newcommand{\CP}{{\mathcal P}}

\newcommand{\CR}{{\mathcal R}}
\newcommand{\CS}{{\mathcal S}}

\newcommand{\CY}{{\mathcal Y}}
\newcommand{\CZ}{{\mathcal Z}}

\newcommand{\ZZ}{{\mathbb Z}}

\newcommand{\NN}{{\ZZ_{\geq 0}}}
\newcommand{\SetQ}{{\mathbb Q}}% Rational numbers

\newcommand{\tx}{\widetilde{x}}
\newcommand{\ty}{\widetilde{y}}

\newcommand{\tA}{\widetilde{A}}
\newcommand{\tB}{\widetilde{B}}

\newcommand{\tJ}{\widetilde{J}}

\newcommand{\tw}{{\sf tw}}
\newcommand{\pw}{\textup{\textsf{pw}}} % Pathwidth
\newcommand{\ar}{{\sf ar}}
\newcommand{\pwar}{\textup{\textsf{pa}}}

\DeclareMathOperator{\ext}{ext}
\DeclareMathOperator{\wt}{wt}

\newcommand{\HWeq}[2][]{\ensuremath{\mathtt{HW}_{= #2}%
		\ifthenelse{\equal{#1}{}}{}{^{(#1)}}}}
	\newcommand{\HWge}[2][]{\ensuremath{\mathtt{HW}_{\ge #2}%
			\ifthenelse{\equal{#1}{}}{}{^{(#1)}}}}
	\newcommand{\HWle}[2][]{\ensuremath{\mathtt{HW}_{\le #2}%
			\ifthenelse{\equal{#1}{}}{}{^{(#1)}}}}
\newcommand{\HWin}[2][]{\ensuremath{\mathtt{HW}_{\in #2}%
		\ifthenelse{\equal{#1}{}}{}{^{(#1)}}}}
\newcommand{\EQ}[1]{\ensuremath{\mathtt{EQ}^{(#1)}}}

\newcommand{\EQset}{\ensuremath{\mathtt{EQ}}\xspace}
\newcommand{\HWsetGen}[1]{\ensuremath{\mathtt{HW}_{#1}}}
\newcommand{\HWset}[1]{\HWsetGen{= #1}}

\newcommand{\alt}[1]{#1^\ddagger}
\newcommand{\Rel}{\text{Rel}}
\newcommand{\altRel}{\alt{\Rel}}

\mathchardef\mhyph="2D
\newcommand{\DomSetGeneral}[4]{\ensuremath{(#1,#2)\mhyph}\textsc{{#3}Dom\-Set\ensuremath{^{#4}}}\xspace}

\newcommand{\DomSet}[2]{\DomSetGeneral{#1}{#2}{}{}}
\newcommand{\srDomSet}{\DomSetGeneral{\sigma}{\rho}{}{}}
\newcommand{\DomSetRel}[3][\Rel]{\DomSetGeneral{#2}{#3}{}{#1}}
\newcommand{\srDomSetRel}[1][\Rel]{\DomSetGeneral{\sigma}{\rho}{}{#1}}
\newcommand{\CountDomSet}[2]{\DomSetGeneral{#1}{#2}{\#}{}}
\newcommand{\srCountDomSet}{\DomSetGeneral{\sigma}{\rho}{\#}{}}
\newcommand{\CountDomSetRel}[2]{\DomSetGeneral{#1}{#2}{\#}{\Rel}}
\newcommand{\srCountDomSetRel}[1][\Rel]{\DomSetGeneral{\sigma}{\rho}{\#}{#1}}

\DeclarePairedDelimiter{\ceil}{\lceil}{\rceil}
\DeclarePairedDelimiter{\floor}{\lfloor}{\rfloor}
\DeclarePairedDelimiter{\abs}{\lvert}{\rvert}
\newcommand{\deff}{\coloneqq}

\renewcommand{\O}{O}
\newcommand{\Oh}{O} % O from big-O notation; maybe change back later

\newcommand*\from{\colon}
\newcommand{\allSets}{\Omega_0}
\newcommand{\allSetsne}{\Omega}
\newcommand{\pwred}{\ensuremath{\le_{\pw}}}
\newcommand{\arred}{\ensuremath{\le^{\ar}}}
\newcommand{\pwarred}{\ensuremath{\le_{\pw}^{\ar}}}

\newcommand{\port}{u} % portal vertex
\newcommand{\Port}{U} % Set of portal vertices

\newcommand{\inverse}[2][]{{\operatorname{\mathsf{inv}}^{#1}({#2})}}
\newcommand{\inv}{\overline}

\newcommand{\ttop}{\textup{top}}
\newcommand{\rhoMax}{r_{\ttop}}
\newcommand{\sigMax}{s_{\ttop}}
\newcommand{\allMax}{t_{\ttop}}

\newcommand{\rhoMin}{r_{\min}}
\newcommand{\sigMin}{s_{\min}}

\newcommand{\rhoStates}{{\mathbb R}}
\newcommand{\sigStates}{{\mathbb S}}
\newcommand{\allStates}{{\mathbb A}}
\newcommand{\B}{\mathbb B}

\newcommand{\rhoStatesExt}{\mathbb R_{{\sf full}}}
\newcommand{\sigStatesExt}{\mathbb S_{{\sf full}}}
\newcommand{\allStatesExt}{\mathbb A_{{\sf full}}}

\newcommand{\encoder}[1]{\ensuremath{#1}\text{-manager}}
\newcommand{\scope}{\operatorname{\mathsf{scp}}}
\newcommand{\rel}{R}

\newcommand{\rank}{rank\xspace}
\newcommand\Bl{B} % A type of blocks in the manager
\newcommand\Br{\inv B} % Another type of blocks in the manager

\usepackage{bm}

\def\fragmentco#1#2{\bm{[}\,#1\,\bm{.\,.}\,#2\,\bm{)}}
\def\fragmentoc#1#2{\bm{(}\,#1\,\bm{.\,.}\,#2\,\bm{]}}
\def\fragmentoo#1#2{\bm{(}\,#1\,\bm{.\,.}\,#2\,\bm{)}}
\def\fragment#1#2{\bm{[}\,#1\,\bm{.\,.}\,#2\,\bm{]}}
\def\nset#1{\numb{#1}}
\def\position#1{\bm{[}\,#1\,\bm{]}}
\def\vposition#1{\bm{[}\,#1\,\bm{]}}
\def\occ#1#2{\ensuremath \#_{#2}(#1)}

\def\mname{{\mathrm{m}}}

\def\degvec#1{\ensuremath \vec{w}(#1)}

\newcommand{\numb}[1]{\fragment{1}{#1}}

\makeatletter
\renewenvironment{cases}{%
  \matrix@check\cases\env@cases
}{%
  \endarray\right.%
}
\def\env@cases{%
  \let\@ifnextchar\new@ifnextchar
  \left\lbrace
  \def\arraystretch{1.1}%
  \array{@{\;}c@{\quad}l@{}}%
}
\makeatother

\def\emptyset{\varnothing}
\def\eps{\varepsilon}
\def\epsilon{\varepsilon}

\usepackage[noabbrev,capitalize]{cleveref}

\numberwithin{equation}{section}
\numberwithin{figure}{section}

\newtheorem{theorem}{Theorem}[section]
\newtheorem{lemma}[theorem]{Lemma}
\newtheorem{fact}[theorem]{Fact}
\newtheorem{corollary}[theorem]{Corollary}

\newtheorem{observation}[theorem]{Observation}
\newtheorem{definition}[theorem]{Definition}

\newtheorem{claim}[theorem]{Claim}

\crefname{mtheorem}{Main Theorem}{Main Theorems}
\Crefname{mtheorem}{Main Theorem}{Main Theorems}
\crefname{obs}{Observation}{Observations}
\Crefname{obs}{Observation}{Observations}
\crefname{fact}{Fact}{Facts}
\Crefname{fact}{Fact}{Facts}
\crefname{problem}{Problem}{Problems}
\Crefname{problem}{Problem}{Problems}
\crefname{conjecture}{Conjecture}{Conjectures}
\Crefname{conjecture}{Conjecture}{Conjectures}
\crefname{claim}{Claim}{Claims}
\Crefname{claim}{Claim}{Claims}

\newenvironment{claimproof}[1][\unskip]{ \begin{proof}[Proof of Claim #1.]
   }{ \end{proof} }

\title{	Tight Complexity Bounds for\\
		Counting Generalized Dominating Sets in\\
		Bounded-Treewidth Graphs\\[1ex]
		\large Part II: Hardness Results
		}
\author[1]{Jacob Focke}
\author[1]{D\'{a}niel Marx}
\author[2]{Fionn {Mc Inerney}}
\author[3]{Daniel Neuen}
\author[4]{\\Govind S. {Sankar}}
\author[1]{Philipp Schepper}
\author[5]{Philip Wellnitz}
\affil[1]{CISPA Helmholtz Center for Information Security}
\affil[2]{Algorithms and Complexity Group, TU Wien}
\affil[3]{School of Computing Science, Simon Fraser University}
\affil[4]{Duke University}
\affil[5]{Max Planck Institute for Informatics, SIC}
\date{}

\usepackage{tocloft}

\begin{document}

\maketitle

\begin{abstract}
  For a well-studied family of domination-type problems, in bounded-treewidth graphs, we investigate whether it is possible to find faster algorithms.
  For sets $\sigma,\rho$ of non-negative integers, a $(\sigma,\rho)$-set of a graph $G$ is a set $S$ of vertices such that $|N(u)\cap S|\in \sigma$ for every $u\in S$, and $|N(v)\cap S|\in \rho$ for every $v\not\in S$.
  The problem of finding a $(\sigma,\rho)$-set (of a certain size) unifies common problems like \textsc{Independent Set}, \textsc{Dominating Set}, \textsc{Independent Dominating Set}, and many others.

  In an accompanying paper, it is proven that, for all pairs of finite or cofinite sets $(\sigma,\rho)$, there is an algorithm that counts $(\sigma,\rho)$-sets in time $(c_{\sigma,\rho})^\tw\cdot n^{\O(1)}$ (if a tree decomposition of width $\tw$ is given in the input).
  Here, $c_{\sigma,\rho}$ is a constant with an intricate dependency on $\sigma$ and $\rho$. 
  Despite this intricacy, we show that the algorithms in the accompanying paper are most likely optimal, i.e.,
  for any pair $(\sigma, \rho)$ of finite or cofinite sets where the problem is non-trivial, and any $\varepsilon>0$, a $(c_{\sigma,\rho}-\varepsilon)^\tw\cdot
  n^{\O(1)}$-algorithm counting the number of $(\sigma,\rho)$-sets would violate the Counting Strong Exponential-Time Hypothesis (\#SETH).
  For finite sets $\sigma$ and $\rho$, our lower bounds also extend to the decision version, showing that those algorithms are optimal in this setting as well.
\end{abstract}

\makeatletter{\renewcommand*{\@makefnmark}{}
	\footnotetext{This work follows an accompanying paper presenting corresponding algorithmic results~\cite{FockeMMNSSW23i}. A preliminary version of this paper appeared in the Proceedings of the 2023 Annual ACM-SIAM Symposium on Discrete Algorithms (SODA). Research supported by the European Research Council (ERC) consolidator grant No.~725978 SYSTEMATICGRAPH and the Austrian Science Foundation (FWF, project P31336).}\makeatother}

\thispagestyle{empty}

\newpage
\tableofcontents

\thispagestyle{empty}

\newpage
\setcounter{page}{1}

\section{Introduction}
\label{sec:intro}
Parameterized complexity has emerged as a vital paradigm used to find tractable instances of NP-hard problems. Specifically, treewidth is an essential parameter in this context since the structure of bounded-treewidth graphs can be exploited to obtain efficient algorithms for many NP-hard problems using a dynamic programming approach~\cite{DBLP:conf/icalp/Bodlaender88,DBLP:journals/dam/ArnborgP89,BERN1987216}. Due to the existence of many such algorithms, there has been a concerted effort to optimize their efficiencies. In particular, for problems solvable in time $c^\tw\cdot n^{\Oh(1)}$, a technique introduced by \citet{LokshtanovMS18} gives, assuming the Strong Exponential Time Hypothesis (SETH), tight lower bounds on the best possible $c$ appearing in the running time~\cite{DBLP:journals/siamcomp/OkrasaR21,DBLP:conf/esa/OkrasaPR20,DBLP:conf/stacs/EgriMR18,DBLP:conf/soda/CurticapeanLN18,DBLP:conf/iwpec/BorradaileL16,DBLP:journals/dam/KatsikarelisLP19,MarxSS21,CurticapeanM16,DBLP:conf/soda/FockeMR22}. For example, \citet{LokshtanovMS18} showed that, assuming SETH, $3^\tw\cdot n^{\Oh(1)}$ is optimal for \textsc{Dominating Set}, that is, there is no algorithm solving \textsc{Dominating Set} in time $(3-\epsilon)^\tw \cdot n^{\Oh(1)}$ for some $\epsilon>0$ when given a graph with a tree decomposition of width $\tw$. Along with an accompanying paper~\cite{FockeMMNSSW23i}, we prove similar tight bounds for generalized domination problems. In this paper, we prove lower bounds that are tightly matching the algorithms presented in~\cite{FockeMMNSSW23i}.

The notion of $(\sigma,\rho)$-sets was introduced by \citet{Telle94} as a generalization of problems related to independent sets and dominating sets.
For sets $\sigma,\rho$ of non-negative integers, a $(\sigma,\rho)$-set of a graph $G$ is a set $S$ of vertices such that $|N(u)\cap S|\in \sigma$ for every $u\in S$, and $|N(v)\cap S|\in \rho$ for every $v\not\in S$.
The problem of finding a $(\sigma, \rho)$-set (of a certain size) generalizes many well-studied algorithmic problems such as \textsc{Independent Set}, \textsc{Strong Independent Set}, \textsc{Dominating Set}, \textsc{Independent Dominating Set}, \textsc{Perfect Code}, \textsc{Total Dominating Set}, \textsc{Perfect Dominating Set}, \textsc{Induced Bounded-Degree Subgraph}, and \textsc{Induced $d$-Regular Subgraph}.

Prior to~\cite{FockeMMNSSW23i},
the best algorithms known for any pair
of finite or cofinite sets $(\sigma,\rho)$ for any of the following variants
were due to van Rooij~\cite{Rooij20}:
decision, minimization/maximization, and counting
(ignoring problems for which polynomial-time algorithms are known).
To give the exact result, we first need some notation.
For a set $\sigma$ of finite or cofinite integers,
let $\sigMax$ denote the maximum element of $\sigma$ if $\sigma$ is finite,
and the maximum missing integer plus one if $\sigma$ is cofinite;
$\rhoMax$ is defined analogously based on $\rho$.
The following was proved by van Rooij:

\begin{theorem}[\citet{Rooij20}]
\label{thm:vanrooij-intro}
Let $\sigma$ and $\rho$ be two finite or cofinite sets. Given a graph $G$ with a tree decomposition of width $\tw$ and an integer $k$, the number of $(\sigma,\rho)$-sets of size exactly $k$ can be counted in time $(\sigMax+\rhoMax+2)^\tw\cdot n^{\Oh(1)}$.
\end{theorem}

In the accompanying paper~\cite{FockeMMNSSW23i}, we prove that for some $(\sigma,\rho)$-sets these algorithms can be improved significantly. Specifically, $(\sigma, \rho)$ is {\em $\mname$-structured} if there is a pair $(\alpha,\beta)$ such that every integer in $\sigma$ is exactly $\alpha$ mod $\mname$, and every integer in $\rho$ is exactly $\beta$ mod $\mname$. For example, the pairs $(\{0,3\},\{3\})$ and $(\{0,3\},\{1,4\})$ are both $3$-structured, but the pair
$(\{0,3\},\{3,4\})$ is not $\mname$-structured for any $\mname\ge 2$. Note that if a set is cofinite, then it cannot be $\mname$-structured for any $\mname\ge 2$.
\begin{definition}
  \label{def:intro:baseOfRunningTime}
Let $\sigma$ and $\rho$ be two finite or cofinite sets of non-negative integers. We define
$c_{\sigma,\rho}$ by setting
  \begin{itemize}
  	\item $c_{\sigma,\rho} \coloneqq \sigMax+\rhoMax+2$ if $(\sigma,\rho)$ is not $\mname$-structured for any $\mname\ge 2$,
  	\item
		$c_{\sigma,\rho} \deff \max\{\sigMax,\rhoMax\}+2$
		if $(\sigma,\rho)$ is $2$-structured, but not $\mname$-structured
		for any $\mname\ge 3$, and $\sigMax=\rhoMax$ is even, and
  	\item $c_{\sigma,\rho} \coloneqq \max\{\sigMax,\rhoMax\}+1$, otherwise.
  \end{itemize}
\end{definition}

With this definition we can formally state
the running time of the improved algorithm.

\begin{theorem}[\cite{FockeMMNSSW23i}]\label{thm:alg-main-intro}
Let $\sigma$ and $\rho$ denote two finite or cofinite sets. Given a graph $G$ with a tree decomposition of width $\tw$ and an integer $k$, the number of $(\sigma,\rho)$-sets of size exactly $k$ can be counted in time $(c_{\sigma,\rho})^\tw\cdot n^{\Oh(1)}$.
\end{theorem}

Due to the peculiar definition of $c_{\sigma,\rho}$, one could believe that further improvements are obtainable. However, in this paper, we prove that, for the counting version,
$c_{\sigma,\rho}$ \emph{precisely} characterizes the best possible base of the exponent.
In particular, our lower bounds hold even when considering pathwidth, a more restrictive parameter.
For the following hardness result, we need to exclude pairs $(\sigma,\rho)$
for which the problem is trivially solvable:
we say that $(\sigma,\rho)$ is \emph{non-trivial}
if $\rho\neq \{0\}$, and $(\sigma,\rho)\neq (\{0,1,\ldots\},\{0,1,\ldots\})$.

\begin{restatable}{theorem}{thmLBmaincounting}\label{thm:lower-main-intro}
	Let $(\sigma,\rho)$ denote a non-trivial pair of finite or cofinite sets.
	If there is an $\epsilon>0$ and an algorithm that
	counts in time $(c_{\sigma,\rho}-\epsilon)^\pw\cdot n^{\Oh(1)}$
	the number of $(\sigma,\rho)$-sets in a given graph
	with a given path decomposition of width $\pw$,
	then the Counting Strong Exponential Time Hypothesis (\#SETH) fails.
\end{restatable}

The running time of the algorithm of \cref{thm:alg-main-intro} is achieved by considering
roughly $(c_{\sigma,\rho})^\tw$ subproblems at each node of the tree decomposition. One can
interpret the lower bound of \cref{thm:lower-main-intro} as showing that at least
that many subproblems need to be considered by any algorithm solving the counting problem.
For finite $\sigma$ and $\rho$, we obtain a matching lower bound for the decision version as well.

\begin{restatable}{theorem}{thmLBmaindecision}\label{thm:lower-main-intro-decision}
	Let $(\sigma,\rho)$ be a pair of finite sets such that $0 \notin \rho$.
	If there is an $\epsilon>0$ and an algorithm that decides in time
	$(c_{\sigma,\rho}-\epsilon)^\pw\cdot n^{\Oh(1)}$
	whether there is a $(\sigma,\rho)$-set in a given graph
	with a given path decomposition of width $\pw$,
	then the Strong Exponential Time Hypothesis (SETH) fails.
\end{restatable}

Note that, for the decision version, we must omit all pairs $(\sigma,\rho)$
with $0 \in \rho$ since the empty set is a $(\sigma,\rho)$-set of any graph in this case.
Interestingly, the lower bounds from \cref{thm:lower-main-intro-decision} do not apply if $\sigma$ or $\rho$ is cofinite. Indeed, in~\cite{FockeMMNSSW23i}, we use the technique
of representative
sets~\cite{FominLPS17,FominLPS16,Monien85,DBLP:conf/wg/PlehnV90,Marx09,DBLP:journals/iandc/BodlaenderCKN15}
to obtain better algorithms in this case.
However, to obtain a tight optimal running time in this case,
there are two significant hurdles:
proving tight upper bounds on the size of representative sets,
and understanding whether they can be handled without matrix-multiplication based methods.

We refer the reader to the accompanying paper~\cite{FockeMMNSSW23i} for a more detailed introduction including further motivation for $(\sigma,\rho)$-sets and our approach, as well as directions for further work.

%%% Local Variables:
%%% mode: latex
%%% TeX-master: "gen-dom-set"
%%% End:

\section{Technical Overview}
\label{sec:technical-overview}
At a very high-level, our approach for deriving our lower bound results
follows previous works on problems parameterized by treewidth
and pathwidth~\cite{%
CurticapeanM16,LokshtanovMS18,MarxSS21,MarxSS22}.
However, the high-level structure is then filled
with a plethora of technical details that are designed specifically for the problem at hand.

We start by presenting a reduction from SAT to \srDomSet (given a graph $G$, decide whether $G$ has a $(\sigma,\rho)$-set)
by constructing a graph which has ``small'' pathwidth.
Instead of giving a direct reduction, we split it into two parts.
The first part shows a lower bound for the intermediate problem
\srDomSetRel
which extends the \srDomSet problem
by the possibility to impose relational constraints on the selection status of vertices.
Such a constraint consists of a scope $D$ of vertices together with a $\abs{D}$-ary relation that specifies which subsets of $D$ are allowed to be selected in a solution.
In a second step, we then show how to model arbitrary relational constraints using the original problem (without relations).

\subsection{Lower Bound for the Problem with Relations}
Let us first consider the decision problem with relations \srDomSetRel.
We explain the lower bound approach based on the example that the sets $\sigma$ and $\rho$ are single-element sets.
Afterward, we handle difficulties in the construction
and explain how to extend the approach to the counting problem and the case of cofinite sets.

\subparagraph*{Naive and Improved Construction.}
We first sketch the naive approach and analyze its downsides.
See \cref{fig:lower:intro-naive} for a visualization of the following construction.
From a SAT instance we define an instance of \srDomSetRel
that uses a graph with a grid-like structure
with one column for each clause of the SAT formula,
and one row for each variable of the formula,
where each edge in each row is subdivided by an \emph{information vertex}. The selection status of the first information vertex in row $i$ encodes whether the corresponding variable $x_i$ is assigned $0$ or $1$.
The crossing points of the grid are \emph{relations}
which check whether the clauses
are satisfied by the corresponding assignments --- the relations in column $j$ enforce the $j$th clause. The relations also propagate the assignment information (i.e., the selection status of the information vertices) from one column to the next, and hence, to the next relation/clause.
As the variable assignment is encoded only by two states of the information vertices (selected or not),
it turns out that this approach yields only a lower bound of $(2-\eps)^{\pw} \cdot n^{\Oh(1)}$.

\begin{figure}[t]
  \centering
  \includegraphics[width=0.8\textwidth]{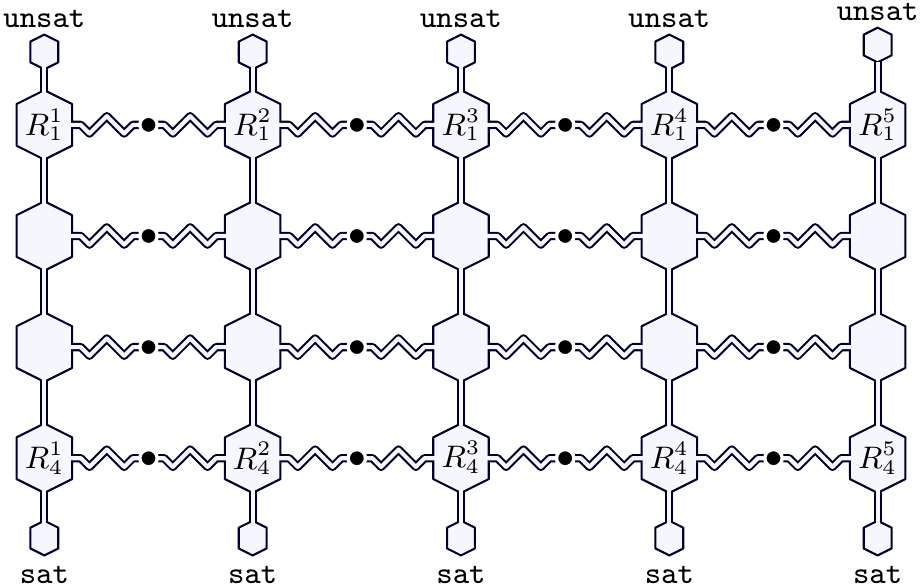}
  \caption{
  The naive construction for the weak $(2-\eps)^{\pw}$ lower bound
  given a SAT formula with five clauses, depicted as columns,
  and four variables, depicted as rows.
  The relations $R_i^j$ check if the assignment for the $i$th variable
  satisfies the $j$th clause.
  }
  \label{fig:lower:intro-naive}
\end{figure}

To obtain stronger lower bounds
of the form $(\abs{\allStates}-\epsilon)^{\pw} \cdot n^{\Oh(1)}$,
the idea is to use states for the information vertices that depend on both the selection status and the number of selected neighbors of the information vertices.
In an improved construction, the information vertices might have selected vertices ``to
the left'', that is, in the previous column, and also ``to the right'', that is, in the next column.
We use the selected vertices to the left to encode the assignment, whereas the selected vertices to the right ensure that the information vertex has a feasible number of selected neighbors overall.
The strongest lower bound can be achieved if the information vertices can take on every state from $\allStates$.
We explain later why this is not always possible.

Before we describe it in more detail,
see \cref{fig:lower:intro} for an illustration of the improved construction.
Assuming that we can model all the states in $\allStates$,
we partition the variables into $n/\log(\abs{\allStates})$ groups
of size $\log(\abs{\allStates})$ each.%
\footnote{For ease of presentation, we ignore rounding and parity issues here, which can be resolved by partitioning into a somewhat smaller number of  somewhat larger groups.}
By this choice, there is a one-to-one correspondence
between assignments to each group and the states in $\allStates$.
Similarly to the naive approach,
the constructed graph has a grid-like structure with $m$ columns,
each corresponding to one clause of the SAT formula,
but only $n/\log(\abs{\allStates})$ rows,
each corresponding to one group of variables.
Moreover, at the crossing points of rows and columns, we have relations
which check whether the assignment satisfies the clause.
The most notable difference is the neighborhood of the information vertices
which we place between two crossing points in the grid.
Let $w_i^j$ be the information vertex
that is in row $i$ between columns $j$ and $j+1$.
% with the crossing between the $i$th row and $j$th column to its left
% and the crossing between the $i$th row and $j+1$th column to its right.
Each such information vertex gets $\max\{\sigMax,\rhoMax\}$ neighbors on the left
and the same number of neighbors on the right.
The state of an information vertex is determined by its selection status
and the number of selected neighbors \emph{to its left}.
As mentioned above,
this state encodes the assignment for the corresponding group of variables.
% Additionally, the relations at the crossing points are connected
% to the information vertices on the left plus their neighbors to the right
% and likewise the information vertices on the right plus their neighbors on the left.
In each column, the edge between two relations
which are assigned to two neighboring crossing points
is used to transfer the information
whether the corresponding clause has already been satisfied.

For the basic idea, it suffices to think of $\sigma$, $\rho$ as single-element sets.
For example, let $\sigma=\{4\}$ and $\rho=\{3\}$.
Whenever the information vertex $w_i^j$ is selected
and has $\ell \in \fragment{0}{\sigMax}$ selected neighbors to the left,
it must have $\sigMax-\ell=4-\ell$ selected neighbors on the right
because of the constraints imposed by $\sigma$,
i.e., in our case a selected vertex needs a total of $4$ selected neighbors.
The relation at crossing point $(i, j+1)$ is defined such that
the next information vertex $w_i^{j+1}$
gets $\ell$ selected neighbors from the left
whenever $w_i^j$ has $4-\ell$ selected neighbors to its right.
By defining all relations which we assign to crossing points in this way,
all information vertices of one row
get the same number of neighbors from the left.
Hence, they can encode the same assignment.
The same arguments work for the case when $w_i^j$ is not selected,
but then $w_i^j$ must have $\rhoMax-\ell=3-\ell$ neighbors to the right
if it has $\ell$ neighbors to the left.
As the behavior of the relation depends on the information vertices,
they are additionally connected to the information vertices
and not only their neighbors toward the relation.
To ensure that clauses that are initially unsatisfied are eventually satisfied,
we add corresponding relations to the first and last row of the graph.

Observe that the information vertices of one column cut the graph in two halves.
Hence, the pathwidth of the graph can be bounded by the number of rows
plus some constant which takes care of the relations and other vertices.
Moreover, when we cut the graph at the information vertices of one column,
we get a direct correspondence to our algorithmic results.
Indeed, the states the information vertices might obtain in a solution
(which are specified by their selection status and the number of selected vertices to the left)
precisely correspond to the strings in the language compatible with the graph.
Consequently, for each such set of information vertices,
the number of states they can have
is trivially upper-bounded by $\abs{\allStates}^{\pw+1}$.
However, as can be seen in the accompanying paper~\cite{FockeMMNSSW23i},
if $(\sigma,\rho)$ is $\mname$-structured, then
the number of states that can actually appear is often much smaller.
In the following paragraph, we explain how this fact comes up in the lower bound construction.

\begin{figure}[t]
  \centering
  \includegraphics{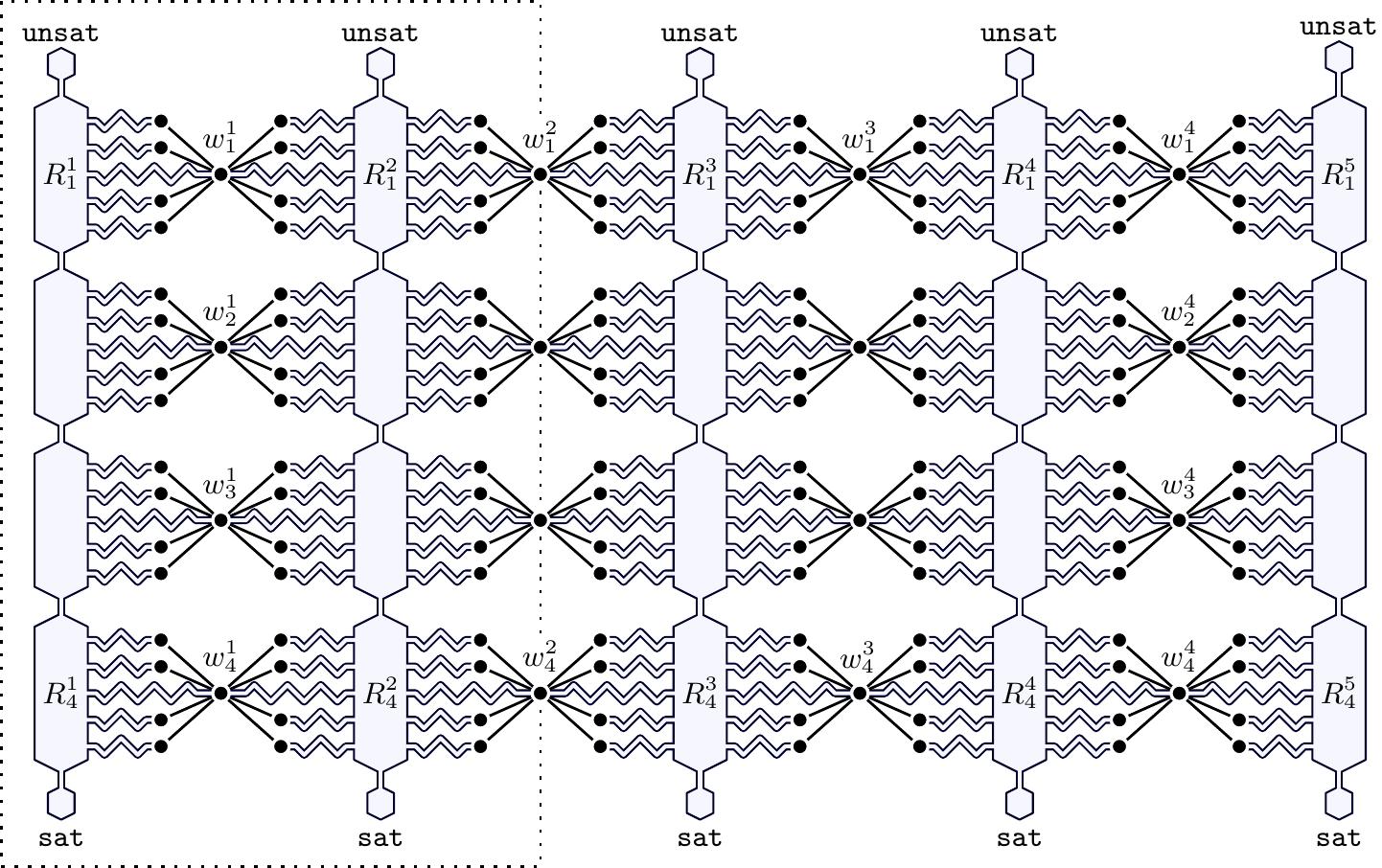}
  \caption{
  The figure illustrates the construction for the lower bound
  given a SAT formula with five clauses,
  depicted as columns,
  and variables which are grouped into four groups,
  depicted as rows.
  The relations $R_i^j$ check if the assignment for the $i$th group
  satisfies the $j$th clause.
  }
  \label{fig:lower:intro}
\end{figure}

\subparagraph*{Difficulties in the Construction.}
Observe that in the previous example, the information vertices always have a feasible number of selected neighbors;
if there are $\ell$ selected vertices on the left,
then there are $\sigMax-\ell$ or $\rhoMax-\ell$ selected vertices on the right,
depending on whether the information vertex is selected or not, respectively.
However, the previous construction ignores a very subtle
but surprisingly crucial issue:
the neighbors of the information vertices
also need a number of selected neighbors that is feasible according to $(\sigma, \rho)$.
More concretely, let $v$ be some information vertex
and $N_L(v)$ be the set of its neighbors on the left.
(The neighbors on the right are treated similarly.)
The idea is to find a graph $J$ that can be attached to $N_L(v)$
such that depending on the selection status of the vertices in $N_L(v)$,
and depending on the selection status of $v$
(which crucially also affects the vertices in $N_L(v)$),
there is some selection of the vertices in $J$
such that all vertices in $J$ (including the vertices in $N_L(v)$)
have a feasible number of selected neighbors overall.
Rephrasing, this means that depending on the state that the information vertex $v$ should obtain,
$J$ must have a feasible selection extending this state.
Note that this property is crucially needed as otherwise no solution exists.
We formalize this concept in \cref{sec:graphsWithRelations} by the definition of a \emph{manager gadget},
which handles the neighbors on the left and right simultaneously.

It turns out that these manager gadgets cannot always be constructed,
and so it might be the case that depending on $(\sigma,\rho)$
such a graph $J$ gives only feasible extensions
for a restricted set of states from $\allStates$.
When proving the lower bound,
for each case mentioned in \cref{def:intro:baseOfRunningTime}, we define a suitable subset $A \subseteq \allStates$, and then show that a so-called \emph{$A$-manager} exists, which
can handle exactly the states in $A$. For technical reasons, we also ensure that the supported set of states $A$ is \emph{closed under the inverse}. This means that whenever $A$ supports a state $\sigma_i$ or $\rho_i$, then it also supports $\sigma_{\sigMax-i}$ or $\rho_{\rhoMax-i}$, respectively. Intuitively, this property is relevant to ensure that if we can handle some set of selected vertices ``to the left'' of an information vertex, then we can handle an ``inverse'' set of selected neighbors ``to the right''.
Here is the formal statement.

\begin{restatable}[Existence of Managers]{lemma}{existenceOfManager}
	\label{lem:lb:existenceOfManager}
	Let $\sigma,\rho \subseteq \NN$ be non-empty, finite or cofinite sets
	with $\rho \neq \{0\}$.
	\begin{enumerate}
		\item
		\label{lem:lb:existenceOfManager:rho}
		There is an \encoder{A} with $\abs{A} = \rhoMax+1$.

		\item
		\label{lem:lb:existenceOfManager:sigma}
		If $\sigMax \ge \rhoMax \ge 1$,
		then there is an \encoder{A} with $\abs{A} = \sigMax+1$.

		\item
		\label{lem:lb:existenceOfManager:even}
		If $(\sigma,\rho)$ is $2$-structured,
		but not $\mname$-structured for any $\mname\ge 3$,
		and $\sigMax$ and $\rhoMax \ge 1$ are even,
		then there is an \encoder{A} with
		$\abs{A}=(\sigMax+\rhoMax)/2 + 2$.

		\item
		\label{lem:lb:existenceOfManager:all}
		If $(\sigma,\rho)$ is not $\mname$-structured for any $\mname\ge 2$,
		% and $\rhoMax \ge 1$,
		then there is an \encoder{A} with
		$\abs{A} = \sigMax + \rhoMax + 2$.
	\end{enumerate}
	Moreover, each \encoder{A} is such
	that $A$ is closed under the inverse with respect to $\sigma,\rho$.
\end{restatable}

The proof of \cref{lem:lb:existenceOfManager} is given in \cref{sec:manager}.
From an $A$-manager we then obtain a lower bound
of the form $(\abs{A}-\epsilon)^{\pw}\cdot n^{\Oh(1)}$.
However, in order to obtain such a lower bound for the decision problem,
we require that $\sigma$ and $\rho$ are finite or \emph{simple cofinite}.
A set is simple cofinite if it contains, for some $k\in \NN$,
precisely the elements greater than or equal to $k$.
In the following paragraph about cofinite sets
we explain how these restrictions come into play.
We state the lower bound for the decision problem with relations.
They hold even if we restrict the problem to instances in which the arity of relations is bounded by some constant.

\begin{restatable}[Lower Bound for \srDomSetRel]{lemma}{lowerBoundForDecDomSetRel}
	\label{lem:lowerBoundWhenHavingSuitableGadget}
	Let $\sigma,\rho \subseteq \NN$ be non-empty, and finite or simple cofinite sets.
	% with $\rho\neq \{0\}$.
	Suppose there is an $A \subseteq \allStates$
	that is closed under the inverse with respect to $\sigma,\rho$
	such that there is an $\encoder{A}$.

	For all $\epsilon >0$,
	there is a constant $d$ such that
	\srDomSetRel on instances of size $n$ and arity at most $d$
	cannot be solved in time
	$(\abs{A} - \epsilon)^{k + \Oh(1)} \cdot n^{\Oh(1)}$,
	even if the input is given with a path decomposition of width $k$,
	% in each bag there are at most $\Oh(1)$ relations,
	% and all relations have arity at most $\Oh(1)$\jac{``bounded by a constant''?},
	unless SETH fails.
\end{restatable}

\cref{lem:lowerBoundWhenHavingSuitableGadget} is proved in \cref{sec:high-level:decision}.
With an eye to \cref{def:intro:baseOfRunningTime},
note that \cref{lem:lb:existenceOfManager}
together with \cref{lem:lowerBoundWhenHavingSuitableGadget}
ensure that there is no
$(c_{\sigma,\rho}-\epsilon)^\pw\cdot n^{\Oh(1)}$ time algorithm
for \srDomSetRel on graphs of size $n$ with pathwidth $\pw$,
as long as $\sigma$ and $\rho$ are finite or simple cofinite sets
with $\rho\neq \{0\}$ (and the SETH holds).
In more detail,
we obtain lower bounds of the form $(c-\epsilon)^\pw\cdot n^{\Oh(1)}$,
from weakest to strongest, as follows:
\begin{itemize}
	\item Using \cref{lem:lb:existenceOfManager}, \cref{lem:lb:existenceOfManager:rho,lem:lb:existenceOfManager:sigma}, the lower bound follows for $c=\rhoMax+1$ and $c=\sigMax+1$ (if $\sigMax\ge \rhoMax$), and consequently for $c=\max\{\rhoMax, \sigMax\}+1$.
	\item Using \cref{lem:lb:existenceOfManager}, \cref{lem:lb:existenceOfManager:even}, if $(\sigma,\rho)$ is $2$-structured,
	but not $\mname$-structured for any $\mname\ge 3$,
	and $\sigMax=\rhoMax$ is even, the lower bound follows for $c=(\sigMax+\rhoMax)/2 + 2$. Under these conditions this means that $c=\max\{\rhoMax, \sigMax\}+2$.
	\item From \cref{lem:lb:existenceOfManager}, \cref{lem:lb:existenceOfManager:all}, if $(\sigma,\rho)$ is not $\mname$-structured for any $\mname\ge 2$,
	the lower bound follows for $c = \sigMax + \rhoMax + 2$.
\end{itemize}

\subparagraph*{Handling Cofinite Sets.}
As mentioned before, if $\sigma$ and $\rho$ are single-element sets,
an information vertex really transfers information:
its selection status and the number of selected neighbors on the left
uniquely determine the number of selected neighbors on the right.
We can interpret a selected vertex as a requirement ``$=\sigMax$'' on the number of selected neighbors.
With some changes, the argument can be made to work with arbitrary finite $\sigma$:
such a set enforces a requirement that is \emph{more restrictive} than ``$\le \sigMax$''.
A similar argument works if $\sigma$ is \emph{simple cofinite}, by which we mean that $\sigma$ contains all elements greater than or equal to some $k\in \NN$, but no elements that are smaller than $k$.
This means that $\sigma=\{\sigMax,\sigMax +1,\ldots\}$,
which enforces \emph{exactly} ``$\ge \sigMax$'' when applied to an information vertex.
However, a cofinite set that is not simple enforces a requirement
that is \emph{less restrictive} than ``$\ge \sigMax$'',
which is not useful for our purposes.

We handle general cofinite sets by presenting a reduction from simple cofinite sets. In
particular, suppose that $\sigma$ is cofinite; then we reduce from the case where
$\sigma'=\{\sigMax, \sigMax+1,\ldots\}$. Crucially, this reduction uses interpolation
techniques, and therefore, works only for the problem of \emph{counting} $(\sigma, \rho)$-sets. Given a graph $G$, the reduction creates a graph $G'$ by attaching to each vertex $v$ a certain gadget with the following properties:
\begin{itemize}
  \item If $v$ is unselected, then the gadget has a unique extension,
  which does not give any additional selected neighbors to $v$.
  \item If $v$ is selected, then every extension of the gadget
  provides at most $\sigMax$ new selected neighbors to $v$.
  \item If $v$ is selected, then the gadget has $d_i$ extensions
  that provide $i$ new selected neighbors to $v$,
  where $i \in \fragment{0}{\sigMax}$.
  \end{itemize}
Given the numbers $d_0$, $\dots$, $d_{\sigMax}$,
it is not very difficult to use the relations to construct a gadget
with exactly these properties and number of extensions.

Let us consider $\sigma=\{1,3,4,6,7,\dots\}$ as an illustrative example,
and suppose for simplicity that we attach such a gadget only to a single vertex $v$.
Note that we have $\sigMax=6$.
Consider now partial solutions of the original graph $G$
where every vertex satisfies the requirements, except potentially $v$.
Suppose that $v$ is selected and has $0$ selected neighbors in such a partial solution;
it needs $1$, $3$, $4$, or at least $6$ further selected neighbors
to satisfy the $\sigma$-constraint.
Therefore, the gadget has exactly $d_1+d_3+d_4+d_6$ extensions
where the degree of $v$ becomes a member of $\sigma$.
Similarly, if $v$ already has one selected neighbor in the partial solution of $G$,
then it needs $0$, $2$, $3$, or \emph{at least} $5$ further selected neighbors,
and thus, the gadget has exactly $d_0+d_2+d_3+d_5+d_6$ extensions
where $v$ satisfies the degree condition, and so on.

We would like to build a gadget
where the integers $d_i$ are chosen in such a way
that the gadget has exactly 1 (or some other constant) extension
if $v$ already has at least 6 selected neighbors,
and 0 extensions if $v$ has less than 6 selected neighbors.
If we can build such a gadget,
then we can effectively force that $v$ has at least 6 selected neighbors.
Based on the discussion above, the $d_i$'s have to be chosen such that
they satisfy the following system of equations:

  \[
    \begin{pmatrix}
      0  & 1 & 0 & 1 & 1 & 0 & 1\\
       1 & 0 & 1 & 1 & 0 & 1 & 1\\
       0 & 1 & 1 & 0 & 1 & 1 & 1\\
       1 & 1 & 0 & 1 & 1 & 1 & 1\\
       1 & 0 & 1 & 1 & 1 & 1 & 1\\
       0 & 1 & 1 & 1 & 1 & 1 & 1\\
       1 & 1 & 1 & 1 & 1 & 1 & 1\\
    \end{pmatrix}
\cdot
    \begin{pmatrix}
      d_0 \\ d_1 \\d_2 \\d_3\\d_4\\d_5\\d_6
    \end{pmatrix}
    =
    \begin{pmatrix}
      0 \\ 0 \\0 \\0\\0\\0\\1
    \end{pmatrix}
  \]
  One can observe that the coefficient matrix has a form
  that guarantees that it is non-singular:
  every element on or below the codiagonal is 1,
  while every element one row above the codiagonal is 0.
  Therefore, the system of equations has a solution,
  so we could construct the appropriate gadget.
  The catch is, however, that some of the $d_i$'s could be negative,
  making it impossible to have a gadget with exactly that many extensions.
  We can solve this issue in the following way:
  if $d_i$ is negative,
  then we design the gadget to have $2^x\cdot |d_i|$ extensions,
  where $x$ is some parameter of the construction.
  If we attach such gadgets to every vertex,
  then it can be observed that the number of solutions
  is a polynomial $P(y)$ where we set $y=2^x$.
  Varying the values for $x$,
  we can use interpolation techniques to recover the coefficients
  of this polynomial $P$, and thus, the entire polynomial $P(y)$.
  Then, we evaluate $P$ at $y=-1$,
  which then simulates exactly $-1 \cdot \abs{d_i} = d_i$ extensions.

	So, we have sketched how to use interpolation to extend the result from \cref{lem:lowerBoundWhenHavingSuitableGadget} to arbitrary cofinite sets for the counting problem. We obtain the following result.

	\begin{restatable}[Lower Bound for \srCountDomSetRel]{lemma}{lowerBoundForCountDomSetRel}
		\label{lem:count:lowerBoundWhenHavingSuitableGadget}
		Let $\sigma,\rho \subseteq \NN$ be non-empty,
		and finite or cofinite sets.
		Suppose there is an $A \subseteq \allStates$
		that is closed under the inverse with respect to $\sigma,\rho$
		such that there is an $\encoder{A}$.

		For all $\epsilon > 0$,
		there is a constant $d$ such that
		\srCountDomSetRel
		on instances of size $n$ and arity at most $d$
		cannot be solved in time
		$(\abs{A} - \epsilon)^{k + \Oh(1)} \cdot n^{\Oh(1)}$,
		even if the input is given with a path decomposition of width $k$,
		% in each bag there are at most $\Oh(1)$ relations,
		% and all relations have arity at most $\Oh(1)$,
		unless \#SETH fails.
	\end{restatable}

	\Cref{lem:count:lowerBoundWhenHavingSuitableGadget}
  is proved in \cref{sec:high-level:counting}.
	By the same arguments that we gave for the decision problem,
  \cref{lem:lb:existenceOfManager} together with
  \cref{lem:count:lowerBoundWhenHavingSuitableGadget} ensure that
  there is no $(c_{\sigma,\rho}-\epsilon)^\pw\cdot n^{\Oh(1)}$ time algorithm
  for \srCountDomSetRel on graphs of size $n$ with pathwidth $\pw$,
  as long as $\sigma$ and $\rho$ are finite
  or (not necessarily simple) cofinite sets with $\rho\neq \{0\}$
  (and the \#SETH holds).

\subsection{Realizing Relations}
In the previous step we explained how to obtain lower bounds under SETH for the intermediate problem $\srDomSetRel$ and its counting version. Our next goal is to show that these lower bounds transfer to the original problem versions (without the relations).
Following previous approaches \cite{CurticapeanM16,MarxSS21,MarxSS22}, we first show how to model arbitrary relations using only ``Hamming weight one'' (\HWset{1}) relations.
To this end, we show a parsimonious reduction from \srDomSetRel
to \srDomSetRel[{\HWset{1}}], which denotes the restriction of \srDomSetRel
to instances that use only \HWset{1} relations.
It is important that the reduction is parsimonious
since we use it both in the decision setting and in the counting setting.

We then aim to express \HWset{1} relations using graph gadgets. Suppose there is some set of vertices $U$ whose desired selection status is subject to some relation $R$. Intuitively, we aim to attach some graph $J$ to the vertices in $U$ such that the $(\sigma,\rho)$-constraints of vertices in $J$ require a selection of vertices in $U$ that satisfies $R$.
If this works out, we say that $J$ is a \emph{realization of \(R\)}.
Crucially, we wish to be able to add such a realization
without invalidating the $\sigma$- and $\rho$-constraints of the vertices in $U$ with respect to the original graph.
Hence, we require that a realization does not add any \emph{selected} neighbors to the vertices in $U$.

\subparagraph*{Relations for the Decision Problem.}
For the decision problem, the \HWset{1} relations can be realized by graphs whose size is bounded by a function in the arity $d$, and in the values of $\rhoMax$ and $\sigMax$.
However, we obtain such realizations only for finite $\sigma$ and $\rho$. As we have seen, this is justified by the fact that, for cofinite sets, algorithmic improvements are possible (see the accompanying paper~\cite{FockeMMNSSW23i}).
With the realizations at hand, we give a pathwidth-preserving reduction
(that is, a reduction in which the pathwidth
increases only by an additive constant)
from the decision problem with relations to the one without,
as long as there is some constant upper bound on the arity of relations. Note that this
reduction holds only if $0\notin \rho$, that is, if \srDomSet is non-trivial.

\begin{restatable}{theorem}{removingRelationsForDecVersion}
	\label{lem:lower:dec:removingRelations}
	Let $\sigma,\rho$ denote finite non-empty sets with $0 \notin \rho$.
	% If all relations have arity at most $\Oh(1)$,
	% then $\srDomSetRel \twred \srDomSet$.
	For all constants $d$, there is a pathwidth-preserving reduction
	from \srDomSetRel on instances with arity at most $d$
	to \srDomSet.
\end{restatable}

\subparagraph*{Relations for the Counting Problem.}
For the counting version, the situation is more complicated.
Indeed, the proof of \cref{lem:lower:count:removingRelations} is very technical and
splits into a number of cases and sequences of reductions.
Given that there are non-trivial pairs $(\sigma, \rho)$ for which the decision problem can be solved in polynomial time, it is not plausible
that for all of these pairs relations can be realized directly by attaching some graph gadget (as this would give lower bounds also for the decision problem).
Instead, for the counting problem, in many places we heavily rely on interpolation
to isolate the number of selections we wish to enforce.
An overview of the corresponding reductions is given in
\cref{fig:count:removingRelationsSimplified}.

\begin{restatable}{theorem}{thmCountRemovingRelations}
	\label{lem:lower:count:removingRelations}
	\label{thm:RelfromDS}
	\ifnum \thesection<8
	Let $\sigma,\rho$ be finite or cofinite, non-empty sets such that $(\sigma, \rho)$ is non-trivial.
	\else
	Let $(\sigma,\rho)\in \allSetsne^2$ be non-trivial.
	\fi
	For all constants $d$, there is a pathwidth-preserving reduction
	from \srCountDomSetRel on instances with arity at most $d$
	to \srCountDomSet.
\end{restatable}

Note that
we aim to establish lower bounds also for cofinite $\sigma$ or $\rho$.
Moreover, recall that there are also fewer finite pairs $(\sigma,\rho)$
for which the counting version is polynomial-time solvable compared to the decision version.
So, also for finite sets we have to cover additional cases when showing lower bounds for the counting version.
In particular, all cases with $0 \in \rho$ are trivial for the decision version,
but are non-trivial for the counting version if $\rho \neq \{0\}$.
It turns out that realizing \HWset{1} relations is much more challenging for this larger class of sets $\sigma$, $\rho$ that we consider for the counting problem, and we use additional techniques to work around this.

We split our approach into the following three cases depending on $\sigma$ and $\rho$:
\begin{enumerate}[label = (\Alph*)]
 \item\label{item:case-1} $\rho\neq \NN$,
 \item\label{item:case-2} $\rho=\NN$ and $\sigma$ is finite, and
 \item\label{item:case-3} $\rho=\NN$ and $\sigma$ is cofinite.
\end{enumerate}
The treatment of the first two cases will be similar, whereas the third case is substantially different.
Let us first talk about the first two cases.

To reduce \srCountDomSetRel[{\HWset{1}}] to \srCountDomSet,
it is convenient to introduce a natural generalization of the problem
for which it is not necessary that every vertex has the same constraint $(\sigma,\rho)$.
Instead, we have an additional set $\CP$ of pairs
and the input contains a mapping specifying for every vertex
if it is constrained by $(\sigma,\rho)$ or by some other pair in $\CP$.
For a fixed set $\CP$ of pairs,
we write \srCountDomSetRel[\CP] for this extension.
As the first step of the proof, we show that if
\begin{itemize}
\item $(\emptyset,\{1\})\in \CP$,
\item $(\emptyset,\{0,1\})\in \CP$,
\item $(\emptyset,\ZZ_{\ge 1})\in \CP$, or
\item $(\{0\}, \NN) \in \CP$,
\end{itemize}
then \srCountDomSetRel[{\HWset{1}}] can be reduced to \srCountDomSetRel[\CP]
in a pathwidth-preserving way because \HWset{1} relations
can be simulated using one of the four pairs listed above.
Then, the main part of the proof is to show that,
for every non-trivial $(\sigma, \rho)$ with $\rho\neq \NN$
(Case \ref{item:case-1})
or $\rho= \NN$ and finite $\sigma$ (Case \ref{item:case-2}),
one of the four pairs can be simulated with the \srCountDomSet problem
(see \cref{fig:count:removingRelationsSimplified}).

\begin{figure}[t]
\centering
\begin{tikzpicture}[
  scale=.615,
  transform shape,
]
  \input{drawings/overview-tikz-style}
  \def\x{4.5cm}
  \def\xGap{1.25cm}
  \def\y{-2 cm}
  \def\xW{4.2cm}
  \def\xWimp{4.2cm-1.12mm}

  \crefname{lemma}{Lem.}{Lem.}

  \tikzset{redTur/.style={red}}
  \tikzset{%
  part/.style = { % some specific problem
    draw=gray,
    fill=gray!15,
    rounded corners=2mm,
  },
  lowerB/.style = { % Lower bound shown for a problem
  },
  important/.style = { % Lower bound shown for a problem
    line width=.66mm,
  },
  case/.style = { % Different cases (reachable by switch
  },
  res/.style = { % Reference for the result
    anchor=west,
    yshift=-0.5*\y,
  },
  }

  % Uncomment the following if we do not want to state the full name each time.
  % \renewcommand{\srCountDomSetRel}[1][\textsc{Rel}]{\ensuremath{#1}}
  % \renewcommand{\DomSetGeneral}[4] {\ensuremath{(#1,#2)\mhyph}\textsc{{#3}DS\ensuremath{^{#4}}}\xspace}

  %% BLOCK ONE (making sigma, rho cofinite)
  \node[problem, lowerB, break, important, minimum width=4*\x+\xWimp+2*\xGap] (domSetRelFull) at (1*\x+\xGap,0*\y)
    {\CountDomSetRel{\widehat\sigma}{\widehat\rho}};

  \node[problem, lowerB, break, important, minimum width=4*\x+\xWimp+2*\xGap] (domSetRel) at (\x+\xGap,2*\y)
    {\CountDomSetRel{\sigma}{\rho}};
  \draw[redTur] (domSetRelFull) -- (domSetRel);
  % \node[res] at (domSetRel) {\cref{sec:high-level:counting}};

  %% BLOCK TWO (replace relations by HW=1)
  \node[case,break] (goodCase) at (.5*\x+.5*\xGap, 3*\y)
    {$\rho \neq \NN \lor \sigma$ finite};
  \draw[switch] (domSetRel.south -| goodCase) -- (goodCase);

  \node[problem, minimum width=3*\x+\xWimp+\xGap,important] (domSetHWone) at (.5*\x+.5*\xGap, 4*\y)
    {\srCountDomSetRel[\HWeq{1}]};
  \draw[red] (goodCase) -- (domSetHWone);
  % \node[res] at (domSetHWone) {\cref{cor:realizing:parsiArbitraryToHWone}};

  %% BLOCK THREE (remove HW=1 for finite, and not everything cofinite)
  % BLOCK THREE-ONE (HW=1)
  \node[case,break] (empty1Case) at (0, 5*\y)
    {$\rho$ finite $\land\; \rhoMax-1\notin \rho$};
  \draw[switch] (domSetHWone.south -| empty1Case) -- (empty1Case);

  \node[problem,break] (empty1) at (0, 6*\y)
    {\srCountDomSetRel[(\emptyset, \{1\})]};
  \draw[red] (empty1Case) -- (empty1);
  % \node[res] at (empty1) {\cref{lem:I}};

  % BLOCK THREE-TWO (HW<=1)
  \node[case,break] (empty01Case) at (-1*\x, 5*\y)
    {$\rho$ finite $\land\; \rhoMax-1\in \rho$};
  \draw[switch] (domSetHWone.south -| empty01Case) -- (empty01Case);

  % \node[problem,break] (domSetHWleOne) at (-1*\x, 8*\y)
  %   {\srCountDomSetRel[\HWsetGen{\le1}]};
  % \draw[redTur] (empty01Case) -- (domSetHWleOne);
  % \node[res] at (domSetHWleOne) {\cref{lem:IIrelversion}};

  \node[problem,break] (empty01) at (-1*\x, 6*\y)
    {\srCountDomSetRel[(\emptyset,\{0,1\})]};
  \draw[redTur] (empty01Case) -- (empty01);
  % \node[res] at (empty01) {\cref{lem:II}};

  % BLOCK THREE-THREE (HW>=1)
  % \node[case,break] (cofCase) at (\x, 7*\y)
  %   {$\rho$ cofinite};
  % \draw[switch] ([xshift=\x] domSetHWone.south) -- (cofCase);

  % \node[problem,break,minimum width=8.75cm+\xGap] (domSetHWgeOne) at (1.5*\x+\xGap/2, 8*\y)
  %   {\srCountDomSetRel[\HWge{1}]};
  % \draw[redTur] (cofCase) -- ([xshift=-.5*\x-\xGap/2] domSetHWgeOne.north);
  % \node[res] at (\x, 8*\y) {\cref{lem:IIIrelversion}};

  \node[case,break] (emptyFrom1Case) at (\x, 5*\y)
    {$\rho$ cofinite $\land\; \rho\neq \NN$};
  \draw[switch] (domSetHWone.south -| emptyFrom1Case) -- (emptyFrom1Case);

  \node[problem,break] (emptyFrom1) at (\x, 6*\y)
    {\srCountDomSetRel[(\emptyset,\ZZ_{\ge 1})]};
  \draw[red] (emptyFrom1Case) -- (emptyFrom1);
  % \node[res] at (emptyFrom1) {\cref{lem:III}};

  % BLOCK THREE-FOUR (remove (sigma, empty)+(empty, rho)
  \node[problem,minimum width=2*\x+\xW] (sigEmptyEmptyRho) at (0, 7*\y)
    {\srCountDomSetRel[(\sigma,\emptyset)+(\emptyset,\rho)]};
  \draw[red] (empty1) -- (sigEmptyEmptyRho);
  \draw[red] (empty01) -- ([xshift=-1*\x] sigEmptyEmptyRho.north);
  \draw[red] (emptyFrom1) -- ([xshift=\x] sigEmptyEmptyRho.north);
  % \node[res, xshift=-1*\x] at (sigEmptyEmptyRho) {\cref{lem:rel:liftingRho}};
  % \node[res              ] at (sigEmptyEmptyRho) {\cref{lem:rel:liftingRho}};
  % \node[res, xshift=\x]    at (sigEmptyEmptyRho) {\cref{lem:rel:liftingRho}};

  %% BLOCK FOUR (rho everything, sigma finite)
  \node[case] (sigFinite) at (2*\x+\xGap, 5*\y)
    {$\rho=\NN \;\land\; \sigma$ finite};
  \draw[switch] (domSetHWone.south -| sigFinite) -- (sigFinite);

  \node[problem,break] (0everything) at (2*\x+\xGap, 6*\y)
    {\srCountDomSetRel[(\{0\},\NN)]};
  \draw[redTur] (sigFinite) -- (0everything);
  % \node[anchor=west] at (2*\x+\xGap, 9.5*\y) {\cref{lem:IV}};

  \node[problem,break] (sigEmpty) at (2*\x+\xGap, 7*\y)
    {\srCountDomSetRel[(\sigma,\emptyset)]};
  \draw[redTur] (0everything) -- (sigEmpty);
  % \node[res] at (sigEmpty) {\cref{lem:sigfin2,lem:sigfin3}};

  %% BLOCK FIVE (rho everything, sigma cofinite)
  \node[case] (sigCofinite) at (3*\x+2*\xGap, 3*\y)
    {$\rho=\NN \;\land\; \sigma$ cofinite};
  \draw[switch] (domSetRel.south -| sigCofinite) -- (sigCofinite);

  \node[problem,break] (cofRelWeight) at (3*\x+2*\xGap, 4*\y)
    {rel-wt. \srCountDomSetRel[(\ZZ_{\ge1},\NN),\altRel]};
  \draw[red] (sigCofinite) -- (cofRelWeight);
  % \node[res] at (cofRelWeight) {\cref{lem:RelfromweightedRelS}};

  % \node[problem,break] (cofVertWeight) at (3*\x+2*\xGap, 2*\y)
  %   {vtx-wt. \srCountDomSetRel[(\ZZ_{\ge1},\NN),\text{Rel}^*]};
  % \draw[red] (cofRelWeight) -- (cofVertWeight);
  % \node[res] at (cofVertWeight) {\cref{lem:vertexWeightedRelS}};
  %
  % \node[problem,break] (cofUnWeight) at (3*\x+2*\xGap, 3*\y)
  %   {\srCountDomSetRel[(\ZZ_{\ge1},\NN),\text{Rel}^*]};
  % \draw[red] (cofVertWeight) -- (cofUnWeight);
  % \node[res] at (cofUnWeight) {\cref{lem:relWeightedRelS}};
  %
  % \node[problem,break] (cofHWeq1) at (3*\x+2*\xGap, 6*\y)
  %   {\srCountDomSetRel[(\ZZ_{\ge1},\NN),\HWeq{1}^*]};
  % \draw[red] (cofUnWeight) -- (cofHWeq1);
  % \node[res] at (cofHWeq1) {\cref{lem:RelSfromHWS}};
  %
  % \node[problem,break] (cofHWge1) at (3*\x+2*\xGap, 8*\y)
  %   {\srCountDomSetRel[(\ZZ_{\ge1},\NN),\HWge{1}^*]};
  % \draw[red] (cofHWeq1) -- (cofHWge1);
  % \node[res] at (cofHWge1) {\cref{lem:HWeqS}};

  \node[problem,break] (cofAllEmpty) at (3*\x+2*\xGap, 6*\y)
    {\srCountDomSetRel[(\ZZ_{\ge1},\NN)+(\NN, \emptyset)]
      \footnotesize{$(\NN, \emptyset)$ 1-bounded}};
  \draw[red] (cofRelWeight) -- (cofAllEmpty);
  % \node[res] at (cofAllEmpty) {\cref{lem:HWgeS}};

  % % \node[problem,break] (cofGe1All) at (3*\x+2*\xGap, 10*\y)
  % %   {\srCountDomSetRel[(\ZZ_{\ge1},\NN)]};
  % % \draw[red] (cofAllEmpty) -- (cofGe1All);
  % % \node[res] at (cofGe1All) {\cref{lem:sigcofauxpairs}};

  \node[problem,break] (cofAllEmptyBounded) at (3*\x+2*\xGap, 7*\y)
    {\srCountDomSetRel[{(\NN,\emptyset)}]
      \footnotesize{$(\NN, \emptyset)$ $\sigMax$-bounded}};
  \draw[red] (cofAllEmpty) -- (cofAllEmptyBounded);
  % \node[res] at (cofAllEmptyBounded) {\cref{lem:sigcof2}};

  %% BLOCK SIX (the last step)
  \node (domSetAnchor) at (0, 8*\y) {};
  \node[problem,lowerB,important,minimum width=4*\x+\xWimp+2*\xGap]
    (domSet) at (1*\x+\xGap,8*\y) {\srCountDomSet};

  % From block 3
  \draw[red] (sigEmptyEmptyRho) -- (domSet.north -| sigEmptyEmptyRho);
  % \node[res] at (0, 12*\y) {\cref{lem:forceboth}};
  % From block 4
  \draw[redTur] (sigEmpty) -- (domSet.north -| sigEmpty);
  % \node[res] at (2*\x+\xGap, 12*\y) {\cref{lem:sigfin1}};
  % From block 5
  \draw[red] (cofAllEmptyBounded) -- (domSet.north -| cofAllEmptyBounded);
  % \node[res] at (3*\x+2*\xGap, 12*\y) {\cref{lem:sigcof1}};

  %% MARKING THE BLOCKS
  \begin{scope}[on background layer]
    \draw[part]
      (-1.6*\x,          0*\y-.25cm) rectangle (3.6*\x+2*\xGap,  2*\y+.25cm);
    \draw[part,fill=red!10]
      (-1.6*\x,          2*\y-.25cm) rectangle (2.6*\x+1*\xGap,  4*\y+.25cm);
    \draw[part,fill=green!50!black!5]
      (-1.6*\x,          4*\y-.25cm) rectangle (1.6*\x,          8*\y+.25cm);
    \draw[part,fill=blue!10]
      ( 1.4*\x+1*\xGap,  4*\y-.25cm) rectangle (2.6*\x+1*\xGap,  8*\y+.25cm);
    \draw[part,fill=orange!10]
      ( 2.4*\x+2*\xGap,  2*\y-.25cm) rectangle (3.6*\x+2*\xGap,  8*\y+.25cm);
  \end{scope}
\end{tikzpicture}
\caption{
A simplified overview of the reductions for the counting version of the problem.
The sets $\widehat\sigma$ and $\widehat\rho$ are simple cofinite sets,
and the sets $\sigma$ and $\rho$ are cofinite sets with
$\max(\NN\setminus \sigma) = \max(\NN\setminus \widehat\sigma)$
and $\max(\NN\setminus \rho) = \max(\NN\setminus \widehat\rho)$.
\\
The reductions from the topmost gray box
are given in \cref{sec:high-level:counting}, and
the reductions in the red box below
are given in \cref{sec:RelToHW1}.
The reductions of the remaining three boxes are covered
in \cref{sec:realizingRelations:counting}.
The green box on the left represents Case~\ref{item:case-1}
(cf.\ \cref{sec:rhoNotAll}), %,sec:forcingboth:ommittedProofs}),
the blue box in the middle represents Case~\ref{item:case-2}
(cf.\ \cref{sec:rhoiseverything}),
and the orange box on the right represents Case~\ref{item:case-3}
(cf.\ \cref{sec:rhoiseverything2}).
}
\label{fig:count:removingRelationsSimplified}
\end{figure}

Toward this goal, we simulate some number of other pairs.
A key intermediate goal is to ``force a vertex to be unselected'',
that is, to simulate an $(\emptyset,\rho)$-vertex.
Similarly, another important goal is to ``force a vertex to be selected''
by simulating a $(\sigma,\emptyset)$-vertex.
Now, let us do the following:
\begin{itemize}
  \item introduce $\rhoMax-1$ cliques, each consisting of $\min\sigma+1$
  vertices of type $(\sigma,\emptyset)$;
  \item introduce one vertex $v$ of type $(\emptyset,\rho)$, adjacent to one vertex in each clique.
  \end{itemize}
  The vertex $v$ can be considered as an always unselected vertex
  that already has $\rhoMax-1$ extra selected neighbors,
  that is, $\rho$ is ``shifted down'' by $\rhoMax-1$.
  Crucially, this works only if $\rhoMax\neq 0$, that is, if $\rho\neq \{0\}$ and $\rho\neq \NN$.
  Because we need to consider only non-trivial $(\sigma, \rho)$, we can assume $\rho\neq \{0\}$.
  Assuming that $\rho\neq \NN$ brings us to Case \ref{item:case-1}.
  In this case, depending on whether $\rho$ is cofinite,
  finite with $\rhoMax-1\in \rho$ or finite with $\rhoMax-1\not\in\rho$,
  the vertex $v$ effectively becomes an
  $(\emptyset,\ZZ_{\ge 1})$-, $(\emptyset,\{0,1\})$-,
  or $(\emptyset,\{1\})$-vertex, respectively.

  The treatment of Case \ref{item:case-2} ($\rho=\NN$ with finite $\sigma$) is similar to Case \ref{item:case-1} and is arguably even slightly simpler.
  Here it suffices to force selected vertices, i.e., to model $(\sigma, \emptyset)$-vertices, in order to ``shift down'' to a $(\{0\}, \NN)$-vertex.

  Let us now highlight some of the technical ideas
  behind the simulation of $(\emptyset,\rho)$- and $(\sigma,\emptyset)$-vertices.
  The following notation is convenient for the description.
  Let $\CG$ be a gadget with a single so-called \emph{portal} vertex $v$.
  We write $\ext_\CG(\rho_i)$ for the number of partial solutions $S$ of $\CG$
  where every vertex except $v$ satisfies the constraints,
  $v\not\in S$, and $v$ has exactly $i$ selected neighbors in $S$.
  In this sense  $\ext_\CG(\rho_i)$ is the number of \emph{extensions} of the state $\rho_i$ of $v$ to the gadget $\CG$.
  The definition of $\ext_\CG(\sigma_i)$ is analogous with $v\in S$.
  The proof proceeds by constructing gadgets
  where the numbers of extensions satisfy certain properties.
  See \cref{fig:intro:steps} for an overview of these gadgets.
  Let us assume first that $\sigma$ and $\rho$ are finite.

  \begin{figure}
    \begin{subfigure}[t]{.25\textwidth}
      \centering
      \includegraphics{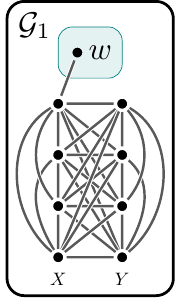}
      \caption{
        Gadget for Step~1, for $\sigMax=3$ and $\rhoMax=4$.
      }
      \label{fig:intro:steps:first}
    \end{subfigure}%
    \hfill
    \begin{subfigure}[t]{.28\textwidth}
      \centering
      \includegraphics{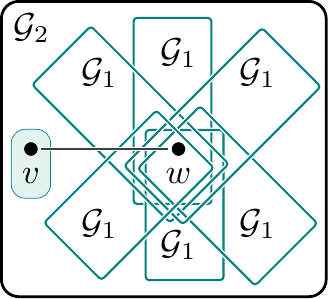}
      \caption{
        The gadget from Step~2 using the gadet from Step~1.
      }
      \label{fig:intro:steps:second}
    \end{subfigure}%
    \hfill
    \begin{subfigure}[t]{.28\textwidth}
      \centering
      \includegraphics{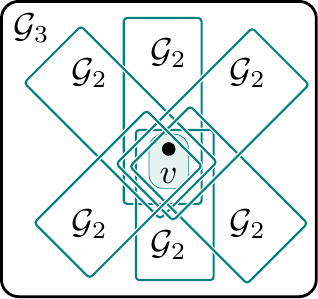}
      \caption{
        The gadget from Step~3 using the gadget from Step~2.
      }
      \label{fig:intro:steps:third}
    \end{subfigure}
    \caption{
      Illustration of the gadgets from Steps~1 to~3.
    }
    \label{fig:intro:steps}
  \end{figure}

  \paragraph{Step 1.}
  First, we construct $\CG_1$ with
  $\ext_{\CG_1}(\rho_0),\ext_{\CG_1}(\rho_1)\ge 1$,
  but $\ext_{\CG_1}(\sigma_i)=0$ for $i\ge 2$.
  The construction is easy: the graph consists of a portal vertex $w$,
  and two sets of vertices $X$, $Y$ with a single edge between $X$ and $v$.
  All we need to ensure is that both $X$ and $Y$ are solutions,
  which can be ensured by having degree $\sigMax$ inside $X$ and $Y$,
  and degree $\rhoMax$ between $X$ and $Y$.

  \paragraph{Step 2.}
  Next, we want to construct a graph $\CG_2$
  where $\ext_{\CG_2}(\rho_0)$ and $\ext_{\CG_2}(\sigma_0)$
  are both at least one, but different.
  In other words, the choice of selecting/not selecting the portal $v$
  influences the number of extensions in the gadget
  where the neighbors of the portal are not selected.
  Intuitively, this should be the case for most graphs,
  but to find a graph where this is provably so, we argue in the following way.
  Take $x$ copies of the previously defined gadget $\CG_1$,
  and identify their portal vertices into a single vertex $w$.
  The portal of $\CG_2$ is a new vertex $v$ adjacent to $w$.
  The number $\ext_{\CG_2}(\rho_0)$ corresponds to partial solutions $S$
  where $v,w\not\in S$.
  Then, each copy of $\CG_1$ has
  $\ext_{\CG_1}(\rho_1)$ extensions that add a new selected neighbor to $w$
  and $\ext_{\CG_1}(\rho_0)$ extensions that do not.
  These extra selected neighbors should satisfy the constraint given by $\rho$ on $w$.
  Therefore, we have
  \[
\ext_{\CG_2}(\rho_0)=\sum_{i\in \rho}\binom{x}{i}\ext_{\CG_1}(\rho_1)^i \ext_{\CG_1}(\rho_0)^{x-i}.
\]
The expression for $\ext_{\CG_2}(\sigma_0)$ is similar, but counts partial solutions containing $v$, giving an extra selected neighbor to $w$. Thus,
\[
\ext_{\CG_2}(\sigma_0)=\sum_{i+1\in \rho}\binom{x}{i}\ext_{\CG_1}(\rho_1)^i \ext_{\CG_1}(\rho_0)^{x-i}.
\]
Observe that dividing both expressions by $\ext_{\CG_1}(\rho_0)^x$
gives two polynomials in $x$, where the first one has degree $\rhoMax$,
and the second one has degree $\rhoMax-1$.
Thus, there has to be an $x$ for which the two polynomials are different.

\paragraph{Step 3.}
The next step is to implement a vertex $(\emptyset,\{0\})$,
that is, a vertex that cannot be selected and should not get any selected neighbors.
This can be used to force any number of other vertices to be unselected.
We construct a gadget $\CG_3$ by creating $x$ copies of $\CG_2$
and identifying  their portal vertices into a single portal vertex $v$.
Suppose that we extend a graph $G$ to a graph $G'$
by identifying the portal of $\CG_3$ with a vertex $v\in V(G)$.
Let $a_d$ be the number of partial solutions in $G$
where $v$ is unselected and has exactly $d$ selected neighbors;
$b_d$ is defined similarly for $v$ being selected.
Now the number of solutions in $G'$ is
\[
\underbrace{\sum_{d=0}^{\rhoMax} a_d\sum_{i+d\in \rho} \binom{x}{i}\ext_{\CG_2}(\rho_1)^i \ext_{\CG_2}(\rho_0)^{x-i}}_{A:=}+
\underbrace{\sum_{d=0}^{\sigMax} b_d\sum_{i+d\in \rho} \binom{x}{i}\ext_{\CG_2}(\sigma_1)^i \ext_{\CG_2}(\sigma_0)^{x-i}}_{B:=}
\]

The important observation here is that
$\ext_{\CG_2}(\rho_0)\neq \ext_{\CG_2}(\sigma_0)$ implies that
the two terms $A$ and $B$ above have very different orders of magnitude.
Assume that $\ext_{\CG_2}(\rho_0)<\ext_{\CG_2}(\sigma_0)$
(the other case is similar).
Then, $2^{|V(G)|}\cdot \ext_{\CG_2}(\rho_0)^{x}$ is a rough upper bound for $A$,
while $B$ is divisible by $\ext_{\CG_2}(\sigma_0)^{x-\rhoMax}$,
which is much larger for sufficiently large values of $x$.
Therefore, if we compute $A+B$ modulo $\ext_{\CG_2}(\sigma_0)^{x-\rhoMax}$,
then $B$ vanishes and we can recover $A$.

Dividing $A$ by $\ext_{\CG_2}(\rho_0)^{x}\neq 0$
gives a polynomial in $x$ of degree $\rhoMax$,
whose coefficients we can recover using interpolation.
(For simplicity, we assume here that $\ext_{\CG_2}(\rho_1)\neq 0$;
otherwise, we treat it as a separate case).
Observe that only $d=0$ contributes a term $x^{\rhoMax}$,
and thus, from the coefficient of $x^{\rhoMax}$ we can recover the value of $a_0$.
This counts the number of partial solutions
where $v$ is not selected and has no selected neighbors,
that is, it simulates an $(\emptyset,\{0\})$-vertex, as we wanted.

\paragraph{Step 4.}
Once we can express a single $(\emptyset,\{0\})$-vertex,
attaching it to any set $S$ of vertices forbids selecting any vertex of $S$
without adding any selected neighbors,
effectively making them $(\emptyset,\rho)$-vertices.
A further interpolation argument of similar flavor
allows us to simulate an arbitrary number of $(\sigma, \emptyset)$-vertices.
We omit the details here.

\paragraph{Step 5.}
There are a number of difficulties with this approach
if $\sigma$ and/or $\rho$ is cofinite.
For example, suppose that $\rho$ is cofinite and consider the term $A$ above.
As $\rho$ is not finite, $A$ cannot be written as the sum
of a constant number of binomial terms.
Instead, we can write $A$ as all possible extensions to $\CG_3$,
minus those extensions that give the wrong number of selected neighbors to $v$.
That is, for $k$ the number of neighbors of $v$ in $G$,
\[
  A \deff \sum_{d=0}^{k}
    a_d
    \left(
      (\ext_{\CG_2}(\rho_0)+\ext_{\CG_2}(\rho_1))^x
      - \sum_{i+d\not\in \rho}
        \binom{x}{i}
        \ext_{\CG_2}(\rho_1)^i
        \ext_{\CG_2}(\rho_0)^{x-i}
    \right).
\]
Observe that, if we divide the expression by $\ext_{\CG_2}(\rho_0)^{x}$, this is a combination of an exponential term and some polynomials. In the cofinite cases, we have to deal with functions of the form $f(x)=c^x-p(x)$, where $p(x)$ is a polynomial. A key technique to handle such expressions is to compute $f(x+1)-c\cdot f(x)=p(x+1)-c\cdot p(x)$, which is a polynomial of the same degree as $p(x)$, and hence, much easier to work with.

\paragraph{\boldmath The Third Case: $\rho=\NN$ with Cofinite $\sigma$.}
Now let us give some intuition why the case $\rho=\NN$ with cofinite $\sigma$ requires a different treatment.
Suppose that $\sigma=\ZZ_{\ge1}$, which turns out to be the key case.
A $(\ZZ_{\ge1},\NN)$-vertex $v$ can be selected or unselected, but if it is selected, then it requires at least one more selected neighbor. Suppose $v$ is subject to a relation $R$. In order to model the relation, we need to distinguish between solutions in which $v$
is selected, and those in which it is not (in the sense that there is a different number of partial solutions extending the two states).
Now note that, since $\rho=\NN$, unselected neighbors of $v$ behave the same independently
of the selection status of $v$. Also, in general, if $v$ has only unselected neighbors, then their number does not affect the selection status of $v$.
Therefore, the fact of whether $v$ is selected or not can be distinguished
only if $v$ has at least one selected neighbor.
However, as soon as the $(\ZZ_{\ge1},\NN)$-vertex $v$ has at least one selected neighbor, then
it has a feasible number of selected neighbors independently of its selection status and independently of how many additional selected neighbors it might obtain.
So, at this point, $v$ is entirely unaffected
by the selection status of its remaining neighbors
and acts essentially as a $(\NN,\NN)$-vertex.
This is why, while we are able to observe the selection status of $v$, and in this sense model the relation $R$ (by attaching some gadgets and interpolating), we did not manage to simultaneously recover the situation where $v$ behaves like a $(\ZZ_{\ge1},\NN)$-vertex with respect to the original graph.
To handle this situation,
we introduce another, different concept of ``$(\sigma,\rho)$-sets with relations''.
Now, all vertices that are subject to some relation
are treated as $(\NN,\NN)$-vertices,
that is, they do not impose any constraints on the solutions
other than that they have to satisfy the relations.
So, we have graphs with two types of vertices:
(1) $(\sigma,\rho)$-vertices that are not subject to any relation,
and (2) $(\NN,\NN)$-vertices that are each subject to at least one relation.
It turns out that the required lower bounds
also hold for this modified problem with relations.

\subsection{Organization of this Article}

We first introduce necessary preliminary definitions in \cref{sec:prelims}.
In \cref{sec:lowerbounds}, we give the proofs of
\cref{thm:lower-main-intro-decision,thm:lower-main-intro}. However, the proofs are
high-level in the sense that they use a number of intermediate results that we prove only
in later sections. This way, the purpose of the different proof ingredients is clear before we present their (in some cases long and technical) proofs.

The proofs of these two main theorems are divided into three steps.
The first step (\cref{sec:provider,sec:manager}) is about the construction of a special type of gadget (a \emph{manager}, formally defined in \cref{def:encoder})
which allows us to give a certain number of (selected) neighbors to vertices.
In this step, we show that, depending on the properties of $\sigma$ and $\rho$,
different versions of the manager gadget exist.
Depending on the ``quality'' of the manager gadget,
this then leads to lower bounds of different ``quality''.

Once we have said managers at hand,
the second step is a lower bound for a generalization of \srDomSet
where we additionally allow relations as constraints on the selection status of those vertices that are in the scope of some relation. The lower bound for this intermediate problem is established in \cref{sec:LBforRelations} for both the decision and the counting versions.

The third and last step is to show how to remove the relations
while preserving their properties in order to show
that the lower bounds for the intermediate problems
extend to the original problems \srDomSet and \srCountDomSet, respectively.
This is done in \cref{sec:realizingRelations,sec:realizingRelations:counting}.

An overview of the steps to obtain the lower bounds
is given in \cref{fig:count:highLevel}.

\section{Preliminaries}
\label{sec:prelims}
%\label{sec:notation}\label{sec:prelims}

\subsection{Basics}

\subsubsection*{Numbers, Sets, Strings, and Vectors}
We use $a \equiv_{\mname} b$ as shorthand for $a \equiv b \pmod \mname$.

We write $\NN = \{0,1,2,3,\dots\}$ to denote the set of non-negative integers and $\ZZ_{>0} = \{1,2,3,\dots\}$ to denote the positive integers.
For integers \(i, j\), we write \(\fragment{i}{j}\) for the set \(\{i,\dots,j\}\), and \(\fragmentco{i}{j}\) for the set \(\{i,\dots,j-1\}\).
The sets \(\fragmentoc{i}{j}\) and \(\fragmentoo{i}{j}\) are defined similarly.
A set $\tau \subseteq \NN$ is \emph{cofinite} if $\NN \setminus \tau$ is finite.
Also, we say that $\tau$ is \emph{simple cofinite} if $\tau = \{n,n+1,n+2,\dots\}$ for some $n \in \NN$.

We write \(s =  s\position{1} s\position{2}\cdots  s\position{n}\) for a \emph{string} of length \(| s| = n\) over an alphabet \(\Sigma\).
We write \(\Sigma^{n}\) for the set (or \emph{language}) of all strings of length \(n\), and
we use $\varepsilon$ to denote the empty string.
For a string \( s \in \Sigma^{n}\) and positions \(i \le j \in \fragment{1}{n}\),
we write \( s\fragment{i}{j} \coloneqq  s\position{i}\cdots  s\position{j}\);
accordingly, we define \( s\fragmentoc{i}{j}\), \( s\fragmentco{i}{j}\), and \( s\fragmentoo{i}{j}\).
Sometimes we are interested in the number of occurrences of an element $a \in \Sigma$ in a string $s$.
To this end, we use the notation $\occ{s}{a}$ for the number of occurrences of $a$ in $s$,
that is, $\occ{s}{a} \coloneqq \abs{\{i \in \abs{s} \mid s\position{i}=a\}}$.
Also, for $A \subseteq\Sigma$, $\occ{s}{A}$ denotes the number of occurrences of elements from $A$ in $s$,
that is, $\occ{s}{A} \coloneqq \sum_{a\in A} \occ{s}{a}$.

For a finite set $X$ (for instance, a set of vertices of a graph), we write \(\Sigma^X \coloneqq \Sigma^{|X|}\) to emphasize that we index the strings in (subsets of) \(\Sigma^X\) with elements from \(X\):
for an \(x_i \in X\), we write \( s\position{x_i} \coloneqq  s\position{i}\).

To improve readability, we sometimes use a ``ranging star'' \(\star\) to range over unnamed objects.
For example, if we wish to define a function $f\colon \NN \times \NN \rightarrow \NN$, then we write $f(\star,4) = 5$ to specify that $f(i,4) = 5$ for all $i \in \NN$.

\subsubsection*{Graphs}

We use standard notation for graphs.
A \emph{graph} is a pair $G = (V(G),E(G))$ with finite vertex set $V(G)$ and edge set $E(G) \subseteq \binom{V(G)}{2}$.
Unless stated otherwise, all graphs considered in this paper are simple (that is, there
are no loops or multiple edges) and undirected.
We use $uv$ as a shorthand for edges $\{u,v\} \in E(G)$.
We write \(N_G(v)\) for the \emph{(open) neighborhood} of a vertex $v \in V(G)$,
that is, $N_G(v) \coloneqq \{w \in V(G) \mid vw \in E(G)\}$.
The \emph{degree} of $v$ is the size of its (open) neighborhood, that is,
$\deg_G(v) \coloneqq |N_G(v)|$.
The \emph{closed neighborhood} is $N_G\position{v} \coloneqq N_G(v) \cup \{v\}$.
We usually omit the index $G$ if it is clear from the context.
For $X \subseteq V(G)$, we write $G\position{X}$ to denote the \emph{induced subgraph} on the vertex set $X$, and $G - X \coloneqq G\position{V(G) \setminus X}$ denotes the induced subgraph on the complement of $X$.

\subsubsection*{Treewidth and Pathwidth}

Next, we define tree decompositions and recall some of their basic properties.
For a more thorough introduction to tree decompositions and their many applications, we refer the reader to~\cite[Chapter 7]{CyganFKLMPPS15}.

Fix a graph $G$.
A \emph{tree decomposition} of $G$ is a pair $(T,\beta)$ that consists
of a rooted tree $T$ and a mapping $\beta\colon V(T) \to 2^{V(G)}$ such that
\begin{enumerate}[label = (T.\arabic*)]
 \item $\bigcup_{t \in V(T)} \beta(t) = V(G)$,
 \item for every edge $vw \in E(G)$, there is some node $t \in V(T)$ such that $\{u,v\} \subseteq \beta(t)$, and
 \item for every $v \in V(G)$, the set $\{t \in V(T) \mid v \in \beta(t)\}$ induces a connected subtree of $T$.
\end{enumerate}
If $T$ is a path, then $(T,\beta)$ is a \emph{path decomposition} of $G$.
The \emph{width} of a tree decomposition (path decomposition, respectively) $(T,\beta)$ is defined as $\max_{t \in V(T)}|\beta(t)|-1$.
The \emph{treewidth} of a graph $G$, denoted by $\tw(G)$, is the minimum width of a tree decomposition of $G$.
The \emph{pathwidth} of a graph $G$, denoted by $\pw(G)$, is the minimum width of a path decomposition of $G$.

\subsubsection*{Treewidth-Preserving and Pathwidth-Preserving Reductions}

To prove hardness results for computational problems on graphs of bounded treewidth
(or variations thereof where treewidth is defined)
we rely on reduction chains that, up to additive constants,
preserve the treewidth of the input graphs.
As we state the lower bounds for pathwidth,
we formally only introduce reductions preserving pathwidth.
By definition of pathwidth and treewidth,
a pathwidth-reduction is automatically treewidth-preserving.

\begin{definition}[pathwidth-preserving reduction]
    \label{def:pwred}
    Let $A$ and $B$ denote two computational problems
    for which (some notion of) pathwidth is defined.

    We write $A\pwred B$
    if there is a \emph{pathwidth-preserving reduction} from $A$ to $B$,
    that is, a polynomial-time Turing reduction from $A$ to $B$ that,
    if executed on an instance $I$ of $A$
    given with a path decomposition of width at most $t$,
    makes oracle calls to $B$ only on instances
    that have size polynomial in $|I|$
    and that are given with a path decomposition of width at most $t + \Oh(1)$.
\end{definition}

Similarly to a treewidth preserving reduction
we define an arity-preserving reduction.
\begin{definition}[arity-preserving reduction]\label{def:arityRed}
  Let $A$ and $B$ denote two computational problems
  for which (some notion of) arity is defined.

  We write $A\arred B$
  if there is an \emph{arity-preserving reduction} from $A$ to $B$,
  that is, a polynomial-time Turing-reduction from $A$ to $B$ that,
  if executed on an instance $I$ of $A$ with arity at most $\ar(I)$,
  makes oracle calls to $B$ only on instances
  that have size polynomial in $\abs{I}$
  and that have arity at most $\ar(I) + \Oh(1)$.

  We say that $A$ is the \emph{source} and $B$ is the \emph{target} of the reduction.
\end{definition}

We combine a pathwidth-preserving and an arity-preserving reduction
into a \pwar-reduction.
\begin{definition}[\pwar-reduction]\label{def:paRed}
  Let $A$ and $B$ denote two computational problems
  for which (some notion of) pathwidth
  and (some notion of) arity is defined.

  We write $A\pwarred B$
  if there is a \emph{\pwar-reduction} from $A$ to $B$,
  that is, a reduction from $A$ to $B$
  which is pathwidth-preserving and arity-preserving.
\end{definition}

\begin{observation}\label{obs:pwredtransitivity}\label{obs:paRedTransitivity}
	A sequence of pathwidth-preserving reductions $A_1\pwred A_2 \pwred \ldots
	\pwred A_k$ gives a pathwidth-preserving reduction from $A_1$ to $A_k$
  if $k\in \Oh(1)$
	(with respect to the input size of $A_1$).
	Similarly, a sequence of \pwar-reductions
	$A_1 \pwarred A_2 \pwarred \dots \pwarred A_k$
	gives a \pwar-reduction from $A_1$ to $A_k$ if $k\in \Oh(1)$.
\end{observation}

\subsection{Generalized Dominating Sets}

In the following, let $\sigma,\rho \subseteq \NN$ denote two sets that are finite or cofinite.

\subsubsection*{Basics}
Fix a graph $G$.
A set of vertices $S \subseteq V(G)$ is a \emph{$(\sigma,\rho)$-set} if $|N(u) \cap S| \in \sigma$ for every $u \in S$, and $|N(v) \cap S| \in \rho$ for every $v \in V(G) \setminus S$. We also refer to these two requirements as the $\sigma$-constraint and the $\rho$-constraint, respectively.
The (decision version of the) \srDomSet problem takes as input a graph $G$, and asks whether $G$ has a $(\sigma,\rho)$-set $S \subseteq V(G)$.
We use \srCountDomSet to refer to the counting version, that is, the input to the problem is a graph $G$, and the task is to determine the number of $(\sigma,\rho)$-sets $S \subseteq V(G)$.

We say $(\sigma,\rho)$ is \emph{trivial} if $\rho = \{0\}$ or $(\sigma,\rho) = (\NN,\NN)$.

\begin{fact}
 \label{fact:trivial-pairs-counting}
 Suppose $(\sigma,\rho)$ is trivial.
 Then, \srCountDomSet can be solved in polynomial time.
\end{fact}

\begin{proof}
 For $(\sigma,\rho) = (\NN,\NN)$, the number of $(\sigma,\rho)$-sets is $2^{|V(G)|}$.
 For $\rho = \{0\}$, the number of $(\sigma,\rho)$-sets is $2^{c}$, where $c$ denotes the number of those connected components of $G$ where every vertex degree is contained in $\sigma$.
\end{proof}

In order to analyze the complexity of \srDomSet (and \srCountDomSet) for non-trivial pairs $(\sigma,\rho)$, we associate the following parameters with $(\sigma,\rho)$.
We define
\begin{equation}\label{eq:rho-sig-max}
\sigMax \deff
\begin{cases}
	\max(\sigma) & \text{if $\sigma$ is finite,} \\
	\max(\ZZ \setminus \sigma) + 1 & \text{if $\sigma$ is cofinite,}
\end{cases}
\text{and }
\rhoMax \deff
\begin{cases}
	\max(\rho) & \text{if $\rho$ is finite,} \\
	\max(\ZZ \setminus \rho) + 1 & \text{if $\rho$ is cofinite.}
\end{cases}
\end{equation}
Moreover, we set $\allMax \deff \max\{\sigMax,\rhoMax\}$.
We also define $\sigMin \deff \min \sigma$ and $\rhoMin \deff \min \rho$.

Our improved algorithms heavily exploit certain structures of a pair $(\sigma,\rho)$;
formally we are interested in whether a pair is what we call ``\(\mname\)-structured''.

\begin{definition}[$\mname$-structured sets]
    Fix an integer $\mname \geq 1$.
    A set $\tau \subseteq \ZZ_{\ge 0}$ is \emph{$\mname$-structured} if there is some number $c^*$ such that
    \[c \equiv_{\mname} c^* \]
    for all $c \in \tau$.
\end{definition}
We say that $(\sigma,\rho)$ is \emph{$\mname$-structured} if both $\sigma$ and $\rho$ are $\mname$-structured.
Observe that $(\sigma,\rho)$ is always $1$-structured.

\subsubsection*{Partial Solutions and States}

For our hardness results, a key ingredient is the description of \emph{partial solutions}.

A \emph{graph with portals} is a pair $(G,U)$, where $G$ is a graph and $U \subseteq V(G)$.
If $U = \{u_1,\dots,u_k\}$, then we also write $(G,u_1,\dots,u_k)$ instead of $(G,U)$.

Intuitively speaking, the idea of this notion is that $G$ may be part of some larger graph that interacts with $G$ only via vertices from $U$.
In particular, in the context of the \srDomSet problem, vertices in $U$ do not necessarily need to satisfy the definition of a $(\sigma,\rho)$-set since they may receive further selected neighbors from outside of $G$.

\begin{definition}[partial solution]\label{def-partial-sol}\label{def:partialsol}
    Fix a graph with portals $(G,U)$.
    A set \(S \subseteq V(G)\) is a \emph{partial solution} (with respect to \(U\)) if
    \begin{enumerate}[label = (PS\arabic*)]
        \item \label{item:realizable-1s}\label{item:realizable-1}
            for each $v \in S \setminus U$, we have $|N(v) \cap S| \in \sigma$,
            and
        \item \label{item:realizable-3s}\label{item:realizable-3}
            for each $v \in V(G) \setminus (S\cup U)$, we have $|N(v) \cap S| \in \rho$.
    \end{enumerate}
\end{definition}

To describe whether vertices from $U$ are selected into partial solutions and how many selected neighbors they already have inside $G$, we associate a state with every vertex from $U$.

Formally, we write $\sigStatesExt \coloneqq \{\sigma_i \mid i \in \NN\}$ for the set of
potential $\sigma$-states, and we write $\rhoStatesExt \coloneqq \{\rho_i \mid i \in \NN\}$
for the set of potential $\rho$-states.
We also write $\allStatesExt \coloneqq \sigStatesExt \cup \rhoStatesExt$ for the set of \emph{all} potential states.

\begin{figure}[t]
    \centering
    \includegraphics[scale=1.3]{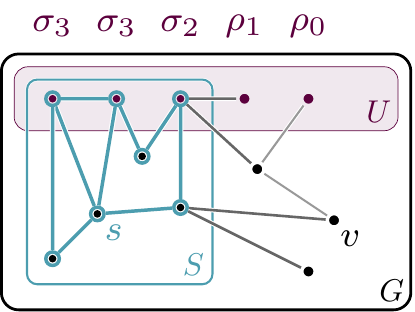}
    \caption{A graph \(G\) and subsets of vertices \(U\) and \(S\).
        For \(\sigma = \{2, 4\}, \rho = \{ 1 \}\), the set \(S\) is a partial solution
        (with respect to \(U\)), as every blue vertex \(s \in S \setminus U\) satisfies
        \(|N(s) \cap S| \in \{2,4\} = \sigma\) and every black vertex \(v \in V(G)
        \setminus (S \cup U)\) satisfies \(|N(v) \cap S| \in \{1\} = \rho\).
        The depicted set \(S\) corresponds to the compatible string
        \(\sigma_3\sigma_3\sigma_2\rho_1\rho_0\) (written above \(G\)).
        Note that \(S\) would not be a partial solution for \(\sigma = \{4\}\), as every
        blue vertex but one has only \(2\) neighbors in~\(S\).
    }\label{fig:partial}
\end{figure}

\begin{definition}[compatible strings]\label{def-compatible-general}\label{def:compatibleString}
    Fix a graph with portals $(G,U)$.
    A string \( x \in \allStatesExt^{U}\) is \emph{compatible with $(G,U)$}
    if there is a partial solution $S_{x} \subseteq V(G)$ such that
    \begin{enumerate}[label = (X\arabic*)]
        \item\label{item:realizable-2s}
            for each $v \in U \cap S_{x}$, we have
            \( x\position{v} = \sigma_{s}\), where $s = |N(v) \cap S_{x}|$, and
        \item\label{item:realizable-4s}
            for each $v \in U \setminus S_{x}$, we have
            \( x\position{v} = \rho_{r}\), where $r = |N(v) \cap S_{x}|$.
    \end{enumerate}
    We also refer to the vertices in $S_{x}$ as being \emph{selected} and say that $S_{x}$ is a \emph{(partial) solution}, \emph{selection}, or \emph{witness} that \emph{witnesses} $x$.
\end{definition}

Consult \cref{fig:partial} for a visualization of an example of a partial solution and its corresponding compatible string.

Observe that, despite $\allStatesExt$ being an infinite alphabet, for every graph with portals $(G,U)$, only finitely many strings $x$ can be realized.
Indeed, if $|V(G)| = n$, then every compatible string can have only characters from $\allStates_n = \sigStates_n \cup \rhoStates_n$, where $\sigStates_n \coloneqq \{\sigma_i \mid i \in \fragment{0}{n}\}$ and $\rhoStates_n \coloneqq \{\rho_i \mid i \in \fragment{0}{n}\}$.

\begin{definition}[realized language, $L$-provider]\label{def:provider}
    For a graph with portals \((G, U)\), we define its realized language as
    \[L(G,U) \coloneqq \{x \in \allStatesExt^U \mid x \text{ is compatible with } (G,U)\}.\]

    For a language \(L \subseteq \allStatesExt^U\), we say that \emph{$(G,U)$}
    is an \(L\)-realizer if \(L = L(G,U)\).

    For a language \(L \subseteq \allStatesExt^U\), we say that \emph{$(G,U)$}
    is an \(L\)-provider if \(L \subseteq L(G,U)\).
\end{definition}

Again, observe that $L(G,U) \subseteq \allStates_n^U$, where $n = |V(G)|$.

In fact, for most of our applications, it makes sense to restrict the alphabet even further.
Recall the definition of $\sigMax$ and $\rhoMax$ from Equation~\eqref{eq:rho-sig-max}.
Suppose that $\sigma$ is finite.
Then, we are usually not interested in partial solutions $S$ where some vertex from $U$ is selected and already has more than $\sigMax$ selected neighbors (as it is impossible to extend this partial solution into a full solution).
Also, if $\sigma$ is infinite, it is usually irrelevant to us whether a selected vertex has exactly $\sigMax$ selected neighbors, or more than $\sigMax$ selected neighbors, since both options lead to the same outcome for all possible extensions of a partial solution.
For this reason, we typically%
\footnote{Sometimes, it turns out to be more convenient to work with the more general
variants; we clearly mark said (rare) occurrences of \(\allStatesExt\) and \(\allStates_n\).}
 restrict ourselves to the alphabets
\[\sigStates \coloneqq \{\sigma_0,\dots,\sigma_{\sigMax}\} \quad\text{and}\quad \rhoStates \coloneqq \{\rho_0,\dots,\rho_{\rhoMax}\}.\]
As before, we define $\allStates \coloneqq \sigStates \cup \rhoStates$.
\begin{definition}[inverse of a state]
    \label{def:inverseOfStates}
    For a state $\sigma_s \in \sigStates$,
    the \emph{inverse of $\sigma_s$ with respect to $\sigma$}
    is the state $\inverse[\sigma]{\sigma_s}=\sigma_{\sigMax-s}$.
    For a state $\rho_r \in \rhoStates$,
    the \emph{inverse of $\rho_r$ with respect to $\rho$}
    is the state $\inverse[\rho]{\rho_r}=\rho_{\rhoMax-r}$.

    We set $\inverse[\sigma]{\rho_r}=\rho_r$
    and $\inverse[\rho]{\sigma_s}=\sigma_s$.
    For all states $a \in \allStates$,
    the inverse of $a$ with respect to $\sigma,\rho$
    is the state $\inverse[\sigma,\rho]{a}=\inverse[\sigma]{\inverse[\rho]{a}}$.

    We extend this in the natural way to strings by applying it coordinate-wise,
    and to sets by applying it to each element of the set.
\end{definition}

\section{High-level Constructions and Proofs of Main Theorems}
\label{sec:lowerbounds}
Before starting with the high-level proofs, we first define some necessary notation. This is done in \cref{sec:graphsWithRelations}. Afterward, we show \cref{thm:lower-main-intro-decision} in \cref{sec:LB-decision}, and we show \cref{thm:lower-main-intro} in \cref{sec:LB-counting}.

%\begin{landscape}
\begin{figure}[t]
\centering
\begin{tikzpicture}[
    scale=0.98,
    transform shape
]
  \input{drawings/overview-tikz-style}
  \def\x{3cm}
  \def\y{-2 cm}

  \node[gadget] (provider) at (0,0) {Providers};

  \node[gadget] (manager) at (0, \y) {Managers};
  \draw[const] (provider) -- (manager);

  \node[problem,minimum width=2cm] (sat) at (-1*\x, 2*\y) {SAT};
  \node[problem,minimum width=2cm] (countSat) at (\x, 2*\y) {\#SAT};

  \node[problem, lowerB, break] (domSetRelFull) at (\x, 3*\y)
    {\CountDomSetRel{\widehat\sigma}{\widehat\rho} \\
    \footnotesize
    (\cref{lem:count:intermediateLowerBoundWhenHavingSuitableGadget})};

  \node[problem, lowerB, break] (domSetRel) at (\x, 4*\y)
        {\CountDomSetRel{\sigma}{\rho} \\
        \footnotesize
        (\cref{lem:count:lowerBoundWhenHavingSuitableGadget})};
  \draw[red] (domSetRelFull) -- (domSetRel) node[lem,right]
    {\cref{sec:high-level:counting}};

  \node[problem, lowerB, break, very thick] (domSet) at (\x, 5*\y)
        {\CountDomSet{\sigma}{\rho} \\
        \footnotesize
        (\cref{thm:lower-main-intro})};
  \draw[red] (domSetRel) -- (domSet) node[lem, right]
    {\cref{lem:lower:count:removingRelations}};

  \draw[red] (countSat) -- (domSetRelFull);
  \node[anchor=west] (countObs) at (\x, 2.5*\y)
      {\cref{obs:lower:parsimoniousAndFull}};
  \draw[const] (manager) |- (countObs);

  %%%%
  %%%% DECISION VERSION
  \node[problem, lowerB, break] (decDomSetRel) at (-1*\x, 4*\y)
    {\DomSetRel{\widehat\sigma}{\widehat\rho} \\
    \footnotesize
    (\cref{lem:lowerBoundWhenHavingSuitableGadget})};

  \node[problem, lowerB, break, very thick] (decDomSet) at (-1*\x, 5*\y)
        {\DomSet{\widehat\sigma}{\widehat\rho} \\
        \footnotesize
        finite $\widehat\sigma, \widehat\rho$
        (\cref{thm:lower-main-intro-decision})};
  \draw[red] (decDomSetRel) -- (decDomSet) node[lem, right]
    {\cref{lem:lower:dec:removingRelations}};

  \draw[red] (sat) -- (decDomSetRel);
  \node[anchor=east] (decProof) at (-1*\x, 2.5*\y)
    {\cref{sec:high-level:decision}};
  \draw[const] (manager |- countObs) -- (decProof);

  %%%%
  %%%% ANNOTATIONS FOR ORGANIZATION
  \node (zeroTop) at (-2*\x-2cm, .5*\y) {};
  \node (zeroBot) at ( 2*\x+2cm, .5*\y) {};
  \node (frstTop) at (-2*\x-2cm,1.5*\y) {};
  \node (frstBot) at ( 2*\x+2cm,1.5*\y) {};
  \node (scndTop) at (-2*\x-2cm,  4*\y) {};
  \node (scndMid) at (      0  ,  4*\y) {};
  \node (scndBot) at ( 2*\x+2cm,  4*\y) {};
  \node (thrdMid) at (      0  ,5.5*\y) {};

  \begin{scope}[on background layer]
    \draw[thick,dotted] (zeroTop) -- (zeroBot);
    \draw[thick,dotted] (frstTop) -- (frstBot);
    \draw[thick,dotted] (scndTop) -- (scndBot);
    \draw[thick,dotted] ([yshift=-.75mm] scndMid.center) -- (thrdMid);
    \node [anchor=south west] at (zeroTop) {\cref{sec:provider}};
    \node [anchor=south west] at (frstTop) {\cref{sec:manager}};
    \node [anchor=south west] at (scndTop) {\cref{sec:high-level}};
    \node [anchor=north west] at (scndTop) {\cref{sec:realizingRelations}};
    \node [anchor=north east] at (scndBot) {\cref{sec:realizingRelations:counting}};
  \end{scope}

\end{tikzpicture}

\caption{
Outline of the proof of the lower bounds.
The left sequence considers the decision version,
and the right sequence considers the counting version.
\\
For the problems in the filled boxes,
the statement refers to the formal lower bound.
% Multiple parallel edges indicate Turing-reductions.
\\
The sets $\sigma, \rho$ are finite or simple cofinite sets.
The sets $\widehat \sigma$ and $\widehat \rho$ are finite or cofinite sets
(not necessarily simple cofinite) unless mentioned otherwise.
}
\label{fig:count:highLevel}
\end{figure}
%\end{landscape}

\subsection{Graphs with Relations and Managers}\label{sec:graphsWithRelations}

For our proofs, it is useful to consider graphs
with additional constraints on the selection of the vertices.
We formalize this by making use of relations which we add to the graph.

\begin{definition}[Graph with Relations] \label{def:lower:graphWithRelations}
	We define a \emph{graph with relations} as a tuple $G = (V,E,\CC)$,
	where $V$ is a set of vertices, $E$ is a set of edges on $V$,
	and $\CC$ is a set of relational constraints,
	that is, each $C\in \CC$ is in itself a tuple $(\scope(C), \rel(C))$.
	Here the \emph{scope} $\scope(C)$ of $C$ is a tuple
	of $\abs{\scope(C)}$ vertices from $V$.
	Then $\rel(C)\subseteq 2^{\scope(C)}$ is a $\abs{\scope(C)}$-ary relation
	specifying possible selections within $\scope(C)$.
	We also say that $\rel(C)$ \emph{observes} $\scope(C)$.

	The size of $G$ is $\abs{G}\coloneqq \abs{V}+\sum_{C\in \CC} \abs{\rel(C)}$.
	Slightly abusing notation, we usually do not distinguish between $G$
	and its underlying graph $(V,E)$. We use $G$ to refer to both objects depending on the context.
\end{definition}

Some relations that we will use frequently are equality and Hamming weight relations.
\begin{definition}
	\label{def:hammingWeightOneAndEquality}
	Let $S$ be a set of vertices with $d=\abs{S}$.
	Then $\EQ{d}=\{\emptyset, S\}$ is the $d$-ary \emph{equality} relation (over $S$).
	For each $\ell\ge 0$, $\HWeq[d]{\ell} \coloneqq \{ Y \subseteq S \mid \abs{Y}=\ell \}$ is the $d$-ary \emph{Hamming weight $\ell$} relation (over $S$).
	Analogously, we define $\HWge[d]{\ell}$ and $\HWle[d]{\ell}$.
\end{definition}

Given \cref{def:lower:graphWithRelations},
we extend the notion of $(\sigma,\rho)$-sets to graphs with relations.

\begin{definition}[$(\sigma,\rho)$-Sets of a Graph with Relations]
	Given a graph with relations $G = (V,E,\CC)$, a \emph{$(\sigma,\rho)$-set} $X$ of $G$ is a $(\sigma,\rho)$-set of the underlying graph $(V,E)$ that
	additionally satisfies $X \cap \scope(C) \in \rel(C)$ for every $C\in \CC$.
\end{definition}

To formally state treewidth-preserving (pathwidth-preserving, respectively) reductions,
we define the treewidth (pathwidth, respectively) of a graph with relations.
The crucial part is that in this definition we treat the scope of a relation as a clique.

\begin{definition}[Width Measures for Graphs with Relations]
	Let $G=(V,E,\CC)$ be a graph with relations.
	Let $G'$ be the graph we obtain from $(V,E)$
	when, for all $C \in \CC$,
	we additionally introduce a complete set of edges on the scope $\scope(C)$.
	The \emph{treewidth of a graph with relations} $G$,
	is the treewidth of the graph $G'$.
	Analogously, we define tree decompositions, path decompositions, and pathwidth of $G$
	as the corresponding concepts in the graph $G'$.
\end{definition}

To make use of the concept of arity-preserving reductions
(see \cref{def:arityRed}),
we define the arity of a graph with relations.
We ensure that this definition agrees with the intuition
that the arity of a graph without relations is 0.

\begin{definition}[Arity of a Graph with Relations]
	Given a graph with relations $G=(V,E,\CC)$,
	the \emph{arity of a graph with relations}
	is $\ar(G)\deff \max (\{0\}\cup \{\abs{\scope(C)} \mid C\in \CC\})$,
	i.e., the maximum arity of some relation of $\CC$ (or 0).
\end{definition}

If, for some constant $d$, a graph with relations has arity at most $d$,
then its size is upper bounded by $\abs{V}+\abs{\CC}\cdot 2^d$.
Hence, for the sake of polynomial time reductions,
this size is essentially $\abs{V}+\abs{\CC}$.
By defining pathwidth and arity of a graph with relations,
we can now use computational problems that take as input a graph with relations
in the context of pathwidth-preserving reductions,
as well as in the context of arity-preserving reductions,
or even \pwar-reductions (\cref{def:paRed}).

Later (in \cref{sec:LBforRelations,sec:realizingRelations}, we additionally need the concept of graphs with relations and portals to generalize the definition of $L$-providers and $L$-realizers from \cref{def:provider} to graphs with relations.

\begin{definition}[Provider/Realizer with Relations]
	\label{def:providersWithRelations}
	Let $G = (V,E,\CC)$ denote a graph with relations and let $U\subseteq V$. Then, $(G,U)$ is a graph with relations and portals.

	A set \(X \subseteq V\) is a \emph{partial solution} of \((G,U)\) if
	\begin{itemize}
		\item
		\(X\) is a partial solution of \((V,E,U)\),
		that is, \(X\) has the properties from \cref{def:partialsol}, and
		\item \(X\) satisfies that $X \cap \scope(C) \in \rel(C)$ for every $C\in \CC$.
	\end{itemize}

	The notions of \emph{compatible strings}, \emph{$L$-providers, and $L$-realizers}
	from \cref{def:compatibleString,def:provider} generalize accordingly.
\end{definition}

As a first example of a graph with relations, we define a special gadget, called a \emph{manager},
which allows us to add selected neighbors to possibly selected vertices.
In fact,
a manager additionally ensures that the vertices
to which it is connected get a valid number of neighbors
such that the $\sigma$-constraints and $\rho$-constraints are always satisfied.
An illustration of such a manager is given in \cref{fig:manager:definition}.

\begin{figure}[t]
  \centering
  \includegraphics[page=3]{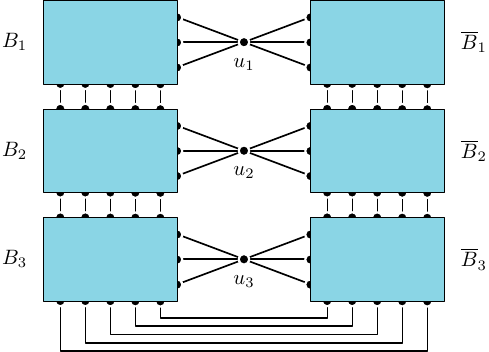}
  \caption{
  An illustration of an \encoder{A} from \cref{def:encoder}
  with $A \subseteq \{\sigma_i, \rho_i \mid i \in \fragment{0}{3} \}$
  for three distinguished vertices ($\ell=3$).
  }
  \label{fig:manager:definition}
\end{figure}

\begin{restatable}[\encoder{A}]{definition}{defEncoder}
  \label{def:encoder}
  For a set $A \subseteq \allStates$, an \emph{\encoder{A}}
  is an infinite family $((G_\ell, \Port_\ell))_{\ell \geq 1}$ of pairs $(G_\ell, \Port_\ell)$ such that
  \begin{itemize}
    \item $G_\ell$ is a graph with relations and
    \item $\Port_\ell=\{\port_1,\dots,\port_{\ell}\} \subseteq V(G_\ell)$
  is a set of $\ell$ distinguished vertices.
  \end{itemize}
  Moreover, there is a non-negative integer $b$ (that depends only on $\sigMax,\rhoMax$)
  such that the following holds for every $\ell \geq 1$:
  \begin{itemize}
    \item
    The vertices from $V(G_\ell) \setminus \Port_\ell$
    can be partitioned into $2\ell$ vertex-disjoint subgraphs
    $\Bl_1,\dots,\Bl_\ell$ and $\Br_1,\dots,\Br_\ell$
    (called \emph{blocks}),
    such that
    \begin{itemize}
      \item $|\Bl_i| \leq b$ and $|\Br_i| \leq b$ for all $i \in \numb{\ell}$,
      \item $N(\port_i) \subseteq \Bl_i \cup \Br_i$ for all $i \in \numb{\ell}$,
      \item there are edges only between the following pairs of blocks:
      $\Bl_i$ and $\Bl_{i+1}$,
      $\Br_i$ and $\Br_{i+1}$, for each $i \in \numb{\ell-1}$,
      and
      $\Bl_\ell$ and $\Br_\ell$.
    \end{itemize}
    \item
    Each $x \in A^\ell \subseteq \allStates^{\ell}$ is \emph{managed} in the sense that
    there is a unique $(\sigma,\rho)$-set $S_x$ of $G_\ell$\footnote{
    Note that all vertices in $G_\ell$
    already have a feasible number of neighbors.
    The graph does not have any potentially ``unsatisfied'' portals.
    In this sense, it is different from the graphs with portals
    that we use in other gadget constructions.}
    such that for all $i\in\numb{\ell}$:
    \begin{itemize}
      \item
      If $x\position{i} = \sigma_s$,
      then $\port_{i} \in S_x$.
      Moreover, $\port_i$ has exactly $s$ neighbors in $\Bl_i\cap S_x$
      and exactly $\sigMax-s$ neighbors in $\Br_i\cap S_x$.
      \item
      If $x\position{i} = \rho_r$,
      then $\port_{i} \notin S_x$.
      Moreover, $\port_i$ has exactly $r$ neighbors in $\Bl_i\cap S_x$
      and exactly $\rhoMax-r$ neighbors in $\Br_i\cap S_x$.
    \end{itemize}
  \end{itemize}
  We refer to $G_\ell$ as the \encoder{A} of \rank $\ell$.
\end{restatable}

Next, we are interested under which conditions manager gadgets exist. The following result gives an
answer to this question; due to its length, we defer the proof to \cref{sec:manager}. Recall the definition of the inverse of a state from \Cref{def:inverseOfStates}.
\existenceOfManager*
Observe that the bounds for the managers
precisely coincide with the bounds for the languages
provided in the accompanying paper~\cite{FockeMMNSSW23i}.
Hence, these managers are the key ingredients to obtain matching lower bounds.

Again, we defer the proof of \cref{lem:lb:existenceOfManager} to \Cref{sec:manager}.
In particular, we prove each case separately; the proof of \cref{lem:lb:existenceOfManager} then follows from
\cref{lem:expanding:onlyRhoStates,lem:expanding:onlyEvenStates,%
lem:expanding:onlySigmaStates,lem:expanding:allStates}.

\subsection{Decision Problem}\label{sec:LB-decision}
As it turns out, we can use the same construction for manager gadgets
for the decision problem and for the counting problem.
However, we need separate constructions for establishing a lower bound for the problem with relations (see \cref{def:lower:domSetRel}),
and we also need separate proof strategies for the reduction from the problem with relations to the original problem without
relations.

\begin{definition}[\srDomSetRel]
  \label{def:lower:domSetRel}
  The problem \srDomSetRel
  takes as input
  a graph with relations $G=(V, E, \CC)$,
  and the task is to decide
  whether there is a $(\sigma,\rho)$-set of $G$.
  We use $\abs{G}$ as the size of the input.

  The counting problem \srCountDomSetRel is defined analogously.
\end{definition}

Our lower bound depends on the ``quality'' of the manager we can achieve.
Thus, we treat the manager as an input, and prove a lower bound based on it.

\lowerBoundForDecDomSetRel*
We defer the proof of \cref{lem:lowerBoundWhenHavingSuitableGadget} to \cref{sec:high-level}.
Lastly, for our lower bound for \srDomSet, we show how to remove the relations from
\srDomSetRel.
\removingRelationsForDecVersion*
The proof of \cref{lem:lower:dec:removingRelations}
is based on the results from \cref{sec:RelToHW1}
and formally given in \cref{sec:realizingRelations:decision}.

Combining the previous results, we obtain \cref{thm:lower-main-intro-decision}, which we
restate for convenience.
\thmLBmaindecision*
\begin{proof}%[Proof of \cref{thm:lower-main-intro-decision}]
  By \cref{def:intro:baseOfRunningTime}, we have to consider three different cases.
  Observe that all the managers from \cref{lem:lb:existenceOfManager}
  satisfy the constraints in \cref{lem:lowerBoundWhenHavingSuitableGadget}.

  \begin{itemize}
    \item
    $(\sigma,\rho)$ is not $\mname$-structured for any $\mname\ge 2$.

    Assume there is an $\eps>0$ such that \srDomSet can be solved in time
    $(c_{\sigma,\rho}-\eps)^{k} \cdot n^{\Oh(1)}$
    on graphs of size $n$ given with a path decomposition of width $k$
    where $c_{\sigma,\rho}=\sigMax+\rhoMax+2$.
    Moreover, assume that the SETH holds.

    By \cref{lem:lb:existenceOfManager},
    Case~\ref{lem:lb:existenceOfManager:all},
    we know that there is an \encoder{A} with $\abs{A}=\sigMax+\rhoMax+2$.
    Observe that $c_{\sigma,\rho}=\abs{A}$.

    Hence, by \cref{lem:lowerBoundWhenHavingSuitableGadget},
    there is a constant $d$
    such that there is no algorithm solving \srDomSetRel
    in time $(c_{\sigma,\rho}-\eps)^{k+\Oh(1)} \cdot n^{\Oh(1)}$
    on instances of size $n$ and arity at most $d$
    given with a tree decomposition of width $k$.

    Toward a contradiction,
    let $G$ denote an instance of \srDomSetRel
    where the arity is at most the constant $d$.
    By the pathwidth-preserving reduction from \cref{lem:lower:dec:removingRelations},
    this instance $G$ can be transformed
    into polynomially (in the size of $G$) many instances $H_i$ of \srDomSet
    such that, for each $i$, $\pw(H_i) \le \pw(G) + \Oh(1)$
    and the size of $H_i$ depends only polynomially on the size of $G$.

    For each $i$, we now apply the faster algorithm for \srDomSet on input $H_i$
    and afterward recover the solution for $G$ in time
    \[
      \sum_i (c_{\sigma,\rho} - \epsilon)^{\pw(H_i)} \cdot \abs{H_i}^{\Oh(1)}
      + \abs{G}^{\Oh(1)}
      \le
      (\sigMax + \rhoMax + 2 - \epsilon)^{\pw(G) + \Oh(1)} \cdot \abs{G}^{\Oh(1)}
      .
    \]
    This immediately contradicts our assumption that the SETH holds
    by \cref{lem:lowerBoundWhenHavingSuitableGadget},
    and hence, the result follows.

    \item
    $(\sigma,\rho)$ is $2$-structured,
    but not $\mname$-structured for any $\mname\ge 3$,
    and $\sigMax=\rhoMax$ is even.

    In this case, we use the same arguments as before
    with the difference being that for the algorithm we have
    $c_{\sigma,\rho} = \max\{\sigMax, \rhoMax\}+2$.
    As a consequence, we use the \encoder{A}
    from \cref{lem:lb:existenceOfManager},
    Case~\ref{lem:lb:existenceOfManager:even}
    with $\abs{A}=(\sigMax+\rhoMax)/2+2=\max\{\sigMax,\rhoMax\}+2$.

    Observe that we can use this encoder
    as $\rho$ being $2$-structured implies that $\rho$ is finite
    but $\rho \neq \{0\}$, by assumption.
    Hence, we directly get $\rhoMax \ge 1$.

    \item
    $(\sigma,\rho)$ is $\mname$-structured
    for an $\mname\ge 3$,
    or $2$-structured with $\sigMax\neq \rhoMax$,
    or $2$-structured with $\sigMax=\rhoMax$ being odd.

    Once more we use the same approach as for the first case,
    but now we have an algorithm with
    $c_{\sigma,\rho}=\max\{\sigMax,\rhoMax\}+1$.

    If $\sigMax \ge \rhoMax$,
    then we want to use the \encoder{A} from \cref{lem:lb:existenceOfManager},
    Case~\ref{lem:lb:existenceOfManager:sigma} with $\abs{A}=\sigMax+1$.
    For this case to be applicable we need $\rhoMax \ge 1$.
    Whenever $\rhoMax = 0$, we either have $\rho=\{0\}$ or $\rho=\NN$.
    The first case is not possible by assumption.
    The second case implies that $\rho$ is $1$-structured,
    but not $\mname$-structured for any $\mname \ge 2$.
    This immediately contradicts our assumption.
    Finally, if $\rhoMax \ge \sigMax$,
    then we can use the \encoder{A} from \cref{lem:lb:existenceOfManager},
    Case~\ref{lem:lb:existenceOfManager:rho} with $\abs{A}=\rhoMax+1$.
    \qedhere
  \end{itemize}
\end{proof}
Observe that if $0 \in \rho$, the empty set is a trivial solution and \srDomSet becomes
trivial. Hence, \(0 \notin \rho\) is indeed a reasonable assumption.
\begin{observation}
  \label{lem:decision:polyCases}
  For all $\sigma, \rho \subseteq \NN$ with \(0 \in \rho\),
  \srDomSet can be solved in \(\Oh(1)\) time.
\end{observation}

\subsection{Counting Problem}\label{sec:LB-counting}

We have mentioned earlier that the lower bounds for the counting version
follow the same ideas as the bounds for the decision version.
In particular, we can directly reuse the construction of managers
for the counting version.
Moreover, the lower bound for the intermediate problem (with relations)
is based on the result for the decision version.
However, as we allow cofinite sets for the counting version,
further investigation is needed to obtain the following result
where we again use the managers in a black-box style.
\lowerBoundForCountDomSetRel*
We prove \cref{lem:count:lowerBoundWhenHavingSuitableGadget}
in \cref{sec:high-level:counting}.

While our approaches for the decision version and for the counting version were so far very similar,
this does not hold for the last step in which we remove the relations.
Although the results are still comparable,
the following reduction is much more involved than the result for the decision version.
With \cref{fact:trivial-pairs-counting} in mind,
recall that a pair $\sigma, \rho$ of non-empty sets is \emph{trivial}
if $\rho=\{0\}$ or $\sigma =\rho=\ZZ_{\ge0}$.
Otherwise, we say that $(\sigma, \rho)$ is \emph{non-trivial}.
\thmCountRemovingRelations*

We defer the proof of \cref{thm:RelfromDS} to \Cref{sec:realizingRelations:counting}.
Combining the previous results,
we obtain the lower bound for the counting version in \cref{thm:lower-main-intro}.
\thmLBmaincounting*
\begin{proof}%[Proof of \cref{thm:lower-main-intro}]
  The proof follows in the same way as the proof of \cref{thm:lower-main-intro-decision};
  the only difference is that we use the intermediate lower bound
  for \srCountDomSetRel from \cref{lem:count:lowerBoundWhenHavingSuitableGadget}
  and the reduction from \cref{lem:lower:count:removingRelations} to remove the relations.
\end{proof}

\section{Constructing Providers}
\label{sec:provider}
In our quest to construct the managers from \cref{lem:lb:existenceOfManager}, we start
with a set of simple, but useful \(L\)-providers (recall \cref{def:provider}).
Later, in \cref{sec:manager}, we then
use said providers in the final construction for \cref{lem:lb:existenceOfManager}.

We start with simple providers in \cref{sec:provider:easy}, which in particular includes
providers for languages that contain only \(\rho\)-states.
Next, in \cref{sec:provider:sigma} we construct providers for languages that contain only
$\sigma$-states; later, we use these providers to
construct the managers from Cases~\ref{lem:lb:existenceOfManager:rho}
and~\ref{lem:lb:existenceOfManager:sigma} in \cref{lem:lb:existenceOfManager},
that is, \encoder{A}s where we have that either
$A \subseteq \rhoStates$ or $A \subseteq \sigStates$.

Finally, in \cref{sec:provider:complex}, we construct  providers
that contain $\sigma$-states together with $\rho$-states.
These providers then allow us to construct \encoder{A}s
for \(A\) that contain states from both $\sigStates$ and $\rhoStates$.
These managers correspond to Cases~\ref{lem:lb:existenceOfManager:even}
and~\ref{lem:lb:existenceOfManager:all} in \cref{lem:lb:existenceOfManager}.

\subsection{\boldmath Simple General-Purpose Providers}
\label{sec:provider:easy}

We start with basic providers that add a single \(\sigma\) state and a single \(\rho\)-state.
This type of provider turns out to be useful in many constructions and is used repeatedly throughout the rest of this work.
For this reason, we also give additional properties on the partial solution sets that are needed later on.

\begin{restatable}{lemma}{lemHappyGadgetCounting}
	\label{lem:happyGadget}
	Let $(\sigma,\rho)$ be non-empty sets
	with $s\in \sigma$, $r\in \rho$, and $r\ge 1$.
	There is a $\{\sigma_s,\rho_r\}$-provider $(G,\{u\})$ that has disjoint $(\sigma,\rho)$-sets $X$ and $Y$ with
	$X\cup Y=V(G)$, $u\in X$ with $\abs{N(u)\cap X}=s$,
	and $u\notin Y$ with $\abs{N(u)\cap Y}=r$.
\end{restatable}
\begin{proof}
	Let $G$ consist of $2r$ cliques $X_1,\dots, X_r, Y_1, \dots Y_r$,
	each of order $s+1$.
	For $i\in \nset{r}$,
	let $x_0^{(i)},\dots, x_s^{(i)}$ be the vertices of $X_i$,
	and let $y_0^{(i)},\dots, y_s^{(i)}$ be the vertices of $Y_i$.
	These cliques are connected in a way that they form bicliques indexwise,
	that is, for each $i\in \nset{r}$ and each $j\in \fragment{0}{s}$,
	the vertex $x_j^{(i)}$ is adjacent not only to the vertices in its clique $X_i$,
	but also to the vertices $y_j^{(1)}, \dots, y_j^{(r)}$.
	Let $u=x_1^{(1)}$.
	Note that both $X=\bigcup_{i=1}^r X_i$ and $Y=\bigcup_{i=1}^r Y_i$
	are $(\sigma,\rho)$-sets of $G$,
	where $u\in X$ with $\abs{N(u)\cap X}=s$,
	and $u\notin Y$ with $\abs{N(u)\cap Y}=r$.
	Consult \cref{sf-73} for a visualization of the construction and the aforementioned (partial) solutions.
\end{proof}

\begin{figure}[t]
    \begin{subfigure}[b]{.49\textwidth}
        \centering
        \includegraphics[scale=1.3]{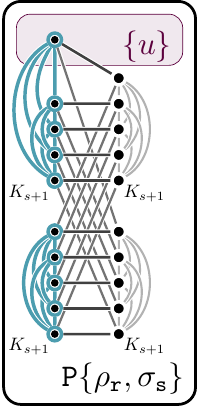}\qquad
        \includegraphics[scale=1.3]{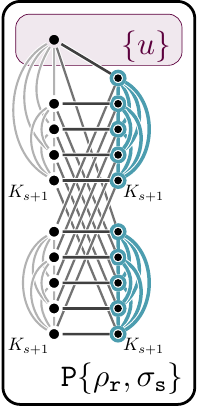}
        \caption{Valid partial solutions for the $\{\sigma_s,\rho_r\}$-provider constructed in \cref{lem:happyGadget} for the case when \(r \geq 1\)
            (\(r=2\) in the figure).
        }\label{sf-73}
    \end{subfigure}
    \begin{subfigure}[b]{.49\textwidth}
        \centering
        \includegraphics[scale=1.7]{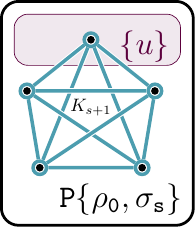}\qquad
        \includegraphics[scale=1.7]{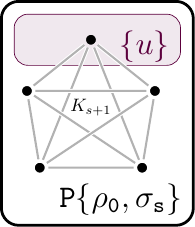}\\
        \vspace{1.5ex}\null
        \caption{Valid partial solutions for the $\{\sigma_s,\rho_r\}$-provider constructed in \cref{lem:fillinggadget} for the case when \(r = 0\).
        }\label{sf-64}
    \end{subfigure}
    \caption{The gadget constructions from \cref{lem:happyGadget,lem:fillinggadget}. We use \(\tt
    P\{\rho_0,\sigma_s\}\) in future figures for appropriate versions of the
    providers of this figure.}
\end{figure}

The last lemma excludes the case $r = 0$, which is why we also give the following lemma.

\begin{lemma}\label{lem:cdgadget}\label{lem:fillinggadget}
    For any $s\in \sigma$ and $r\in \rho$, there is a $\{\sigma_s,\rho_r\}$-provider.
\end{lemma}

\begin{proof}
    We define a graph $G$ with a single portal $u$.
    If $r \neq 0$, then we define $(G,\{u\})$ to be the graph constructed in \cref{lem:happyGadget}.

    So, suppose $r = 0$.
    We choose \(G\) to be a clique on $s+1$ vertices, and we declare any of its vertices to be $u$.
    Observe that both selecting all vertices and selecting no vertices constitute valid
    partial solutions, and hence, the strings \(\sigma_s\) and \(\rho_0 = \rho_r\) are
    compatible with \((G,\{u\})\). Hence, $(G,\{u\})$ is a $\{\sigma_s,\rho_r\}$-provider.
    Consult \cref{sf-64} for a visualization of the construction and the aforementioned (partial) solutions.
\end{proof}

Next, we construct a provider for all possible elements from \(\rhoStates\).

\begin{lemma}\label{lem:trivialrhogadget}
    For non-empty sets $\sigma$ and $\rho$, there is a $\{\rho_0,\rho_1, \ldots, \rho_{\rhoMax}\}$-provider.
\end{lemma}
\begin{figure}[t]
    \centering
    \includegraphics[scale=1.5]{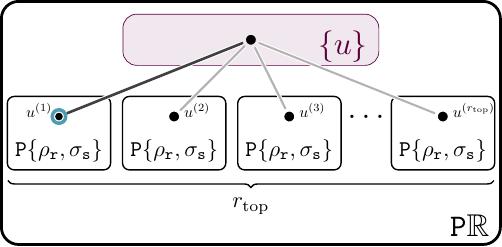}\quad
    \includegraphics[scale=1.5]{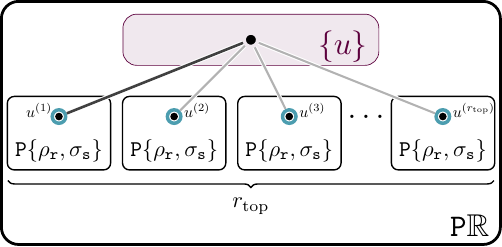}
    \caption{The gadget constructions from \cref{lem:trivialrhogadget} with exemplary
    partial solutions that correspond to \(\rho_1\) and \(\rho_{\rhoMax}\).
    For any $\{\sigma_s,\rho_r\}$-provider \((H^{(i)}, \{\port^{(i)}\})\) used in the
    construction, we depict only its portal vertex and label it with \(u^{(i)}\).
    We use \(\tt P{\rhoStates}\) in future figures for appropriate versions of the providers
    of this figure.}\label{fig:trivialrhogadget}
\end{figure}
\begin{proof}
    Fix $s\in \sigma$ and $r\in\rho$.
    We define a graph $G$ with a single portal $\port$.
    To that end, we take \(\rhoMax\) independent copies
    \((H^{(i)}, \{\port^{(i)}\})\) of the
    $\{\sigma_s,\rho_r\}$-provider from \cref{lem:cdgadget}, and connect each vertex
    \(\port^{(i)}\) to the vertex \(\port\).
    Consult \cref{fig:trivialrhogadget} for a visualization.

    By \cref{lem:cdgadget}, each provider \((H^{(i)},\{\port^{(i)}\})\) has at least a partial solution
    that selects \(\port^{(i)}\) and a partial solution that does not select
    \(\port^{(i)}\). Hence, for each \(j \in \fragment{0}{\rhoMax}\),
    the constructed graph \(G\) has at least one partial solution that selects exactly
    \(j\) neighbors of \(\port\); this completes the proof.
\end{proof}

\subsection{\boldmath A Provider for \texorpdfstring{$\sigma$}{Sigma}-States}
\label{sec:provider:sigma}

\def\modgraph#1#2{M^{(#1)}_{#2}}
For our next construction, we need regular (bipartite) graphs, which are fortunately easy
to construct.
\begin{lemma}
    \label{la:construct-regular-graph}
    For any non-negative integer \(d \le n\), the bipartite graph \(\modgraph{d}{n}\),
    \begin{align*}
        V(\modgraph{d}{n}) &\coloneqq \{v_0,\dots,v_{n-1}\} \cup \{w_0,\dots,w_{n-1}\},\\
        E(\modgraph{d}{n}) &\coloneqq \{\{v_i, w_j\} \mid ((i - j) \bmod n) \in
        \fragmentco{0}{d}\}
    \end{align*}
    is \(d\)-regular and contains an edge \(\{v_i,w_i\}\) for each \(i \in
    \fragmentco{0}{n}\).
\end{lemma}
\begin{proof}
    Observe that, for each \(b \in \fragmentco{0}{d}\) and each \(i \in
    \fragmentco{0}{n}\), there is exactly one solution each to
    the equations \(i - x \equiv_{n} b\) and \(x - i \equiv_n b\); the claim follows.
\end{proof}

Further, recall from \cref{sec:prelims} that, for a string \(x\) and a character
$a\in \allStates$, we write $\occ{x}{a}$ to denote the number of occurrences of $a$ in $x$.

\begin{lemma}
    \label{lem:sigmaGadget}
    For an \(r \ge 1\), consider the language
    \[L_{r} \coloneqq \{
            x \in \{\sigma_0,\sigma_1\}^{4r}
    \mid \occ{x}{\sigma_1\!} \in \{0,2r\}\}.\]

    For any set \(\rho\) that contains \(r\), and any set \(\sigma\) that contains an
    \(s \ge r\), there is an $L_r$-provider.
    Moreover, the closed neighborhoods of portals in \(\Port\) are pairwise disjoint.
\end{lemma}

\begin{proof}
    Let $\Port = \{\port_1,\dots,\port_{4r}\}$ denote the set of portal vertices.
    We define an $L$-provider  \((G, \Port)\) as follows.
    For every $A \subseteq\numb{4r}$ such that $\abs{A} = 2r$,
    $G$ has the following vertices and edges in addition to its portals.
    It has $(2s+2)r$ vertices
    which are partitioned into blocks $V_{A}^1,\dots,V_{A}^{2s+2}$,
    where each block $V_{A}^q$
    contains $r$ vertices $v_{A,1}^q,\dots,v_{A,r}^q$.
    Moreover, $G$ has $2sr + 2(r-1)$ vertices
    which are partitioned into blocks $W_{A}^1,\dots,W_{A}^{2s+2}$,
    where the first two blocks $W_{A}^q$, $q \in \{1,2\}$,
    contain $r-1$ vertices $w_{A,1}^q,\dots,w_{A,r-1}^q$,
    and all other blocks $W_{A}^q$, $q \in\fragment{3}{2s+2}$, contain $r$ vertices $w_{A,1}^q,\dots,w_{A,r}^q$.
    We define $V_A \coloneqq \bigcup_{q \in\numb{2s+2}}V_A^q$
    and $W_A \coloneqq \bigcup_{q \in\numb{2s+2}}W_A^q$.
    Suppose $A = \{k_1,\dots,k_{2r}\}$.
    Then, $G$ has edges to form the following:
    \begin{itemize}
        \item A complete bipartite graph between $V_A^q$ and $W_A^q$ for all $q \in\numb{2s+2}$.
        \item A graph on $V_A$ such that every $v \in V_A^1 \cup V_A^2$ has degree $s-1$,
            and every $v \in V_A \setminus (V_A^1 \cup V_A^2)$ has degree $s$.
            This is possible by \Cref{la:construct-regular-graph}.
            Indeed, $\abs{V_A}$ is even and ${\abs{V_A}}/{2} \geq s$.
            So, we can build an $s$-regular graph on $V_A$.
            Additionally, by numbering vertices appropriately, there is a matching between $V_A^1$ and $V_A^2$, and we may simply omit the corresponding edges to decrease the degree of every vertex from $V_A^1 \cup V_A^2$ by one to achieve the desired result.
        \item An $s$-regular graph on $W_A$ which is possible by \Cref{la:construct-regular-graph}
            since $|W_A|$ is even and ${\abs{W_A}}/{2} \geq s$.
        \item $\port_{k_p}v_{A,p}^1$ for all $p \in\numb{r}$.
        \item $\port_{k_p}v_{A,p-r}^2$ for all $p \in\fragment{r+1}{2r}$.
    \end{itemize}
    This completes the description of the graph $G$.
    From the last two items, it follows
    that no two portals share an edge
    and that their neighborhoods are disjoint.

    We start with some simple observations.
    Fix an $A \subseteq\numb{4r}$ such that $\abs{A} = 2r$.
    Then, $G[W_A]$ is $s$-regular.
    We have $\abs{N_G(v) \cap V_A} = s-1$ for all $v \in V_A^1 \cup V_A^2$,
    and $\abs{N_G(v) \cap V_A} = s$ for all $v \in V_A \setminus V_A^1 \cup V_A^2$.
    Also, $\abs{N_G(v) \cap W_A} = r-1$ for all $v \in V_A^1 \cup V_A^2$,
    and $\abs{N_G(v) \cap W_A} = r$ for all $v \in V_A \setminus V_A^1 \cup V_A^2$.
    Moreover, $\abs{N_G(w) \cap V_A} = r$ for all $w \in W_A$.
    Finally, $\abs{N_G(v) \cap \Port} = 1$ for every $v \in V_A^1 \cup V_A^2$.

	Now, fix $x \in \{\sigma_0,\sigma_1\}^{4r}$
	such that $\occ{x}{\sigma_1} \in \{0,2r\}$,
    and write $x = \sigma_{i_1}\sigma_{i_2}\dots\sigma_{i_{4r}}$.
	We need to show that $x \in L(G,\Port)$.
	Let $A_x \coloneqq \{p \in\numb{4r} \mid i_p = 1\}$.

	First, suppose that $A_x = \emptyset$.
	Then, we define a solution set
	\[S_x \coloneqq \Port \cup \bigcup_{A \subseteq\numb{4r}\colon \abs{A} = 2r} W_A.\]
	Using the basic observations outlined above,
	it is easy to verify that this solution set indeed witnesses that $x \in L(G,\Port)$.

	Otherwise, $\abs{A_x} = 2r$.
	Now, we define a solution set
	\[S_x \coloneqq \Port \cup V_{A_x} \cup \bigcup_{A \subseteq\numb{4r}\colon \abs{A} = 2r, A \neq A_x} W_A.\]
	Again, using the basic observations outlined above,
	it is easy to verify that this solution set indeed witnesses that $x \in L(G,\Port)$.
\end{proof}

\subsection{\boldmath Providers for Combinations of \texorpdfstring{$\sigma$}{Sigma}- and \texorpdfstring{$\rho$}{Rho}-States}
\label{sec:provider:complex}

Complementary to the previous section,
we establish $L$-providers in this section,
where $L$ contains states from both $\sigStates$ and $\rhoStates$.
These providers are relevant if $(\sigma,\rho)$ is $\mname$-structured for $\mname \le 2$,
but not $\mname$-structured for any $\mname\ge 3$.

We start by obtaining two auxiliary providers that are used in later constructions.

\begin{lemma}\label{lem:3stategadgetsimple}
    Fix integers $s\in\sigma$, $r\in \rho$ with $s,r\ge 1$, and set
    \[L_{s,r} \coloneqq \{\sigma_{s-1}\sigma_{s-1}\rho_{r-1},\quad \rho_r\rho_r\rho_r,\quad \rho_{r-1} \rho_r\sigma_s\}.\]
    Then, there is an $L_{s,r}$-provider.
\end{lemma}
\begin{proof}
    We wish to construct a graph \(G\) with three portals
    \(\Port= \{\port_1, \port_2,  \port_3\}\), such
    that \(G\) has (at least) three partial solutions: one partial solution with exactly two portals (\(\port_1, \port_2\)) selected, one partial solution with no
    selected portal, and one partial solution with exactly one portal (\(\port_3\)) selected (in each case with additional constraints on
    the number of selected neighbors).
    Intuitively, we construct \(G\) to consist of three parts \(X, Y, Z\) that, when
    (fully) selected (and when nothing else is selected), each correspond to one of the desired partial solutions; that is, we
    have \(\port_1, \port_2 \in X\) and \(\port_3 \in Z\).

    Hence, for each of the parts \(X, Y, Z\), we need to be able to either fully select
    the part or to not select any vertex of the part.
    To ensure the former, each vertex in \(X, Y, Z\) is part of a clique on \((s + 1)\)
    vertices, and to ensure the latter, each vertex in \(X, Y, Z\) has \(r\) neighbors in each
    of the corresponding other parts.\footnote{For the vertices
    \(\port_1,\port_2,\port_3\), slightly adapted requirements hold.}
    We achieve this property by introducing  a total of \(r\) cliques (each of order $s+1$) per part. For each pair of different parts, we add edges to form
    a complete \(r\)-regular bipartite graph between vertices that have the same
    index in different cliques.

    Formally, we define $G$ as follows.
    For the vertices, we set
    \begin{align*}
        X &\coloneqq \{x_i^j \mid i\in \numb{r}, j\in \numb{s+1}\}, \\
        Y &\coloneqq \{y_i^j \mid i\in \numb{r}, j\in \numb{s+1}\},\quad\text{and}\\
        Z &\coloneqq \{z_i^j \mid i\in \numb{r}, j\in \numb{s+1}\}.
    \end{align*}
    Further, we set
    \begin{align*}
        \port_1 \coloneqq x_1^1,\quad
        \port_2 &\coloneqq x_1^2,\quad\text{and}\quad
        \port_3 \coloneqq z_1^1.
    \end{align*}
    Now, for each $i\in \numb{r}$, we connect
    the vertices $\{x_i^j \mid j\in \numb{s+1}\}$
    ($\{y_i^j \mid j\in \numb{s+1}\}$ and $\{z_i^j \mid j\in \numb{s+1}\}$, respectively)
    to form a clique on $s+1$ vertices.
    Then, between each pair of sets from $X,Y,Z$,
    we connect the vertices with the same index $j\in \numb{s+1}$ to form a complete bipartite graph
    (with $r$ vertices in each part).
    Finally, we remove the edge between $\port_1 = x_1^1$ and $\port_2 = x_1^2$, and the
    edge between $\port_1 = x_1^1$ and $\port_3 = z_1^1$.

    It remains to show that $(G, \{x_1^1, x_1^2, z_1^1\})$ is an $L_{s,r}$-provider.
    We give partial solutions corresponding to the three elements of $L_{s,r}$.
    \begin{description}
        \item[\boldmath $\sigma_{s-1} \sigma_{s-1} \rho_{r-1}$.]
            Select the vertices from $X$. Thus, $x_1^1$ and $x_1^2$ are selected, but $z_1^1$ is not.
            To satisfy the $\sigma$-constraints, each vertex in $X$, with the exception of $x_1^1$ and $x_1^2$, obtains $s$ selected neighbors from (their respective clique in) $X$.
            Because of the missing edge, both $x_1^1$ and $x_1^2$ have only $s-1$ selected neighbors each.
            To satisfy the $\rho$-constraints, each vertex from $Y$ and $Z$, with the exception of $z_1^1$, obtains $r$ selected neighbors from $X$.
            Because of the missing edge to $x_1^1$, $z_1^1$ has only $r-1$ selected neighbors.

        \item[\boldmath $\rho_r\rho_r\rho_r$.]
            Select the vertices from $Y$. As before, we can check that this selection witnesses $\rho_r\rho_r\rho_r$.

        \item[\boldmath $\rho_{r-1}\rho_r\sigma_s$.]
            Select the vertices from $Z$. As before, we can check that this selection witnesses $\rho_{r-1}\rho_r\sigma_s$.
            \qedhere
    \end{description}
\end{proof}

For ease of notation, we abbreviate,
for all $\tau_t \in \allStates$ and all $\delta \in \NN$,
the state $\tau_t \dots \tau_t \in \allStates^\delta$
of a graph with $\delta$ portals
by the notation $(\tau_t)^\delta$.

\begin{lemma}\label{lem:3stategadgetbig}
    \newcommand{\dg}{\delta} 
    Consider $s, s'\in\sigma$, $r\in \rho$ with $s,r\ge 1$ and $s' > \sigMin$, where $\sigMin\coloneqq\min(\sigma)$.
    For any $k\in \ZZ_{> 0}$ and $\dg \coloneqq k(s'-\sigMin)$,
    there is a $\{(\rho_r)^{\dg},\quad (\rho_{r-1})^{\dg},
    \quad(\sigma_s)^{\dg}\}$-provider (with $\dg$ portals).
\end{lemma}
\begin{proof}
    \newcommand{\dg}{\delta}
    We define a graph $G$ with $\dg=k(s'-\sigMin)$ portals as follows.
    \begin{itemize}
        \item There are three disjoint sets of vertices $A=\{a_i \mid i\in \numb{\dg}\}$,
            $B=\{b_i \mid i\in \numb{\dg}\}$, and $C=\{c_i \mid i\in \numb{\dg}\}$,
            where the vertices of $C$ form the $\dg$ portals of $G$.
        \item For each $i\in \numb{\dg}$,
            $(a_i, b_i, c_i)$ are (in this order) the portals of an attached
            $\{\sigma_{s-1}\sigma_{s-1}\rho_{r-1},\allowbreak \rho_r\rho_r\rho_r, \rho_{r-1}\rho_r\sigma_s\}$-provider $J_i$,
            which exists by \Cref{lem:3stategadgetsimple} and the fact that $s,r\ge 1$.
        \item There is a set $D_A$ of size $k$, where each vertex of $D_A$ is the portal of an attached $\{\sigma_{\sigMin},\rho_r\}$-provider, which exists by \Cref{lem:fillinggadget}.
            We introduce edges between $A$ and $D_A$ in such a way that each vertex in $A$ has precisely $1$ neighbor in $D_A$, and each vertex in $D_A$ has precisely $s'-\sigMin$ neighbors in $A$.
            This is possible as $\dg=k\cdot (s'-\sigMin)$.
        \item Similarly, there is a set $D_B$ of size $k$, where each vertex of $D_B$ is the portal of an attached $\{\sigma_{\sigMin},\rho_r\}$-provider, each vertex in $B$ has precisely $1$ neighbor in $D_B$, and each vertex in $D_B$ has precisely $s'-\sigMin$ neighbors in $B$.
    \end{itemize}

	We show that $(G, C)$ is a
	$\{(\rho_r)^{\dg}, (\rho_{r-1})^{\dg}, (\sigma_s)^{\dg}\}$-provider.
	We give solutions corresponding to the three states.
	\begin{description}
		\item[\boldmath $(\rho_{r})^{\dg}$]
		None of the vertices in $A,B,C,D_A,D_B$ are selected. For each $i\in \numb{\dg}$, select vertices in $J_i$ according to the solution that witnesses the state $\rho_r\rho_r\rho_r$. In the $\{\sigma_{\sigMin},\rho_r\}$-providers attached to $D_A$ and $D_B$, select vertices according to the $\rho_{r}$-state.

		For each $i$, the vertices $a_i$ and $b_i$, as well as the portal $c_i$, are unselected and each obtain $r$ selected neighbors from $J_i$. Each vertex in $D_A$ and $D_B$ is also unselected and obtains $r$ selected neighbors from the attached $\{\sigma_{\sigMin},\rho_r\}$-provider.

		\item[\boldmath $(\rho_{r-1})^{\dg}$]
		Select the vertices $A\cup B \cup D_A \cup D_B$.
		For each $i\in \numb{\dg}$, select vertices in $J_i$
		according to the solution that witnesses
		the state $\sigma_{s-1}\sigma_{s-1}\rho_{r-1}$.
		In the $\{\sigma_{\sigMin},\rho_r\}$-providers attached to $D_A$ and $D_B$, select vertices according to the $\sigma_{\sigMin}$-state.

		For each $i$, the portal $c_i$ is unselected and obtains $r-1$ selected neighbors from $J_i$.
		The vertices $a_i$ and $b_i$ are selected and each obtain $s-1$ selected neighbors from $J_i$, and $a_i$ also obtains $1$ selected neighbor from $D_A$, and $b_i$ obtains $1$ selected neighbor from $D_B$. Thus, both obtain a total of $s\in \sigma$ selected neighbors. The vertices in $D_A$ are selected and each obtain $s'-\sigMin$ selected neighbors from $A$, as well as $\sigMin$ selected neighbors from the attached $\{\sigma_{\sigMin},\rho_r\}$-provider, for a total of $s'\in \sigma$. Similarly, the vertices in $D_B$ are selected and obtain $s'\in \sigma$ selected neighbors.

		\item[\boldmath $(\sigma_{s})^{\dg}$]
		Select the vertices $C \cup D_A$.
		For each $i\in \numb{\dg}$, select vertices in $J_i$
		according to the solution that witnesses
		the state $\rho_{r-1} \rho_r\sigma_s$.
		In the $\{\sigma_{\sigMin},\rho_r\}$-providers attached to $D_A$, select vertices according to the $\sigma_{\sigMin}$-state. In the $\{\sigma_{\sigMin},\rho_r\}$-providers attached to $D_B$, select vertices according to the $\rho_{r}$-state.

		For each $i$, the portal $c_i$ is selected and obtains $s$ selected neighbors from $J_i$.
		The vertex $a_i$ is unselected, obtains $r-1$ selected neighbors from $J_i$,
		and obtains $1$ selected neighbor from $D_A$.
		The vertex $b_i$ is unselected and obtains $r$ selected neighbors from $J_i$.
		The vertices in $D_A$ are selected and each obtain
		$\sigMin$ selected neighbors from the attached $\{\sigma_{\sigMin},\rho_r\}$-provider.
		The vertices in $D_B$ are unselected
		and each obtain $r$ selected neighbors from the attached $\{\sigma_{\sigMin},\rho_r\}$-provider.
		\qedhere
	\end{description}
\end{proof}

Before we continue with the next gadget, we first prove a technical lemma
which ensures the existence of bipartite graphs with certain degree sequences.

\begin{lemma}
	\label{lem:bipartite-with-degree-sequence}
	Let $c_1,\dots,c_\ell,d_1,\dots,d_k \in \NN$ such that
	\[\sum_{i \in [\ell]} c_i = \sum_{j \in [k]} d_j.\]
	Also, suppose $a$ is another natural number such that $a \geq \max\{c_1,\dots,c_\ell,d_1,\dots,d_k\}$.
	Then, there is some $s_0 \geq 0$ such that, for every $s \geq s_0$, there is a bipartite graph $G = (V,W,E)$ such that
	\begin{enumerate}
		\item $V = \{v_1,\dots,v_\ell,x_1,\dots,x_s\}$,
		\item $\deg_G(v_i) = c_i$ for all $i \in [\ell]$, and $\deg_G(x_i) = a$ for all $i \in [s]$,
		\item $W = \{w_1,\dots,w_k,y_1,\dots,y_s\}$, and
		\item $\deg_G(w_i) = d_i$ for all $i \in [k]$, and $\deg_G(y_i) = a$ for all $i \in [s]$.
	\end{enumerate}
\end{lemma}
\begin{proof}
	Pick $s_0 \coloneqq a$.
	We give an iterative construction for $G$.
	Initially, $E(G) \coloneqq \emptyset$.
	We define $c(v_i) \coloneqq c_i$ for all $i \in [\ell]$, and $c(x_j) \coloneqq a$ for $j \in [s]$.
	Similarly, define $c(w_i) \coloneqq d_i$ for all $i \in [k]$, and $c(y_j) \coloneqq a$ for $j \in [s]$.
	Throughout the construction, $c(u)$ denotes the number of neighbors still to be added to $u \in V \cup W$.
	We shall maintain the property that
	\begin{enumerate}[label = (\roman*)]
		\item\label{item:bipartite-with-degree-sequence-1} $c(u) \geq 0$ for all $u \in V \cup W$,
		\item\label{item:bipartite-with-degree-sequence-2} $\sum_{v \in V} c(v) = \sum_{w \in W} c(w)$, and
		\item\label{item:bipartite-with-degree-sequence-3} $|\{v \in V \mid c(v) \in \{c_{\max},c_{\max} - 1\}\}| \geq a$, where $c_{\max} \coloneqq \max_{v \in V} c(v)$.
	\end{enumerate}
	Observe that this property is initially satisfied.

	If $c(w) = 0$ for all $w \in W$, then $c(v) = 0$ for all $v \in V$
	by Conditions~\ref{item:bipartite-with-degree-sequence-1}
	and~\ref{item:bipartite-with-degree-sequence-2},
	and the construction is complete.
	So, fix some $w \in W$ such that $c(w) > 0$.
	Let $d \coloneqq c(w)$.
	Let $v_1',\dots,v_{\ell+s}'$ be a list of all the vertices from $V$ ordered according to the current capacities, i.e., $c(v_i') \geq c(v_{i+1}')$ for all $i \in [\ell + s - 1]$.

	First, we claim that $c(v_i') \geq 1$ for all $i \in [d]$.
	If $c_{\max} \geq 2$, then this follows directly from
	Condition~\ref{item:bipartite-with-degree-sequence-3} because $d \leq a$.
	Otherwise, $c_{\max} = 1$.
	But, then $|\{v \in V \mid c(v) = c_{\max}\}| \geq d$
	by Condition~\ref{item:bipartite-with-degree-sequence-2}.
	We add an edge between $w$ and $v_i'$ for all $i \in [d]$.
	Also, we update $c(w) \coloneqq 0$ and decrease $c(v_i')$ by one for all $i \in [d]$.
	Clearly, Conditions~\ref{item:bipartite-with-degree-sequence-1}
	and~\ref{item:bipartite-with-degree-sequence-2} are still satisfied.
	For Condition \ref{item:bipartite-with-degree-sequence-3}, it can be observed that every $v \in V$ such that $c(v) \in \{c_{\max},c_{\max} - 1\}$ before the update still satisfies this Condition after the update.

	Repeating this process until all capacities are equal to zero, we obtain the desired graph $G$.
\end{proof}

The following lemma crucially exploits that $(\sigma,\rho)$ is $1$-structured or $2$-structured.
While in the former case we can obtain a provider with only one portal,
this is not possible for the latter case.
On an intuitive level, the property of being $2$-structured raises parity issues
which does not allow us to have two independent copies
of the provider for the $1$-structured case.
Hence, we design a provider where the state of one portal
depends on the state of the other portal, and
in fact, both states must be equal for this provider.

\begin{lemma}\label{lem:gadgettechnical}
	Let $\rhoMax\ge 1$.
	Suppose there is a maximum value $\mname$
	such that $(\sigma, \rho)$ is $\mname$-structured.
	\begin{enumerate}
		\item If $\mname=1$,
		then there is a $\{\rho_0, \rho_1, \sigma_0\}$-provider.
		\item If $\mname=2$,
		then there is a $\{\rho_0, \rho_2, \sigma_0\}$-provider.
		There is also a $\{\rho_0\rho_0, \rho_1\rho_1, \sigma_0\sigma_0\}$-provider.
	\end{enumerate}
\end{lemma}
\begin{proof}
	\newcommand{\dg}{b}
	\newcommand{\tdg}{\widetilde b}
	We first claim that there are finite subsets $\sigma' \subseteq \sigma$ and $\rho' \subseteq \rho$ such that there exists a maximum value $\mname'$ such that $(\sigma', \rho')$ is $\mname'$-structured, and moreover $\mname' = \mname$.
	Indeed, if $\sigma$ is finite, we simply set $\sigma' \coloneqq \sigma$.
	Otherwise, $\sigma$ is cofinite which means there is some $s \in \NN$ such that $s,s+1 \in \sigma$.
	In particular, it follows that $\mname = 1$.
	We set $\sigma' \coloneqq \{s,s+1\}$ which implies that $\sigma'$ is $1$-structured, but not $\mname'$-structured for any $\mname \geq 2$.
	The set $\rho'$ is defined analogously.

	Since every $L$-provider with respect to $(\sigma',\rho')$ is also an $L$-provider with respect to $(\sigma,\rho)$, it suffices to prove the lemma for the pair $(\sigma',\rho')$.
	For ease of notation, let us suppose that $(\sigma,\rho) = (\sigma',\rho')$, that is, in the remainder of the proof, we suppose that both $\sigma$ and $\rho$ are finite.

	Let $\sigma=\{s_1, \ldots, s_{|\sigma|}\}$ and $\rho=\{r_1, \ldots, r_{|\rho|}\}$, where we assume that elements are ordered with respect to size.
	The fact that $\mname$ is the maximum value
	such that $(\sigma, \rho)$ is $\mname$-structured is equivalent to
	\begin{equation}
		\gcd(\{r-r_1 \mid r\in \rho\}\cup \{s-s_1 \mid s\in \sigma\})=\mname \label{eq:sigmadifference}.
	\end{equation}

	We claim that there are non-negative integers
	$x_{i}, \tx_i$ (for $i\in \fragment{2}{\abs{\rho}}$)
	and non-negative integers $y_{j}, \ty_{j}$
	(for $j\in\numb{\abs{\sigma}}$)
	such that
	\begin{equation}\label{eq:mixedchoice1}
		\mname+ \sum_{i=2}^{|\rho|} (r_i-r_1) (x_i-\tx_i) + \sum_{j=2}^{|\sigma|} (s_{j}-s_1) (y_j-\ty_j) = 0
	\end{equation}
	We argue that such a choice is possible.
	Note that~\eqref{eq:mixedchoice1} corresponds to a Diophantine equation
	with variables $(x_i-\tx_i)$ and $(y_j-\ty_j)$
	(for $2\le i \le |\rho|$, $j\in \numb{\abs{\sigma}}$),
	which has solutions by~\eqref{eq:sigmadifference}.
	We can choose non-negative values for the indeterminates accordingly.

	Using these integer solutions, we define a graph $G$ with one portal as follows. 
	\begin{itemize}
		\item The vertices of $G$ are partitioned into two sets $L$ and $R$ together with some attached providers.
		\item $L$ consists of the following.
		\begin{itemize}
			\item The portal vertex $\port$ which is intended to have $\mname$ neighbors in $R$.
			\item For each $i\in \fragment{2}{|\rho|}$, a set $A_i$ of $x_i$ vertices, each of which are intended to have $r_i-r_1$ neighbors in $R$.
			\item For some yet to be determined positive integer $k$, a set $A^*$ of $k$ vertices, each of which are intended to have $r_{|\rho|}-r_1$ neighbors in $R$.
			\item  For each $j\in \fragment{2}{|\sigma|}$,
			a set $B_j$ of $\dg_j\coloneqq (s_{j}-s_1) y_{j}$ vertices,
			each of which is intended to have $1$ neighbor in $R$.
		\end{itemize}
		We set $A=A^*\cup \bigcup_{i\in \fragment{2}{|\rho|}} A_i$ and $B=\bigcup_{i\in \fragment{2}{|\sigma|}} B_i$.
		\item There are some providers attached to $L$:
		\begin{itemize}
			\item Each vertex in $A$ is the portal of an attached $\{\sigma_{s_1},\rho_{r_1}\}$-provider (exists by \Cref{lem:cdgadget}).
			\item If $|\sigma|\ge 2$, then $\sigMax\ge 1$.
			Also, $\rhoMax\ge 1$ by assumption.
			So, we can choose some $r\in \rho$ and $s\in \sigma$ with $r,s\ge 1$.
			For each $j\in \fragment{2}{|\sigma|}$,
			the vertices of $B_j$ are the $\dg_j$ portals of an attached
            $\{(\rho_{r})^{\dg_j}, (\rho_{r-1})^{\dg_j},(\sigma_{s})^{\dg_j}\}$-provider $J_j$
			(exists by \Cref{lem:3stategadgetbig} using the fact that $s,r\ge 1$).
		\end{itemize}
		\item Analogously, $R$ consists of the following.
		\begin{itemize}
			\item For each $i\in \fragment{2}{|\rho|}$, a set $\tA_i$ of $\tx_i$ vertices, each of which are intended to have $r_i-r_1$ neighbors in $L$.
			\item A set $\tA^*$ of $k$ vertices, each of which are intended to have $r_{|\rho|}-r_1$ neighbors in $L$.
			\item For each $j\in \fragment{2}{|\sigma|}$, a set $\tB_j$ of $\tdg_j\coloneqq(s_{j}-s_1) \ty_{j}$ vertices, each of which is intended to have $1$ neighbor in $L$.
		\end{itemize}
		We set $\tA=\tA^*\cup \bigcup_{i\in \fragment{2}{|\rho|}} \tA_i$ and $\tB=\bigcup_{i\in \fragment{2}{|\sigma|}} \tB_i$.
		\item There are some providers attached to $R$:
		\begin{itemize}
			\item Each vertex in $\tA$ is the portal of an attached $\{\sigma_{s_1},\rho_{r_1}\}$-provider.
			\item For each $j\in \fragment{2}{|\sigma|}$,
			the vertices in $\tB_j$ are the $\tdg_j$ portals of an attached \linebreak
			$\{(\rho_r)^{\tdg_j}, (\rho_{r-1})^{\tdg_j},(\sigma_s)^{\tdg_j}\}$-provider $\tJ_j$.
			(Here we can use the same $s$ and $r$ as in the definition of $L$.)
		\end{itemize}
		\item Using~\eqref{eq:mixedchoice1}, we verify that the intended number of edges going from $L$ to $R$ is the same as the intended number of edges going from $R$ to $L$:
		\begin{align*}
			\mname
			+ &\sum_{i=2}^{|\rho|} (r_i-r_1) x_i
			+ (r_{|\rho|}-r_1)k
			+ \sum_{i=j}^{|\sigma|} (s_i-s_1)y_j\\
			=&\sum_{i=2}^{|\rho|} (r_i-r_1) \tx_i
			+ (r_{|\rho|}-r_1)k
			+ \sum_{j=2}^{|\sigma|} (s_i-s_1) \ty_j
		\end{align*}
		Note that $r_{|\rho|}-r_1$ is the maximum of the intended degrees.
		Therefore, by Lemma~\ref{lem:bipartite-with-degree-sequence}
		applied to $a=r_{|\rho|}-r_1$, for sufficiently large $k$, edges can be introduced such that the intended degrees in the bipartite graph induced by $L\cup R$ are met.
	\end{itemize}

	We show that $(G, \{\port\})$ is a $\{\rho_0, \rho_{m},\sigma_0\}$-provider.
	We give solutions corresponding to the three states.
	\begin{itemize}
		\item[\boldmath $\rho_0$]
		Both $L$ and $R$ are unselected.
		In each of the $\{\sigma_{s_1},\rho_{r_1}\}$-providers attached to $A$ and $\tA$,
		select vertices according to the $\rho_{r_1}$-state.
		For each $i\in \fragment{2}{s_{|\sigma|}}$,
		select vertices in $J_j$ and $\tJ_j$ according to the $(\rho_r^{\dg_j})$- and $(\rho_r^{\tdg_j})$-state, respectively.

		The portal $\port$ is unselected and has no selected neighbors, as required.
		Moreover, the vertices in $A$ and $\tA$ are unselected and obtain $r_1\in \rho$ selected neighbors from the attached $\{ \sigma_{s_1},\rho_{r_1}\}$-providers.
		The vertices in $B$ and $\tB$ are unselected and obtain $r\in \rho$ selected neighbors from the attached $J_j$'s and $\tJ_j$'s, respectively.

		\item[\boldmath $\rho_\mname$]
		$R$ is selected, but $L$ is unselected.
		In each of the $\{\sigma_{s_1},\rho_{r_1}\}$-providers attached to $A$,
		select vertices according to the $\rho_{r_1}$-state.
		In each of the $\{\sigma_{s_1},\rho_{r_1}\}$-providers attached to $\tA$,
		select vertices according to the $\sigma_{s_1}$-state.
		For each $i\in \fragment{2}{s_{|\sigma|}}$,
		select vertices in $J_j$ according to the $(\rho_{r-1}^{\dg_j})$-state,
		and select vertices in $\tJ_j$ according to the $(\sigma_{s}^{\tdg_j})$-state.

		The portal $\port$ is unselected and has $\mname$ selected neighbors in $R$, as required.
		For $i\in \fragment{2}{|\rho|}$, each vertex in $A_i$ is unselected, obtains $r_1\in \rho$ selected neighbors from the attached $\{\sigma_{s_1},\rho_{r_1}\}$-provider, and additionally obtains $r_i-r_1$ selected neighbors in $R$, for a total of $r_i\in \rho$ selected neighbors. The vertices in $\tA$ are selected and obtain $s_1\in \sigma$ selected neighbors from the attached $\{\sigma_{s_1},\rho_{r_1}\}$-providers. They have no further selected neighbors.
		For $i\in \fragment{2}{|\sigma|}$, each vertex in $B_i$ is unselected, obtains $r-1\in \rho$ selected neighbors from $J_j$, and additionally obtains $1$ selected neighbor from $R$, for a total of $r\in \rho$.
		The vertices in $\tB$ are selected and obtain $s\in \sigma$ selected neighbors from the attached $\tJ_j$'s.

		\item[\boldmath $\sigma_0$]
		The selection is symmetric to the previous state:
		$L$ is selected, but $R$ is unselected.
		In each of the $\{\sigma_{s_1},\rho_{r_1}\}$-providers attached to $A$,
		select vertices according to the $\sigma_{s_1}$-state.
		In each of the $\{\sigma_{s_1},\rho_{r_1}\}$-providers attached to $\tA$,
		select vertices according to the $\rho_{r_1}$-state.
		For each $i\in \fragment{2}{s_{|\sigma|}}$,
		select vertices in $J_j$ according to the $(\sigma_{s}^{\dg_j})$-state,
		and select vertices in $\tJ_j$ according to the $(\rho_{r-1}^{\tdg_j})$-state.

		The portal $\port$ is selected and has no selected neighbors, as required.
		The analysis is symmetric to the previous case.
	\end{itemize}
	So, for $\mname\in \{1,2\}$,
	we have shown the existence of a $\{\rho_0, \rho_\mname, \sigma_0\}$-provider.
	Note that, for $\mname=2$,
	if we replace the single portal $\port$ of degree $2$ with two portals, each of degree $1$,
	then the construction yields a $\{\rho_0\rho_0, \rho_1\rho_1, \sigma_0\sigma_0\}$-provider.
\end{proof}

In the remaining part, we extend these providers to a setting
where the selected portals can also get neighbors.
That is, we also have the state $\sigma_1$ for the portal.
We first handle the case when $(\sigma,\rho)$ is $1$-structured,
and then, independently, the case when the sets are $2$-structured.

\subsubsection{\texorpdfstring%
{\boldmath $1$-Structured $(\sigma,\rho)$}%
{1-Structured (Sigma, Rho)}}
In this section, we consider the case where $m=1$ is the largest $m$ for which $(\sigma, \rho)$ is $m$-structured. We first construct an auxiliary provider
which we combine afterward with the provider from \cref{lem:gadgettechnical}.

\begin{lemma}\label{lem:casem1}
	Suppose there is a $\{\rho_0, \rho_1, \sigma_0\}$-provider.
	If $\sigMax\ge 1$,
	then there is a $\{\rho_0\rho_0, \rho_0\sigma_0,\allowbreak \sigma_0\sigma_0, \sigma_1\sigma_1\}$-provider.
	Moreover, the closed neighborhoods of both portals are disjoint.
\end{lemma}
\begin{proof}
	Let $(J,\{\widehat \port\})$ be a $\{\rho_0, \rho_1, \sigma_0\}$-provider.
	We define a graph $G$ with two portals $\port_1$ and $\port_2$ as follows.
	Let $K$ be a clique with vertices $v_1,\dots,v_{\sigMax+1}$
	where we remove the edge between $v_1$ and $v_2$. (Note that $K$ contains at least $2$ vertices since $\sigMax\ge 1$.)
	Then, $G$ consists of $K$ where we add the two edges $\port_1v_1$ and $\port_2v_2$.
	Moreover, for every $i\in\numb{\sigMax+1}$,
	we add $\rhoMax$ copies of $J$ to $G$,
	where we use $v_i$ as the portal vertex for all copies.

	It remains to show that $(G, \{\port_1,\port_2\})$ is a
	$\{\rho_0\rho_0, \rho_0\sigma_0, \sigma_0\sigma_0, \sigma_1\sigma_1\}$-provider.
	We give solutions corresponding to the four states.
	\begin{description}
		\item[\boldmath $\rho_0\rho_0$]
		No vertex from $K$ is selected, and so, $\port_1$ and $\port_2$ have no selected neighbors.
		To satisfy the $\rho$-constraints of the vertices in $K$, each $v_i$ obtains $\rhoMax$ selected neighbors from its attached copies of $J$ (for each of the $\rhoMax$ attached $\{\rho_0, \rho_1, \sigma_0\}$-providers, we select vertices corresponding to the $\rho_1$-state).

		\item[\boldmath $\rho_0\sigma_0$]
		No vertex from $K$ is selected.
		Since $\port_2$ is selected, $v_2$ already has one selected neighbor.
		From the copies of $J$ attached to $v_2$, it receives another $\rhoMax-1$ selected neighbors, for a total of $\rhoMax$ ($\rhoMax-1$ of the attached $\{\rho_0, \rho_1, \sigma_0\}$-providers are in the $\rho_1$-state, and one is in the $\rho_0$-state).
		The remaining vertices of $K$ get $\rhoMax$ selected neighbors via the attached copies of $J$.

		\item[\boldmath $\sigma_0\sigma_0$]
		No vertex from $K$ is selected.
		Now, both $v_1$ and $v_2$ already have one selected neighbor each, and receive another $\rhoMax-1$ selected neighbors through their attached copies of $J$.
		The remaining vertices of $K$ again get $\rhoMax$ selected neighbors via the attached copies of $J$.

		\item[\boldmath $\sigma_1\sigma_1$]
		All vertices from $K$ are selected.
		The copies of $J$ do not provide further neighbors to these vertices (all of them are in the $\sigma_0$-state).
		Both $v_1$ and $v_2$ have $\sigMax-1$ selected neighbors in $K$ (since the edge between $v_1$ and $v_2$ is missing),
		and they are adjacent to the selected $\port_1$ and $\port_2$, respectively.
		Hence, they each have $\sigMax$ selected neighbors in total.
		By construction, the remaining vertices of $K$ have $\sigMax$ selected neighbors in $K$, and thus, in total.
		\qedhere
	\end{description}
\end{proof}

Combining \cref{lem:casem1} with \cref{lem:gadgettechnical},
we obtain the following result.
\begin{lemma}
	\label{lem:mainm1}
	Let $\sigMax\ge 1$.
	Suppose that $\mname=1$ is the maximum value
	such that $(\sigma, \rho)$ is $\mname$-structured.

	If $\rhoMax\ge 1$, then there is a
	$\{\rho_0\sigma_0, \rho_1\sigma_0, \sigma_0\sigma_0, \sigma_1\sigma_1\}$-provider.
	If $\rhoMax =0$ but $\rho\neq\{0\}$, then there is a
	$\{\rho_0\sigma_0, \sigma_0\sigma_0, \sigma_1\sigma_1\}$-provider.
	Moreover, the closed neighborhoods of the portals in both providers are disjoint.
\end{lemma}
\begin{proof}
	Assume $\rhoMax\ge 1$.
	Let $(G_1, \{\port_1\})$ and $(G_2,\{\port_2\})$ be two copies of a
	$\{\rho_0, \rho_1, \sigma_0\}$-provider as constructed in \Cref{lem:gadgettechnical}.
	Let $(G_3, \{\port_1,\port_2\})$ be a
	$\{\rho_0\rho_0, \rho_0\sigma_0, \sigma_0\sigma_0, \sigma_1\sigma_1\}$-provider which exists by \Cref{lem:casem1}.
	It is easy to check that $(G_1\cup G_2 \cup G_3, \{\port_1,\port_2\})$
	is a $\{\rho_0\sigma_0, \rho_1\sigma_0, \allowbreak
	\sigma_0\sigma_0, \sigma_1\sigma_1\}$-provider.

	Observe that the construction from \cref{lem:casem1}
	does not use any $\{\rho_0,\rho_1,\sigma_0\}$-provider if $\rhoMax=0$.
	Hence, the exact same proof works in the case when $\rhoMax=0$
	and gives us a
	$\{\rho_0\rho_0, \rho_0\sigma_0,\allowbreak \sigma_0\sigma_0, \sigma_1\sigma_1\}$-provider.
\end{proof}

\subsubsection{\texorpdfstring%
{\boldmath $2$-Structured $(\sigma,\rho)$}
{2-Structured (Sigma, Rho)}}
We now treat the case where $(\sigma,\rho)$ is $2$-structured, but not $m$-structured for any $m$ larger than $2$. We again start with an auxiliary result.
\begin{lemma}\label{lem:casem2}
	Let $\sigMax\ge 1$. If
	a $\{\rho_0,\rho_2,\sigma_0\}$-provider and
	a $\{\rho_0\rho_0,\rho_1\rho_1,\sigma_0\sigma_0\}$-provider exist,
	then there is a $\{\rho_0,\rho_2,\sigma_0,\sigma_2\}$-provider.
\end{lemma}
\begin{proof}
	Let $(G', \{\port'\})$ be a $\{\rho_0,\rho_2,\sigma_0\}$-provider
	and assume we can construct a $\{\rho_0,\sigma_0,\sigma_2\}$-provider $(G,\{\port\})$.
	By setting $\port=\port'$,
	$(G \cup G', \{\port\})$ is a $\{\rho_0,\rho_2,\sigma_0,\sigma_2\}$-provider.
	In the remainder of this proof,
	we show how to obtain a $\{\rho_0,\sigma_0,\sigma_2\}$-provider $(G,\{\port\})$.

	Let $K$ be a complete bipartite graph with vertex partition $\{w_1,\dots,w_{\sigMax}\}$ and $\{w'_1,\dots,w'_{\sigMax}\}$.
	We remove the edge between $w_1$ and $w'_1$ (observe that $\sigMax\ge 1$, so these vertices always exist).
	For each $i\in \numb{\sigMax}$, let $w_i$ and $w'_i$ be the two portals of $\rhoMax$ attached copies of a
	$\{\rho_0\rho_0,\rho_1\rho_1,\sigma_0\sigma_0\}$-provider.
	The portal $\port$ is adjacent to $w_1$ and $w'_1$.
	Then, $G$ is the graph consisting of $K$,
	the $\sigMax\cdot\rhoMax$ attached
	$\{\rho_0\rho_0,\rho_1\rho_1,\sigma_0\sigma_0\}$-providers,
	and the portal $\port$.
	We show that $(G, \{\port\})$ is a $\{\rho_0,\sigma_0,\sigma_2\}$-provider
	by giving solutions corresponding to the three states.

	\begin{itemize}
		\item[\boldmath $\rho_0$]
		No vertex from $K$ is selected, and so, $\port$ has no selected neighbor.
		To satisfy the $\rho$-constraints of the vertices in $K$,
		all vertices $w_i$ and $w'_i$ receive $\rhoMax$ selected neighbors
		from the attached $\{\rho_0\rho_0,\rho_1\rho_1,\sigma_0\sigma_0\}$-providers
		(they are all in state $\rho_1\rho_1$).

		\item[\boldmath $\sigma_0$]
		No vertex from $K$ is selected, and so, $\port$ has no selected neighbor.
		However, now $\port$ is selected, and
		so, $w_1$ and $w'_1$ already have one selected neighbor.
		Therefore, $w_1$ and $w'_1$ receive only $\rhoMax-1$ selected neighbors
		from the attached $\{\rho_0\rho_0,\rho_1\rho_1,\sigma_0\sigma_0\}$-providers
		($\rhoMax-1$ of these are in state $\rho_1\rho_1$, and the remaining one is in state $\rho_0\rho_0$).
		For $i\ge 2$, $w_i$ and $w'_i$ receive $\rhoMax$ selected neighbors
		from the attached $\{\rho_0\rho_0,\rho_1\rho_1,\sigma_0\sigma_0\}$-providers, as before.

		\item[\boldmath $\sigma_2$]
		All vertices from $K$ are selected.
		Since $K$ is a complete bipartite graph (minus one edge),
		the vertices $w_2,\dots,w_{\sigMax}$ and $w'_2,\dots,w'_{\sigMax}$ are adjacent to $\sigMax$ selected neighbors.
		The vertices $w_1$ and $w'_1$ have $\sigMax-1$ selected neighbors in $K$ and, additionally, the selected portal $\port$ as a neighbor.
		All of the $\{\rho_0\rho_0,\rho_1\rho_1,\sigma_0\sigma_0\}$-providers are in state $\sigma_0\sigma_0$,
		and so, do not give any further selected neighbors to vertices in $K$.
		\qedhere
	\end{itemize}
\end{proof}

As a last step, we combine \cref{lem:casem2} with \cref{lem:gadgettechnical}
to obtain the final provider for the $2$-structured case.
\begin{lemma}\label{lem:mainm2}
	Let $\rhoMax\ge 1$.
	Suppose that $\mname=2$ is the maximum value
	such that $(\sigma, \rho)$ is $\mname$-structured.
	Suppose further that the elements of $\rho$ and $\sigma$ are even.
	Then, for
	$L = \{\rho_{i}  \in \rhoStates \mid i \in 2\NN \}
	\cup \{\sigma_{i}\in \sigStates \mid i \in 2\NN \}$,
	there is an $L$-provider.
\end{lemma}
\begin{proof}
	First assume that $\sigMax \ge 1$.
	Combining \Cref{lem:gadgettechnical,lem:casem2}, there is a $\{\rho_0, \rho_2, \sigma_0,\sigma_2\}$-provider $(G, \{\port\})$.
	Create $\allMax/2$ copies $(G_i, \{\port_i\})_{i\in\numb{\allMax/2}}$ of $(G,\{\port\})$.
	We set $\port_1=\dots=\port_{\allMax/2}$,
	that is, we identify all portals with the same vertex.
	It is straightforward to verify that $(G_1\cup \dots\cup G_{\allMax/2}, \{\port_1\})$ is an $L$-provider.

	If we have that $\sigMax = 0$, then we can use the same construction
	except that we use the $\{\rho_0, \rho_2,\sigma_0\}$-provider
	from \cref{lem:gadgettechnical} instead.
\end{proof}

\section{Constructing Managers}
\label{sec:manager}
The goal of this section is to establish four \encoder{A}s
with different $A \subseteq \allStates$.
Formally, we prove \cref{lem:lb:existenceOfManager}
which we restate here for convenience.
\existenceOfManager*

Before starting with the construction,
we restate the definition of a manager here. See \cref{fig:manager:example}
for an illustration of this definition.
\defEncoder*

Note that \cref{def:encoder} does not a priori rule out solutions
that do \emph{not} correspond to values $x \in A^\ell$.

The proof of \cref{lem:lb:existenceOfManager} is split into four cases,
each case corresponding to one manager.
\Cref{lem:expanding:onlyRhoStates,,lem:expanding:onlyEvenStates,,%
lem:expanding:onlySigmaStates,,lem:expanding:allStates}
handle each one of these cases
and together imply the above lemma.

For the construction of the managers, we mostly rely on the providers
introduced in \cref{sec:provider}.
In the following, we first construct the managers from the first two cases
as they follow rather directly from the constructed providers.
Then, we show how to obtain the remaining two managers
assuming a general construction for managers based on specific providers.
As a last step, we provide this general construction
in \cref{sec:expanding:basicToExpanding}.

\begin{lemma}
	\label{lem:expanding:onlyRhoStates}
	For non-empty sets $\sigma$ and $\rho$,
	there is an $\encoder{\rhoStates}$.
\end{lemma}
\begin{proof}
	Let $\ell$ be a fixed \rank of the manager.
	For each $i\in \numb{\ell}$, $G$ contains a distinguished vertex $\port_i$
	that is the portal of two independent copies of
	a $\{\rho_0,\rho_1, \ldots, \rho_{\rhoMax}\}$-provider
	(which exists by \Cref{lem:trivialrhogadget}).
	The first and second copy of the providers neighboring $\port_i$
	form the sets $\Bl_i$ and $\Br_i$, respectively.
	To each copy, we add a relational constraint whose scope are all vertices of the respective copy. This constraint ensures that the selection of vertices for each state $\rho_i$ is unique.
	This defines an $\encoder{\rhoStates}$ of rank $\ell$
	since, for each $r \in \fragment{0}{\rhoMax}$,
	both $\rho_r$ and $\rho_{\rhoMax-r}$
	are compatible with the $\{\rho_0,\rho_1, \ldots, \rho_{\rhoMax}\}$-provider.
\end{proof}

As a next step, we construct the first manager with $\sigma$ and $\rho$-states.
\begin{lemma}
	\label{lem:expanding:onlyEvenStates}
	Let $\rhoMax \ge 1$.
	Suppose that $\mname=2$ is the maximum value
	such that $(\sigma, \rho)$ is $\mname$-structured.
	Suppose further that all elements of $\rho$ and $\sigma$ are even.
	Then, there is an $\encoder{L}$ where
	\[
		L \coloneqq \{\rho_{i}   \in \rhoStates \mid i \in 2\NN \}
	         \cup \{\sigma_{i} \in \sigStates \mid i \in 2\NN \}.
	\]
\end{lemma}
\begin{proof}
	The construction is analogous to that
	in the proof of \cref{lem:expanding:onlyRhoStates},
	but we use the $L$-provider from \cref{lem:mainm2}
	instead of the $\{\rho_0,\rho_1, \ldots, \rho_{\rhoMax}\}$-provider.
	Note that $\rhoMax$ and $\sigMax$ are even,
	and thus, $\rhoMax-k$ and $\sigMax-k$ are even if $k$ is even.
\end{proof}

Before providing the last two managers, we first introduce some notation.
Let $x$ be a vector (string) of elements from $\allStates$.
Recall that, for $a\in \allStates$, $\occ{x}{a}$ denotes the number of occurrences of $a$ in $x$.
Similarly, for $A \subseteq \allStates$,
$\occ{x}{A}$ denotes the number of occurrences of elements from $A$ in $x$.
\begin{definition}
	\label{def:auxiliaryForManager}
	Let $d\ge 1$ and $\beta \in \fragment{0}{2d}$. We define a set
	\[
	L^{(2d)}_{\beta} \coloneqq \{x\in \{\rho_0, \rho_1, \sigma_0, \sigma_1\}^{2d}
	\mid \occ{x}{\sigma_1\!}\in\{0,d\}, \occ{x}{\rhoStates}\le\beta\}.
	\]
\end{definition}

Using this definition,
we can state the following general construction for a manager.
\begin{lemma}
	\label{lem:expanding:basicToExpanding}
	Suppose for some $d\ge 1$ and $\beta\in \{0,1\}$ that
	there is an $L^{(2d)}_{\beta}$-provider
	where the closed neighborhoods of the portals are disjoint.
	Then, there is an $\encoder{\sigStates}$,
	and if $\beta=1$, then there is also an $\encoder{\allStates}$.
\end{lemma}

Before proving \cref{lem:expanding:basicToExpanding} in \cref{sec:expanding:basicToExpanding},
we show how to use it to prove
\cref{lem:expanding:onlySigmaStates,,lem:expanding:allStates},
that is, the remaining two managers from \cref{lem:lb:existenceOfManager}.
For the first manager, the distinguished vertices have to always be selected.

\begin{lemma}
	\label{lem:expanding:onlySigmaStates}
	If there is $c \in \rho$ with $\sigMax \geq c \ge 1$,
	then there is an $\encoder{\sigStates}$.
\end{lemma}
\begin{proof}
	By \Cref{lem:sigmaGadget}, there is an $L$-provider where
	\[L \coloneqq \{
	x \in \{\sigma_0,\sigma_1\}^{4r}
	\mid \occ{x}{\sigma_1\!} \in \{0,2r\}
	\}.\]
	Note that this precisely corresponds to an $L^{(4r)}_0$-provider.
	Moreover, we get from the lemma that
	the closed neighborhoods of the portals are disjoint.
	Then, the claim follows by \Cref{lem:expanding:basicToExpanding}.
\end{proof}

The following manager concludes the series of constructions
by providing an \encoder{\allStates}.
\begin{lemma}
	\label{lem:expanding:allStates}
	Let $\rho \neq \{0\}$.
	Suppose that $\mname=1$ is the maximum value
	such that $(\sigma, \rho)$ is $\mname$-structured.
	Then, there is an $\encoder{\allStates}$.
\end{lemma}
\begin{proof}
	We first handle the case when $\sigMax,\rhoMax\ge 1$.
	By \cref{lem:mainm1}, there is a $Z$-provider
	with $Z=\{(\rho_0,\sigma_0), (\rho_1,\sigma_0), (\sigma_0,\sigma_0), (\sigma_1,\sigma_1)\}$.
	Moreover,
	the closed neighborhoods of both portals are disjoint.
	From this, we construct an $L^{(4)}_1$-provider as follows.
	We introduce four vertices $\port_1,\dots,\port_4$ that serve as portals of the $L^{(4)}_1$-provider.
	For each pair $(\port_i,\port_j)$ of distinct portals,
	we attach an independent copy of a $Z$-provider
	where $\port_i$ is the first portal, and $\port_j$ is the second portal vertex.
	It is straightforward to verify
	that the resulting gadget is an $L^{(4)}_1$-provider
	where the closed neighborhoods of the portals are disjoint.
	The lemma then follows from \cref{lem:expanding:basicToExpanding}.

	The next case is when $\rhoMax \ge 1$ and $\sigMax=0$, i.e., $\sigma=\{0\}$.
	By \cref{lem:gadgettechnical}, there is a $\{\rho_0,\rho_1,\sigma_0\}$-provider $J$.
	For all $i$, the blocks $\Bl_i$ and $\Br_i$
	consist of $\allMax$ copies of $J$ each
	using $\port_i$ as a portal. For each copy we introduce a relational constraint whose scope are all vertices of that respective copy. This constraint ensures that for each state there is only a unique solution.
	Since $\sigma=\{0\}$
	this finishes the construction of the $\encoder{\allStates}$ of \rank $\ell$
	and proves its correctness.

	It remains to handle the case when $\rhoMax=0$.
	If we have $\sigMax=0$,
	then we are interested only in the states $\sigma_0$ and $\rho_0$.
	In this case, the empty graph already gives an \encoder{\allStates}.
	Thus, assume $\sigMax \ge 1$.
	We claim that it suffices to have a $Z'$-provider
	with $Z'=\{(\rho_0,\sigma_0), (\sigma_0,\sigma_0), (\sigma_1,\sigma_1)\}$.
	Then, we could use the construction for $\rhoMax,\sigMax \ge 1$
	as we never need the combination $(\rho_1,\sigma_0)$.
	The desired $Z'$-provider follows from \cref{lem:mainm1}
	and finishes the proof.
\end{proof}

\subsection{Proof of \texorpdfstring{\Cref{lem:expanding:basicToExpanding}}{Lemma \ref{lem:expanding:basicToExpanding}}:
A Blueprint for Managers}
\label{sec:expanding:basicToExpanding}

Now, we turn to the proof of \cref{lem:expanding:basicToExpanding}.
That is, given an $L^{(2d)}_{\beta}$-provider for some $d\ge 1$ and $\beta\in \{0,1\}$,
we construct an $\encoder{\sigStates}$, and if $\beta=1$, then also an $\encoder{\allStates}$.

{%\begin{proof}
	\def\hell{\ell^*}
	\def\newx{x^*}
	\newcommand\Xl{X}
	\newcommand\Xr{\inv X}
	\newcommand\Fl{\mu}
	\newcommand\Fr{\inv \mu}
	\newcommand{\LSize}{\delta_L}
	\newcommand{\Fil}{F}
	\newcommand{\Fir}{\inv F}

	\begin{figure}[t]
		\centering
		\includegraphics[page=2]{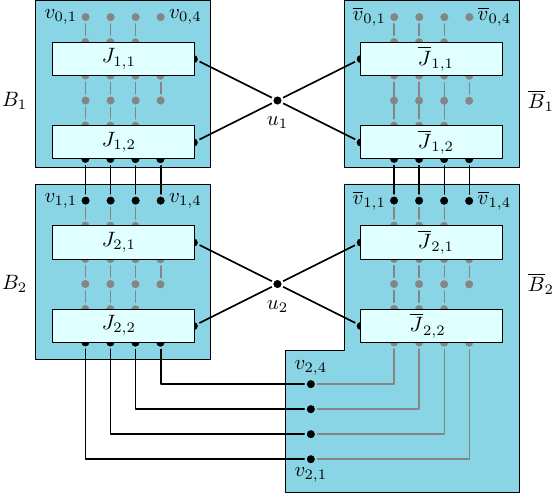}
		\caption{
		Construction of the manager from \cref{lem:expanding:basicToExpanding}.
		The big dark blue boxes represent the blocks,
		while the smaller light blue boxes represent the $L$-providers.
		The $L$-providers $\Fil_{i,z}^j,\Fir_{i,z}^j$ are not shown
		to keep the figure simple.
		}
		\label{fig:manager:example}
	\end{figure}

	To simplify notation,
	let $L=L^{(2d)}_{\beta}$ in the following.

	\paragraph{Construction of the Graph.}
	Let $\ell$ be a fixed \rank of the manager.
	For ease of notation, we omit $\ell$ from $G_\ell$ and $\Port_\ell$
	in the following, and just write $G$ and $\Port$.
	We first define the graph $G$
	with its set of distinguished vertices $\Port=\{\port_1,\dots,\port_{\ell}\}$.
	An illustration is given in \cref{fig:manager:example}.
	Afterward, we establish that it is an $\encoder{\sigStates}$
	or $\encoder{\allStates}$ for $\beta=0$ or $\beta=1$, respectively.

	Apart from the vertices in $\Port$,
	$G$ also contains $\hell-\ell$ auxiliary vertices $\port_{\ell+1},\dots,\port_{\hell}$,
	which are used to circumvent some parity issues that would otherwise prevent us from obtaining proper solutions.
	In the construction of $G$, we treat these auxiliary vertices in the same way as the vertices in $\Port$.
	However, for the encoding that we want to model (see \Cref{def:encoder}),
	we are interested only in the vertices $\port_1,\ldots, \port_\ell$
	and their selected neighbors.
	Later, we set the number $\hell$ (as a number depending on $\beta$).
	For each $i\in\numb{\hell}$,
	there are some subgraphs $\Xl_i$ and $\Xr_i$,
	which we define subsequently.
	They provide the neighbors for $\port_i$ in $G$,
	and with some minor modifications,
	these subgraphs $\Xl_i$ and $\Xr_i$ form the blocks $\Bl_i$ and $\Br_i$ of the manager $G$.

	For all $i\in \fragment{0}{\hell}$ and all $j\in\numb{d}$,
	we create vertices $v_{i,j}$ and $\inv v_{i,j}$.
	For $j\in\numb{d}$, we identify $v_{\hell,j}=\inv v_{\hell,d+1-j}$.
	Intuitively, the case $i=\hell$ is special
	as it is the connection between the $v_{i,j}$'s and the $\inv v_{i,j}$'s.
	For each $i\in\fragment{0}{\hell}$,
	the vertices $v_{i,1}, \ldots, v_{i,d}$ serve as portals of $\sigMax$ copies of the given $L$-provider.
	Let us name these providers $\Fil^1_i,\ldots, \Fil^{\sigMax}_i$.
	(We use $\Fil$ as in ``fill''
	since these providers are used to fill up the count of selected neighbors
	of their portals as far as possible.)
	Note that the $L$-provider has $2d$ portals,
	and therefore, each vertex in $v_{i,1}, \ldots, v_{i,d}$ serves as two portals for each of these $L$-providers.
	Identifying two portals of an $L$-provider is possible without forming loops or multiple edges by the assumption that the portals of the given $L$-provider are independent and have no common neighbors.
	Analogously, for each $i\in\fragment{0}{\hell}$,
	the vertices $\inv v_{i,1}, \ldots, \inv v_{i,d}$
	serve as portals of $\sigMax$ $L$-providers $\Fir^1_i,\dots,\Fir^{\sigMax}_i$.

	\newcommand{\hx}{x} % Helper vertices in $X$
	For each $i\in\numb{\hell}$, there is a subgraph $\Xl_i$ defined as follows.
	It contains a set of vertices
	$\{\hx_{i,z,j} \mid z\in\fragment{0}{\allMax}, j\in\numb{d}\}$.
	For $z\notin \{0,\allMax\}$,
	these $\hx_{i,z,j}$'s are new vertices,
	but, for all $j\in\numb{d}$, we set $\hx_{i,0,j}=v_{i-1,j}$
	and $\hx_{i,\allMax,j}=v_{i,j}$.
	Just like the $v_{i,j}$'s,
	for each $z\in\fragment{1}{\allMax-1}$,
	the vertices $\hx_{i,z,1}, \ldots, \hx_{i,z,d}$ serve as portals
	of $\sigMax$ copies $\Fil^1_{i,z},\dots,\Fil^{\sigMax}_{i,z}$
	of the given $L$-provider.
	We set $\Fil^y_{i,0}=\Fil^y_{i-1}$ and $\Fil^y_{i,\allMax}=\Fil^y_i$
	for all $y\in\numb{\sigMax}$.
	Furthermore, for all $z\in\numb{\allMax}$,
	$\Xl_i$ contains an $L$-provider $J_{i,z}$
	with portals $\port_i$, $\hx_{i,z-1,1},\dots,\hx_{i,z-1,d-1}$,
	and $\hx_{i,z,1},\dots,\hx_{i,z,d}$.

	For each $i$, we similarly define a subgraph $\Xr_i$
	with vertices $\{\inv \hx_{i,z,j} \mid z\in\fragment{0}{\allMax}, j\in\numb{d}\}$,
	and $\inv \hx_{i,0,j}=\inv v_{i-1,j}$, $\inv \hx_{i,\allMax,j}=\inv v_{i,j}$
	(for all $j\in\numb{d}$).
	Moreover, for all $z\in\numb{\allMax}$,
	$\Xr_i$ contains an $L$-provider $\inv J_{i,z}$
	with portals $\port_i$, $\inv \hx_{i,z-1,1},\dots,\inv \hx_{i,z-1,d-1}$,
	and $\inv \hx_{i,z,1},\dots,\inv \hx_{i,z,d}$.
	For all $z\in\numb{\allMax-1}$,
	we create $\sigMax$ $L$-providers
	$\Fir^1_{i,z},\dots,\Fir^{\sigMax}_{i,z}$,
	and additionally set for all $y\in\numb{\sigMax}$,
	$\Fir^y_{i,0}=\Fir^y_{i-1}$
	and $\Fir^y_{i,\allMax}=\Fir^y_{i}$.
	As a last step, we add a relational constraint
	whose scope is the union of the vertices in $\Xl_i$ with the vertex $\port_i$.
	This constraint ensures that there is exactly one solution
	for each state the vertex $\port_i$ can obtain.
	For this, we use the solution which we construct in the following.
	We apply the same to the vertices of $\Xr_i$ together with $\port_i$.
	This completes the definition of $G$.

	\paragraph{Partitioning the Graph.}
	We  use $G$ for both cases $\beta=0$ and $\beta=1$.
	We  adjust only the value of $\hell$ accordingly.
	We have to show that there is an integer $b$,
	and a partition of $G$ into blocks
	that fulfill the requirements stated in \Cref{def:encoder}.

	Let $\LSize$ denote the number of vertices of the $L$-provider (including the $2d$ portal vertices).
	We set
	\begin{equation}
	 \label{eq:block-size}
	 b \coloneqq (\hell - \ell + 1) \cdot (\sigMax + 1) \cdot (\allMax + 1) \cdot \LSize.
	\end{equation}
	For all $i\in\numb{\ell-1}$,
	the block $\Bl_i$ consists of the vertices of $\Xl_i$
	minus the ones from $\Xl_{i+1}$ (and minus $\port_i$).
	Note that the common vertices of $\Xl_i$ and $\Xl_{i+1}$
	are $v_{i,1},\dots,v_{i,d}$
	together with the $L$-providers $\Fil^1_i,\ldots, \Fil^{\sigMax}_i$.
	Likewise, the block $\Br_i$ consists of the vertices of $\Xr_i$
	minus the vertices from $\Xr_{i+1}$.
	The final blocks $\Bl_{\ell}$ and $\Br_{\ell}$
	hold all the remaining vertices of $V(G)\setminus \Port$,
	that is, $\Bl_{\ell}$ consists of all $\Xl_i$
	with $i\in \fragment{\ell}{\hell}$
	minus the vertices from $\Xl_{\hell} \cap \Xr_{\hell}$.
	(These are $v_{\hell,1}= \inv v_{\hell,d},\dots,v_{\hell,d}=\inv v_{\hell,1}$
	together with the $L$-providers $\Fil^1_{\hell},\ldots, \Fil^{\sigMax}_{\hell}$.)
	Then, $\Br_{\ell}$ consists of all $\Xr_i$
	with $i\in \fragment{\ell}{\hell}$.

	It is straightforward to verify that these blocks form a partition of the vertices of $G$.
	Each set $\Xl_i$ contains $\sigMax \cdot (\allMax + 1)$ many $L$-providers of the form $\Fil^{s}_{i,z}$, and $\allMax$ many $L$-providers of the form $J_{i,z}$.
	So, in total, this gives at most $(\sigMax + 1) \cdot (\allMax + 1)$ many $L$-providers.
	Since each vertex $x_{i,z,j}$ is a portal of at least one of those $L$-providers, we conclude that $\Xl_i$ contains at most $(\sigMax + 1) \cdot (\allMax + 1) \cdot \LSize$ many vertices.
	The same bound also holds for all subgraphs $\Xr_i$ for all $i\in \numb{\hell}$.
	Since each block $\Bl_i$ or $\Br_i$, $i\in \numb{\ell}$, contains vertices from at most $(\hell - \ell + 1)$ many subgraphs $\Xl_i$ or $\Xr_i$, it follows that $|\Bl_i| \leq b$ and $|\Br_i| \leq b$ for all $i\in \numb{\ell}$.
	Finally, by construction, it is immediately clear that $N(\port_i) \subseteq \Bl_i \cup \Br_i$ for all $i \in \numb{\ell}$,
	and there are edges only between $\Bl_i$ and $\Bl_{i+1}$, $\Br_i$ and $\Br_{i+1}$, for each $i \in \numb{\ell-1}$, and $\Bl_\ell$ and $\Br_\ell$.

	\paragraph{Constructing a Solution.}
	To show that the corresponding graph $G$
	together with the distinguished vertices $\Port=\{\port_1,\dots,\port_{\ell}\}$
	is an $\encoder{\allStates}$ (for $\beta=1$)
	or an $\encoder{\sigStates}$ (for $\beta=0$)
	of \rank $\ell$, we additionally need to prove the second property from \Cref{def:encoder}.

	For a given $x\in\allStates^\ell$,\footnote{
	For $\beta=0$, we restrict ourselves to strings from $\sigStates^\ell$.}
	we define an $\newx \in \allStates^{\hell}$
	with $\newx\position{i}=x\position{i}$ for all $i\in\numb{\ell}$,
	while the remaining entries are set depending on $\beta$ (and $x$).
	We iteratively construct (partial) solutions $S_{i,z}$
	for all $i\in \numb{\hell}$ and $z\in \fragment{0}{\allMax}$
	such that $S_x \deff S_{\hell,\allMax}$ corresponds to the final solution.
	Based on these $S_{i,z}$, we define two functions
	$\Fl,\Fr\from\numb{\hell}\times\fragment{0}{\allMax}\to\fragment{0}{d-1}$
	such that
	$\Fl(i,z)=k$ if and only if
	the $k$ vertices $\hx_{i,z,1},\dots,\hx_{i,z,k}$
	need (exactly) one more neighbor
	in $S_{i,z}$,
	respectively for $\Fr(i,z)=k$
	and the vertices $\inv \hx_{i,z,1},\dots,\inv \hx_{i,z,k}$.
	Due to the construction of $G$,
	we set $S_{i,\allMax}=S_{i+1,0}$,
	and therefore, $\Fl(i,\allMax)=\Fl(i+1,0)$ and $\Fr(i,\allMax)=\Fr(i+1,0)$
	for all $i\in\numb{\hell-1}$.
	We later set $\hell$ and $\newx$
	such that the partial solutions maintain the following invariant
	for all $i\in\numb{\hell}$:
	\begin{equation}
		\label{eq:bToExanding:invariant}
		\Fl(i,\allMax)+\Fr(i,\allMax)
		\equiv_d \sigMax \cdot \occ{\newx\fragment{1}{i}}{\sigStates}
		\quad
		\text{and}
		\quad
		\Fl(1,0)=\Fr(1,0)=0
	\end{equation}

	$S_{1,0}$ consists of the vertices $\hx_{1,0,j}$ and $\inv \hx_{1,0,j}$
	for all $j\in\numb{d}$.
	Moreover, the $L$-providers $\Fil^1_{1,0},\dots,\Fil^{\sigMax}_{1,0}$
	and $\Fir^1_{1,0},\dots,\Fir^{\sigMax}_{1,0}$
	give each of these vertices $\sigMax$ neighbors.
	We directly get $\Fl(1,0)=\Fr(1,0)=0$ which satisfies the invariant.

	Now, for each $i$ from $1$ to $\hell$,
	we iterate over all $z$ from $1$ to $\allMax$,
	and apply the following selection process.
	Unless mentioned otherwise,
	the portals of the $L$-providers get state $\sigma_0$.
	\begin{itemize}
		\item
		$S_{i,z}$ consists of all selected vertices from $S_{i,z-1}$.
		If $\newx\position{i} \in \sigStates$,
		then we select vertex $\port_i$.
		Moreover,
		we select the vertices $\hx_{i,z,j}$ and $\inv \hx_{i,z,j}$
		for all $j\in\numb{d}$.
		We use the $L$-providers $\Fil^y_{i,z},\Fir^y_{i,z}$
		for all $y\in\numb{\sigMax-1}$
		to add $\sigMax-1$ neighbors to each of these vertices.

		\item
		Selection process for $\Bl_i$ if $\newx\position{i} =\sigma_s$.
		\begin{itemize}
			\item
			If $\port_i$ has less than $s$ neighbors in $\Bl_i\cap S_{i,z-1}$,
			then we give it one more neighbor in $S_{i,z}$.
			The $L$-provider $J_{i,z}$ gives state $\sigma_1$ to $\port_i$
			and each of
			$\hx_{i,z-1,1},\dots,\hx_{i,z-1,\Fl(i,z-1)}$
			and $\hx_{i,z,\Fl(i,z-1)+2},\dots,\hx_{i,z,d}$.
			Hence, the $\Fl(i,z-1)+1$ vertices $\hx_{i,z,1},\dots,\hx_{i,z,\Fl(i,z-1)+1}$ need one more neighbor.

			If $\Fl(i-1,z)+1=d$,
			then we use the additional $L$-provider $\Fil^{\sigMax}_{i,z}$
			to give every vertex one additional neighbor.
			Therefore,
			$\Fl(i,z)=\Fl(i,z-1)+1 \mod d$.

			\item
			If $\port_i$ has $s$ neighbors in $\Fl_i\cap S_{i,z-1}$,
			then we do not give it one more neighbor in $S_{i,z}$.
			The $L$-provider $J_{i,z}$ gives state $\sigma_0$ to $\port_i$,
			and gives state $\sigma_1$ to each of
			$\hx_{i,z-1,1},\dots,\hx_{i,z-1,\Fl(i,z-1)}$
			and $\hx_{i,z,\Fl(i,z-1)+1},\dots,\hx_{i,z,d}$.
			Hence,
			$\Fl(i,z)=\Fl(i,z-1)$.
		\end{itemize}

		\item
		Selection process for $\Br_i$ if $\newx\position{i} =\sigma_s$.
		\begin{itemize}
			\item
			If $\port_i$ has less than $\sigMax-s$ neighbors in $\Br_i\cap S_{i,z-1}$,
			then we give it one more neighbor in $S_{i,z}$.
			The $L$-provider $\inv J_{i,z}$ gives state $\sigma_1$ to $\port_i$
			and each of
			$\inv \hx_{i,z-1,1},\dots,\inv \hx_{i,z-1,\Fr(i,z-1)}$
			and $\inv \hx_{i,z,\Fr(i,z-1)+2},\dots,\inv \hx_{i,z,d}$.
			Hence, the $\Fr(i,z-1)+1$ vertices $\inv \hx_{i,z,1},\dots,\inv \hx_{i,z,\Fr(i,z-1)+1}$ need one more neighbor.

			If $\Fr(i-1,z)+1=d$,
			then we use the additional $L$-provider $\Fir^{\sigMax}_{i,z}$
			to give every vertex one additional neighbor.
			Therefore,
			$\Fr(i,z)=\Fr(i,z-1)+1 \mod d$.

			\item
			If $\port_i$ has $\sigMax-s$ neighbors in $\Br_i\cap S_{i,z-1}$,
			then we do not give it one more neighbor in $S_{i,z}$.
			The $L$-provider $\inv J_{i,z}$ gives state $\sigma_0$ to $\port_i$,
			and gives state $\sigma_1$ to each of
			$\inv \hx_{i,z-1,1},\dots,\inv \hx_{i,z-1,\Fr(i,z-1)}$
			and $\inv \hx_{i,z,\Fr(i,z-1)+1},\dots,\inv \hx_{i,z,d}$.
			Hence,
			$\Fr(i,z)=\Fr(i,z-1)$.
		\end{itemize}

		\item
		Selection process for $\Bl_i$ if $\newx\position{i} =\rho_r$.
		(Note that this case does not occur when $\beta=0$)
		\begin{itemize}
			\item
			If $\port_i$ has less than $r$ neighbors in $\Bl_i\cap S_{i,z-1}$,
			then we give it one more neighbor in $S_{i,z}$.
			The $L$-provider $J_{i,z}$ gives state $\rho_1$ to $\port_i$,
			and gives state $\sigma_1$ to each of
			$\hx_{i,z-1,1},\dots,\hx_{i,z-1,\Fl(i,z-1)}$
			and $\hx_{i,z,\Fl(i,z-1)+1},\dots,\hx_{i,z,d}$.
			Since the $\Fl(i,z-1)$ vertices $\hx_{i,z,1},\dots,\hx_{i,z,\Fl(i,z-1)}$ need one more neighbor,
			$\Fl(i,z)=\Fl(i,z-1)$.

			\item
			If $\port_i$ has $r$ neighbors in $\Bl_i\cap S_{i,z-1}$,
			then we do not give it one more neighbor in $S_{i,z}$.
			The $L$-provider $J_{i,z}$ gives state $\rho_0$ to $\port_i$,
			and gives state $\sigma_1$ to each of
			$\hx_{i,z-1,1},\dots,\hx_{i,z-1,\Fl(i,z-1)}$
			and $\hx_{i,z,\Fl(i,z-1)+1},\dots,\hx_{i,z,d}$.
			Hence,
			$\Fl(i,z)=\Fl(i,z-1)$.
		\end{itemize}

		\item
		Selection process for $\Br_i$ if $\newx\position{i} =\rho_r$.
		(Note that this case does not occur when $\beta=0$)
		\begin{itemize}
			\item
			If $\port_i$ has less than $\rhoMax-r$ neighbors in $\Br_i\cap S_{i,z-1}$,
			then we give it one more neighbor in $S_{i,z}$.
			The $L$-provider $\inv J_{i,z}$ gives state $\rho_1$ to $\port_i$,
			and gives state $\sigma_1$ to each of
			$\inv \hx_{i,z-1,1},\dots,\inv \hx_{i,z-1,\Fr(i,z-1)}$
			and $\inv \hx_{i,z,\Fr(i,z-1)+1},\dots,\inv \hx_{i,z,d}$.
			Since the $\Fr(i,z-1)$ vertices $\inv \hx_{i,z,1},\dots,\inv \hx_{i,z,\Fr(i,z-1)}$ need one more neighbor,
			$\Fr(i,z)=\Fr(i,z-1)$.

			\item
			If $\port_i$ has $\rhoMax-r$ neighbors in $\Br_i\cap S_{i,z-1}$,
			then we do not give it one more neighbor in $S_{i,z}$.
			The $L$-provider $\inv J_{i,z}$ gives state $\rho_0$ to $\port_i$,
			and gives state $\sigma_0$ to each of
			$\inv \hx_{i,z-1,1},\dots,\inv \hx_{i,z-1,\Fr(i,z-1)}$
			and $\inv \hx_{i,z,\Fr(i,z-1)+1},\dots,\inv \hx_{i,z,d}$.
			Hence,
			$\Fr(i,z)=\Fr(i,z-1)$.
		\end{itemize}
	\end{itemize}
	\paragraph{Correctness of the Solution.}
	We first show that these steps preserve the invariant in \Cref{eq:bToExanding:invariant}.
	The claim immediately follows if $\newx\position{i}=\rho_r$
	because, in this case, $\Fl(i,z-1)=\Fl(i,z)$ for all $z\in\numb{\allMax}$.
	Likewise, we get $\Fr(i,z-1)=\Fr(i,z)$.

	For the case $\newx\position{i}=\sigma_s$,
	we get that $\Fl(i,z)\equiv_d \Fl(i,0)+z$ for all $z\in\numb{s}$,
	and $\Fl(i,z)=\Fl(i,z-1)$ for all $z\in\fragmentoc{s}{\allMax}$.
	Conversely, we get
	$\Fr(i,z)\equiv_d \Fr(i,0)+z$ for all $z\in\numb{\sigMax-s}$,
	and $\Fr(i,z)=\Fr(i,z-1)$ for all $z\in\fragmentoc{\sigMax-s}{\allMax}$.
	This directly yields
	\begin{align*}
		\Fl(i,\allMax)+\Fr(i,\allMax)
		&\equiv_d \Fl(i,0)+\Fr(i,0) + s + (\sigMax-s) \\
		&\equiv_d \occ{\newx\fragment{1}{i-1}}{\sigStates} \cdot \sigMax + \sigMax \\
		&\equiv_d \occ{\newx\fragment{1}{i}}{\sigStates} \cdot\sigMax,
	\end{align*}
	proving that the invariant is preserved.

	\paragraph{Setting the Parameters.}
	Assume we can set $\hell$ and $\newx\fragmentoc{\ell}{\hell}$
	such that the solution satisfies
	$\Fl(\hell,\allMax)+\Fr(\hell,\allMax)\in\{0,d\}$.
	We then show that $S_{\hell,\allMax}$ is indeed a valid solution.
	From $\Fl(\hell,\allMax)=k$, we know
	that the vertices $v_{\hell,1},\dots,v_{\hell,k}$ need one more neighbor.
	From $\Fr(\hell,\allMax)=k'$, we know
	that the vertices $\inv v_{\hell,1},\dots,\inv v_{\hell,k'}$ need one more neighbor.
	Recall that $v_{\hell,j}=\inv v_{\hell,d+1-j}$ for all $j$.
	If $\Fl(\hell,\allMax)+\Fr(\hell,\allMax)=d$,
	then $k'=d-k$ and the vertices $v_{\hell,1},\dots,v_{\hell,k}$
	get a neighbor from $\inv J_{\hell,\allMax}$.
	Conversely, the vertices $\inv v_{\hell,1}=v_{\hell,d},\dots,\inv v_{\hell,d-k}=v_{\hell,k+1}$
	get a neighbor from $J_{\hell,\allMax}$.
	Thus, none of these vertices need an additional neighbor.

	If $\Fl(\hell,\allMax)+\Fr(\hell,\allMax)=0$,
	then the $L$-provider $\Fil^{\sigMax}_{\hell}=\Fir^{\sigMax}_{\hell}$
	already provided one additional neighbor for these vertices.

	It remains to set $\hell$ and $\newx\fragmentoc{\ell}{\hell}$
	depending on $\beta$
	such that the claim holds.
	\begin{itemize}
		\item
		$\beta=1$:
		We set $\hell=\ell+d$.
		Let $\newx\position{i}=\sigma_0$
		for all $i\in\fragmentoc{\ell}{\ell+\ell'}$
		and $\newx\position{i}=\rho_0$
		for all $i\in\fragmentoc{\ell+\ell'}{\hell}$,
		where $\ell' = d-(\occ{x}{\sigStates} \mod d)$.
		From the invariant in \Cref{eq:bToExanding:invariant}, we get:
		\begin{align*}
			\Fl(\hell,\allMax)+\Fr(\hell,\allMax)
			&\equiv_d \sigMax \cdot \occ{\newx}{\sigStates} \\
			&\equiv_d \sigMax \cdot (\occ{x}{\sigStates} + \occ{\newx\fragmentoc{\ell}{\hell}}{\sigStates}) \\
			&\equiv_d \sigMax \cdot (\occ{x}{\sigStates} + \ell') \\
			&\equiv_d \sigMax \cdot (\occ{x}{\sigStates} + d-(\occ{x}{\sigStates} \mod d))
			 \equiv_d 0
		\end{align*}

		\item
		$\beta=0$:
		We set $\hell=\ell+d-(\ell \mod d)$
		and $\newx\position{i}=\sigma_0$ for all $i\in\fragmentoc{\ell}{\hell}$.
		By the invariant in \Cref{eq:bToExanding:invariant}
		and the fact that $\newx\in \sigStates^{\hell}$, we get:
		\begin{align*}
			\Fl(\hell,\allMax)+\Fr(\hell,\allMax)
			&\equiv_d \sigMax \cdot \occ{\newx}{\sigStates} \\
			&\equiv_d \sigMax \cdot \hell \\
			&\equiv_d \sigMax \cdot (\ell + d-(\ell \mod d))
			 \equiv_d 0
			 \qedhere
		\end{align*}
	\end{itemize}
	We complete the proof of \cref{lem:expanding:basicToExpanding} by observing that
    $\hell - \ell + 1 \leq d + 1$ in both cases, which means that $b$ (see Equation
    \eqref{eq:block-size}) depends only on $\sigMax$ and $\rhoMax$.

}%\end{proof}

\section{Lower Bound for the Problem with Relations}
\label{sec:high-level}\label{sec:LBforRelations}
In \cref{sec:high-level:decision}, we prove the intermediate lower bound for the decision version,
that is, we prove \cref{lem:lowerBoundWhenHavingSuitableGadget}.
We achieve this by giving a reduction from SAT to \srDomSetRel
assuming a suitable \encoder{A} is given.
In \cref{sec:high-level:counting}, we reuse this construction to show the lower bound for the counting version,
that is, \cref{lem:count:lowerBoundWhenHavingSuitableGadget}.
As the construction for the decision version is only for the case
when $\sigma$ and $\rho$ are finite or simple cofinite,
we provide some modifications and extensions
(exploiting properties of the counting version)
such that the sets might also be cofinite without being simple cofinite.

Observe that each $\encoder{A}$ from
\cref{lem:expanding:onlyRhoStates,lem:expanding:onlyEvenStates,lem:expanding:onlySigmaStates,lem:expanding:allStates},
satisfies our conditions on $A$ from the lower bounds,
that is, $A$ is closed under the inverse with respect to $\sigma,\rho$
(see \cref{def:inverseOfStates} for the definition of the inverse).
Hence, we can use them in the following constructions.

\subsection{\boldmath Proof of
\texorpdfstring{\cref{lem:lowerBoundWhenHavingSuitableGadget}}
{Lemma \ref{lem:lowerBoundWhenHavingSuitableGadget}}:
Lower bound for \texorpdfstring{\srDomSetRel}
{(Sigma,Rho)-Dominating Set with Relations}}
\label{sec:high-level:decision}

In this section we prove \cref{lem:lowerBoundWhenHavingSuitableGadget},
which we restate here for convenience. We will use a
reduction from SAT to \srDomSetRel that keeps
the pathwidth low.
For this, we follow previous approaches as in \cite{LokshtanovMS18,CurticapeanM16,MarxSS21}.
Recall the relations defined in \cref{def:hammingWeightOneAndEquality}.

\lowerBoundForDecDomSetRel*

	Let $\phi$ be a CNF-SAT formula with $n$ variables $x_1,\dots,x_n$ and $m$ clauses $C_1,\dots,C_m$.
	We first define a corresponding instance $G_\phi$ of \srDomSetRel.
	Afterward, we establish necessary properties of $G_\phi$.
	Finally, we put the pieces together to show that a faster algorithm
	for $\srDomSetRel$ gives an algorithm for SAT that would contradict the SETH.

\subsubsection*{\boldmath Construction of $G_\phi$}\label{sec:construction}
\def\numGroups{t} % number of groups of variables
\def\sizeGroup{q} % number of variables per group
\def\numVertPG{g} % number of information-vertices per group
\def\numNeighb{T} % number of neighbors of the information-vertices

We group the variables of $\phi$
into $\numGroups \deff \ceil{n/\sizeGroup}$ groups $F_1,\dots,F_{\numGroups}$ of $\sizeGroup$ variables each,
where $\sizeGroup$ is chosen later.
Later, we also choose some $\numVertPG$ with $2^{\sizeGroup} \le \abs{A}^{\numVertPG}$.

We now define a corresponding graph $G_\phi$ that serves as an instance of \srDomSetRel.
The construction is illustrated in
\Cref{fig:lower:construction}.

In the following, the indices we use
are superscripts if they refer to ``columns'' of the construction,
and they are subscripts if they refer to ``rows'' of the construction.
Whenever we say that some set of vertices $V'$ is subject to a relation $R$,
we mean that there is a relational constraint $C$ with $\scope(C)=V'$ and $\rel(C)=R$.
In a separate step, we show how to remove these relational constraints.
Here is the construction of $G_\phi$:

\begin{itemize}
  \item
  Introduce vertices $w_{i,\ell}^j$
  for all $i \in\numb{\numGroups}$,
  $\ell \in\numb{\numVertPG}$,
  and $j\in\fragment{0}{m}$.
  \item
  Introduce vertices $c_{i}^j$
  for all $i\in\fragment{0}{\numGroups}$ and $j\in\numb{m}$.
  \item
  For all $j\in\fragment{0}{m}$,
  we create a copy $J^j$ of the given $\encoder{A}$ of \rank $\numGroups \numVertPG$
  for which $\{ w_{i,\ell}^j \mid i \in\numb{\numGroups},
  \ell\in\numb{\numVertPG}\}$ act as distinguished vertices.
  Let $\Bl_{i,\ell}^j$ and $\Br_{i,\ell}^j$
  be the corresponding blocks.
  \item For each $w_{i,\ell}^j$,
  $N_{\Bl}(w_{i,\ell}^j)$ and $N_{\Br}(w_{i,\ell}^j)$
  are the neighborhoods of $w_{i,\ell}^j$
  in $\Bl_{i,\ell}^j$ and $\Br_{i,\ell}^j$, respectively.
  \item
  For all $i\in\fragment{0}{\numGroups}$ and $j\in\numb{m}$,
  there is a $\{\sigma_{\sigMin},\rho_{\rhoMin}\}$-provider $Q_i^j$
  (exists by \cref{lem:fillinggadget})
  with $c_i^j$ as portal.
  We add a relation to the vertices of $Q_i^j$ and $c_i^j$
  which ensures that the provider only has two solutions,
  namely one unique solution corresponding to the state $\sigma_{\sigMin}$,
  and one unique solution for the state $\rho_{\rhoMin}$.
  \item
  For all $j\in\numb{m}$, $c_0^j$ is additionally subject to a relation $\HWeq[1]{0}$,
  and $c_t^j$ is additionally subject to $\HWeq[1]{1}$.
  \item
  For each $i\in\numb{\numGroups}$ and $j\in\numb{m}$,
  the set of vertices
  \[ Z_i^j\coloneqq \{c_{i-1}^j,c_i^j\} \cup
  \bigcup_{\ell \in\numb{\numVertPG}} \bigl(
    w_{i,\ell}^{j-1} \cup N_{\Br}(w_{i,\ell}^{j-1})
    \cup w_{i,\ell}^j \cup N_{\Bl}(w_{i,\ell}^j)
  \bigr)
  \]
  is subject to a relation $R_i^j$, which we define in a moment.
  \item
  Similarly, for each $i\in\numb{\numGroups}$,
  \[ Z_i^0\coloneqq
  \bigcup_{\ell \in\numb{\numVertPG}} \bigl(
    w_{i,\ell}^1 \cup N_{\Bl}(w_{i,\ell}^1)
  \bigr)
  \quad
  \text{and}
  \quad
  Z_i^{m+1}\coloneqq
  \bigcup_{\ell \in\numb{\numVertPG}} \bigl(
    w_{i,\ell}^{m} \cup N_{\Br}(w_{i,\ell}^{m})
  \bigr),
  \]
  are subject to relations $R_i^0$
  and $R_i^{m+1}$, respectively.
\end{itemize}
This completes the definition of $G_\phi$.
In order to define the relations $R_i^j$, we need some more notation.
For each group of variables $F_i$, there are at most $2^{\sizeGroup}$ corresponding partial assignments. We encode these assignments by states from $A^{\numVertPG}$.
For technical reasons, instead of using all states,
we only use those states that have a certain property.
For now it suffices to think of this property as having a certain weight.

To this end, for a vector \(x \in \allStates^n\),
we define the \emph{weight-vector of \(x\)}
as \(\degvec{x} \in \ZZ_{\geq 0}^n\)
with \( \degvec{x}\vposition{i} \deff c \)
where \( x\position{i} \in \{\sigma_c, \rho_c \}\).
Moreover,
the weight $\wt(x)$ of a vector $x \in \allStates^{\numVertPG}$ is defined as
the $1$-norm of $\degvec{x}$.
Formally, we define
\[
  \wt(x) \deff \sum_{i=1}^{\numVertPG} \degvec{x}\position{i}.
\]
Let $\max A=\max\{i \mid \sigma_i \in A \lor \rho_i \in A\}$.
Observe that the vectors in $A^g$ can have at most $\numVertPG\cdot\max A+1$ possible weights between $0$ and $\numVertPG\cdot\max A$.
There are $\abs{A^{\numVertPG}}$ vectors in total.

% The encodings with the same weight:
\newcommand{\Encsvar}{\mathfrak A}
% The encodings which are actually used
\newcommand{\Encs}{\widehat{ \mathfrak A}}

\begin{claim}
  There is a weight $w$ and a set $\Encsvar \subseteq A^{\numVertPG}$ such that
  $\abs{\Encsvar} \ge
  \lceil{\abs{A^{\numVertPG}}} / (\numVertPG\cdot\max A+1)\rceil$,
  and, for all $x \in \Encsvar$, we have $\wt(x) = w$.
\end{claim}
\begin{claimproof}
  If $\numVertPG\cdot\max A+1 \ge \abs{A^{\numVertPG}}$,
  then we only need that there is one weight which appears at least once.
  This is true as $A$ is non-empty.
  Otherwise,
  there are more vectors than weights, and
  $\numVertPG\cdot\max A+1 \le \abs{A^{\numVertPG}}$.
  The pigeonhole principle provides some weight $w$
  such that at least the claimed fraction of vectors has weight $w$.
\end{claimproof}

Recall the definition of the inverse of a state from \Cref{def:inverseOfStates}.
Using this definition,
we define a new set $\Encs \subseteq \allStates^g$ as the inverse of $\Encsvar$
with respect to $\sigma$ and $\rho$ if the sets $\sigma$ and $\rho$ are simple cofinite, respectively.
We initially set $\Encs = \Encsvar$.
\begin{itemize}
  \item
  If $\sigma$ is simple cofinite,
  then update $\Encs \deff \inverse[\sigma]{\Encs}$.
  \item
  If $\rho$ is simple cofinite,
  then update $\Encs \deff \inverse[\rho]{\Encs}$.
\end{itemize}
Since $A$ is closed under inversion with respect to $\sigma,\rho$,
we get $\Encs \subseteq A^g$.
Moreover, as the inverse is a bijective function,
we have $\abs{\Encs}=\abs{\Encsvar}$.

We set $\sizeGroup=\floor{\numVertPG \log \abs{A} - \log(\numVertPG \cdot \max A +1)}$
such that $2^{\sizeGroup} \le \abs{A}^{\numVertPG} / (\numVertPG\cdot\max A + 1) \le \abs{\Encs}$,
that is,
even when using vectors from $\Encs$ for the encodings,
we can still encode all partial assignments.
As a last step, we choose an (injective) mapping
$e \colon 2^{\sizeGroup} \to \Encs$ to fix the encoding.

Given a selection $S$ of vertices from $G_\phi$, for each $i \in\numb{\numGroups}$,
$\ell \in\numb{\numVertPG}$, and $j\in\fragment{0}{m}$,
we define two states $a^j_{i,\ell}$ and $\inv a^j_{i,\ell}$ as follows.
Let $T^j_{i,\ell}\coloneqq S\cap N_{\Bl}(w^{j}_{i,\ell})$
be the number of selected neighbors of $w^{j}_{i,\ell}$ in $\Bl^j_{i,\ell}$,
and let $\inv T^j_{i,\ell}\coloneqq S\cap N_{\Br}(w^{j}_{i,\ell})$
be the number of selected neighbors of $w^{j}_{i,\ell}$ in $\Br^j_{i,\ell}$.
\begin{equation}\label{equ:encodedstates}
	\begin{aligned}
		&\text{If $w^{j}_{i,\ell}$ is selected, then we set $a^j_{i,\ell}=\sigma_{T^j_{i,\ell}}$ and $\inv a^j_{i,\ell}=\sigma_{\inv T^j_{i,\ell}}$.} \\
		&\text{If $w^{j}_{i,\ell}$ is \emph{not} selected, then we set $a^j_{i,\ell}=\rho_{T^j_{i,\ell}}$ and $\inv a^j_{i,\ell}=\rho_{\inv T^j_{i,\ell}}$.}
	\end{aligned}
\end{equation}

Note that, for each $i\in\numb{\numGroups}$ and $j\in\numb{m}$,
a selection of vertices $S^j_i$ from $Z_i^j$
determines the states $a^j_{i,\ell}$ and $\inv a^{j-1}_{i,\ell}$
for each $\ell \in\numb{\numVertPG}$.

Using these states, we now have everything in place to conveniently define the relations $R^j_i$.
For each $i\in\numb{\numGroups}$ and $j\in\numb{m}$,
$S_i^j$ is in $R_i^j$ if and only if all of the following hold:

\begin{itemize}
	\item
  $a^j_{i,1}\ldots a^j_{i,\numVertPG}$ is in $\Encs \subseteq A^{\numVertPG}$,
  and it is the encoding $e(\pi_i)$ of a partial assignment $\pi_i$ of the group of variables $F_i$.
	\item For each $\ell\in \numb{\numVertPG}$, the state $\inv a^{j-1}_{i,\ell}$ complements $a^j_{i,\ell}$ in the sense that
  $\inv a^{j-1}_{i,\ell} = \inverse[\sigma,\rho]{a^j_{i,\ell}}$.
	Note that therefore $w^j_{i,\ell}$ is selected
  if and only if $w^{j-1}_{i,\ell}$ is selected.
	\item If the vertex $c_{i-1}^{j}$ is selected,
	then the vertex $c_{i}^j$ must also be selected.
	\item If $c_{i-1}^{j}$ is not selected, then
	$c_i^j$ is selected
	if and only if
	the encoded assignment $\pi_i$ satisfies the clause $C_{j}$.
\end{itemize}

Similarly, we define $R_i^0$ and $R_i^{m+1}$ for each $i\in\numb{\numGroups}$.
A selection of vertices $S^0_i$ from $Z^0_i$
determines, for each $\ell \in\numb{\numVertPG}$, the state $a^0_{i,\ell}$; and a selection of vertices $S^m_i$ from $Z^m_i$
determines, for each $\ell \in\numb{\numVertPG}$, the state $\inv a^m_{i,\ell}$
where we follow the above procedure.

$S_i^0$ is in $R_i^0$ if and only if
\begin{itemize}
	\item $a^0_{i,1}\ldots a^0_{i,\numVertPG}$ is in $\Encs \subseteq A^{\numVertPG}$, and it is the encoding $e(\pi_i)$ of a partial assignment $\pi_i$ of the group of variables $F_i$.
\end{itemize}

In order to define $R_i^{m+1}$, we define, for each $\ell\in \numb{\numVertPG}$,
the auxiliary state $a^{m+1}_{i,\ell}$ with
$a^{m+1}_{i,\ell} \deff \inverse[\sigma,\rho]{\inv a^{m}_{i,\ell}}$.
Then, $S_i^{m+1}$ is in $R_i^{m+1}$ if and only if
\begin{itemize}
	\item $a^{m+1}_{i,1}\ldots a^{m+1}_{i,\numVertPG}$ is in $\Encs \subseteq A^{\numVertPG}$,
  and it is the encoding $e(\pi_i)$ of a partial assignment $\pi_i$ of the group of variables $F_i$.
\end{itemize}

\begin{figure}
  \includegraphics[width=\linewidth]{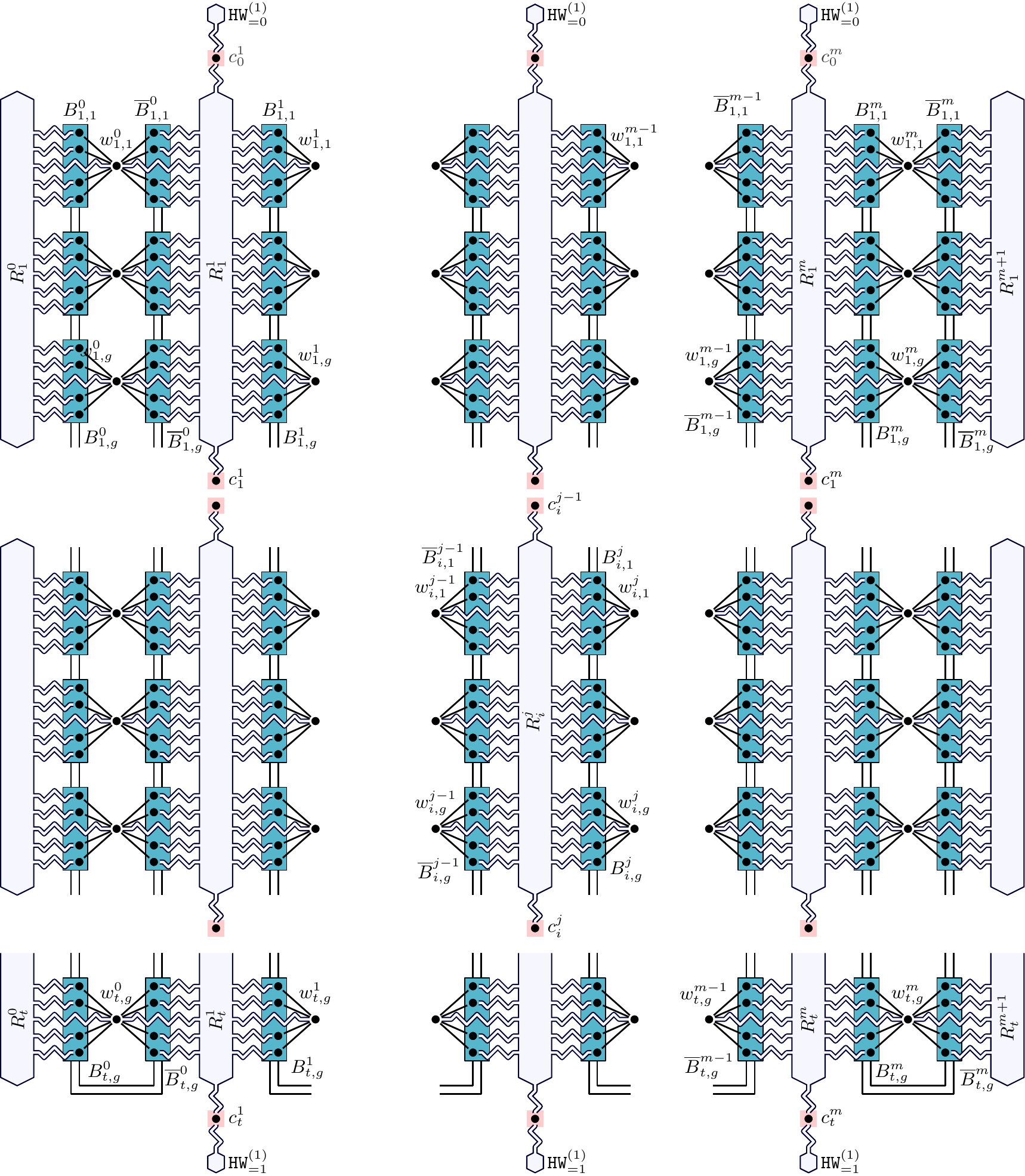}
  \caption{The construction in \cref{sec:high-level:decision}
  for the lower bound for the problem with relations.
  Vertices are shown by dots,
  $\{\sigma_{\sigMin},\rho_{\rhoMin}\}$-providers are shown in red,
  the blocks from the $\encoder{A}$ are shown in blue,
  and relations by hexagons.
  }
  \label{fig:lower:construction}
\end{figure}

\subsubsection*{\boldmath Properties of $G_\phi$}
\label{sec:high-level:decision:correctness}

\begin{claim}
  If $\phi$ is satisfiable,
  then \srDomSetRel has a solution on $G_\phi$.
\end{claim}
\begin{claimproof}
  Let $\pi$ be a satisfying assignment for $\phi$, and let $\pi_1, \ldots, \pi_{\numGroups}$ be the corresponding partial assignments to the variable groups $F_1, \ldots, F_{\numGroups}$.
  For each $i\in\numb{\numGroups}$,
  $e_i\coloneqq e(\pi_i) \in \Encs \subseteq A^{\numVertPG} \subseteq \allStates^{\numVertPG}$
  denotes the corresponding encoding.
  So, these encodings as a whole form an element $x\in A^{\numGroups \numVertPG}$.
  By the definition of an $\encoder{A}$ of \rank $\numVertPG \numGroups$,
  for each $x\in A^{\numGroups \numVertPG}$,
  there is a unique solution $S_x$ that manages $x$.

  We now define a selection $S$ of vertices of the graph $G_\phi$,
  and afterward show that $S$ is a solution for this instance of \srDomSetRel.
   For each $j\in \numb{m}$, the following vertices are selected:
   \begin{itemize}
   	\item From the $\encoder{A}$ $\numVertPG \numGroups$ $J^j$ of \rank,
    we select vertices according to the solution $S_x$.
   	\item Let $i^*\in\numb{\numGroups}$ be the smallest index
   	such that $F_{i^*}$ contains a variable
    that satisfies the clause $C_{j}$ under the assignment $\pi$.
   	A vertex $c_{i}^j$ is selected if and only if $i \ge i^*$.
   	\item We lift the selection
   	to the $\{\sigma_{\sigMin},\rho_{\rhoMin}\}$-providers:
    if $c_i^j$ is selected, then we select vertices according to
    the $\sigma_{\sigMin}$-state of the attached $\{\sigma_{\sigMin},\rho_{\rhoMin}\}$-provider,
    and otherwise, according to the $\rho_{\rhoMin}$-state.
   \end{itemize}

   In order to show that this selection of vertices is a feasible solution of \srDomSetRel, we have to verify that
   \begin{enumerate}
   	\item  for each $j\in\numb{m}$,
    the corresponding $\HWeq[1]{0}$- and $\HWeq[1]{1}$-relations are satisfied,
    i.e., $c_0^j$ is unselected, and $c_t^j$ is selected,\label{item:HWs}
   	\item for each $i$ and $j$, the relation $R_i^j$ is satisfied, and\label{item:Rij}
   	\item
    all vertices of $G_\phi$ have a feasible number of selected neighbors,
    according to their own selection status.
    \label{item:validsel}
   \end{enumerate}

	Since $\pi$ is a satisfying assignment, for each $j\in \numb{m}$,
  there is an index $i^*\in\numb{\numGroups}$
  such that $\pi_{i^*}$ satisfies $C_{j}$.
  This gives \Cref{item:HWs}.

	In order to verify \Cref{item:Rij},
  let $i\in\numb{\numGroups}$ and $j\in\numb{m}$.
  We check that $R_i^j$ is satisfied.
  The relevant selected vertices are $S_i^j=S\cap Z_i^j$.
  Recall from \Cref{equ:encodedstates} that $S^j_i$ determines the states
  $a^j_{i,\ell}$ and $\inv a^{j-1}_{i,\ell}$
  for each $\ell \in\numb{\numVertPG}$.
	The states $a^j_{i,1}, \ldots, a^j_{i,\numVertPG}$ are determined
  by the selected vertices of $J^j$,
  i.e.,\ an $\encoder{A}$ of \rank $\numVertPG \numGroups$,
  and they are determined by the selection status of its distinguished vertices
  and their selected neighbors in the corresponding $\Bl$-blocks.
	As we selected vertices of $J^j$ according to the solution $S_x$, the fact that $a^j_{i,1}\ldots a^j_{i,\numVertPG}$ is the encoding $e_i$ of $\pi_i$ follows from the fact that $S_x$ manages $x$ (see \Cref{def:encoder}).

	Similarly, the states $\inv a^{j-1}_{i,1}, \ldots, \inv a^{j-1}_{i,\numVertPG}$
  are determined by the selected vertices of $J^{j-1}$, and
  in particular, by the selection status of its distinguished vertices
  and their selected neighbors in the corresponding $\Br$-blocks.
	Let $\ell \in\numb{\numVertPG}$.
	Since we selected vertices of $J^{j-1}$ according to $S_x$ as well, we have $a^{j-1}_{i,\ell} = a^{j}_{i,\ell}$. By \Cref{def:encoder}, $\inv a^{j-1}_{i,\ell}$ complements the state $a^{j-1}_{i,\ell} = a^{j}_{i,\ell}$ as required by $R^j_i$.

   	It is clear that in our selection we include $c_i^j$ into $S$
    whenever we include $c_{i-1}^j$.
    Finally, suppose $c_{i-1}^j$ is not selected.
    If $c_i^j$ is selected, then we have $i=i^*$,
    which means that $\pi_i$ satisfies $C_{j}$ by the choice of $i^*$.
  	Conversely, if we do not select $c_i^j$, then $i<i^*$,
    which means that $\pi_i$ does not satisfy $C_{j}$.
  	This shows that the selection $S$ satisfies $R_i^j$.

  	It remains to show \cref{item:validsel}.
    Let $v$ be a vertex in $G_\phi$.
    If $v$ is a non-portal vertex of 
    a $\{\sigma_{\sigMin},\rho_{\rhoMin}\}$-provider, then,
    by the definition of such providers (\cref{def:provider}),
    $v$ has a feasible number of selected neighbors
    (and, as it is not a portal, it does not have any neighbors outside of the respective provider).
    In the remaining cases, $v$ is either a vertex of some $J^j$,
    i.e., an $\encoder{A}$ of \rank $\numVertPG \numGroups$,
    or it is a vertex of the form $c_i^j$.
  	Suppose $v$ is in some $J^j$.
    By \cref{def:encoder}, $v$ has a feasible number of neighbors within $J^j$.
    If $v$ is in some set $Z_i^j$ then it does not have any selected neighbors other than the feasible number within $J^j$.
  	Finally, the same argument holds for $v=c_i^j$:
    it has a feasible number of selected neighbors within the
    attached $\{\sigma_{\sigMin},\rho_{\rhoMin}\}$-provider,
    but it does not receive any more selected neighbors.
\end{claimproof}

In the following, we show the reverse direction from the proof of correctness.

\begin{claim}
  If \srDomSetRel has a solution on $G_\phi$,
  then $\phi$ is satisfiable.
\end{claim}
\begin{claimproof}
	Let $S$ be a selection of vertices from $G_\phi$ that is a solution of \srDomSetRel.
	Recall that $S$ determines the states $a^j_{i,\ell}$, $\inv a^j_{i,\ell}$ and numbers $T^j_{i,\ell}$, $\inv T^j_{i,\ell}$, as described in \Cref{equ:encodedstates}.

  Observe that, for all $i,\ell$, we have that
  $w_{i,\ell}^j$ is selected
  if and only if $w_{i,\ell}^{j-1}$ is selected.
  This follows from the definition of the relation $R_i^j$.
  First, suppose the vertices $w_{i,\ell}^0, \ldots, w_{i,\ell}^{m}$ are selected.

  Assume further that $\sigma$ is finite.
  As $S$ is a solution,
  $w^j_{i,\ell}$ can have at most $\sigMax$ selected neighbors.
  Hence, for each $j\in\fragment{0}{m}$, we have
  $T_{i,\ell}^j+\inv T_{i,\ell}^j \le \sigMax$.
  Moreover, by the definition of the relations $R_i^j$,
  for each $j\in\numb{m}$, we get
  $\inv T_{i,\ell}^{j-1}+ T_{i,\ell}^j = \sigMax$.
  Combining both constraints yields
  $T_{i,\ell}^0 \le \dots \le T_{i,\ell}^{m}$
  and $\inv T_{i,\ell}^0 \ge \dots \ge \inv T_{i,\ell}^{m}$.

  Now, assume $\sigma$ is simple cofinite.
  As $S$ is a solution,
  $w^j_{i,\ell}$ has at least $\sigMax$ selected neighbors.
  Hence, for each $j\in\fragment{0}{m}$, we have
  $T_{i,\ell}^j+\inv T_{i,\ell}^j \ge \sigMax$.
  As before, we get by the definition of the relations $R_i^j$ that,
  for each $j\in\numb{m}$,
  $\inv T_{i,\ell}^{j-1}+ T_{i,\ell}^j = \sigMax$.
  Combining both constraints yields
  $T_{i,\ell}^0 \ge \dots \ge T_{i,\ell}^{m}$
  and $\inv T_{i,\ell}^0 \le \dots \le \inv T_{i,\ell}^{m}$
  in the simple cofinite case.

  It is easy to check that we get the same constraints
  if the vertices $w^j_{i,\ell}$ are not selected.

  We define new values $\widehat T_{i,\ell}^j$.
  \begin{equation}
    \label{equ:TsToModifiedTs}
    \widehat T_{i,\ell}^j \deff
      \begin{cases}
        \sigMax - T_{i,\ell}^j & w^j_{i,\ell} \text{ is selected }\land
          \sigma \text{ is simple cofinite} \\
        \rhoMax - T_{i,\ell}^j & w^j_{i,\ell} \text{ is not selected }\land
          \rho \text{ is simple cofinite} \\
        T_{i,\ell}^j & \text{else}
      \end{cases}
  \end{equation}
  One can easily check that we get $\widehat T_{i,\ell}^0 \le \dots \le \widehat T_{i,\ell}^{m}$
  independently from $\sigma$ or $\rho$ being finite or cofinite.

  Recall that we only use encodings from $\Encs$
  (all other encodings are not accepted by the relations $R_i^j$),
  and hence, they have a one-to-one correspondence
  to the encodings in $\Encsvar$ which all have the same weight.
  Therefore,
  for all $i\in\numb{\numGroups}$ and $j\in\fragment{0}{m+1}$, we get:
  $\sum_{\ell=1}^{\numVertPG} \widehat T^j_{i,\ell} = w$.
  Now, assume that for some
  $i\in\numb{\numGroups}$, $j\in\numb{m}$, and $\ell\in\numb{\numVertPG}$,
  we have $\widehat T^j_{i,\ell} < \widehat T^{j+1}_{i,\ell}$.
  This implies that there is some $\ell'$
  such that $\widehat T^j_{i,\ell'} > \widehat T^{j+1}_{i,\ell'}$,
  which contradicts the above assumption.
  Hence, we obtain $\widehat T_{i,\ell}^0 = \dots = \widehat T_{i,\ell}^{m}$
  for each $i$ and $\ell$.
  Combined with \cref{equ:TsToModifiedTs}, this implies
  \begin{equation}
    \label{equ:Ts}
    T_{i,\ell}^0 = \dots = T_{i,\ell}^{m}
    \quad
    \text{and}
    \quad
    \inv T_{i,\ell}^0 = \dots = \inv T_{i,\ell}^{m}
    \quad
    \text{for each $i$ and $\ell$.}
  \end{equation}

  We define an assignment $\pi$ of $\phi$ as follows.
  For each $i$, $a^1_{i,1}\ldots a^1_{i,\numVertPG}\in A^\numVertPG$ is subject to $R_i^1$, and therefore, it is the encoding of a partial assignment $\pi_i$ of the group of variables $F_i$. Let $\pi$ be the assignment comprised of these partial assignments.

  It remains to verify that $\pi$ satisfies $\phi$.
  For $j\in \numb{m}$, we verify that $\pi$ satisfies the clause $C_j$.
  Consider the vertices $c_0^j,\ldots, c_{\numGroups}^j$. Since $c_0^j$ and $c_{\numGroups}^j$ are subject to a $\HWeq[1]{0}$- and $\HWeq[1]{1}$-relation, respectively, we have $c_0^j \notin S$ and $c_{\numGroups}^j \in S$. Hence, there is an $i\in \numb{\numGroups}$ for which $c_{i-1}^j$ is not selected, but $c_{i}^j$ is. As $c_{i-1}^j$ and $c_{i}^j$ are subject to $R_i^j$, it follows that $a^j_{i,1}\ldots a^j_{i,\numVertPG}$ encodes a partial assignment that satisfies the clause $C_j$.
  The equalities from \Cref{equ:Ts} imply $a^1_{i,1}\ldots a^1_{i,\numVertPG}=a^j_{i,1}\ldots a^j_{i,\numVertPG}$. Therefore, $\pi_i$ satisfies $C_j$, and consequently $\pi$ satisfies $C_j$.
\end{claimproof}

This finishes the correctness of the construction.
As a last step, we analyze the graph and its pathwidth.

\begin{claim}
  \label{clm:lower:boundForPathwidth}
  There is some function $f$ that depends only on $\numVertPG$, $\sigMax$, and $\rhoMax$ such that the graph $G_\phi$ has size at most
  $m \cdot \numGroups \cdot f(\numVertPG, \sigMax,\rhoMax)$ vertices,
  pathwidth
  at most
  $\numVertPG\numGroups +f(\numVertPG,\sigMax,\rhoMax)$,
  and arity at most $f(\numVertPG, \sigMax, \rhoMax)$.
\end{claim}
\begin{claimproof}
	Every vertex in $G_\phi$ is part of (at least) one of the following:
	\begin{itemize}
		\item an $\encoder{A}$ of rank $\numVertPG \numGroups$
    (and there are $m+1$ of these),
		\item a $\{\sigma_{\sigMin},\rho_{\rhoMin}\}$-provider $Q_i^j$
    (and there are $(\numGroups +1)m$ of these)
	\end{itemize}
	We analyze the number of vertices in these components.
	\begin{itemize}
		\item
    By \Cref{def:encoder}, the size of an $\encoder{A}$ of \rank $\numVertPG \numGroups$
    is of the form  $\numVertPG \numGroups \cdot b$,
    where $b$ is an upper bound on the size of the blocks.
    Thus, $b$ depends only on $\sigMax$ and $\rhoMax$.
		\item
    The sizes of the $\{\sigma_{\sigMin},\rho_{\rhoMin}\}$-providers $Q_i^j$
    depend only on $\sigMin$ and $\rhoMin$
    (see the construction in \cref{lem:fillinggadget}).
	\end{itemize}
		Recall that the arity of the graph is defined
    as the largest degree of any relation in the graph.
    Hence, the arity is upper bounded by the size of the sets $Z^j_i$
    or some function depending on the size of the blocks from the manager.
		Recall that
		\[ Z_i^j= \{c_{i-1}^j,c_i^j\}
    \cup \bigcup_{\ell \in\numb{\numVertPG}} \bigl(
      w_{i,\ell}^{j-1} \cup N_{\Br}(w_{i,\ell}^{j-1})
      \cup w_{i,\ell}^j \cup N_{\Bl}(w_{i,\ell}^j)
    \bigr).\]
		The size of the sets $N_{\Br}(w_{i,\ell}^{j-1})$ and $N_{\Bl}(w_{i,\ell}^j)$
    is bounded by that of the blocks of
    the $\encoder{A}$ of \rank $\numVertPG \numGroups$.
    So, once again, their size depends only on $\sigMax$ and $\rhoMax$.
    So, there is some function $f''$
    such that, for each $i$ and $j$, we have
    $\abs{Z^j_i}\le g\cdot f''(\sigMax, \rhoMax)$.
    This proves the desired bound on the arity of $G_\phi$.

    So for the bound on the size of $G_\phi$ it suffices to bound the number of vertices plus the number of relations. We have already bounded the number vertices. For the number of relations, note that
    there are $2m$ relations of the form $\HWeq[1]{0}$ or $\HWeq[1]{1}$,
    $\numGroups (m+2)$ relations of the form $R_i^j$,
    and $(\numGroups+1)m$ relations attached to $Q_i^j$ with $c_i^j$.
	This proves the bound on the size of $G_\phi$.

  We use a node search strategy to bound the pathwidth of $G_\phi$ (see~\cite[Section 7.5]{CyganFKLMPPS15}).
  For each $i\in \numb{\numGroups}$ and $j\in \numb{m}$,
  we define a set $Y^j_i$ that contains the following vertices:
  \begin{itemize}
  	\item
    The vertices $Z^j_i$ that are subject to relation $R_i^j$.
    \item
    The vertices in $\Br^{j-1}_{i,\ell}$ and $\Bl^{j}_{i,\ell}$
    for all $\ell \in \numb{\numVertPG}$.
    (Some of these vertices are also contained in $Z^j_i$.)
  	\item The vertices of the $\{\sigma_{\sigMin},\rho_{\rhoMin}\}$-providers
    $Q_{i-1}^j$ and $Q_i^j$ attached to $c^j_{i-1}$ and $c^j_i$.
  \end{itemize}
  We similarly define, for each $i \in \numb{\numGroups}$,
  the sets $Y_i^0$ and $Y_i^{m+1}$.
	To simplify notation, we consider $Y^j_i$ to be the empty set
  if $i\notin \numb{\numGroups}$ or $j\notin \fragment{0}{m+1}$.

	We now describe $m+2$ stages of selecting vertices as positions for searchers,
  where each stage consists of $\numGroups$ rounds.
	Let $i\in \numb{\numGroups}$.
	In the $i$th round of the $j$th stage, where $j\in \fragment{0}{m+1}$,
  the selected vertices are $Y^j_i \cup Y^j_{i+1} \cup Y^{j-1}_{\numGroups}$
  together with,
  for each $z\in \fragment{1}{i-1}$ and $\ell\in \numb{\numVertPG}$,
  the vertex $w^{j}_{z,\ell}$ (if it exists),
  and for each $z\in \fragment{i+1}{\numGroups}$ and $\ell\in \numb{\numVertPG}$,
  the vertex $w^{j-1}_{z,\ell}$ (if it exists).

  Using \cref{fig:lower:construction},
  it is straightforward to verify that selecting vertices
  according to the described stages cleans the graph as required.
  Note that we select the set $Y^{j-1}_{\numGroups}$
  in all rounds of the $j$th stage
  to prevent ``recontamination'' via the edges
  between the blocks $\Bl^{j-1}_{\numGroups,\numVertPG}$
  and $\Br^{j-1}_{\numGroups,\numVertPG}$.
  Moreover, for each relational constraint, there is some stage in some round
  where the entire scope of the constraint is covered
  (that is, all vertices int he scope are in the same bag of the decomposition).

  Now, consider the number of searchers/the selected vertices in each stage.
  By the previously specified bounds on the size of
  blocks, and $\{\sigma_{\sigMin},\rho_{\rhoMin}\}$-providers,
  there is some function $f$
  that depends only on $\numVertPG$, $\sigMax$, and $\rhoMax$,
  such that the size of $Y^j_i$ is bounded from above by
  $f(\numVertPG, \sigMax, \rhoMax)$ for each $i$ and $j$.
  Hence, at each stage we select at most
  $\numVertPG \numGroups + \Oh(1) \cdot f(\numVertPG, \sigMax, \rhoMax)$ vertices.
  This is an upper bound on the node search number of $G_\phi$,
  and consequently, an upper bound on its pathwidth.
\end{claimproof}

The following observation follows directly from the construction
and the previous proofs.
\begin{observation}
  \label{obs:lower:parsimoniousAndFull}
  There is a one-to-one correspondence between
  satisfying assignments for $\phi$ and solutions of \srDomSetRel on $G_\phi$.
  In all solutions for $G_\phi$, each vertex $w_{i,\ell}^j$
  has exactly $\sigMax$ neighbors if it is selected,
  and exactly $\rhoMax$ neighbors if it is not selected.
\end{observation}

\subsubsection*{Putting the Pieces together}
\label{sec:high-level:decision:bound}
\newcommand{\states}{A}

	We now use the previous construction to solve SAT in time exponentially faster than $2^n$ (contradicting the SETH)
	assuming an algorithm for \srDomSetRel is given that beats the lower bound stated in \cref{lem:lowerBoundWhenHavingSuitableGadget}.
	We fix an $\eps>0$.
  To use the construction from the previous sections,
  it remains to choose the values of $\sizeGroup$ and $\numVertPG$.
  For ease of notation,
  we define $\epsilon' = \log_{\abs{\states}} (\abs{\states}-\epsilon) < 1$.
  We choose some $\alpha >1$
  such that $\alpha \cdot \epsilon' \le \delta' = \log (2-\delta) <1$
  for some $\delta>0$.
  We choose $\numVertPG$ sufficiently large such that it satisfies
  $\numVertPG \log{\abs{\states}} \le \alpha \floor{\numVertPG \log{\abs{\states}} - \log(\numVertPG\cdot\max \states+1)}$.
  Finally, set $\sizeGroup = \floor{\numVertPG \log{\abs{\states}} - \log(\numVertPG\cdot\max \states+1)}$.
  Observe that we can actually encode all partial assignments of one group.

  Using these parameters, we can construct a \srDomSetRel instance $G_\phi$
  based on $\phi$.
  For the size bound and pathwidth, observe
  that $\abs{\states} \le \abs{\allStates}$,
  and thus, $\abs{\states}$ can be bounded in terms of $\sigMax$ and $\rhoMax$.
  Moreover,
  $\numVertPG$ only depends on $\epsilon$, $\delta$, and $\abs{\states}$.
  Since $\sigma,\rho$ are fixed,
  any term only depending on $\epsilon$, $\delta$, $\sigMax$, and $\rhoMax$
  can be treated as a constant.
  Based on this,
	from \cref{clm:lower:boundForPathwidth} it follows that
	there is some fixed function $f$ such that $G_\phi$ has pathwidth at most
  $\numGroups\numVertPG+f(\numVertPG,\sigMax,\rhoMax) = \numGroups\numVertPG+\Oh(1)$,
  where $\numGroups = \ceil{ \frac{n}{\sizeGroup} }$,
	and that a path decomposition of this size can be computed efficiently.
  Using \cref{def:lower:graphWithRelations}, it is straightforward to check that $G_\phi$ has size in
  \[
   \Oh(
    \ceil{\tfrac{n}{\sizeGroup}} \cdot m \cdot f(\numVertPG,\sigMax,\rhoMax)
    )
   =
   \Oh(n \cdot m \cdot f'(\numVertPG, \sigMax,\rhoMax))
   =
   \Oh(n \cdot m)
  \]
  and arity at most $f'(\numVertPG, \sigMax, \rhoMax)$.
  We choose the constant $d$ from the lemma statement as the arity of $G_\phi$
  which is a constant by the above arguments.

  Now assume there is an algorithm solving \srDomSetRel on instances of size $N$ and degree at most $d$
  given a path decomposition of width at most $k$ in time
  $(\abs{\states} - \eps)^{k+\Oh(1)} \cdot N^{\Oh(1)}$.
  We execute this algorithm on the graph $G_\phi$ which we just constructed:
  \begin{align*}
    (\abs{\states} - \epsilon)^{\pw(G_\phi)+\Oh(1)} \cdot \abs{G_\phi}^{\Oh(1)}
    \le & (\abs{\states}-\epsilon)^
      {\numGroups\numVertPG+\Oh(1)}
      \cdot (n \cdot m)^{\Oh(1)} \\
    \le & (\abs{\states}-\epsilon)^
      {\ceil{ \frac{n}{\sizeGroup}} \numVertPG+\Oh(1)}
      \cdot (n \cdot m)^{\Oh(1)} \\
    \le & (\abs{\states}-\epsilon)^
      {\frac{n}{\sizeGroup} \numVertPG+\numVertPG+\Oh(1)}
      \cdot (n \cdot m)^{\Oh(1)} \\
    \intertext{
    By the same argument as above,
    we can ignore the term of \(\numVertPG+\Oh(1)\) in the exponent
    as it only contributes a constant factor to the runtime.
    }
    \le & (\abs{\states}-\epsilon)^
      {\frac{n}{\sizeGroup} \cdot \numVertPG}
      \cdot (n \cdot m)^{\Oh(1)} \\
    \le & 2^
      {\log{\abs{\states}} \cdot \epsilon'
        \cdot \frac{n}{\sizeGroup} \cdot \numVertPG}
      \cdot (n \cdot m)^{\Oh(1)} \\
    \le & 2^
      {\epsilon'
        \cdot \frac{n} {\floor{\numVertPG \log{\abs{\states}} - \log(\numVertPG\cdot\max \states+1)}}
        \cdot \numVertPG \cdot \log{\abs{\states}}}
      \cdot (n \cdot m)^{\Oh(1)} \\
    \intertext{
    By our choice of \(\numVertPG\) we get:
    }
    \le & 2^{\epsilon' \alpha n} \cdot (n \cdot m)^{\Oh(1)}
    \le   2^{\delta' n} \cdot (n \cdot m)^{\Oh(1)}
    =   (2-\delta)^{n} \cdot (n \cdot m)^{\Oh(1)}.
  \end{align*}
  This directly contradicts SETH and finishes the proof of \cref{lem:lowerBoundWhenHavingSuitableGadget}.

\subsection{\boldmath Proof of
\texorpdfstring{\cref{lem:count:lowerBoundWhenHavingSuitableGadget}}
{Lemma \ref{lem:count:lowerBoundWhenHavingSuitableGadget}}:
Lower Bound for \texorpdfstring{\srCountDomSetRel}
{Counting (Sigma,Rho)-Dominating Set with Relations}}
\label{sec:high-level:counting}

Now, we move to the proof of the intermediate lower bound for the counting version.
By \cref{obs:lower:parsimoniousAndFull}, the previous reduction is parsimonious.
Hence, the lower bound for the decision version from
\cref{lem:lowerBoundWhenHavingSuitableGadget} directly transfers to the counting version.
\begin{corollary}
  \label{lem:count:intermediateLowerBoundWhenHavingSuitableGadget}
  Let $\sigma,\rho \subseteq \NN$ be two fixed and non-empty sets
  which are \emph{finite} or \emph{simple cofinite}.
  Suppose there is an $A \subseteq \allStates$
  that is closed under the inverse with respect to $\sigma,\rho$
  such that there is an $\encoder{A}$.

  For all $\eps>0$, there is a constant $d$ such that
  \srCountDomSetRel on instances of size $n$ and arity at most $d$
  cannot be solved in time
  $(\abs{A} - \epsilon)^{k+\Oh(1)} \cdot n^{\Oh(1)}$,
  even if the input is given with a path decomposition of width $k$,
  unless \#SETH fails.
\end{corollary}

In what follows, we extend this result to \srCountDomSetRel,
where we allow the sets $\sigma$ and $\rho$ to also be arbitrary cofinite sets
and not only \emph{simple cofinite} sets.
See \cref{fig:count:simpleCofToCof} for an illustration of the process.

\begin{figure}[t]
\centering
\begin{tikzpicture}[
  scale=.615,
  transform shape,
]
  \input{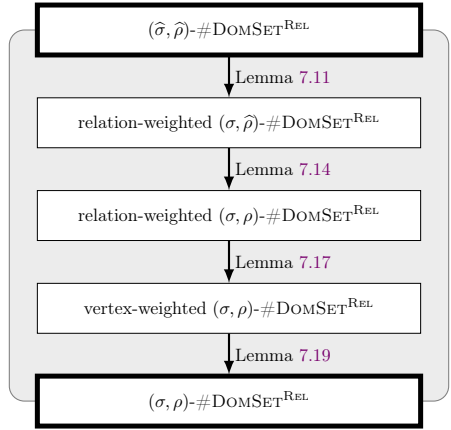}
  \def\x{8cm}
  \def\xGap{1cm}
  \def\y{-2 cm}

  % \crefname{lemma}{Lem.}{Lem.}

  \tikzset{redTur/.style={red}}
  \tikzset{%
  part/.style = { % some specific problem
    draw=gray,
    fill=gray!15,
    rounded corners=2mm,
  },
  lowerB/.style = { % Lower bound shown for a problem
  },
  important/.style = { % Lower bound shown for a problem
    line width=.66mm,
  },
  case/.style = { % Different cases (reachable by switch
  },
  res/.style = { % Reference for the result
    anchor=west,
    yshift=-0.5*\y,
  },
  }

  % Uncomment the following if we do not want to state the full name each time.
  % \renewcommand{\srCountDomSetRel}[1][\textsc{Rel}]{\ensuremath{#1}}
  % \renewcommand{\DomSetGeneral}[4] {\ensuremath{(#1,#2)\mhyph}\textsc{{#3}DS\ensuremath{^{#4}}}\xspace}

  %% BLOCK ONE (making sigma, rho cofinite)
  \node[problem, lowerB, noBreak, important] (domSetRelFull) at (0,-1*\y)
    {\CountDomSetRel{\widehat\sigma}{\widehat\rho}};

  \node[problem,noBreak] (relWeiDomSetRelSigma) at (0,0*\y)
    {relation-weighted \CountDomSetRel{\sigma}{\widehat\rho}};
  \draw[red] (domSetRelFull) -- (relWeiDomSetRelSigma);
  \node[res] at (relWeiDomSetRelSigma) {\cref{lem:count:makeSigmaCofinite}};

  \node[problem,noBreak] (relWeiDomSetRel) at (0,1*\y)
    {relation-weighted \CountDomSetRel{\sigma}{\rho}};
  \draw[red] (relWeiDomSetRelSigma) -- (relWeiDomSetRel);
  \node[res] at (relWeiDomSetRel) {\cref{lem:count:makeRhoCofinite}};

  \node[problem,noBreak] (verWeiDomSetRel) at (0,2*\y)
    {vertex-weighted \CountDomSetRel{\sigma}{\rho}};
  \draw[red] (relWeiDomSetRel) -- (verWeiDomSetRel);
  \node[res] at (verWeiDomSetRel)
    {\cref{lem:count:relationWeightedToVertexWeighted}};

  \node[problem, lowerB, noBreak, important] (domSetRel) at (0,3*\y)
    {\CountDomSetRel{\sigma}{\rho}};
  \draw[redTur] (verWeiDomSetRel) -- (domSetRel);
  \node[res] at (domSetRel) {\cref{lem:count:vertexWeightedToUnweighted}};

  %% MARKING THE BLOCKS
  \begin{scope}[on background layer]
    \draw[part]
      (-.5*\x-0.7cm, -1*\y) rectangle (0.5*\x+0.7cm, 3*\y);
  \end{scope}
\end{tikzpicture}
\caption{
An overview of how the reductions in \cref{sec:high-level:counting}
are combined to obtain the lower bound
from \cref{lem:count:lowerBoundWhenHavingSuitableGadget}
when $\sigma$ and $\rho$ are cofinite sets that are not simple cofinite.
The other cases where at least one set is finite or simple cofinite
are easier as some steps of the reduction can be omitted.
The sets $\widehat\sigma$ and $\widehat\rho$ are simple cofinite sets
% and the sets $\sigma$ and $\rho$ are cofinite sets
with
$\max(\NN\setminus \sigma) = \max(\NN\setminus \widehat\sigma)$
and $\max(\NN\setminus \rho) = \max(\NN\setminus \widehat\rho)$.
}
\label{fig:count:simpleCofToCof}
\end{figure}

It is instructive to generalize the definition of graphs with relations to include weights.
\begin{definition}
  A tuple $G=(V, E, \CC)$ is a \emph{graph with weighted relations}
  if $G$ is defined as in \cref{def:lower:graphWithRelations}
  except that, each constraint $C\in \CC$
  is now a function $f_C\from 2^{\scope(C)} \to \SetQ$.

  The definitions of pathwidth and arity of a graph with weighted relations
  transfer directly from \cref{def:lower:graphWithRelations}.
\end{definition}

We use the \emph{relation-weighted} version of \srCountDomSetRel, which is defined as follows.

\begin{definition}
  [$q$-Relation-weighted \srCountDomSetRel]
  Let $q$ be a non-negative integer.
  In the \emph{$q$-relation-weighted \srCountDomSetRel} problem,
  we are given a graph with weighted relations $G=(V, E, \CC)$
  such that $\abs{\{ f_C(U) \mid C \in \CC, U \subseteq \scope(C) \} } \le q$,
  that is, at most $q$ different weights are used by the constraints of $G$.

  Let $\Omega$ be the set of $(\sigma,\rho)$-sets of $G$. The task is to compute
  \[
    \sum_{U\in \Omega}
      \prod_{C\in\CC}
        f_C(U \cap \scope(C))
    .
  \]
\end{definition}
Note that, for $f_C \from 2^{\scope(C)} \to \{0,1\}$, this corresponds to the unweighted \srCountDomSetRel problem.
Therefore, unweighted \srCountDomSetRel trivially reduces to
$2$-relation weighted \srCountDomSetRel.
As we allow negative weights,
it is possible that solutions cancel each other out.
We exploit this fact to remove unwanted solutions from the sum.

In the following, we reuse some of the providers
which we have seen in the previous parts about the decision version.
As we are now dealing with the counting version,
we have to be more careful about the number of (partial) solutions of such providers.
We are especially interested in $L$-providers which
do not have any ``unwanted'' (partial) solutions that might change the output.
We make this formal by the following definition.
\begin{definition}
  [Parsimonious Provider]
  Let $G$ be a graph (possibly with relations) together with a set of portals $U\subseteq V(G)$.
  For a language $L \subseteq \allStatesExt$,
  we say that $(G,U)$ is a \emph{parsimonious $L$-provider}
  if $(G,U)$ is an $L$-realizer (as defined in \cref{def:providersWithRelations})
  and, for all $x \in L$,
  there is a unique partial solution that witnesses $x$.
\end{definition}
Observe that the requirement of being an $L$-\emph{realizer}
rules out the possibility of having (partial) solutions
which witness an element not contained in $L$.
By \cref{lem:fillinggadget}, there is a $\{\sigma_s, \rho_r\}$-provider
for all $s \in \sigma$ and $r\in \rho$.
We extend this now to a parsimonious provider.
\begin{lemma}\label{lem:count:parsFillingGadget}
  For any $s \in \sigma$ and $r \in \rho$,
  there is a parsimonious $\{\sigma_s, \rho_r\}$-provider.
\end{lemma}
\begin{proof}
  Let $(G,\port)$ be the $\{\sigma_s, \rho_r\}$-provider
  from \cref{lem:fillinggadget} with portal $\port$.
  Let $S_0$ be a partial solution corresponding to state $\rho_r$,
  and let $S_1$ be a partial solution corresponding to the state $\sigma_s$.
  We add a relation $R$ to the graph which observes all vertices of $G$,
  including $\port$.
  Relation $R$ accepts only two possible selections of vertices,
  namely $S_0$ whenever $\port$ is unselected,
  and $S_1$ if $\port$ is selected.
  The properties of the parsimonious provider follow immediately.
\end{proof}

Now, we have everything ready to formally state the first reduction.

\begin{lemma}
  \label{lem:count:makeSigmaCofinite}
  Let $\rho$ be a fixed and non-empty set.
  Let $\sigma,\sigma'$ be two fixed and non-empty sets
  where $\sigma$ is cofinite and $\sigma'$ is simple cofinite such that
  $\max(\NN\setminus \sigma) = \max(\NN\setminus \sigma')$.
  There is a non-negative, constant integer $q_0$
  such that the following holds:

  For all constants $q$,
  there is a \pwar-reduction
  from $q$-relation-weighted \CountDomSetRel{\sigma'}{\rho}
  to $(q+q_0)$-relation-weighted \CountDomSetRel{\sigma}{\rho}.
\end{lemma}
\begin{proof}
  We can assume that $\sigma \neq \sigma'$ as otherwise the reduction is trivial. It follows that $\sigma\neq \NN$, and consequently $\sigMax>0$.
  We use the following auxiliary provider in the proof.
  \begin{claim}
    \label{lem:count:smallHelperGadget:sigma}
    Let $\sigma,\rho$ be fixed and non-empty.
    If $\sigma$ is cofinite, then there is
    a parsimonious $\{\sigma_0, \sigma_1, \rho_0\}$-provider
    which uses relations.
  \end{claim}
  \begin{claimproof}
    We use a set of vertices $V=\{\port, w, v_1, \ldots, v_{\allMax}\}$,
    where $\port$ is the portal vertex, and $\port$ is adjacent only to $w$.
    For each $i\in \numb{\allMax}$,
    $w$ is adjacent to $v_i$, and $v_i$ serves as the portal
    of a new parsimonious $\{\sigma_{\sigMax}, \rho_{\rhoMin}\}$-provider $J_i$,
    which exists by \cref{lem:count:parsFillingGadget}.
    We show that there is a partial solution for each state.

    \begin{itemize}
    \item[\boldmath $\rho_0$]
    For each $i\in \numb{\rhoMax}$, the vertex $v_i$ is selected,
    which is extended to the $\sigma_{\sigMax}$-state of $J_i$.
    The remaining vertices of $V$ are unselected,
    and this is extended to the $\rho_{\rhoMin}$-states
    of $J_{\rhoMax+1}, \ldots, J_{\allMax}$.
    Hence, the vertices $v_1,\dots,v_{\rhoMax}$ obtain $\sigMax$ selected neighbors from the attached $\{\sigma_{\sigMax}, \rho_{\rhoMin}\}$-providers, and
    the vertices $v_{\rhoMax+1},\dots,v_{\allMax}$ obtain $\rhoMin$ selected neighbors.
    Moreover, $w$ is unselected and has $\rhoMax$ selected neighbors $v_1,\dots,v_{\rhoMax}$.

    \item[\boldmath $\sigma_0$]
    For each $i\in \numb{\rhoMax-1}$, the vertex $v_i$ is selected,
    which is extended to the $\sigma_{\sigMax}$-state of $J_i$.
    The vertex $\port$ is also selected.
    The remaining vertices in $V$ are unselected,
    and this is extended to the $\rho_{\rhoMin}$-states
    of $J_{\rhoMax}, \ldots, J_{\allMax}$.
    The vertex $w$ is unselected and has $\rhoMax$ selected neighbors ($v_1,\dots,v_{\rhoMax-1}$ and $\port$).
    The vertices $v_1,\dots,v_{\allMax}$ again obtain a feasible number of selected neighbors
    from the attached $\{\sigma_{\sigMax}, \rho_{\rhoMin}\}$-providers.

    \item[\boldmath $\sigma_1$]
    Each vertex in $V$ is selected,
    and this is extended to the $\{\sigma_{\sigMax}, \rho_{\rhoMin}\}$-providers
    $J_1, \ldots, J_{\allMax}$.
    Each vertex in $V\setminus \{\port\}$ is selected with $\sigMax+1$ selected neighbors, which is feasible since $\sigma$ is cofinite.%
    \footnote{
    The state $\rho_1$ could also be obtained by selecting all of $V\setminus\{\port\}$.
    }
    \end{itemize}
    We add a relation that observes all vertices of $V$. 
    The relation ensures that exactly three partial solutions are feasible,
    each corresponding to one of the mentioned states.
    As the $\{\sigma_{\sigMax},\rho_{\rhoMin}\}$-providers are also parsimonious,
    all other potential partial solutions are ruled out.
  \end{claimproof}

  Let $G$ be the graph from an instance of \CountDomSetRel{\sigma'}{\rho}.
  Observe that every $(\sigma',\rho)$-set is also a $(\sigma,\rho)$-set,
  but the converse is not necessarily true as in a $(\sigma,\rho)$-set
  a selected vertex might also have fewer than $\sigMax$ selected neighbors.
  We modify $G$ with the goal of modeling the situation where no selected vertex can receive fewer than $\sigMax$ selected neighbors. To this end,
  we intend to use weighted relations to cancel out unwanted selections.
  The underlying idea is to add neighbors to each vertex $v \in V(G)$
  that can be selected almost arbitrarily.
  Based on the set $\sigma$,
  some selections in which $v$ receives fewer than $\sigMax$ selected neighbors are allowed and some are not.
  We add weights based on the number of selected neighbors of $v$
  such that in total the vertex behaves as if it had at least $\sigMax$ selected neighbors.

  Formally, we apply the following modification for each vertex $v \in V(G)$.
  To simplify notation, set $Z = \{\sigma_0, \sigma_1, \rho_0\}$.
  Note that, by \cref{lem:count:smallHelperGadget:sigma}, a $Z$-provider exists
  (if we allow additional relations).
  We introduce $\sigMax$ copies of the $Z$-provider with $v$ as the portal vertex.
  Let $W=\{w_1,\dots,w_{\sigMax}\}$ be the set of neighbors of $v$ in these copies (i.e., the vertex $w$ from each copy).
  We add a relation $R$ observing $v, w_1,\dots,w_{\sigMax}$ that only allows the following selections: if $v$ is unselected, then all vertices in $W$ are unselected; if $v$ is selected, then $R$ allows a selection for each $\gamma\in \fragment{0}{\sigMax}$ such that $w_1,\dots,w_\gamma$ are selected
  and $w_{\gamma+1},\dots,w_{\sigMax}$ are unselected.
  Moreover, if exactly the vertices $v,w_1,\dots,w_\gamma$ are selected,
  then the relation accepts with a weight $f_\gamma$ which we define in a moment.

  Let us first check that selecting vertices according to $R$ actually can be extended to partial solutions.
  If $v$ and all vertices in $W$ are unselected, then we obtain a partial solution if all of the $\sigMax$ $Z$-providers are in state $\rho_0$.
  Similarly, if $v$ is selected, the first $\gamma$ $w_i$'s are selected, and the remaining $w_i$'s are unselected, then we obtain a partial solution if the first $\gamma$ $Z$-providers are in state $\sigma_1$,
  and the remaining $Z$-providers are in state $\sigma_0$.

  Now let us define the weights.
  Note that if $v$ is selected and has $\alpha$ selected neighbors in the original graph
  (ignoring the vertices in $W$),
  then by the $\sigma$-constraint and the definition of $R$, there is some $\gamma$ with $\alpha+\gamma \in \sigma$ such that from $W$ precisely $w_1,\ldots, w_\gamma$ are selected.
  We now aim to define weights $f_i$,
  for $i \in \fragment{0}{\sigMax}$,
  such that valid weights sum up to 0 whenever $\alpha<\sigMax$.
  Formally, we want the following equation to hold for all $\alpha \in \fragment{0}{\sigMax}$:
  \begin{equation}
    \label{eqn:count:replaceVtcsByGadgets:property}
    \sum_{\substack{\gamma=0: \\ \alpha+\gamma \in \sigma}}^{\sigMax}
      f_\gamma
      =
      \begin{cases}
        1 & \alpha \ge \sigMax \\
        0 & \text{else}
      \end{cases}
    .
  \end{equation}

  \begin{claim}
    We can (efficiently) find the values for the $f_i$'s such that
    they satisfy the constraints from \cref{eqn:count:replaceVtcsByGadgets:property}.
  \end{claim}
  \begin{claimproof}
    Observe that there are $\sigMax+1$ unknown values
    and $\sigMax+1$ constraints.
    Recall that $\sigMax > 0$, and therefore, $\sigMax-1 \notin \sigma$.
    From \cref{eqn:count:replaceVtcsByGadgets:property},
    the sums for $\alpha=\sigMax$ and $\alpha=\sigMax-1$
    differ by exactly one weight, namely $f_0$.
    Hence, we can determine the value for $f_0$
    and eliminate it from the constraints.
    Repeating this procedure yields a solution for the equations
    and provides the values for all $f_i$.
  \end{claimproof}

  Note that there are at most $\sigMax+1$ different weights.
  Hence, the above reduction constructs an instance of
  $q+\sigMax+1$-relation weighted \srCountDomSetRel
  whenever the input was an instance of $q$-relation weighted \srCountDomSetRel.
  Thus, we have $q_0 \deff \sigMax+1$.

  It remains to show that the reduction is pathwidth-preserving.
  The relation $R$ has a scope of size $\sigMax+1$.
  Hence, we can duplicate one bag containing $v$
  and add all $\sigMax$ vertices $w_1,\dots,w_{\sigMax}$ to that bag.
  Then, the scope of the relation $R$ is contained in one bag.
  As $\sigma$ is fixed, $\sigMax$ is a constant,
  and hence, the pathwidth increases by a constant only.

  Moreover, the reduction is arity-preserving.
  Indeed, the existing relations are not changed
  and the relation $R$ has arity at most $\sigMax+1$
  which is a constant as $\sigma$ is fixed.
\end{proof}

Similarly to the previous result, we show the following dual result,
that is, we allow $\rho$ to be cofinite.
\begin{lemma}
  \label{lem:count:makeRhoCofinite}
  Let $\sigma$ be a fixed and non-empty set.
  Let $\rho,\rho'$ be two fixed sets
  where $\rho$ is cofinite and $\rho'$ is simple cofinite such that
  $\max(\NN\setminus \rho) = \max(\NN\setminus \rho')$.
  There is a non-negative constant $q_0$ such that the following holds:

  For all constants $q$,
  there is a \pwar-reduction
  from $q$-relation-weighted \CountDomSetRel{\sigma}{\rho'}
  to $(q+q_0)$-relation-weighted \CountDomSetRel{\sigma}{\rho}.
\end{lemma}
\begin{proof}
  % [Proof of \cref{lem:count:makeRhoCofinite}]
  We use the following auxiliary provider in the proof
  which complements the result from \cref{lem:count:smallHelperGadget:sigma}.
  \begin{claim}
    \label{lem:count:smallHelperGadget:rho}
    Let $\sigma,\rho$ be fixed and non-empty.
    If $\rho$ is cofinite, then there is
    a parsimonious $\{\sigma_0, \rho_0, \rho_1\}$-provider
    which uses relations.
  \end{claim}
  \begin{claimproof}
    Let $\port$ be the portal vertex, which is adjacent to a new vertex $w$
    which serves as the portal of a
    parsimonious $\{\sigma_{\sigMin}, \rho_{\rhoMax}\}$-provider
    (\cref{lem:count:parsFillingGadget}).
    For states $\sigma_0$ and $\rho_0$,
    $v$ is unselected which is extended to the $\{\sigma_{\sigMin}, \rho_{\rhoMax}\}$-provider.
    Hence, $v$ has at least $\rhoMax$ neighbors.
    For state $\rho_1$,
    $v$ is selected which is again extended to the $\{\sigma_{\sigMin}, \rho_{\rhoMax}\}$-provider.
    Hence, $v$ has $\sigMin$ neighbors as $\port$ is not selected.

    We add a relation
    which ensures that there are exactly three partial solutions,
    each corresponding to one of the mentioned states.
    By this, all other potential partial solutions are ruled out.
  \end{claimproof}

  The proof follows precisely the same idea as before.
  The most notable difference is that we now use $Z$-providers
  with $Z=\{\sigma_0, \rho_0, \rho_1\}$,
  that is, the state $\sigma_1$ is not allowed,
  but the state $\rho_1$ is allowed.
  Observe that this is not an issue as we never need to add neighbors to $v$
  if the vertex is selected, and
  actually the relation $R$ forbids to add selected neighbors to $v$ in this case.

  Then, the remainder of the proof is analogous to that of \cref{lem:count:makeSigmaCofinite}.
\end{proof}

In the next step, we remove the weighted relations
by allowing weighted vertices.
In this setting, the vertices contribute with their weight to the solution
whenever they are selected.
We refer to this new variant as \emph{vertex-weighted \srCountDomSetRel}.
\begin{definition}
  [$q$-Vertex-weighted \srCountDomSetRel]
  Let $q$ be a non-negative integer.
  In the \emph{$q$-vertex-weighted \srCountDomSetRel} problem,
  we are given a graph with relations $G=(V, E, \CC)$
  and a weight function $\wt\from V \to \SetQ$
  such that $\abs{\{ \wt(v) \mid v \in V \} } \le q$,
  that is, there are at most $q$ different vertex weights.

  Let $\Omega$ be the set of $(\sigma,\rho)$-sets of $G$. The task is to compute
  \[
    \sum_{U\in \Omega}
      \prod_{v \in U}
        \wt(v)
    .
  \]
\end{definition}
Recall that a $(\sigma,\rho)$-set of $G$
has to satisfy the relations of $G$. Therefore, this requirement appears only implicitly in the formula.
Note that to obtain the original \srCountDomSetRel problem we can set
$\wt(v)=1$ for all $v\in V$.

Now, we can formally state the next step of the reduction.
This step crucially relies on the fact
that only a constant number of different weights are used,
as otherwise the pathwidth of the used construction would be too large.
\begin{lemma}
  \label{lem:count:relationWeightedToVertexWeighted}
  For all constants $q$,
  there is a \pwar-reduction from
  $q$-relation-weighted \srCountDomSetRel
  to $q$-vertex-weighted \srCountDomSetRel.
\end{lemma}
\begin{proof}
  For a fixed $q$,
  let $w_1,\dots,w_{q'}$ be the $q' \le q$ different weights used by the relations
  of a $q$-relation-weighted \srCountDomSetRel instance $G$.
  Formally, we have
  $\{ w_1, \dots, w_{q'} \} = \{ f_C(U) \mid C \in \CC, U \subseteq \scope(C) \}$.
  We assume without loss of generality that $q'=q$ in the following.

  We apply the following modification for all weighted constraints in $\CC$.
  Let $f_C$ be the weight function of some constraint $C\in \CC$.
  We introduce $q$ new vertices $v_1,\dots,v_q$.
  Each of these vertices is the portal of a parsimonious
  $\{\sigma_{\sigMin}, \rho_{\rhoMin}\}$-provider with relations,
  which exists by \cref{lem:count:parsFillingGadget}.
  For all $i\in \numb{q}$,
  we assign weight $w_i$ to vertex $v_i$.
  Moreover, we replace the constraint $C$ by a relational (unweighted) constraint $C'$ with $\scope(C')\coloneqq \scope(C)\cup \{v_1,\ldots,v_q\}$.
  The new relation $\rel(C')$ accepts a selection $U\subseteq \scope(C')$ if and only if
  \begin{itemize}
  	\item the corresponding weight $f_C(U\cap \scope(C))$ is nonzero, 
  	\item the vertex $v_i$ with $w_i=f_C(U\cap \scope(C))$ is selected, and
  	\item all $v_j$ with $i\neq j$ are unselected.
  \end{itemize}

  It remains to argue that the reduction
  increases the pathwidth and the arity only by a constant.
  By assumption, there is a bag containing $\scope(C)$.
  We duplicate this bag,
  make it a neighbor of the original bag in the path decomposition,
  and add the new vertices $v_1,\dots,v_q$ to the copy
  which increases the size of the bag by $q$.
  For each $v_i$ (one after the other), we duplicate this new bag
  and also add all the vertices from its attached $\{\sigma_{\sigMin}, \rho_{\rhoMin}\}$-provider.
  Observe that the size of each parsimonious
  $\{\sigma_{\sigMin}, \rho_{\rhoMin}\}$-provider
  with relations is bounded by a constant,
  as the size depends only on $\sigma$ and $\rho$, and these two sets are fixed.
  Hence, the width of the decomposition increases by
  at most an additive constant
  depending only on $\sigma$, $\rho$, and $q$.

  To bound the arity of the new graph,
  observe that the arity of each relation increases by at most $q$,
  which is a constant.
  Moreover, the arity of the relations in the
  parsimonious
  $\{\sigma_{\sigMin},\rho_{\rhoMin}\}$-provider with relations
  is bounded by a constant as it is clearly bounded by the size of such a provider.
  Hence, the reduction increases the arity of the graph
  at most by an additive constant.
\end{proof}

The last step now finally removes the weights from the vertices,
and thus, also from the instance.
This procedure makes use of the following known fact
about recovering multivariate polynomials.
\begin{fact}
  [Lemma~2.5 in \cite{Curticapean18}]
  \label{fct:count:recoveringPolynomial}
  Let $P \in \SetQ[x_1,\dots,x_q]$ be a multivariate polynomial
  such that, for all $i \in \numb{q}$,
  the degree of $x_i$ in $P$ is bounded by $d_i \in \NN$.
  Furthermore, assume we are given sets $\Xi_i \subseteq \SetQ$
  for $i \in \numb{q}$
  such that $\abs{\Xi_i} = d_i + 1$ for all $i \in \numb{q}$.
  Consider the Cartesian product of these sets, that is,
  \[
    \Xi \deff \Xi_1 \times \dots \times \Xi_q
    .
  \]
  Then, we can compute the coefficients of $P$
  with $\Oh(\abs{\Xi}^3)$ arithmetic operations
  when given as input the set $\{ (\xi, P(\xi)) \mid \xi \in \Xi \}$.
\end{fact}

\begin{lemma}
  \label{lem:count:vertexWeightedToUnweighted}
  For all constants $q$,
  there is a \pwar-reduction from
  $q$-vertex-weighted \srCountDomSetRel
  to unweighted \srCountDomSetRel.
\end{lemma}
\begin{proof}
  We use known interpolation techniques
  to remove the weights.
  Let $G=(V,E,\CC)$ be the graph with relations, and let $\wt$ be the vertex weight function
  of a given $q$-vertex-weighted \srCountDomSetRel instance $I$.
  Let $w_1, \dots, w_q$ be the distinct values of the vertex weights,
  that is, $\{w_1, \dots, w_q\} = \{ \wt(v) \mid v \in V \}$.

  We replace each weight $w_i$ by a variable $x_i$,
  and treat the sought-after weighted sum of solutions of $I$ as a polynomial $P$ in the $q$ variables $x_1,\dots,x_q$.
  Observe that there are at most $n$ vertices,
  and hence, the total degree is at most $n$ for each variable.
  Hence, if we can realize all combinations of $n+1$ different weights for each variable $x_i$,
  then we can use \cref{fct:count:recoveringPolynomial}
  to recover the coefficients of $P$ in time $\Oh((n+1)^{3q})$
  which is polynomial in $n$ as $q$ is a constant.
  Then, we can output $P(w_1, \dots, w_q)$ to recover the original solution.

  It remains to realize $n+1$ different weights for each variable.
  For this, it suffices
  to realize weights of the form $2^\ell$.

  Let $v$ be the vertex for which we want to realize weight $\wt(v)=2^\ell$.
  For this, we introduce $\ell$ new vertices $v_1,\dots,v_\ell$
  which are all portals of a parsimonious
  $\{\sigma_{\sigMin}, \rho_{\rhoMin},\}$-provider with relations,
  which exists by \cref{lem:count:parsFillingGadget}.
  Moreover, for each $j=\numb{\ell}$, we add a relational constraint $C_j$ with scope $\{v,v_j\}$ that ensures that if $v$ is unselected,
  then $v_j$ must also be unselected, but if $v$ is selected, then $v_j$ can be either selected or unselected.
  Whenever $v$ is selected,
  this construction contributes a factor of $2^\ell$ to the solution,
  whereas if $v$ is unselected, it only contributes a factor of $1$.

  Observe that the arity of the resulting graph increases by at most a constant
  (or stays unchanged if it is already larger)
  as $C_j$ has arity $2$ and the degrees of the relations
  in the parsimonious $\{\sigma_{\sigMin},\rho_{\rhoMin}\}$-provider
  are bounded by constants as the size of the provider is constant.

  It is straightforward to see that this modification
  does not change the pathwidth too much.
  Indeed, for each weighted vertex $v$ where we realize weight $2^\ell$,
  we pick an arbitrary bag containing $v$
  and duplicate this bag for each $j\in \numb{\ell}$.
  Then, we add $v_j$
  and the corresponding $\{\sigma_{\sigMin}, \rho_{\rhoMin}\}$-provider
  together with the relations to the $j$th copy of the bag containing $v$.
  By the size bound for the $\{\sigma_{\sigMin}, \rho_{\rhoMin}\}$-provider,
  the claim follows.
\end{proof}

Combining all the previous results,
we obtain the lower bound for the counting version in terms of $\encoder{A}$s
(\cref{lem:count:lowerBoundWhenHavingSuitableGadget}),
which we restate here for convenience.
\lowerBoundForCountDomSetRel*
\begin{proof}
	The basic idea of the proof is to combine the reductions from
	\cref{lem:count:makeSigmaCofinite,lem:count:makeRhoCofinite,%
		lem:count:relationWeightedToVertexWeighted,%
		lem:count:vertexWeightedToUnweighted}
	with the lower bound from
	\cref{lem:count:intermediateLowerBoundWhenHavingSuitableGadget}.
  If both $\sigma$ and $\rho$ are finite or simple cofinite, then
  the lower bound follows directly from
  \cref{lem:count:intermediateLowerBoundWhenHavingSuitableGadget}.
  So, let us assume that at least one of $\sigma$ or $\rho$ is cofinite, but not simple cofinite.
  In fact, let us assume that both $\sigma$ and $\rho$ are cofinite, but not simple cofinite
  as the other cases can be treated analogously.

  Let $\sigma' \subseteq \sigma$ and $\rho'\subseteq \rho$
  be simple cofinite sets such that
  $\max(\NN \setminus \sigma) = \max(\NN \setminus \sigma')$ and
  $\max(\NN \setminus \rho) = \max(\NN \setminus \rho')$.

  Fix some arbitrary $\eps > 0$
  and assume that \#SETH holds.
  We know by \cref{lem:count:intermediateLowerBoundWhenHavingSuitableGadget},
  that there is a constant $d_0$ such that
  \CountDomSetRel{\sigma'}{\rho'} cannot be solved in time
  $(\abs{A}-\eps)^{k+\Oh(1)} \cdot n^{\Oh(1)}$ on graphs of size $n$
  and arity at most $d_0$,
  even if the input is given with a path decomposition of width $k$.

  Let $G$ be an \CountDomSetRel{\sigma'}{\rho'} instance
  with arity at most $d_0$.
  Note that, by definition, $G$ is an instance of
  $2$-relation weighted \CountDomSetRel{\sigma'}{\rho'}.
  We sequentially use the \pwar-reductions from
	\cref{lem:count:makeSigmaCofinite,lem:count:makeRhoCofinite,%
		lem:count:relationWeightedToVertexWeighted,%
		lem:count:vertexWeightedToUnweighted}
	to obtain polynomially (in $n$) many instances $H_i$ of \srCountDomSetRel
  whose solutions can be used to compute a solution for
  \CountDomSetRel{\sigma'}{\rho'} on input $G$.

  First observe that, for each intermediate problem involving weights,
  the number of different weights is bounded by a constant.
  Since the \pwar-reductions are pathwidth-preserving
  (see \cref{obs:pwredtransitivity}),
	we have $\pw(H_i) \le \pw(G) + \Oh(1)$
	and $\abs{H_i} \le \abs{G}^{\Oh(1)}$.
	Similarly, since all reductions are also arity-preserving, there is a constant $d'$, such that, for each $i$, we have $\ar(H_i) \le \ar(G) + d'$.

  We set $d\deff d_0 + d'$ as the constant in the lemma statement.
  Now assume there is an algorithm solving \srCountDomSetRel
  on instances of size $N$ and arity at most $d$ in time
  $(\abs{A}-\eps)^{k + \Oh(1)} \cdot N^{\Oh(1)}$
  if the input is given with a path decomposition of width $k$.

	We apply this algorithm to all instances $H_i$
  to recover the solution for \CountDomSetRel{\sigma'}{\rho'} on input $G$.
  This can be done in time at most
	\begin{align*}
		\sum_i (\abs{A} - \epsilon)^{\pw(H_i) + \Oh(1)} \cdot \abs{H_i}^{\Oh(1)} + \abs{G}^{\Oh(1)}
		= (\abs{A} - \epsilon)^{\pw(G) + \Oh(1)} \cdot \abs{G}^{\Oh(1)}
		.
	\end{align*}
	By \cref{lem:count:intermediateLowerBoundWhenHavingSuitableGadget},
	this immediately contradicts \#SETH, and hence, finishes the proof.
\end{proof}

\section{Realizing Relations: Basics and Decision Version}
\label{sec:realizingRelations}
In \cref{sec:LBforRelations}, we established lower bounds under SETH
for the intermediate problem \srDomSetRel and its counting version.
In order to show that these lower bounds transfer
to the original problem versions (without the relations),
we show in this section how to express arbitrary relations
using graphs with portals.
In particular, for a
given \(d\)-ary relation~\(R\), we wish to construct a graph with portals whose compatible language
is \emph{equivalent} to \(R\). We call such a graph a \emph{realization of \(R\)}, formally defined in \cref{def:realization}.
Crucially, we wish to add (a realization of) a relation to a set of vertices,
without invalidating the $\sigma$- and $\rho$-constraints of these vertices.
Hence, we require that a realization does not add any \emph{selected} neighbors to the
portal vertices.

We first show in \cref{sec:RelToHW1} how to model arbitrary relations using $\HWeq{1}$ relations only. We show a parsimonious reduction since we use it both in the decision setting and in the counting setting. A crucial requirement for this reduction is the restriction to arbitrary relations whose arity is bounded by a constant. As a side note, this is also the reason why we had to establish the lower bound in \cref{lem:lowerBoundWhenHavingSuitableGadget} for this setting.

Afterward, in \cref{sec:realizingRelations:decision},
we prove the result for the decision version,
that is, \cref{lem:lower:dec:removingRelations}.
For this, we show how to realize $\HWeq{1}$ relations using graphs with portals, where the size of such a realization is always
bounded by a function of the arity $d$,
and the values of $\rhoMax$ and $\sigMax$.
This then gives the sought-after pathwidth-preserving reduction
from the decision problem with relations to the one without.

For the counting version, the main result is \cref{thm:RelfromDS}, and \cref{sec:realizingrelationsforcounting} is devoted to proving this result. Here, the situation is more complicated.
Indeed, the proof of \cref{lem:lower:count:removingRelations} is very technical and
splits into a number of cases and sequences of reductions. In the counting version, it is not plausible that relations can be directly realized by attaching some graph gadget. Instead, we heavily rely on interpolation to isolate the number of selections we wish to determine.
An overview of the corresponding reductions is given in \cref{fig:count:removingRelations}.

\subsection{\boldmath Realizing Relations using
\texorpdfstring{\HWeq{1}}{Hamming Weight One Relations}}
\label{sec:RelToHW1}

Recall from \cref{def:hammingWeightOneAndEquality},
that we denote by $\HWeq[d]{1}$ the $d$-ary \emph{Hamming weight one} relation
where exactly one portal must be selected,
and by $\EQ{d}$ the $d$-ary \emph{equality} relation
where either all or no portals must be selected.
By \HWset{1} relations, we refer to the set of relations \HWeq[d]{1} for some arity $d$. Similarly, we talk about \EQset relations when we mean the set of equality relations of any arity.

In a first step, we show how to replace arbitrary relations
by gadgets using only \HWset{1} relations and \EQset relations.
In a second step, we then remove the \EQset relations
by modeling them using only \HWset{1} relations.

To formally state the steps of the reductions,
we introduce some notation.
Recall that \srCountDomSetRel allows arbitrary relations to appear.
For some set of relations $\CR$,
we write \srCountDomSetRel[\CR] for the restricted problem
where only relations in $\CR$ are used.
Slightly abusing notation, we write \srCountDomSetRel[{\HWset{1}},\EQset]
for the restriction of \srCountDomSetRel
to instances in which all relations are either \EQset relations or \HWset{1} relations (of any arity).

In order to ensure that certain vertices can obtain a feasible number of selected
neighbors, we again use (parsimonious) $\{\sigma_s,\rho_r\}$-providers for some $s\in \sigma$ and $r\in \rho$.
Recall that said providers exist by \cref{lem:count:parsFillingGadget}.
However, the proof of \cref{lem:count:parsFillingGadget} uses arbitrary relations to
obtain a parsimonious provider.
As we intend to remove just these arbitrary relations, we first show in \cref{lem:count:parsFillingGadgetWithHWone} how to obtain parsimonious $\{\sigma_s,\rho_r\}$-providers that use only \HWset{1} relations.
The proof is based on \cref{lem:happyGadget}, which we restate for the reader's convenience.

\lemHappyGadgetCounting*

\begin{lemma}
	\label{lem:count:parsFillingGadgetWithHWone}
	\label{clm:happyGadget}
  Let $\sigma,\rho$ denote non-empty sets with $\rho \neq \{ 0 \}$.
	There are $s\in \sigma$ and $r \in \rho$ such that
	there is a parsimonious $\{\sigma_s, \rho_r\}$-provider
	that uses only \HWeq[2]{1} relations.
\end{lemma}
\begin{proof}
	There are $s\in \sigma$ and $r\in \rho$ with $s\ge 0$ and $r\ge 1$
	since both $\sigma$ and $\rho$ are non-empty,
	and additionally $\rho\neq \{0\}$.
	Then, by \cref{lem:happyGadget}, there is a graph $G$ with a vertex $w$
	such that there are two solutions $S_0$ and $S_1$
	partitioning the vertex set of $G$ and satisfying $w \notin S_0$,
	$\abs{N(w) \cap S_0} = r$,
	$w \in S_1$, and $\abs{N(w) \cap S_1} = s$.

	Then each pair of vertices $v\in S_0$ and $u\in S_1$
	is subject to a relation \HWeq[2]{1}.
	These relations ensure that in each partial solution
	either all vertices from $S_0$ are selected and none from $S_1$,
	or all vertices from $S_1$ are selected and none from $S_0$.
	The gadget does not contain any other vertices so
	it has precisely these two partial solutions.
\end{proof}

\paragraph*{Handling Arbitrary Relations.}
Now our goal is to
replace arbitrary relations
by graphs using only \HWset{1} and \EQset relations
such that the ``behavior of the graph'' does not change.
For this, we formally introduce the notion of a realization of a relation.
This is a graph with portals (and possibly some restricted set of relations) that ``simulates'' the relation.

\begin{definition}[Realization of a relation]\label{def:realization}
	For a set of vertices $S$ with $d=\abs{S}$,
	let $R\subseteq 2^{S}$ denote a $d$-ary relation.
	For an element \(r \in R\), we
	write \(x_r\) for the length-\(d\) string
	that is \(\sigma_0\) at every position \(v\in r\),
	and \(\rho_0\) at the remaining positions, i.e.,
	\[
		x_r\position{v} \deff
			\begin{cases}
				\sigma_0 & \text{if \(v \in r\)},\\
				\rho_0   & \text{otherwise.}
			\end{cases}
	\]
	We set $L_R \deff \{ x_r \mid r \in R \}$.
	Let $G=(V,E,\CC)$ be a graph with relations.
	Let $U=\{ \port_1,\dots,\port_d \}\subseteq V$ be a set of portals of $G$.

	Slightly overloading notation, we say that $H$ \emph{realizes} $R$
	if $H$ realizes $L_R$ (as defined in \cref{def:providersWithRelations}).
	We say that $R$ is realizable
	if there is a graph with $d$ portals that realizes $R$.
	Such a realization is \emph{parsimonious} if, for each $x\in L_R$,
	there is exactly one partial solution in $H$ that witnesses $x$.
\end{definition}

With this definition at hand,
we first show that we can realize arbitrary relations
using \HWset{1} and \EQset relations.

\begin{lemma}
  \label{lem:realizing:arbitraryRelations}
	Let $\sigma,\rho$ denote non-empty sets with $\rho \neq \{ 0 \}$,
  and let $Q \subseteq 2^{\numb{d}}$ denote an arbitrary relation.
  Then, $Q$ is parsimoniously realizable by a graph $\CG$
	with relations and portals with the following properties:
  \begin{itemize}
  	\item All relations used in $G$ are \HWset{1} or \EQset relations of arity at most $2^d+1$.
  	\item The size of $\CG$ is $\Oh(2^d\cdot d)$.
  \end{itemize}
\end{lemma}

\begin{figure}[p]
  \centering
  \includegraphics[scale=2.55]{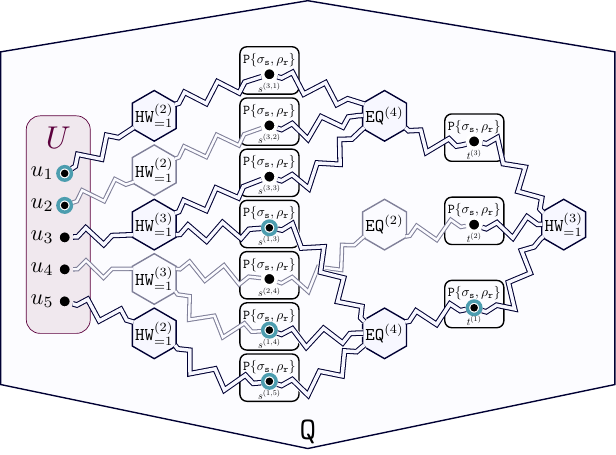}
  \includegraphics[scale=1.25]{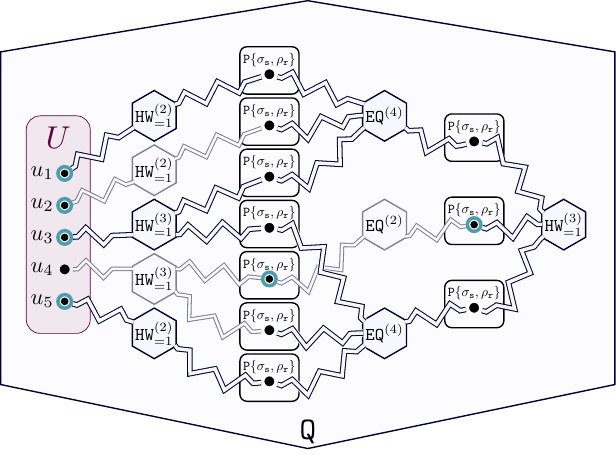}
  \includegraphics[scale=1.25]{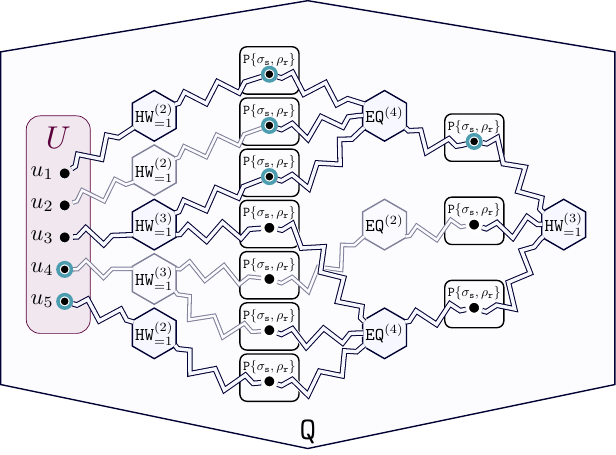}
  \caption{
    The realization of $Q=\{ \{1,2\}, \{1,2,3,5\},\{4,5\} \}$
    and its three partial solutions.
    For providers, we depict only their portal vertices;
    also we write \(\tt P\{\sigma_s,\rho_r\}\)
    for a parsimonious \( \{\sigma_{s},\rho_{r}\}\)-provider
		from \cref{lem:count:parsFillingGadgetWithHWone}.
    Further, hexagons depict relations
    that are realized between vertices.
  }
  \label{fig:arbitraryrelation}
\end{figure}

Observe that the proof uses the construction from \cite[Lemma~3.3]{CurticapeanM16}.
However, we need to modify the construction from \cite[Lemma~3.3]{CurticapeanM16} to account for the fact that we are selecting vertices rather than edges.

\begin{proof}
	Suppose $Q=\{q_1,\dots,q_{\abs{Q}}\}$ where $q_i \subseteq \numb{d}$.
	We construct a graph with portals $\CG=(G,\Port)$ with $\Port=\{\port_1,\dots,\port_d\}$.
	For each set \(q_i\),
	we introduce \( k_i\deff \abs{\numb{d} \setminus q_i} \)
	independent copies \( (H^{(i,j)}, \{s^{(i,j)}\}) \) ($j\in \numb{d}\setminus q_i$)
	plus an extra copy \( (H^{(i)}, \{t^{(i)}\}) \)
	of a parsimonious \(\{\sigma_s,\rho_r\}\)-provider
	with \HWeq[2]{1} relations  to \(G\) (as given by \cref{lem:count:parsFillingGadgetWithHWone}).
	They are connected as follows.
	We put an \EQ{k_i + 1} relation
	between the vertices $\{s^{(i,j)} \mid j\in \numb{d}\setminus q_i\}$ and \(t^{(i)}\).
	Additionally, we put a \(\HWeq[\abs{Q}]{1}\) relation between the vertices
	$\{t^{(1)}, \ldots, t^{(\abs{Q})}\}$.
	Finally, for each \(j \in \numb{d}\),
	we put a \(\HWeq{1}\) relation between \(\port_j\)
	and all vertices of the form \(s^{(\star,j)}\).
	Consult \cref{fig:arbitraryrelation} for a visualization of an example.

	\begin{claim}
		The constructed graph \(H\) parsimoniously realizes \(Q\).
	\end{claim}
	\begin{claimproof}
		Let $q_i \in Q$. Then, the following selection is a solution that ``witnesses'' $q_i$.
		We select the portal vertices $\port_j$ with $j \in q_i$.
		In addition, we select $t^{(i)}$ and the vertices \(s^{(i,j)}\) for $j \in \numb{d}\setminus q_i$.
		We do no select any other vertices \(\port_{\star}, s^{(\star,\star)}\) or
		\(t^{(\star)}\).
		Observe that, for each $i$ and $j$, the parsimonious
		\(\{\sigma_s,\rho_r\}\)-provider \( (H^{(i,j)}, \{s^{(i,j)}\}) \) ensures
		that both selecting or not selecting \(s^{(i,j)}\) can be extended to satisfy the $\sigma$- or $\rho$-constraints of \(s^{(i,j)}\), respectively. The same holds for the vertex $t^{(i)}$. So, vertices from the providers are selected such that all the corresponding $\sigma$- and $\rho$-constraints are satisfied.

		It is clear that this selection witnesses $q_i$. Next, we show that the selection is a solution, and for this, we need to show
		that our selection satisfies all relation constraints on
		the vertices of \(H\) that we introduced.

		The vertices $\{t^{(1)}, \ldots, t^{(\abs{Q})}\}$ are subject to the relation
        \(\HWeq[\abs{Q}]{1}\), and this is fine since we selected only \(t^{(i)}\) from this set.
        Since we also selected the vertices \(s^{(i,j)}\) for $j \in \numb{d}\setminus q_i$, the \(\EQset\)-relation on these vertices and \(t^{(i)}\) is satisfied;
		for all remaining \(\EQset\) relations, all corresponding vertices are unselected
		(which in turn also satisfies said \(\EQset\) relations). Finally, for each \(j \in
		\numb{d}\), we selected precisely one of \(\port_j\) (if \(j \in q_i\)) or \(s^{(i,j)}\) (if \(j
		\in \numb{d}\setminus q_i\)). Thus, the remaining \(\HWset{1}\) relations are also satisfied.
		In total, we conclude that our selection is a valid partial solution for \(H\).

		In order to show the claim, it remains to show that the realization is parsimonious. Since the used providers are parsimonious and the selection status of all vertices outside the providers is determined by the relations and the selection of vertices from $\Port$,
		the constructed partial solution is the only one that witnesses $q_i$.

		Finally, we show that there are no partial solutions that extend selections from $\Port$ that do not belong to $Q$. For this, assume that we are given some partial solution.
		As the relations are satisfied,
		there is exactly one vertex $t^{(i)}$ that is selected.
		By the $\EQset$ relations, we also get
		that all $s^{(i,j)}$ are selected for $j\in \numb{d}\setminus q_i$,
		and that all remaining $s^{(\star,\star)}$ are unselected.
		Finally, from the remaining \(\HWset{1}\) relations, we get that, for each \(j \in
		\numb{d} \setminus q_i\), the portal \(\port_j\) is unselected (as \(s^{(i,j)}\) is
		already selected), and that, for each \(j \in q_i\), the portal \(\port_j\) is selected
		(as all vertices \(s^{(i',j)}\) are unselected in this case, as \(i' \ne i\)).
		In particular, the indices of the selected portals correspond to \(q_i\), which
		completes the proof of the claim.
	\end{claimproof}

	For the bound on the size of $H$, observe that $\abs{Q}\le 2^d$, and
	the size of a parsimonious $\{\sigma_s,\rho_r\}$-provider is upper-bounded
	by a function in $\rhoMax$ and $\sigMax$
	(see \cref{lem:count:parsFillingGadgetWithHWone}), which is a constant.
	As there are at most $2^d \cdot d$ vertices $s^{(i,j)}$
	which are all the portal of some $\{\sigma_s,\rho_r\}$-provider,
	the size of the graph is bounded by $\Oh(2^d \cdot d)$.

  Moreover, the arity of the graph is bounded by $2^d+1$
	as the \HWset{1} relations have arity at most $\abs{Q}+1 \le 2^d+1$
	and the \EQset relations have arity at most $d+1$.
\end{proof}

\begin{lemma}
	\label{lem:realizing:parsiArbitraryToHWoneAndEQ}
	Let $\sigma,\rho$ denote non-empty sets with $\rho \neq \{ 0 \}$.
	For all constants $d$,
	there is a parsimonious pathwidth-preserving reduction from \srDomSetRel
	on instances of arity at most $d$
	to \srDomSetRel[{\HWset{1}},\EQset]
	on instances of arity at most $2^d+1$.
\end{lemma}
\begin{proof}
	Let $G$ be an instance of \srDomSetRel.
	We replace each relation in $G$ that is neither a \HWset{1} relation nor an \EQset relation
	with the corresponding realizer from \cref{lem:realizing:arbitraryRelations}.
	Let $H$ be the resulting
	\srDomSetRel[{\HWset{1}},\EQset] instance.
	Indeed, by \cref{lem:realizing:arbitraryRelations},
	$H$ contains only \HWset{1} or \EQset relations.
	Moreover, by the construction of the realization,
	the number of solutions is preserved, and hence, the reduction is parsimonious.

	We first analyze the arity of $H$.
	We know that the arity of $G$ is at most $d$.
	From \cref{lem:realizing:arbitraryRelations}, we know
	that the arity for each realization is at most $2^d+1$.
	Hence, the arity of $H$ itself is at most $2^d+1$.

	It remains to argue about the pathwidth of $H$.
	Assume that we are given a path decomposition of $G$.
	Let $R$ be a relation in $G$ which we replaced by a realization.
	By definition, there is a bag $B$ in the path decomposition of $G$
	containing all vertices in the scope of $R$.
	We duplicate bag $B$ to obtain a copy $B_R$
	and add all vertices of the realization of $R$ to $B_R$.
	Observe that this is a valid path decomposition of the resulting graph
	of width at most $\pw(H) \le \pw(G) + \Oh(2^d d) \le \pw(G) + \Oh(1)$
	as $d$ is a constant.

	Note that copying the original bag for each relation separately is crucial.
	Indeed, we cannot simply add the vertices of all realizations to the same bag
	as this could increase the size of the bag by more than a constant (since there might be many relations).
\end{proof}

\paragraph*{Handling Equality Relations.}
The next step is to remove the \EQset relations
and replace them by realizers which use only \HWset{1} relations.

\begin{lemma}
  \label{lem:realizing:eqConditional}
  Let $\sigma,\rho$ denote non-empty sets with $\rho \neq \{0\}$.
  Then, for any $k\geq 1$, there is a graph $G$ with relations and portals that parsimoniously realizes $\EQ{k}$ such that $G$ has the following properties:
	\begin{itemize}
		\item All relations used in $G$ are \HWeq[2]{1} relations.
		\item The size of $G$ is $\Oh(k)$.
		\item $G\setminus \Port$ has pathwidth $\Oh(1)$.
	\end{itemize}
\end{lemma}
\begin{figure}[t]
  \centering
  \includegraphics[scale=2]{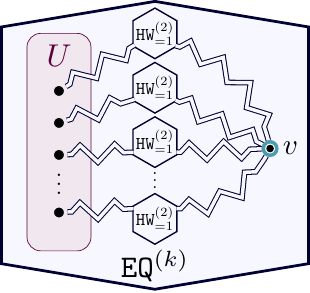}\qquad
  \includegraphics[scale=2]{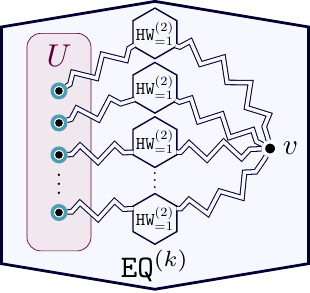}
  \caption{
    The two solutions of the
    gadget construction from \cref{lem:realizing:eqConditional}
    for realizing \(\EQ{k}\).
    For realizers, we depict only their portal vertices; we use hexagons to
    depict that some relation is realized between the vertices connected to said
    hexagons.
  }
  \label{fig:EQ}
\end{figure}
\begin{proof}
  First, observe that we can trivially realize $\EQ{1}$
  by a graph that consists of a single portal vertex.

  Suppose that $k\ge 2$.
  We construct a graph with portals \((G, \Port)\) with
  \(\Port = \{\port_1,\dots,\port_k\}\).
  We add a single parsimonious $\{\sigma_s,\rho_r\}$-provider
  \((H, \{v\})\) with \HWeq[2]{1} relations
  (as given by \cref{lem:count:parsFillingGadgetWithHWone}).
  Finally, we add a \(\HWeq[2]{1}\) relation pairwise between \(v\)
  and each of \(\port_1,\dots,\port_k\).
  Consult \cref{fig:EQ} for a visualization.

  To see that $G$ realizes $\EQ{k}$,
	first observe that if at least one portal vertex \(\port_i\) is selected,
	then $v$ is forced to be unselected by the $\HWeq[2]{1}$ relations.
  This forces the remaining portal vertices to be selected.
  On the other hand, if no portal \(\port_i\) is selected,
  then we can get a valid solution by selecting $v$.

	The size bound follows trivially as we added only one parsimonious $\{\sigma_s,\rho_r\}$-provider
  \((H, \{v\})\) with \HWeq[2]{1} relations,
	and $k$ \HWeq[2]{1} relations.
	Likewise, the bound on arity and pathwidth follow.
\end{proof}

\begin{lemma}
  \label{lem:realizing:parsiHWoneAndEQToHWone}
  Let $\sigma,\rho$ denote non-empty sets with $\rho \neq \{0\}$.
	Then, there is a parsimonious \pwar-reduction
  from \srDomSetRel[{\HWset{1}}, \EQset]
  to \srDomSetRel[{\HWset{1}}].
\end{lemma}
\begin{proof}
  Let $G$ be an instance of \srDomSetRel[{\HWset{1}}, \EQset].
  We use the construction from \cref{lem:realizing:eqConditional}
  to replace all \EQset relations in $G$ by the appropriate realizers
  which use \HWset{1} relations.
  Let $H$ be the obtained instance of \srDomSetRel[{\HWset{1}}].
  By \cref{lem:realizing:eqConditional}, all relations used in $H$ are $\HWset{1}$ relations. Moreover, the realizations ensure that the number of solutions is preserved, and thereby, the reduction is parsimonious.
	Note that this reduction is now also arity-preserving
	as the realization of the \EQset relations has arity at most $2$.

	It remains to argue about the pathwidth of $H$.
	Assume we are given a path decomposition of $G$.
	Let $R$ be a relation in $G$ which we replaced by a realization.
	By the definition of the path decomposition,
	there is a bag $B$ containing all vertices in the scope of $R$.
	Let $X_1,\dots,X_\ell$ be the bags of width at most $\Oh(1)$
	from the path decomposition of the realization of $R$.
	We duplicate the bag $B$ $\ell$ times
	and let $B, B_1,\dots,B_\ell$ denote the resulting copies.
	We combine the bags $B_i$ and $X_i$ into a new bag
	of the path decomposition of $H$.
	Observe that this is a valid path decomposition
	of width at most $\pw(H) \le \pw(G) + \Oh(1)$.
\end{proof}

This concludes the first step of removing the relations.
We formally summarize the result by the following corollary
which follows directly from
\cref{lem:realizing:parsiArbitraryToHWoneAndEQ,%
lem:realizing:parsiHWoneAndEQToHWone}.
\begin{corollary}
	\label{cor:realizing:parsiArbitraryToHWone}
	\label{lem:RelfromDShelper}\label{lem:Irelversion}
	Let $\sigma,\rho$ denote non-empty sets with $\rho \neq \{ 0 \}$.
	For all constants $d$,
	there is a parsimonious pathwidth-preserving reduction from \srDomSetRel
	on instances of arity at most $d$
	to \srDomSetRel[{\HWset{1}}]
	on instances of arity at most $2^d+1$.
\end{corollary}

In the remainder of this section,
we show how to replace the \HWset{1} relations by realizations.
For this, we distinguish between the decision version (\Cref{sec:realizingRelations:decision}) and the counting version (\Cref{sec:realizingRelations:counting}).

\subsection{Decision Version: Proof of \texorpdfstring{\cref{lem:lower:dec:removingRelations}}{Lemma \ref{lem:lower:dec:removingRelations}}}
\label{sec:realizingRelations:decision}

In this section, we show how to replace \HWeq{1} relations
by gadgets which do not use any relations at all. These gadgets ensure the existence of a solution that respects the relations, and are not necessarily parsimonious.
We again make use of the notion of a realization from \cref{def:realization}.
Crucially, however, while \cref{def:realization} defines realizers and providers as graphs with relations and portals, in this section, the constructed providers and realizers do not use relations, i.e., they are simply graphs with portals.
We refrain from stating this in each result --- it is easily verifiable as the constructions do not use relations.

We start with a helpful gadget to ensure that a vertex \(v\) is selected (with \(\sigMax\) selected
neighbors) in any partial solution, that is, how to realize the language
\(\{\sigma_{\sigMax}\}\).

\begin{lemma}
    \label{lem:forcingGadget}
    Let \(\sigma\) and \(\rho\) denote finite and non-empty sets with $0 \notin \rho$.
    Then, there is a \(\{\sigma_{\sigMax}\}\)-realizer.
\end{lemma}
\begin{proof}
    We start with a useful helper gadget
    which ensures that the portal is always selected
    and receives at least $\sigMin$ selected neighbors in the gadget.
    \begin{claim}\label{cl-helper-gadget}
        There is a \(\{\sigma_{\sigMin}\}\)-provider \( (H, \{u\}) \)
        such that
        \(L(H, \{u\}) \subseteq \{ \sigma_i \in \sigStates \mid \sigMin \le i\}\).
    \end{claim}
    \begin{claimproof}
        \begin{figure}[t]
            \centering
            \begin{subfigure}[t]{.48\textwidth}
                \centering
                \includegraphics[scale=1.3]{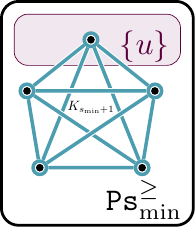}\qquad
                \includegraphics[scale=1.3]{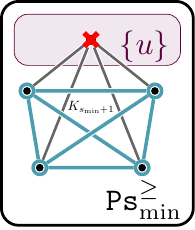}
                \caption{The \(\{\sigma_{\sigMin}\}\)-provider for Case 1 when \(0 < \sigMin < \rhoMin\).
                    Observe that not selecting the portal vertex violates the
                \(\rho\)-constraints of the portal vertex.}\label{sf-139}
            \end{subfigure}%
            \hfill
            \begin{subfigure}[t]{.48\textwidth}
                \centering
                \includegraphics[scale=1.3]{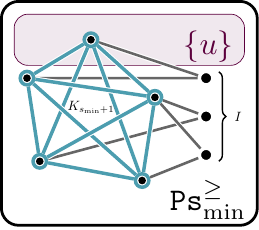}\qquad
                \includegraphics[scale=1.3]{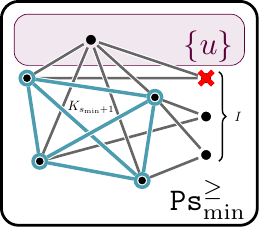}
                \caption{The \(\{\sigma_{\sigMin}\}\)-provider for Case 2 when \(\sigMin \ge \rhoMin \ge 1\).
                    Observe that not selecting the portal vertex violates the
                \(\rho\)-constraints of at least one vertex from the independent set \(I\).
            }\label{sf-148}
            \end{subfigure}%
            \hfill
            \begin{subfigure}[t]{.49\textwidth}
                \centering
                \includegraphics[scale=1.3]{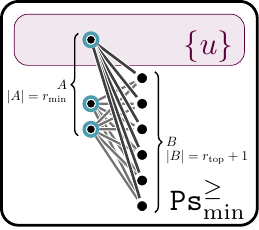}\qquad
                \includegraphics[scale=1.3]{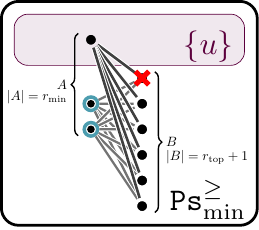}
                \caption{The \(\{\sigma_{\sigMin}\}\)-provider for Case~3 when \(\sigMin=0\), \(\rhoMin \ge 1\).
                    Observe that not selecting the portal vertex violates the
                \(\rho\)-constraints of at least one vertex from \(B\).
            }\label{sf-182}
            \end{subfigure}
            \caption{The gadget constructions from \cref{cl-helper-gadget}.}
        \end{figure}
        Recall that we need to construct a graph \(H\)
        with a single portal vertex \(u\) that satisfies
        $\{\sigma_{\sigMin}\}\subseteq L(H,\{u\})
        \subseteq \{ \sigma_i \in \sigStates \mid \sigMin \le i \}$.
        We consider three cases.

        \paragraph{Case 1: ${0 < \sigMin < \rhoMin}$.}
        We choose \(H\) to be a clique of order $\sigMin+1$;
        we declare any of the vertices to be \(u\); consult \cref{sf-139} for a
        visualization.

        Now, suppose that one vertex \( v \neq u \) of \(H\) is unselected;
        this vertex exists as \( \sigMin+1 \ge 2 \).
        Then, \(v\) needs at least $\rhoMin$ selected neighbors.
        However, \(v\) can have at most $\sigMin < \rhoMin$ selected neighbors,
        and thus, \(v\) cannot be unselected.
        Since \(v\) is selected, it needs at least $\sigMin$ selected neighbors,
        which implies that the remaining $\sigMin$ vertices of the clique
        (including \(u\)) must be selected as well.

        \paragraph{Case 2: \({\sigMin \ge \rhoMin \ge 1}\).}
        We construct \(H\) as follows. Starting from a clique \(C = K_{\sigMin + 1}\) (of
        \(\sigMin + 1\) vertices) and an independent set
        \(I = I_{\ceil{(\sigMin+1)/\rhoMin}}\) (of \(\ceil{(\sigMin+1)/\rhoMin}\) vertices),
        we add \(\rhoMin \cdot \ceil{(\sigMin+1)/\rhoMin}\) edges between \(C\) and \(I\) such
        that each vertex in \(I\) has a degree of exactly \(\rhoMin\), and such that
        every vertex in \(C\) is adjacent to at least one vertex in \(I\). It is readily
        verified that this is always possible.
        Finally, we pick any vertex from \(C\) to be the portal vertex \(u\).
        Consult \cref{sf-148} for a visualization.

        Consider an arbitrary vertex $v$ of the clique $C$.
        Assume that $v$ is unselected.
        Let $v'$ be a neighbor of $v$ in the independent set $I$.
        If $v'$ is selected,
        it needs (at least) $\sigMin\ge \rhoMin$ selected neighbors in $C$,
        and if $v'$ is unselected,
        it needs (at least) $\rhoMin$ selected neighbors in $C$.
        By our assumption that $v$ is unselected and the construction of $H$,
        $v'$ can have at most $\rhoMin-1$ selected neighbors in $C$,
        which cannot be turned into a valid solution.
        Hence, $v$ must be selected.
        As this argument applies to each vertex $v$ of the clique $C$,
        the whole clique $C$ has to be selected,
        which gives $\rhoMin$ selected neighbors to every vertex of $I$.

        \paragraph{Case 3: ${\sigMin=0, \rhoMin \ge 1}$.}
        We choose $H$ to be a complete bipartite graph with bipartition $(A,B)$
        such that $\abs{A}=\rhoMin$ and $\abs{B}=\rhoMax+1$.
        We pick an arbitrary vertex $u$ in $A$ as the portal
        vertex \(u\).
        Consult \cref{sf-182} for a visualization.

        Suppose that $u$ is not selected.
        Then, every vertex in $B$ can have at most
        $\rhoMin-1$ selected neighbors,
        which means that every vertex in $B$ must be selected.
        In this case, $u$ would have $\rhoMax+1$ selected neighbors,
        which is not feasible. Therefore, $u$ must be selected.
        If \(u\) is indeed selected, then selecting every vertex in
        $A$, and leaving every vertex in $B$ unselected
        yields a solution for which $u$ has $\sigMin=0$ selected neighbors.
    \end{claimproof}
    \begin{figure}[t]
        \centering
        \includegraphics[scale=1.5]{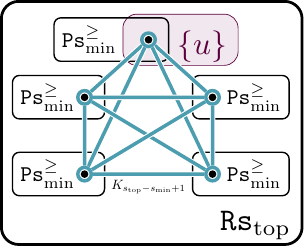}
        \caption{The main construction of \cref{lem:forcingGadget}: we use the providers
            from \cref{cl-helper-gadget} to obtain \(({\tt R\sigMax},\{u\})\). Observe that
            \(u\) has exactly \(\sigMin + (\sigMax-\sigMin) = \sigMax\) selected
            neighbors. Indeed, it has \(\sigMin\) selected neighbors from inside the provider, and
            \((\sigMax-\sigMin)\) selected neighbors from the selected portals of the other
        instances of the provider. For providers, only their portal vertex is depicted.}
    \end{figure}

    Finally, we use \cref{cl-helper-gadget} to construct a
    \(\{\sigma_{\sigMax}\}\)-realizer \( (G, \{u\}) \) as follows.
    As our graph \( G \), we use \(\sigMax-\sigMin+1\ge 1\)
    independent copies \( (H^{(i)}, \{u^{(i)}\}) \) of the provider
    from \cref{cl-helper-gadget}, and connect the vertices \(u^{(i)}\) to form an
    \((\sigMax-\sigMin + 1)\)-clique. We choose any of the vertices \(u^{(i)}\) to be the
    portal vertex \(u\).

    By \cref{cl-helper-gadget}, all vertices \(u^{(i)}\) are selected in any partial
    solution for \( (G,\{u\}) \) with the additional guarantee that each vertex \(u^{(i)}\)
    already has at least \(\sigMin\) selected neighbors inside of the corresponding helper
    gadgets.
    Hence, each vertex \(u^{(i)}\) (and in particular \(u\))
    has at least \( \sigMax \) selected neighbors.
    More precisely, as \( \sigma \) is finite,
    each such vertex \( u^{(i)} \) has exactly \( \sigMax \) selected neighbors
    and, in particular, exactly \( \sigMin \) selected neighbors in \( H^{(i)} \),
    which is a valid solution by \cref{cl-helper-gadget}.
    This completes the proof.
\end{proof}

As a next step, we realize \HWeq{1} relations.

\begin{lemma}
    \label{lem:realizing:hwExactly1}\label{lem:HW1}
    Let $\sigma,\rho$ denote finite non-empty sets with $0 \notin \rho$.
    Then, for all $k\geq 1$,
    $\HWeq[k]{1}$ is realizable.
\end{lemma}

\begin{figure}[t]
    \centering
    \begin{subfigure}[b]{0.48\textwidth}
        \centering
        \includegraphics[width=\textwidth]{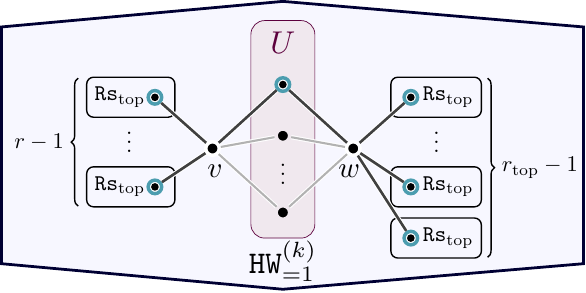}
        \caption{The construction in Case~1 when there is an $r\in \rho$
        with $r\ge 2$ and $r-1\notin \rho$.}
        \label{fig:hwExactly1:gap}
    \end{subfigure}
    \hfill
    \begin{subfigure}[b]{0.48\textwidth}
        \centering
        \includegraphics[width=\textwidth]{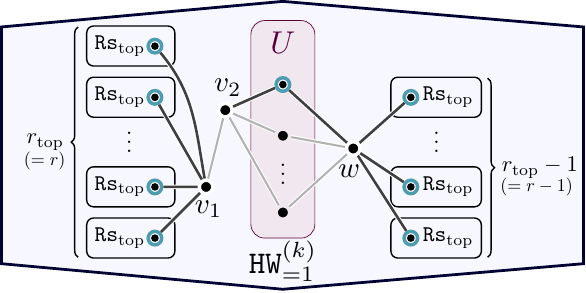}
        \caption{The construction in Case~2 when $\rho=\numb{r}$ for some $r\geq 2$.}
        \label{fig:hwExactly1:noGap}
    \end{subfigure}
    \begin{subfigure}[b]{0.48\textwidth}
        \centering
        \includegraphics[width=\textwidth]{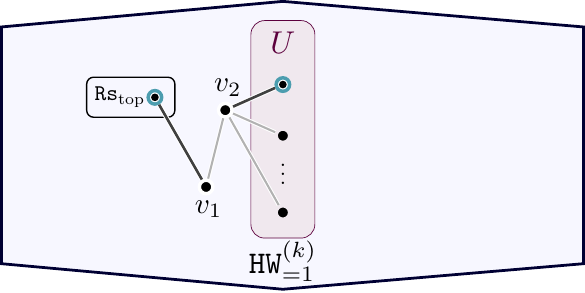}
        \caption{The construction in Case~3 when $\rho=\{1\}$.}
        \label{fig:hwExactly1:c}
    \end{subfigure}
    \caption{The gadget constructions from \cref{lem:HW1} for realizing \(\HWeq[k]{1}\).
        For realizers, we depict only their portal vertices; we depict an arbitrary vertex
        from \(U\) as being selected.}
\end{figure}

\begin{proof}
    We distinguish three different cases for the definition of the realization
    $\CG=(G,\Port)$ with portals $\Port = \{\port_1,\dots,\port_k\}$.

    \paragraph{\boldmath Case 1: There is an $r\in \rho$ with $r\ge 2$ and $r-1\notin \rho$.}
    We create two additional vertices $v$ and $w$,
    and make both adjacent to each portal $\port_i$.
    Further, we take \(r-1\) independent copies
    \( (H^{(i,1)}, \port^{(i,1)}) \) of the \( \{\sigma_{\sigMax}\} \)-realizer
    from \cref{lem:forcingGadget},
    and connect \(v\) to the \(r-1\) vertices \(\port^{(i,1)}\).
    Next, we take \(\rhoMax-1\) independent copies
    \( (H^{(i,2)}, \port^{(i,2)}) \)
    of the \( \{\sigma_{\sigMax}\} \)-realizer from \cref{lem:forcingGadget},
    and connect \(w\) to the \(\rhoMax-1\) vertices \(\port^{(i,2)}\).
    Consult \cref{fig:hwExactly1:gap} for a visualization.%

    To see that $\CG$ realizes $\HWeq[k]{1}$,
    first observe that by \cref{lem:forcingGadget},
    any vertex \(\port^{(\star,\star)}\) is always selected with \(\sigMax\) selected
    neighbors (inside of the corresponding realizer).
    Hence, no neighbor  of any \(\port^{(\star,\star)}\) outside the corresponding
    realizer can be selected. In particular, \(v\) and \(w\) have to be unselected.

    Since $r-1\not\in \rho$,
    at least one portal vertex \(\port_1,\dots,\port_k\) must be selected into the solution
    to satisfy the \(\rho\)-constraint of $v$.
    Further, as $\rho$ is finite and $w$ already has $\rhoMax-1$ selected neighbors,
    at most one of the $k$ portal vertices \(\port_1,\dots,\port_k\) can be selected.
    Thus, we must select exactly one portal vertex into the solution, completing the proof
    for this case.

    \paragraph{\boldmath Case 2: $\rho=\numb{r}$ for $r\geq 2$.}
    We create three additional vertices $v_1, v_2$, and $w$,
    and make \(v_2\) and \(w\) adjacent to each portal $\port_i$.
    Further, we take \(\rhoMax = r\) independent copies
    \( (H^{(i,1)}, \port^{(i,1)}) \) of the \( \{\sigma_{\sigMax}\} \)-realizer
    from \cref{lem:forcingGadget},
    and connect \(v_1\) to the \(r\) vertices \(\port^{(i,1)}\); we also connect \(v_1\)
    and \(v_2\).
    Next, we take \(\rhoMax-1 = r - 1\) independent copies
    \( (H^{(i,2)}, \port^{(i,2)}) \)
    of the \( \{\sigma_{\sigMax}\} \)-realizer from \cref{lem:forcingGadget},
    and connect \(w\) to the \(r-1\) vertices \(\port^{(i,2)}\).
    Consult \cref{fig:hwExactly1:noGap} for a visualization.

    As in Case~1, $v_1$ and $w$ must be unselected. Now, observe that \(v_2\) must also be
    unselected as \(v_1\) already has \(r\) selected neighbors and \(r + 1 \notin \rho\).
    As \(0 \notin \rho\), at least one of the neighbors of \(v_2\) (but not \(v_1\)) must be
    selected, that is, at least one of \(\port_1,\dots,\port_k\) must be selected.
    Finally, observe that \(\rhoMax = r\), so as in Case~1, the vertex \(w\) ensures that at
    most one of the vertices \(\port_1,\dots,\port_k\) is selected. This completes the
    proof in this case.

    \paragraph{\boldmath Case 3: $\rho=\{1\}$.}
    We use the same construction as in Case~2,
    but we remove the vertex $w$.
    Consult \cref{fig:hwExactly1:c} for a visualization.

    The constraint that exactly one portal is selected
    follows from the properties of $v_2$,
    as it needs exactly one selected neighbor.
    This completes the proof of \cref{lem:HW1}.
\end{proof}

The sought-after reduction now directly follows from \cref{lem:HW1}

\begin{lemma}
  \label{lem:realizing:removingHWone}
  Let $\sigma,\rho$ denote finite non-empty sets with $0 \notin \rho$.
  For all constants $d$,
  there is a pathwidth-preserving reduction
  from \srDomSetRel[\HWsetGen{=1}]
  on instances of arity at most $d$
  to \srDomSet.
\end{lemma}
\begin{proof}
  Let $G$ be an instance of \srDomSetRel[\HWsetGen{=1}]
  where all relations have arity at most $d$.

  We use \cref{lem:realizing:hwExactly1}
  to replace each \HWeq{1} relation by its realization. It is easy to check that these realizations do not use relations, i.e., they are simply graphs with portals.
  Let $H$ be the resulting graph.
  A given path decomposition of $G$ is modified as follows.
  For each \HWeq{1} relation $R$ of $G$,
  there is a bag $B$ in the path decomposition of $G$ that contains $R$
  and the scope of $R$.
  We duplicate the bag $B$ and let $B_R$ denote the copy.
  We add all the vertices of the realization of $R$ to $B_R$.

  As the arity of each relation is bounded by a constant
  and we add only vertices of one realization to a bag,
  we get $\pw(H) \le \pw(G) + \Oh(1)$.
  Hence, the reduction is pathwidth-preserving.
\end{proof}

As the final step, we prove \cref{lem:lower:dec:removingRelations}.

\removingRelationsForDecVersion*

\begin{proof}
  The proof follows from combining the reductions in
  \cref{cor:realizing:parsiArbitraryToHWone,lem:realizing:removingHWone}.

  As the arity of the initial instance of \srDomSetRel
  is bounded by a constant,
  the arity of all constructed instances is also bounded by a constant
  (by the statement of \cref{cor:realizing:parsiArbitraryToHWone}).
  Hence, it follows from the transitivity of pathwidth-preserving reductions
  that the final reduction is also pathwidth-preserving.
\end{proof}

\section{Realizing Relations: Counting Version}
\label{sec:realizingrelationsforcounting}
\label{sec:realizingRelations:counting}

In this section, we show that if the counting problem \srCountDomSet is not polynomial-time solvable,
then it is sufficiently expressive to realize arbitrary relations, that is, we prove \cref{thm:RelfromDS}.
First, we introduce some notation which we use throughout \cref{sec:realizingRelations:counting}.

\paragraph*{Notation.}
Let $\allSets$ be the set of all finite and cofinite subsets of $\NN$,
and let $\allSetsne\deff\allSets\setminus\{\emptyset\}$.

Let $(\sigma, \rho)\in \allSetsne^2$.
For some non-negative integer $k$,
let $\CP \subseteq \allSets^2$ be a set of pairs
$\{(\sigma^{(i)}, \rho^{(i)}) \in \allSets^2\mid i\in \nset{k}\}$.
We now define a generalization of \srCountDomSet
that allows us to augment the problem with additional constraints 
specified by $\CP$ in order to obtain more modeling power.
For ease of notation, let $(\sigma^{(0)},\rho^{(0)})\deff (\sigma, \rho)$.

The problem \srCountDomSetRel[\CP] takes as input some graph $G$
together with a mapping $\lambda \from V(G) \to \fragment{0}{k}$.
Then, a $\lambda$-set of $G$ is a subset $S\subseteq V(G)$
with the following properties:
\begin{itemize}
	\item For each $v\in S$, we have $|N(v)\cap S| \in \sigma^{(\lambda(v))}$.
	\item For each $v\notin S$, we have $|N(v)\cap S| \in \rho^{(\lambda(v))}$.
\end{itemize}
\srCountDomSetRel[\CP] then asks for the number of $\lambda$-sets of $G$.

If $\CP$ is the empty set, then we also drop the superscript
and simply write \srCountDomSet.
In this case, it also suffices to take as input only the graph $G$
and assume that $\lambda$ is constant $0$.
A $\lambda$-set in this case is simply a $(\sigma, \rho)$-set.
We also refer to the set $S$ as a \emph{solution} of \srCountDomSetRel[\CP]
(on input $G,\lambda$).
For a less formal, but more convenient specification of $\lambda$, we say that a vertex $v$ with $\lambda(v)=i$ is a \emph{$(\sigma^{(i)}, \rho^{(i)})$-vertex}.

Sometimes we restrict the problem \srCountDomSetRel[\CP] to instances
for which certain pairs from $\CP$ are used only a constant number of times
(w.r.t.~the number of vertices in $G$).
To this end, for some positive integer $c$,
we say that a pair $(\sigma^{(i)},\rho^{(i)}) \in\CP$ is \emph{$c$-bounded}
if we restrict \srCountDomSetRel[\CP] to instances $G,\lambda$
with $|\lambda^{-1}(i)|\le c$.
In order to define pairs in $\CP$, the following notation is useful.
For a set $\tau$ of non-negative integers (think of $\sigma$ or $\rho$)
and an integer $s$, let $\tau-s \deff \{k-s \mid k\in \tau, k-s\ge 0\}$.

When working with \srCountDomSetRel[\CP],
it is helpful to generalize some definitions from \cref{sec:prelims}.
For a graph $G$ with portals $U$
and a mapping $\lambda \from V(G) \to \fragment{0}{k}$,
the tuple $\CG=(G,U,\lambda)$ is a \emph{gadget} for \srCountDomSetRel[\CP].
Again, we drop $\lambda$ if $\CP=\emptyset$.
Then, a \emph{partial solution} of $(G,U,\lambda)$
is a set $S\subseteq V(G)$ with:
\begin{itemize}
	\item For each $v\in S\setminus U$,
	we have $|N(v)\cap S| \in \sigma^{(\lambda(v))}$.
	\item For each $v\notin S\cup U$,
	we have $|N(v)\cap S| \in \rho^{(\lambda(v))}$.
\end{itemize}

Let $n=|V(G)|$ be the number of vertices in $G$.
Recall that
$\allStates_n\deff \{\sigma_0,\dots, \sigma_n, \rho_0, \dots, \rho_n\}$.
Recall that a partial solution $S$ \emph{witnesses} the string $x\in \allStates_n^U$ if, for each $u\in U$ and $z \deff |N(u) \cap S|$, we have
\[x\position{u} = \begin{cases}
	\sigma_{z} &\text{if $u\in S$}\\
	\rho_{z}	 &\text{otherwise.}
\end{cases}\]
Then, for $x\in \allStates_n^U$, $\ext_{\CG}(x)$ is the number of partial solutions of $\CG$ that witness $x$. We also refer to this as the number of \emph{extensions} of $x$ (to the gadget $\CG$).

Given a set of pairs $\CP$ from $\allSets^2$ and a pair $(\sigma',\rho')\in \allSets^2$, we use the shorthand $\CP+(\sigma',\rho')$ for the set $\CP\cup \{(\sigma',\rho')\}$.

\subsection{Proof of \texorpdfstring{\cref{thm:RelfromDS}}{Theorem \ref{thm:RelfromDS}}}

\begin{figure}
\centering
\begin{tikzpicture}[
  scale=.615,
  transform shape,
]
  \input{drawings/overview-tikz-style}
  \def\x{4.5cm}
  \def\xGap{1.25cm}
  \def\y{-2 cm}
  \def\xW{4.2cm}
  \def\xWimp{4.2cm-1.12mm}

  \crefname{lemma}{Lem.}{Lem.}

  \tikzset{redTur/.style={red}}
  \tikzset{%
  part/.style = { % some specific problem
    draw=gray,
    fill=gray!15,
    rounded corners=2mm,
  },
  lowerB/.style = { % Lower bound shown for a problem
  },
  important/.style = { % Lower bound shown for a problem
    line width=.66mm,
  },
  case/.style = { % Different cases (reachable by switch
  },
  res/.style = { % Reference for the result
    anchor=west,
    yshift=-0.5*\y,
  },
  }

  % Uncomment the following if we do not want to state the full name each time.
  % \renewcommand{\srCountDomSetRel}[1][\textsc{Rel}]{\ensuremath{#1}}
  % \renewcommand{\DomSetGeneral}[4] {\ensuremath{(#1,#2)\mhyph}\textsc{{#3}DS\ensuremath{^{#4}}}\xspace}

  %% BLOCK ONE (making sigma, rho cofinite)
  % \node[problem, lowerB, break, important] (domSetRelFull) at (0,-1*\y)
  %   {\CountDomSetRel{\widehat\sigma}{\widehat\rho}};

  % \node[problem,break] (relWeiDomSetRelSigma) at (0,0*\y)
  %   {relation-wt.\ \CountDomSetRel{\sigma}{\widehat\rho}};
  % \draw[red] (domSetRelFull) -- (relWeiDomSetRelSigma);
  % \node[res] at (relWeiDomSetRelSigma) {\cref{lem:count:makeSigmaCofinite}};

  % \node[problem,break] (relWeiDomSetRel) at (0,1*\y)
  %   {relation-wt.\ \CountDomSetRel{\sigma}{\rho}};
  % \draw[red] (relWeiDomSetRelSigma) -- (relWeiDomSetRel);
  % \node[res] at (relWeiDomSetRel) {\cref{lem:count:makeRhoCofinite}};

  % \node[problem,break] (verWeiDomSetRel) at (0,2*\y)
  %   {vertex-wt.\ \CountDomSetRel{\sigma}{\rho}};
  % \draw[red] (relWeiDomSetRel) -- (verWeiDomSetRel);
  % \node[res] at (verWeiDomSetRel)
  %   {\cref{lem:count:relationWeightedToVertexWeighted}};

  \node[problem, lowerB, break, minimum width=4*\x+\xWimp+2*\xGap, important]
    (domSetRel) at (\x+\xGap, 0)
  % \node[problem, lowerB, break, important] (domSetRel) at (0,3*\y)
    {\CountDomSetRel{\sigma}{\rho}};
  % \draw[redTur] (verWeiDomSetRel) -- (domSetRel);
  % \node[res] at (domSetRel) {\cref{lem:count:vertexWeightedToUnweighted}};

  %% BLOCK TWO (replace relations by HW=1)
  \node[case,break] (goodCase) at (.5*\x+\xGap/2, \y)
    {$\rho \neq \NN \lor \sigma$ finite};
  \draw[switch] (domSetRel.south -| goodCase) -- (goodCase);

  \node[problem,break] (domSetEqHWone) at (.5*\x+\xGap/2, 5*\y)
    {\srCountDomSetRel[\HWsetGen{=1},\EQset]};
  \draw[red] (goodCase) -- (domSetEqHWone);
  \node[res] at (domSetEqHWone)
    {\cref{lem:realizing:parsiArbitraryToHWoneAndEQ}};

      %13.25cm
  \node[problem,minimum width=3*\x+\xWimp+\xGap,important] (domSetHWone) at (.5*\x+\xGap/2, 6*\y)
    {\srCountDomSetRel[\HWsetGen{=1}]};
  \draw[red] (domSetEqHWone) -- (domSetHWone);
  \node[res] at (domSetHWone) {\cref{lem:realizing:parsiHWoneAndEQToHWone}};

  %% BLOCK THREE (remove HW=1 for finite, and not everything cofinite)
  % BLOCK THREE-ONE (HW=1)
  \node[case,break] (empty1Case) at (0, 7*\y)
    {$\rho$ finite $\land\; \rhoMax-1\notin \rho$};
  \draw[switch] (domSetHWone.south -| empty1Case) -- (empty1Case);

  \node[problem,break] (empty1) at (0, 10*\y)
    {\srCountDomSetRel[(\emptyset, \{1\})]};
  \draw[red] (empty1Case) -- (empty1);
  \node[res] at (empty1) {\cref{lem:I}};

  % BLOCK THREE-TWO (HW<=1)
  \node[case,break] (empty01Case) at (-1*\x, 7*\y)
    {$\rho$ finite $\land\; \rhoMax-1\in \rho$};
  \draw[switch] (domSetHWone.south -| empty01Case) -- (empty01Case);

  \node[problem,break] (domSetHWleOne) at (-1*\x, 8*\y)
    {\srCountDomSetRel[\HWsetGen{\le1}]};
  \draw[redTur] (empty01Case) -- (domSetHWleOne);
  \node[res] at (domSetHWleOne) {\cref{lem:IIrelversion}};

  \node[problem,break] (empty01) at (-1*\x, 10*\y)
    {\srCountDomSetRel[(\emptyset,\{0,1\})]};
  \draw[redTur] (domSetHWleOne) -- (empty01);
  \node[res] at (empty01) {\cref{lem:II}};

  % BLOCK THREE-THREE (HW>=1)
  \node[case,break] (cofCase) at (1.5*\x+.5*\xGap, 7*\y)
    {$\rho$ cofinite};
  \draw[switch] (domSetHWone.south -| cofCase) -- (cofCase);

  \node[problem,break,minimum width=\x+\xW+\xGap] (domSetHWgeOne) at (1.5*\x+\xGap/2, 8*\y)
    {\srCountDomSetRel[\HWsetGen{\ge1}]};
  \draw[redTur] (cofCase) -- (cofCase |- domSetHWgeOne.north);
  \node[res] at (domSetHWgeOne) {\cref{lem:IIIrelversion}};

  \node[case,break] (emptyFrom1Case) at (\x, 9*\y)
    {$\rho$ cofinite $\land\; \rho\neq \NN$};
  \draw[switch] (domSetHWgeOne.south -| emptyFrom1Case) -- (emptyFrom1Case);

  \node[problem,break] (emptyFrom1) at (\x, 10*\y)
    {\srCountDomSetRel[(\emptyset,\ZZ_{\ge 1})]};
  \draw[red] (emptyFrom1Case) -- (emptyFrom1);
  \node[res] at (emptyFrom1) {\cref{lem:III}};

  % BLOCK THREE-FOUR (remove (sigma, empty)+(empty, rho)
  \node[problem,minimum width=2*\x+\xW] (sigEmptyEmptyRho) at (0, 11*\y)
    {\srCountDomSetRel[(\sigma,\emptyset)+(\emptyset,\rho)]};
  \draw[red] (empty1) -- (sigEmptyEmptyRho);
  \draw[red] (empty01) -- ([xshift=-1*\x] sigEmptyEmptyRho.north);
  \draw[red] (emptyFrom1) -- ([xshift=\x] sigEmptyEmptyRho.north);
  \node[res, xshift=-1*\x] at (sigEmptyEmptyRho) {\cref{lem:rel:liftingRho}};
  \node[res              ] at (sigEmptyEmptyRho) {\cref{lem:rel:liftingRho}};
  \node[res, xshift=\x]    at (sigEmptyEmptyRho) {\cref{lem:rel:liftingRho}};

  %% BLOCK FOUR (rho everything, sigma finite)
  \node[case] (sigFinite) at (2*\x+\xGap, 9*\y)
    {$\rho=\NN \;\land\; \sigma$ finite};
  \draw[switch] (domSetHWgeOne.south -| sigFinite) -- (sigFinite);

  \node[problem,break] (0everything) at (2*\x+\xGap, 10*\y)
    {\srCountDomSetRel[(\{0\},\NN)]};
  \draw[redTur] (sigFinite) -- (0everything);
  \node[res] at (0everything) {\cref{lem:IV}};

  \node[problem,break] (sigEmpty) at (2*\x+\xGap, 11*\y)
    {\srCountDomSetRel[(\sigma,\emptyset)]};
  \draw[redTur] (0everything) -- (sigEmpty);
  \node[res] at (sigEmpty) {\cref{lem:sigfin2,lem:sigfin3}};

  %% BLOCK FIVE (rho everything, sigma cofinite)
  \node[case] (sigCofinite) at (3*\x+2*\xGap, 1*\y)
    {$\rho=\NN \;\land\; \sigma$ cofinite};
  \draw[switch] (domSetRel.south -| sigCofinite) -- (sigCofinite);
  % \draw[switch] (domSetRel) -- ++ (1.5*\x+\xGap/2,0)
  %   -- (1.5*\x+\xGap/2, -0.75*\y) -| (sigCofinite);

  \node[problem,break] (cofRelWeight) at (3*\x+2*\xGap, 2*\y)
    {rel-wt. \srCountDomSetRel[(\ZZ_{\ge1},\NN),\altRel]};
  \draw[red] (sigCofinite) -- (cofRelWeight);
  \node[res] at (cofRelWeight) {\cref{lem:RelfromweightedRelS}};

  \node[problem,break] (cofVertWeight) at (3*\x+2*\xGap, 3*\y)
    {vtx-wt. \srCountDomSetRel[(\ZZ_{\ge1},\NN),\altRel]};
  \draw[red] (cofRelWeight) -- (cofVertWeight);
  \node[res] at (cofVertWeight) {\cref{lem:vertexWeightedRelS}};

  \node[problem,break] (cofUnWeight) at (3*\x+2*\xGap, 4*\y)
    {\srCountDomSetRel[(\ZZ_{\ge1},\NN),\altRel]};
  \draw[red] (cofVertWeight) -- (cofUnWeight);
  \node[res] at (cofUnWeight) {\cref{lem:relWeightedRelS}};

  \node[problem,break] (cofHWeq1) at (3*\x+2*\xGap, 6*\y)
    {\srCountDomSetRel[(\ZZ_{\ge1},\NN),\alt{\HWsetGen{=1}}]};
  \draw[red] (cofUnWeight) -- (cofHWeq1);
  \node[res] at (cofHWeq1) {\cref{lem:RelSfromHWS}};

  \node[problem,break] (cofHWge1) at (3*\x+2*\xGap, 8*\y)
    {\srCountDomSetRel[(\ZZ_{\ge1},\NN),\alt{\HWsetGen{\ge1}}]};
  \draw[red] (cofHWeq1) -- (cofHWge1);
  \node[res] at (cofHWge1) {\cref{lem:HWeqS}};

  \node[problem,break] (cofAllEmpty) at (3*\x+2*\xGap, 10*\y)
    {\srCountDomSetRel[(\ZZ_{\ge1},\NN)+(\NN, \emptyset)]%
      \footnotesize{\\\vspace{-.4em}$(\NN, \emptyset)$ 1-bounded}};
  \draw[red] (cofHWge1) -- (cofAllEmpty);
  \node[res] at (cofAllEmpty) {\cref{lem:HWgeS}};

  % \node[problem,break] (cofGe1All) at (3*\x+2*\xGap, 10*\y)
  %   {\srCountDomSetRel[(\ZZ_{\ge1},\NN)]};
  % \draw[red] (cofAllEmpty) -- (cofGe1All);
  % \node[res] at (cofGe1All) {\cref{lem:sigcofauxpairs}};

  \node[problem,break] (cofAllEmptyBounded) at (3*\x+2*\xGap, 11*\y)
    {\srCountDomSetRel[{(\NN,\emptyset)}]
      \footnotesize{\\\vspace{-.1em}$(\NN, \emptyset)$ $\sigMax$-bounded}};
  \draw[red] (cofAllEmpty) -- (cofAllEmptyBounded);
  \node[res] at (cofAllEmptyBounded) {\cref{lem:sigcof2}};

  %% BLOCK SIX (the last step)
  \node (domSetAnchor) at (\x+\xGap, 12*\y) {};
  \node[problem, lowerB, minimum width=4*\x+\xWimp+2*\xGap, important]
    (domSet) at (domSetAnchor) {\srCountDomSet};

  % From block 3
  \draw[red] (sigEmptyEmptyRho) -- (domSet.north -| sigEmptyEmptyRho);
  \node[res] at (0, 12*\y) {\cref{lem:forceboth}};
  % From block 4
  \draw[redTur] (sigEmpty) -- (domSet.north -| sigEmpty);
  \node[res] at (2*\x+\xGap, 12*\y) {\cref{lem:sigfin1}};
  % From block 5
  \draw[red] (cofAllEmptyBounded) -- (domSet.north -| cofAllEmptyBounded);
  \node[res] at (3*\x+2*\xGap, 12*\y) {\cref{lem:sigcof1}};

  %% MARKING THE BLOCKS
  \begin{scope}[on background layer]
    % \draw[part]
    %   (-0.6*\x,         -1*\y)       rectangle (0.6*\x,          3*\y+.25cm);
    \draw[part,fill=red!10]
      (-1.6*\x,          0*\y-.25cm) rectangle (2.6*\x+1*\xGap,  6*\y+.25cm);
      % (-0.1*\x+\xGap/2, -1*\y-.25cm) rectangle (1.1*\x+\xGap/2,  6*\y+.25cm);
    \draw[part,fill=green!50!black!5]
      (-1.6*\x,          6*\y-.25cm) --        ( 2.6*\x+1*\xGap, 6*\y-.25cm)--
      (2.6*\x+1*\xGap,   8*\y+.25cm) --        ( 1.6*\x,         8*\y+.25cm)--
      (1.6*\x,          12*\y+.25cm) --        (-1.6*\x,        12*\y+.25cm)
      -- cycle;
      % (-1.6*\x,          6*\y-.25cm) rectangle (1.6*\x,         12*\y+.25cm);
    \draw[part,fill=blue!10]
      ( 1.4*\x+1*\xGap,  8*\y-.25cm) rectangle (2.6*\x+1*\xGap, 12*\y+.25cm);
    \draw[part,fill=orange!10]
      ( 2.4*\x+2*\xGap,  0*\y-.25cm) rectangle (3.6*\x+2*\xGap, 12*\y+.25cm);
  \end{scope}
\end{tikzpicture}
\caption{
An overview of the reductions for the counting version of the problem.
The sets
% $\widehat\sigma$ and $\widehat\rho$ are simple cofinite sets, and the sets
$\sigma$ and $\rho$ are cofinite sets.
% with $\max(\NN\setminus \sigma) = \max(\NN\setminus \widehat\sigma)$
% and $\max(\NN\setminus \rho) = \max(\NN\setminus \widehat\rho)$.
\\
% The reductions from the topmost gray box
% are given in \cref{sec:high-level:counting}, and
The reductions in the top-left, red box
are given in \cref{sec:RelToHW1}.
The reductions of the remaining three boxes are covered
in \cref{sec:realizingRelations:counting}.
The green box on the left is covered in \cref{sec:rhoNotAll},
%sec:forcingboth:ommittedProofs},
the blue box in the middle is covered in \cref{sec:rhoiseverything},
and the orange box on the right is covered in \cref{sec:rhoiseverything2}.
}
\label{fig:count:removingRelations}
\end{figure}

We restate the main result of this section using the new notation.
\thmCountRemovingRelations*

In its entirety, the proof of \cref{thm:RelfromDS} is very long and involved, which is why we try to break it down as much as possible. Consult \cref{fig:count:removingRelations} for a more detailed overview of the steps involved in this proof.
Below, we give the proof relying on three intermediate results that are proved in subsequent sections, namely \cref{lem:relrhonoteverything,lem:rhoeverything,lem:rhoeverything2}.

The proof is by a case distinction on $\sigma$ and $\rho$. At the top level, we consider three cases:
\begin{enumerate}[label = (\Alph*)]
 \item\label{item:relations-top-level-case-1} $\rho\neq \NN$,
 \item\label{item:relations-top-level-case-2} $\rho=\NN$ and $\sigma$ is finite, and
 \item\label{item:relations-top-level-case-3} $\rho=\NN$ and $\sigma$ is cofinite.
\end{enumerate}
Corresponding to these three cases, we state three intermediate results.

\begin{restatable}{lemma}{lemrelrhonoteverything}\label{lem:relrhonoteverything}
	Let $(\sigma,\rho)\in \allSetsne^2$ be non-trivial.
	If $\rho\neq \NN$,
	then we have \srCountDomSetRel[\HWeq1]~$\pwred$ \srCountDomSet.
\end{restatable}

The proof of \cref{lem:relrhonoteverything} is given in \cref{sec:rhoNotAll}.

\begin{restatable}{lemma}{lemrhoeverything}\label{lem:rhoeverything}
	Let $(\sigma,\rho)\in \allSetsne^2$ be non-trivial.
	If $\rho= \ZZ_{\ge0}$ and $\sigma$ is finite,
	then \srCountDomSetRel[\HWeq1]~$\pwred$ \srCountDomSet.
\end{restatable}
The proof of \cref{lem:rhoeverything} is given in \cref{sec:rhoiseverything}.

\begin{restatable}{lemma}{lemrhoeverythingtwo}\label{lem:rhoeverything2}
	Let $(\sigma,\rho)\in \allSetsne^2$ be non-trivial.
	If $\rho= \ZZ_{\ge0}$ and $\sigma$ is cofinite,
	then, for all constants $d$, there is a pathwidth-preserving reduction
	from \srCountDomSetRel on instances with arity at most $d$ to \srCountDomSet.
\end{restatable}
\Cref{sec:rhoiseverything2} contains the proof of this lemma.
With these three results, we have everything ready
to prove the main result of this section.
\begin{proof}[Proof of \cref{thm:RelfromDS}]
	Recall that we have already established in \cref{sec:RelToHW1}
	that it suffices to model $\HWeq1$ relations in order to model arbitrary relations (\cref{lem:Irelversion}). We use this result in the first two of the following three cases:
	\begin{itemize}
		\item
		If $\rho\neq\NN$,
		then the proof follows from \cref{lem:Irelversion,lem:relrhonoteverything}.
		\item
		If $\rho=\NN$ and $\sigma$ is finite,
		then the proof follows from
		\cref{lem:Irelversion,lem:rhoeverything}.
		\item
		If $\rho=\NN$ and $\sigma$ is cofinite,
		then the proof follows from \cref{lem:rhoeverything2}.
		\qedhere
	\end{itemize}
\end{proof}

Before going into the detailed proofs of \cref{lem:relrhonoteverything,lem:rhoeverything,lem:rhoeverything2} in the next sections, let us give some rough outline to point out how they relate, or how they differ.

We know from \cref{sec:RelToHW1}
that it suffices to model $\HWeq1$ relations,
and we can use \srCountDomSetRel[\HWeq1] as a starting point for our reductions.
This is what we do in Case \ref{item:relations-top-level-case-1}
($\rho\neq \NN$) and Case \ref{item:relations-top-level-case-2} ($\rho=\NN$, $\sigma$ finite).
However, Case \ref{item:relations-top-level-case-3} ($\rho=\NN$, $\sigma$ cofinite) is substantially different. Let us first talk about Cases \ref{item:relations-top-level-case-1} and \ref{item:relations-top-level-case-2}.

A first key ingredient is augmenting the problem by additional constraint pairs,
i.e., using the problem \srCountDomSetRel[\CP] for certain sets of pairs $\CP$.
In particular, we show that each of the following pairs can be used to model $\HWset{1}$ constraints:
\begin{itemize}
	\item $(\emptyset, \{1\})$,
	\item $(\emptyset, \{0,1\})$,
	\item $(\emptyset, \ZZ_{\ge 1})$,
	\item $(\{0\}, \NN)$.
\end{itemize}
The first of these pairs, for instance, refers to vertices that are never selected (since the first entry is $\emptyset$, and if they are unselected they have exactly 1 selected neighbor).

As a second key ingredient, we show that in case $\rho\neq\NN$, if we can force the selection status of $(\sigma, \rho)$-vertices, that means, if we can model $(\emptyset, \rho)$-vertices as well as $(\sigma, \emptyset)$-vertices, then we can model one of the constraint pairs in the list. In fact, we can model one of $(\emptyset, \{1\})$,
$(\emptyset, \{0,1\})$,
or $(\emptyset, \ZZ_{\ge 1})$. The chain of reductions is then finished by showing that we can in fact force the selection status of $(\sigma, \rho)$-vertices.

The treatment of Case \ref{item:relations-top-level-case-2} ($\rho=\NN$, $\sigma$ finite) is similar to the first case and arguably even slightly simpler. It turns out that here it suffices to force selected vertices, i.e., to model $(\sigma, \emptyset)$-vertices, in order to model $(\{0\}, \NN)$-vertices.

Now, let us talk about Case \ref{item:relations-top-level-case-3} ($\rho=\NN$ with cofinite $\sigma$).
This case turns out to be somewhat special. To give some a-priori intuition for why this is the case, suppose that $\sigma=\ZZ_{\ge1}$, which turns out to be the key case.
Think of some $(\ZZ_{\ge1},\NN)$-vertex $v$ that is subject to a relation $R$. In order to model the relation, we need to distinguish between solutions in which $v$
is selected and those in which it is not (in the sense that there is a different number of partial solutions extending the two states).
Note that, since $\rho=\NN$, unselected neighbors of $v$ behave the same independently of the selection status of $v$. Also, in general, if $v$ has only unselected neighbors, then their number does not affect the selection status of $v$.
Therefore, the fact of whether $v$ is selected can be distinguished only
if $v$ has at least one selected neighbor.
However, as soon as the $(\ZZ_{\ge1},\NN)$-vertex $v$ has at least one selected neighbor,
it has a feasible number of selected neighbors independently of its selection status or how many additional selected neighbors it might obtain.
So, at this point, $v$ is entirely unaffected
by the selection status of its remaining neighbors
and acts essentially as a $(\NN,\NN)$-vertex.
Since $v$ now behaves the same no matter what,
it is hard to imagine how one could recover the situation
where $v$ behaves like a $(\ZZ_{\ge1},\NN)$-vertex
(while having a handle on its selection status).
This is why while we are able to observe the selection status of $v$ and, in this sense, model the relation $R$ (by attaching some gadgets and interpolating), we did not manage to simultaneously recover the situation where $v$ behaves like a $(\ZZ_{\ge1},\NN)$-vertex with respect to the original graph.
To handle this situation, we take a different view on the concept of
``$(\sigma,\rho)$-sets with relations''. In particular, we diverge from
\cref{def:lower:graphWithRelations} in the sense that all vertices that are subject to
some relation are now treated as $(\NN,\NN)$-vertices, that is, they do not impose any constraints on the solutions other than that they have to satisfy the relations.
For the details, we refer to \cref{sec:rhoiseverything2}.

Generally speaking, working with $(\NN,\NN)$-vertices seems particularly suited for the setting where both $\sigma$ and $\rho$ are cofinite because, once some gadget provides $\max\{\sigMax,\rhoMax\}$ selected neighbors to a vertex, this vertex ``is always satisfied'' in the sense that it acts as a $(\NN,\NN)$-vertex with respect to the rest of the graph (outside the gadget). However, in the context of finite $\sigma$ or $\rho$, it is harder to imagine how to model the concept of being ``always satisfied'', so it seems less useful.

\subsection{\boldmath The Case
\texorpdfstring{$\rho\neq \NN$}{Rho is Restricted}}
\label{sec:RelfromRhoStates}
\label{sec:rhoNotAll}

In this section, we prove \cref{lem:relrhonoteverything}, that is, we reduce from \srCountDomSetRel[\HWeq1] to \srCountDomSet
when $\rho$ is not equal to $\NN$.
We do this with the following two steps.
\begin{enumerate}
  \item
  We first model $\HWeq{1}$ relations using $(\sigma,\emptyset)$-vertices
  and $(\emptyset,\rho)$-vertices. Intuitively, these are (normal $(\sigma,\rho)$-)vertices
  that we force to be selected and such vertices that we force to be unselected, respectively.

  \item
  We replace the $(\sigma,\emptyset)$-vertices
  and the $(\emptyset,\rho)$-vertices
  using special gadgets which we later refer to
  as \emph{winners} and \emph{strong candidates}.
\end{enumerate}

The first step is formally described by the following lemma.

\begin{restatable}{lemma}{lemrelfromforcing}\label{lem:Relfromforcing}
  Let $(\sigma,\rho)\in \allSetsne^2$ be non-trivial with $\rho\neq \ZZ_{\ge0}$.
  Then, \srCountDomSetRel[\HWeq1]~$\pwred$
  \srCountDomSetRel[\{(\emptyset,\rho),(\sigma,\emptyset)\}].
\end{restatable}

The second step proves the following result.

\begin{restatable}{lemma}{lemforceboth}\label{lem:forceboth}
  Let $(\sigma,\rho)\in \allSetsne^2$ be non-trivial with $\rho\neq \ZZ_{\ge0}$.
  Let $\CP$ be some (possibly empty) set of pairs from $\allSets^2$.
  Then, $\srCountDomSetRel[\CP+(\emptyset,\rho)+(\sigma,\emptyset)]
  \pwred \srCountDomSetRel[\CP]$.
\end{restatable}

Given these two results we directly obtain the main result of this section.

\begin{proof}[Proof of \cref{lem:relrhonoteverything}]
  The lemma follows from \cref{lem:Relfromforcing} together with \cref{lem:forceboth} (applied to an empty set of pairs $\CP$).
\end{proof}

We first focus on the proof of \cref{lem:Relfromforcing}, and afterward, we prove \cref{lem:forceboth}.

\subsubsection{Realizing Relations by Forcing the Selection Status}
To prove \cref{lem:Relfromforcing}, we first show the following result
which allows us to model $\HWeq{1}$ relations
using certain additional $(\emptyset, \rho')$-pairs.
We show that one can realize $\HWeq{1}$ with
\begin{itemize}
  \item $(\emptyset, \{1\})$ in \cref{lem:I},
  \item $(\emptyset, \{0,1\})$ in \cref{lem:II}, and
  \item $(\emptyset, \ZZ_{\ge1})$ in \cref{lem:IIIrelversion,lem:III}.
\end{itemize}

The first case is straightforward.
It merely requires to replace each \HWeq1 relation by an $(\emptyset,\{1\})$-vertex.

\begin{lemma}\label{lem:I}
  Let $(\sigma,\rho)\in \allSetsne^2$ be non-trivial.
  Then, \srCountDomSetRel[{\HWeq{1}}]~$\pwred$
  \srCountDomSetRel[{(\emptyset, \{1\})}].
\end{lemma}
\begin{proof}
  Let $I$ be an instance of \srCountDomSetRel[\HWeq{1}].
  We create an instance $I'$ of \srCountDomSetRel[(\emptyset, \{1\})] as follows.
  For each relation $\HWeq{1}$ with scope $S$, we completely connect the vertices in $S$ to a new $(\emptyset, \{1\})$-vertex.

  It is straightforward to see that the solutions of $I$ and $I'$ are in a one-to-one-correspondence.
  Since each scope of a relation in $I$
  is considered as a clique in the definition of the pathwidth of $I$,
  the pathwidth of $I'$ is at most that of $I$ plus 1.
\end{proof}

In the next step, we establish that $\HWset{1}$ relations can also be modeled using $(\emptyset, \{0,1\})$-vertices.
However, this is less straightforward than the previous result as we are using interpolation techniques and the auxiliary result from
\cref{lem:happyGadget} to obtain the reduction.
We restate \cref{lem:happyGadget} for convenience.
\lemHappyGadgetCounting*

Note that the requirements of \cref{lem:happyGadget} are always fulfilled
if $(\sigma,\rho)$ is a non-trivial pair in $\allSetsne^2$.
In this case, both $\sigma$ and $\rho$ are non-empty, and $\rho\neq \{0\}$.

Recall the definition of $\HWge[d]{1}$ and $\HWge[d]{1}$ from \cref{def:hammingWeightOneAndEquality}. If the arity is irrelevant/unspecified, then we often simply write $\HWge{1}$.
Analogously, we define $\HWle{1}$ relations.
We show that, for the counting problem, $\HWeq{1}$ relations can be modeled using $\HWle{1}$ relations (\cref{lem:IIrelversion}) or $\HWge{1}$ relations (\cref{lem:IIIrelversion}).
Analogously to $\srCountDomSetRel[\HWeq1]$ (definition in \cref{sec:RelToHW1}),
we define \srCountDomSetRel[\HWle{1}] and \srCountDomSetRel[\HWge{1}].

\begin{lemma}\label{lem:IIrelversion}
	Let $(\sigma,\rho)\in \allSetsne^2$ be non-trivial.
  Then, \srCountDomSetRel[\HWeq{1}]~$\pwred$ \srCountDomSetRel[\HWle{1}].
\end{lemma}
\begin{proof}
	\begin{figure}[t]
		\centering
    \raisebox{4.475ex}{%
		\includegraphics[scale=1.9]{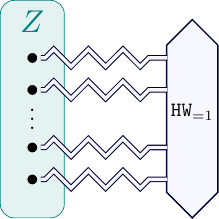}}\qquad\raisebox{15ex}{\scalebox{3}{\(\leadsto\)}}\qquad
		\includegraphics[scale=1.9]{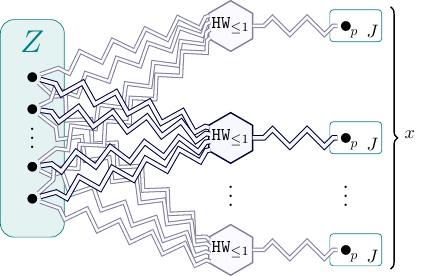}
		\caption{
			The construction of instance $I'_x$ from \cref{lem:IIrelversion}
			to replace a \(\HWeq{1}\) relation between the
			vertices of a single set \(Z\).
		}\label{fig:z-1}
	\end{figure}
	\renewcommand{\CP}{{(\emptyset, \{0,1\})}}
	Let $I$ be an instance of \srCountDomSetRel[\HWeq{1}].
  Let $\CC$ be the set of relational constraints of $I$
  (all of which enforce $\HWeq{1}$ relations), and let $\CZ=\{\scope(C) \mid C\in \CC\}$.
	We remove the relational constraints $\CC$ fromt he instance $I$ to form a new instance $I'$.
  $I'$ can be cast as an instance of \srCountDomSet.
  Then, the number of solutions of $I$ is identical
  to the number of those solutions of $I'$
  that select precisely one vertex from each $Z\in \CZ$.

	By \cref{lem:happyGadget}, there is a graph $J$ that contains a vertex $p$ such that $J$ has $\alpha\ge 1$ $(\sigma,\rho)$-sets that contain $p$, and it has $\beta\ge1$ $(\sigma,\rho)$-sets that do not contain $p$. It is not clear whether or not $\alpha=\beta$.

	For a positive integer $x$,
  let $I'_x$ be the instance of \srCountDomSetRel[\HWle{1}]
  obtained from $I'$ by adding, for each set $Z\in \CZ$,
  $x$ copies $J^Z_1, \ldots, J^Z_x$ of the graph $J$. Let $p^Z_1, \ldots, p^Z_x$ be the corresponding copies of $p$. For each $i\in \numb{x}$, we add a relational constraint with scope $Z\cup \{p^Z_i\}$ that enforces a $\HWle{1}$-relation on this scope.
  Consult \cref{fig:z-1} for a
	visualization of this construction (for a single set \(Z\)).

	Note that a solution $S'$ of $I'$ for \srCountDomSet
  can be extended only to a solution of $I'_x$ for \srCountDomSetRel[\HWle{1}]
  if, in each set $Z\in \CZ$, at most one vertex is selected
  (because of the attached $\HWle{1}$-relations).
  Let us say that solutions of $I'$ with this property are \emph{good}.
	Suppose in a good solution $S'$, a set $Z\in \CZ$ is entirely unselected. Then, there are $f_0\deff (\alpha+\beta)^x$ feasible extensions to its attached copies of $J$ ($\alpha+\beta$ for each graph $J^Z_i$ since $p_i^Z$ can either be selected or not).
	If, in $S'$, exactly $1$ vertex of $Z$ is selected,
	then there are $f_1\deff \beta^x$ extensions to the attached copies of $J$
	(as we can use only the selections of the $J^Z_i$'s
	for which $p_i^Z$ is unselected).

	Let $a_i$ be the number of good solutions of $I'$
	in which precisely $i$ of the sets in $\CZ$ are entirely unselected.
	(In the remaining $\abs{\CZ}-i$ sets from $\CZ$,
	exactly one vertex is selected.)
	Let $\#S(I'_x)$ denote the number of solutions of $I'_x$. Then,
	\begin{align*}
		\#S(I'_x)
		&=\sum_{i=0}^{\abs{\CZ}} a_i\cdot f_0^i f_1^{\abs{\CZ}-i}\\
		&= f_1^{\abs{\CZ}} \cdot\sum_{i=0}^{\abs{\CZ}} a_i\cdot \left(\frac{f_0}{f_1}\right)^i\\
		&= \beta^{x\abs{\CZ}} \cdot\sum_{i=0}^{\abs{\CZ}} a_i\cdot \left(\left(\frac{\alpha+ \beta}{\beta}\right)^x\right)^{i}.
	\end{align*}

	Recall that $\beta\ge 1$, and therefore, $\#S(I'_x)/\beta^{x\abs{\CZ}}$ is a polynomial  in
	$((\alpha+ \beta)/{\beta})^x$ of degree $\abs{\CZ}$.
	Since $(\alpha+\beta)/\beta >1$, it suffices to choose $\abs{\CZ}$ different values of $x> 0$ for the interpolation. This way we can recover the coefficients $a_i$. In particular, we can recover $a_0$, which corresponds to the number of good solutions of $I'$ in which none of the sets in $\CZ$ are entirely unselected, i.e., in which all of these sets contain precisely one selected vertex. This is exactly the sought-after number of solutions of $I$.

	Since the scope of a relational constraint in $I$
	is considered as a clique in the definition of the pathwidth of $I$,
	for $Z\in \CZ$, we can add all vertices from $J^Z_1, \ldots, J^Z_x$
	one after another to a copy of the original bag containing the clique $Z$.
	Hence, the pathwidth of $I'_x$ is at most that of $I$
	plus an additive constant (depending only on $\sigma$ and $\rho$).
\end{proof}

\begin{lemma}\label{lem:II}
  Let $(\sigma,\rho)\in \allSetsne^2$ be non-trivial.
  Then,
  \srCountDomSetRel[\HWle{1}]~$\pwred$ \srCountDomSetRel[(\emptyset, \{0,1\})].
\end{lemma}
\begin{proof}
	From an instance $I$ of \srCountDomSetRel[\HWle{1}],
	we create an instance $I'$ of \srCountDomSetRel[(\emptyset, \{0,1\})]
	by completely connecting each scope $Z$ of some $\HWle{1}$-relation in $I$ to a new $(\emptyset, \{0,1\})$-vertex $z'$.

	It is straightforward to see that the solutions of $I$ and $I'$
	are in a one-to-one correspondence.
	Since each scope of a relation in $I$
	is considered as a clique in the definition of the pathwidth of $I$,
	the pathwidth of $I'$ is at most that of $I$ plus 1.
\end{proof}

\begin{lemma}\label{lem:IIIrelversion}
  Let $(\sigma,\rho)\in \allSetsne^2$ be non-trivial.
  Then, \srCountDomSetRel[\HWeq{1}]~$\pwred$ \srCountDomSetRel[\HWge{1}].
\end{lemma}
\begin{proof}
  The proof is similar to that of \cref{lem:IIrelversion}, but uses a different gadget, which leads to different numbers of extensions that have to be considered.
  Let $I$ be an instance of \srCountDomSetRel[\HWeq{1}].
  We define $\CZ$, $I'$, $J$, $\alpha$, and $\beta$ precisely as we did in the proof of \cref{lem:IIrelversion}.

  When defining $I'_x$, we divert from the proof of \cref{lem:IIrelversion}.
  For a positive integer $x$,
  let $I'_x$ be the instance of \srCountDomSetRel[\HWge{1}]
  obtained from $I'$ by making each $Z\in \CZ$ the scope of a $\HWge{1}$ relation.
  In addition, we introduce, for each subset $Z'$ of $Z$ with $\abs{Z'}=\abs{Z}-1$,
  a total of $x$ copies $J^{(Z,Z')}_1, \ldots, J^{(Z,Z')}_x$ of the graph $J$.
  For each corresponding copy
  $p^{(Z,Z')}_i$ of $p$ we introduce a $\HWge{1}$ relation with scope $Z'\cup \{p^{(Z,Z')}_i\}$.
  Note that some set $Z'$ may receive multiple such attachments for different supersets $Z$.
  For an illustration see \cref{fig:z-3}.
  \begin{figure}[t]
		\centering
    \raisebox{3ex}{%
		\includegraphics[scale=1.9]{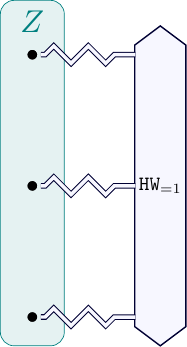}}\hfill\raisebox{20ex}{\scalebox{3}{\(\leadsto\)}}\hfill
		\includegraphics[scale=1.9]{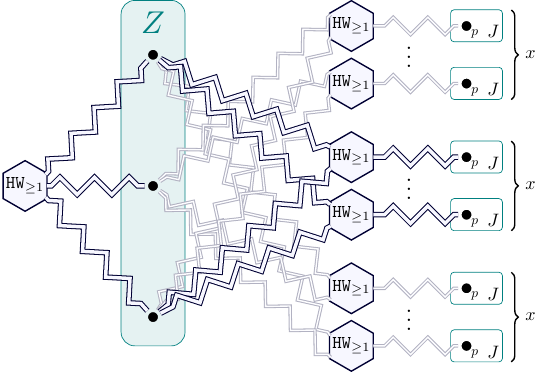}
		\caption{
			The construction of instance $I'_x$ from \cref{lem:IIIrelversion}
			to replace a \(\HWeq{1}\) relation between the
			vertices of a single set \(Z\).
		}\label{fig:z-3}
	\end{figure}

  Note that a solution $S'$ of $I'$ for \srCountDomSet
  can be extended only to a solution of $I'_x$ for \srCountDomSetRel[\HWge{1}]
  if, in each set $Z\in \CZ$, at least one vertex is selected
  (because of the $\HWge{1}$ relation attached to $Z$).
  Let us say that solutions of $I'$ with this property are \emph{good}.
  If, in a good solution $S'$ of $I'$, some $Z'$ has at least one selected vertex, then there are $f_1\deff (\alpha+\beta)^x$ feasible extensions to the graphs $J^{(Z,Z')}_1, \ldots, J^{(Z,Z')}_x$ ($\alpha+\beta$ for each graph $J^{(Z,Z')}_i$ since the corresponding copy of $p$ can either be selected or not).
  If, in $S'$, no vertex of $Z'$ is selected,
  then there are only $f_0\deff \alpha^x$ extensions
  (each copy of $p$ in some gadget $J^{(Z,Z')}_i$ has to be selected to fulfill the corresponding \HWge{1} constraints).
  Continuing from this, we say that if a set $Z\in \CZ$ contains at least 2 selected vertices (from $S'$), then it is \emph{undesired}, and if it contains precisely 1 selected vertex, then it is \emph{desired}.

  For an undesired set $Z\in \CZ$, there is at least one vertex selected in each of the $\abs{Z}$ corresponding subsets $Z'$ (this gives a total of $f_1^{\abs{Z}}$ extensions to the graphs $J^{(Z, \star)}_{\star}$), whereas, for a desired set $Z$, there is exactly one subset $Z'$ that is entirely unselected (which gives a total of $f_1^{(\abs{Z}-1)}\cdot f_0$ extensions).
  For $\CZ'\subseteq \CZ$, let $a_{\CZ'}$ be the number of good solutions of $I'$ for which the sets in $\CZ'$ are undesired and the sets in $\CZ\setminus \CZ'$ are desired. Let $a_i$ be the number of good solutions of $I'$ in which precisely $i$ of the sets in $\CZ$ are undesired. Let $\#S(I'_x)$ denote the number of solutions of $I'_x$. Then,
  \begin{align*}
    \#S(I'_x)
    &=\sum_{\CZ'\subseteq \CZ} a_{\CZ'}
        \cdot \left( \prod_{Z\in \CZ'} f_1^{\abs{Z}}
               \cdot \prod_{Z\notin \CZ'} f_1^{(\abs{Z}-1)} f_0
              \right)\\
    &= f_1^{\sum_{Z\in \CZ}(\abs{Z}-1)} \cdot \sum_{\CZ'\subseteq \CZ} a_{\CZ'}\cdot  f_1^{\abs{\CZ'}} f_0^{\abs{\CZ\setminus\CZ'}}\\
    &=
    \underbrace{f_1^{\sum_{Z\in \CZ}(\abs{Z}-1)}f_0^{\abs{\CZ}}
    }_{\eqqcolon F}
    \cdot \sum_{i=0}^{\abs{\CZ}} a_i
          \cdot \left( \frac{(\alpha+\beta)^x}{\alpha^x} \right)^i.
  \end{align*}
  Note that $\#S(I'_x)/F$ is a polynomial in $(\alpha+\beta)^x/\alpha^x$ of degree $\abs{\CZ}$.
  Since $\alpha, \beta \ge 1$,
  the values of $((\alpha+\beta)/\alpha)^x$ are defined
  and distinct for different $x>0$.
  It suffices to choose $\abs{\CZ}$ different values of $x>0$
  for the interpolation.
  This way we can recover the coefficients $a_i$.
  In particular, we can recover $a_0$,
  which corresponds to the number of good solutions of $I'$
  in which none of the sets in $\CZ$ are undesired,
  which means that all of the sets in $\CZ$ are desired,
  i.e., in which all of these sets contain precisely one selected vertex.
  This is exactly the sought-after number of solutions of $I$.

  As in the proof of \cref{lem:IIrelversion},
  we can argue that the pathwidth of $I'_x$ is at most that of $I$
  plus an additive term in $\Oh(1)$.
\end{proof}

As a next step, we replace the \HWge{1} relations by some appropriate vertices
as we did previously for the \HWeq{1} relations.

\begin{lemma}\label{lem:III}
  Let $(\sigma,\rho)\in \allSetsne^2$ be non-trivial.
  Then, \srCountDomSetRel[\HWge{1}]~$\pwred$
  \srCountDomSetRel[(\emptyset, \ZZ_{\ge 1})].
\end{lemma}
\begin{proof}
  From an instance $I$ of \srCountDomSetRel[\HWge{1}],
  we create an instance $I'$ of \srCountDomSetRel[(\emptyset, \ZZ_{\ge 1})]
  by completely connecting
  each scope $Z$ of some $\HWge{1}$-relation in $I$ to a new $(\emptyset, \ZZ_{\ge 1})$-vertex $z'$.

  It is straightforward to see that the solutions of $I$ and $I'$
  are in a one-to-one correspondence.
  Since each scope of a relation in $I$
  is considered as a clique in the definition of the pathwidth of $I$,
  the pathwidth of $I'$ is at most that of $I$ plus 1.
\end{proof}

To conclude the proof of \cref{lem:Relfromforcing},
we need one more auxiliary result, which shows how to use $(\sigma,\emptyset)$-vertices
to replace $(\emptyset, \rho-i)$-vertices by $(\emptyset,\rho)$-vertices.

\begin{lemma}
  \label{lem:rel:liftingRho}
  Let $(\sigma,\rho)\in \allSetsne^2$ be non-trivial,
  and let $\CP$ be some set of pairs from $\allSets^2$.
  Then, for all $0 \le i \le \rhoMax$,
  \srCountDomSetRel[\CP + (\emptyset,\rho - i)]~$\pwred$
  \srCountDomSetRel[\CP + (\sigma, \emptyset) + (\emptyset, \rho)].
\end{lemma}
\begin{proof}
  We first construct a gadget $(H,\{\port\})$
  which simulates a $(\sigma', \emptyset)$-vertex,
  for some $\sigma' \subseteq \NN$ with $0 \in \sigma'$.
  The gadget $H$ consists of an $(\sigMin +1)$-clique
  of $(\sigma, \emptyset)$-vertices. Fix one vertex in this clique
  as the distinguished portal $\port$.
  It is immediate that $\port$
  behaves like a $(\sigma-\sigMin, \emptyset)$-vertex
  which is, by definition of $\sigMin$, a $(\sigma',\emptyset)$-vertex
  with $0\in \sigma'$.

  Now, we modify a given \srCountDomSetRel[\CP + (\emptyset,\rho - i)] instance $I$ as follows.
  We turn each $(\emptyset, \rho-i)$-vertex $v$ into a $(\emptyset, \rho)$-vertex that is
  adjacent to the portals of $i$ new copies of $H$.
  Since the portals of the freshly added gadgets are selected,
  the vertex $v$ already has $i$ selected neighbors,
  and thus, behaves like an $(\emptyset, \rho-i)$-vertex with respect to the vertices of the original instance.
  Let $I'$ be the resulting instance of
  \srCountDomSetRel[\CP + (\sigma, \emptyset) + (\emptyset, \rho)].

  It remains to argue that this modification preserves the pathwidth.
  For each $(\emptyset, \rho-i)$-vertex $v$,
  there is a bag in the path decomposition of $I$ containing $v$.
  We duplicate this bag $i$ times for the path decomposition of $I'$,
  and add one copy of $H$ to precisely one bag.
  As the size of $H$ is bounded by a function in $\sigMin$,
  the size of $H$ is constant
  and the pathwidth of $I'$ increases only by an additive constant.
\end{proof}

Now, we have everything ready to prove \cref{lem:Relfromforcing}.
\lemrelfromforcing*
\begin{proof}
  Since $\rho \neq \{0\}$ and $\rho \neq \NN$, we have $\rhoMax > 0$.
  We claim that
  \srCountDomSetRel[\HWsetGen{=1}]~$\pwred$
  \srCountDomSetRel[(\emptyset,\rho-(\rhoMax-1))].

  First, suppose that $\rho$ is finite.
  If $\rhoMax-1 \notin \rho$,
  then $\rho-(\rhoMax-1)$ is equal to $\{1\}$,
  and the reduction follows from \cref{lem:I}.
  Otherwise, we have that $\rhoMax-1 \in \rho$.
  In this case, $\rho-(\rhoMax-1)$ is equal to $\{0,1\}$,
  and the reduction follows from \cref{lem:IIrelversion,lem:II}.

  Second, suppose that $\rho$ is cofinite.
  Then, $\rhoMax$ is the largest integer missing from $\rho$ plus $1$, and consequently $\rhoMax-1\notin \rho$.
  Therefore, $\rho-(\rhoMax-1)$ is equal to $\ZZ_{\ge1}$,
  and the reduction follows from \cref{lem:IIIrelversion,lem:III}.

  Thus, we have that
  \srCountDomSetRel[\HWsetGen{=1}]~$\pwred$
  \srCountDomSetRel[(\emptyset,\rho-(\rhoMax-1))].
  Now, we use \cref{lem:rel:liftingRho} with $i=\rhoMax-1$
  to obtain \srCountDomSetRel[\CP+(\emptyset,\rho-(\rhoMax-1))]~$\pwred$
  \srCountDomSetRel[\CP + (\sigma, \emptyset) + (\emptyset, \rho)], which (when applied to an empty set of pairs $\CP$) concludes the proof.
\end{proof}

\subsubsection{How to Force the Selection Status}
\label{sec:forcingboth}

Now, the main work is to show that selected and unselected vertices
can always be modeled, that is, \cref{lem:forceboth}.
We show that this is the case whenever gadgets with certain extension counts exist. We
call the sought-after gadgets ``winners'' (defined momentarily). Unfortunately, such
winners do not always exist. Sometimes, we can ensure only a weaker version, namely a ``strong candidate'' (see \cref{lem:candidate}).

\begin{restatable}[Candidates and Winners]{definition}{defCandidatesWinner}
	\label{def:candidates}
Let $(\sigma,\rho)\in \allSetsne^2$, and let $\CP$ be some possibly empty set of pairs from $\allSets^2$.
Let $\CG$ be a gadget for \srCountDomSetRel[\CP] with a single portal.
We say that $\CG$ is a \emph{candidate} if it has the following properties:
\begin{enumerate}[label = (\Roman*)]
  \item $\ext_\CG(\rho_0), \ext_\CG(\sigma_0)\ge 1$,
	\label{item:candidate:rho0sig0ge1}
  \item $\ext_\CG(\rho_0)\neq \ext_\CG(\sigma_0)$,
	\label{item:candidate:item:rho0NEQsig0}
  \item $\ext_\CG(\sigma_i)=\ext_\CG(\rho_i)=0$ for all $i\ge 2$.
	\label{item:candidate:item:rhoiEQsigiEQ0}
\end{enumerate}
We say that a candidate $\CG$ is a \emph{winner} if additionally it holds that
\begin{enumerate}[label = (\Roman*)]
  \setcounter{enumi}{3}
  \item $\ext_\CG(\rho_1)=\ext_\CG(\sigma_1)=0$. \label{item:rho1equsig1equ0}
\end{enumerate}
We say that a candidate $\CG$ is a \emph{strong candidate} if (instead of \cref{item:rho1equsig1equ0}) it additionally holds that
\begin{enumerate}[label = (\Roman*)]
  \setcounter{enumi}{4}
  \item $\ext_\CG(\rho_1)\ge 1$. \label{item:strongagg0}
  \item $\ext_\CG(\rho_0)\neq \ext_\CG(\sigma_0)+\ext_\CG(\sigma_1)$.\label{item:strongagg}
\end{enumerate}
\end{restatable}

We defer the proof of the following \cref{lem:candidate} to \cref{sec:candidateorwinner}.

\begin{restatable}{lemma}{lemcandidate}\label{lem:candidate}
  Let $(\sigma,\rho)\in \allSetsne^2$  be non-trivial.
	If $\rho\neq \ZZ_{\ge0}$, then there is a gadget $\CG=(G,\{p\})$
	for \srCountDomSet that is either a winner or a strong candidate.
\end{restatable}

A second important ingredient for the proof of \cref{lem:forceboth}, is the fact that, given a strong candidate gadget, one can model a single $(\emptyset, \{0\})$-vertex (or a constant number thereof), see \cref{lem:forcesingle00}.

\begin{restatable}{lemma}{forcesinglezerozero}\label{lem:forcesingle00}
  Let $(\sigma,\rho)\in \allSetsne^2$ be non-trivial,
	and let $\CP$ be some (possibly empty) set of pairs from $\allSets^2$.
	Suppose that there is a strong candidate $\CJ$ for \srCountDomSet.
  Then,
  \[
  	\srCountDomSetRel[\CP+(\emptyset, \{0\})] \pwred \srCountDomSetRel[\CP],
  \]
  where $(\emptyset, \{0\})$ is $1$-bounded in the source problem.
\end{restatable}

We defer the proof of \cref{lem:forcesingle00} to \cref{sec:forcesingle00}, and first
discuss how we use it to obtain \cref{lem:forceboth}.

\lemforceboth*
\begin{proof}
  Let $\CJ=(J,\{v\})$ be the gadget given by \cref{lem:candidate}
	(using that $\rho\neq \NN$).
	Then, $\CJ$ is either a winner or a strong candidate for \srCountDomSet.
	Note that $\CJ$ is also a winner or a strong candidate
	for \srCountDomSetRel[\CP] (that does not use vertices from $\CP$).

  We first give the reduction assuming that $\CJ$ is a winner, and consequently, $\ext_\CJ(\rho_1)=\ext_\CJ(\sigma_1)=0$.
  Afterward, we give a modified reduction for the case where $\CJ$ is a strong candidate.
	In both cases, the proof is split into two steps
	where we first show the reduction
	from \srCountDomSetRel[\CP+(\emptyset,\rho)]
	to \srCountDomSetRel[\CP],
	that is, how to force vertices to be unselected,
	and then the reduction
	from \srCountDomSetRel[\CP+(\sigma,\emptyset)+(\emptyset,\rho)]
	to \srCountDomSetRel[\CP+(\emptyset,\rho)],
	that is, how to force vertices to be selected given that one can already force unselected vertices.

	\begin{claim}
		If $\CJ$ is a winner,
		then $\srCountDomSetRel[\CP+(\emptyset,\rho)]\pwred \srCountDomSetRel[\CP]$.
	\end{claim}
	\begin{claimproof}
  If $\CJ$ is a winner, then the only states of $\CJ$ with non-zero extensions are $\rho_0$ and $\sigma_0$.
  In particular, $v$ never receives any selected neighbors within $\CJ$.

  Let $I=(G,\lambda)$ be an instance of
	\srCountDomSetRel[\CP+(\emptyset,\rho)],
	and let $U=\{u_1, \ldots, u_k\}$ be the set of $(\emptyset,\rho)$-vertices
	in $G$.
  For a positive integer $x$, we define an instance $I_x=(G_x,\lambda_x)$
	of \srCountDomSetRel[\CP] by attaching $x$ new copies of the gadget $\CJ$
	to each vertex $u \in U$,
	where $u$ is identified with the portal of each attached copy.
  The function $\lambda_x$ is identical to $\lambda$ on the vertices in $V(G)\setminus U$, and maps the remaining vertices to $0$ (i.e., the vertices in $U$ and in their attached copies of $\CJ$ are $(\sigma,\rho)$-vertices).

  Let $\#S(I_x)$ denote the number of solutions of \srCountDomSetRel[\CP]
	on input $I_x$.
	Further, let $a_i$ be the number of solutions of \srCountDomSetRel[\CP]
	on instance $I'=(G,\lambda_x\vert_{V(G)})$
	for which precisely $i$ of the vertices in $U$ are selected.
  Intuitively, the instance $I'$ is obtained from $I$
	by replacing the $(\emptyset,\rho)$-vertices in $U$
	by $(\sigma,\rho)$-vertices.
	Then, our goal is to compute $a_0$
	since this is precisely the number of solutions on instance $I$.

  We recover $a_0$ from $\#S(I_x)$ using interpolation.
  We observe that
  \begin{align*}
   \#S(I_x)
   &=\sum_{i=0}^{|U|} a_i\cdot \ext_{\CJ}(\sigma_0)^{xi}
	 				\cdot \ext_{\CJ}(\rho_0)^{x(|U|-i)}\\
   &=\ext_{\CJ}(\rho_0)^{x|U|}\cdot \sum_{i=0}^{|U|} a_i
	 		\cdot \left(\frac{\ext_{\CJ}(\sigma_0)}{\ext_{\CJ}(\rho_0)}\right)^{xi}.
  \end{align*}
  Hence, $\#S(I_x)/\ext_{\CJ}(\rho_0)^{x|U|}$ is a polynomial of degree $|U|$
	with coefficients $a_0, \dots, a_{|U|}$
	and indeterminates $(\ext_{\CJ}(\sigma_0)/\ext_{\CJ}(\rho_0))^{x}$.
  Since $\CJ$ is a winner, we have $\ext_\CJ(\rho_0), \ext_\CJ(\sigma_0)\ge 1$, and $\ext_\CJ(\rho_0) \neq \ext_\CJ(\sigma_0)$.
  Thus, we can recover the coefficients using polynomial interpolation
	using $|U|$ distinct values of $x$,
	which gives $|U|$ distinct values of the indeterminates.
  In particular, we obtain $a_0$ as required.
	\end{claimproof}

	For the next claim, we use exactly the same approach as for the previous one. Note
    that every winner for \srCountDomSetRel[\CP] is also a winner for
    \srCountDomSetRel[\CP+(\emptyset,\rho)].
		(It does not even use $(\emptyset,\rho)$-vertices.)
    So, the following claim states a slightly stronger statement than we need at this
    point. This stronger variant is useful for the second case where $\CJ$ is only a strong candidate.
	\begin{claim}\label{clm:winner2}
		If $\CJ$ is a winner for \srCountDomSetRel[\CP+(\emptyset,\rho)],
		then
		\[
		\srCountDomSetRel[\CP+(\sigma,\emptyset)+(\emptyset,\rho)]
		\pwred \srCountDomSetRel[\CP+(\emptyset,\rho)].
		\]
	\end{claim}
	\begin{claimproof}
		Observe that the coefficient $a_{|U|}$ in the proof of the previous claim
		corresponds to the number of solutions
		for which all vertices in $U$ are selected.
	  This is precisely what we need in order to model the vertices in $U$ as $(\sigma,\emptyset)$-vertices.
		Hence, with the same proof, we obtain the reduction.
		Moreover, since we are reducing to \srCountDomSetRel[\CP+(\emptyset,\rho)],
		it suffices that $\CJ$ be a gadget/winner for this problem.
	\end{claimproof}

	The proof of the statement when $\CJ$ is a winner
	follows from the combination of both claims.

  Now, we revisit our assumption about the gadget $\CJ$
	and consider the remaining case where $\CJ$ is not a winner,
	but a strong candidate.
	In this case, $\CJ$ fulfills the requirements of \cref{lem:forcesingle00}
	with which we obtain $\srCountDomSetRel[\CP+(\emptyset, \{0\})]
	\pwred \srCountDomSetRel[\CP]$,
	where $(\emptyset, \{0\})$ is $1$-bounded in the source problem.
  Together with the following claim,
	we obtain
	\[
		\srCountDomSetRel[\CP+(\emptyset,\rho)] \pwred \srCountDomSetRel[\CP].
	\]

	\begin{claim}%\label{lem:forceunselectedfrom00}
		It holds that
		$\srCountDomSetRel[\CP+(\emptyset,\rho)] \pwred
		\srCountDomSetRel[\CP+(\emptyset,\{0\})]$,
		even if $(\emptyset,\{0\})$ is $1$-bounded in the target problem.
	\end{claim}
	\begin{claimproof}
		The construction is straightforward:
		it suffices for all vertices that should be unselected
		(that should be $(\emptyset,\rho)$-vertices)
		to be adjacent to a single $(\emptyset,\{0\})$-vertex $p$.
		The key observation is that $p$ forces its neighbors to be unselected,
		but since $p$ itself is unselected, it does not alter the original solutions.
	\end{claimproof}

  In order to show the sought-after
	\srCountDomSetRel[\CP+(\sigma,\emptyset)+(\emptyset,\rho)]~$\pwred$ 
	\srCountDomSetRel[\CP],
	it now suffices to show
	\srCountDomSetRel[\CP+(\sigma,\emptyset)+(\emptyset,\rho)]~$\pwred$
	\srCountDomSetRel[\CP+(\emptyset,\rho)].
	We reuse \cref{clm:winner2} to obtain this reduction by applying it to the gadget $\CJ'$ as given by the following claim.
	\begin{claim}
		If $\CJ$ is a strong candidate for \srCountDomSetRel[\CP],
		then there is a winner $\CJ'$ for \srCountDomSetRel[\CP+(\emptyset, \rho)].
	\end{claim}
	\begin{claimproof}
  	If in the gadget $\CJ$ we replace all neighbors of its portal
		by $(\emptyset,\rho)$-vertices,
		then we obtain a gadget $\CJ'$ for \srCountDomSetRel[\CP+(\emptyset,\rho)]
  that inherits the properties $\ext_{\CJ'}(\rho_0), \ext_{\CJ'}(\sigma_0)\ge 1$, and  $\ext_{\CJ'}(\rho_0) \neq \ext_{\CJ'}(\sigma_0)$ from $\CJ$.
  But now, since all neighbors of the portal have to be unselected, we also have $\ext_{\CJ'}(\rho_i)=\ext_{\CJ'}(\sigma_i)=0$ for all $i\ge 1$.
  Thus, $\CJ'$ is now a \emph{winner}.
	\end{claimproof}

	This completes the proof for the case when $\CJ$ is a strong candidate.
\end{proof}

\subsubsection{Omitted Proofs}
\label{sec:forcingboth:ommittedProofs}
In the following, we give the proofs of \cref{lem:candidate,lem:forcesingle00}.
\subsubsection*{\boldmath Constructing Strong Candidates and Winners: Proof of \texorpdfstring{\cref{lem:candidate}}{Lemma \ref{lem:candidate}}}
\label{sec:candidateorwinner}

We start with the proof of \cref{lem:candidate}.
Recall the definitions of (strong) candidates and winners.
\defCandidatesWinner*

In the following, we show the existence of such gadgets.

\lemcandidate*

\begin{proof}
We start the construction with an auxiliary gadget.

\begin{claim}\label{clm:extensions1}
  Let $(\sigma,\rho)\in \allSetsne^2$ be non-trivial.
  Then, there is a gadget $\CJ=(J,\{p\})$ for \srCountDomSet
  such that $\ext_\CJ(\rho_0), \ext_\CJ(\rho_1)\ge 1$,
  but $\ext_\CJ(\rho_i)=0$ for each $i\ge 2$.
\end{claim}
\begin{claimproof}
  Since $\rho\neq \{0\}$ by non-triviality, there is an $r\in \rho$ with $r\ge 1$. Hence, we can apply \cref{lem:happyGadget} to obtain a graph $G$ with a vertex $u$ such that $G$ has two disjoint solutions $S_0$ and $S_1$
  for which $u\notin S_0$ with $\abs{N(u)\cap S_0}=r$, and $u\in S_1$ with $\abs{N(u)\cap S_1}=s\in \sigma$.
  Observe that the size of $G$ depends only on $\sigma$ and $\rho$, and therefore, the size of all gadgets defined in this proof ultimately depend only on $\sigma$ and $\rho$.

  Then, $J$ is obtained by attaching a new vertex $p$, the portal of $J$,
  to the vertex $u$ in $G$.
  Note that each partial solution of the gadget $(J,\{p\})$
  in which $p$ is unselected is a solution of \srCountDomSet
  when restricted to $G$
  ($u$ has a feasible number of selected neighbors in $G$).
  Since such solutions exist independently of the selection status of $u$
  (\cref{lem:happyGadget}), both states $\rho_0$ and $\rho_1$
  (of the portal $p$) can be extended by $(J,\{p\})$.
  Since $p$ only has a single neighbor in $J$,
  there are no extensions for $\rho_i$ if $i\ge 2$.
\end{claimproof}

\begin{figure}[t]
  \begin{subfigure}[b]{.3\textwidth}
    \centering
    \includegraphics{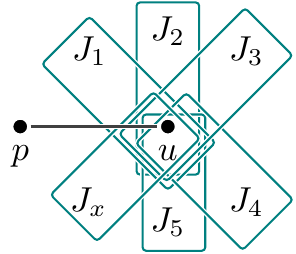}
    \caption{
      The gadget $\CZ=(Z, \{p\})$ with portal $p$ for $x=6$.
    }
    \label{fig:candidate:first}
  \end{subfigure}%
  \hfill
  \begin{subfigure}[b]{.68\textwidth}
    \centering
    \includegraphics{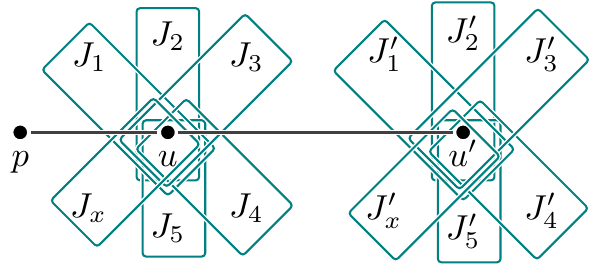}
    \caption{
      The modified gadget $\CZ^*=(Z^*,\{p\})$
      resulting from combining the gadget $\CZ=(Z,\{p\})$
      with a copy $\CZ'=(Z',\{p'\})$ where $p'=u$.
    }
    \label{fig:candidate:second}
  \end{subfigure}
  \centering
  \caption{
    The constructions of the gadgets from \cref{lem:candidate},
    where $J$ denotes the gadget from \cref{clm:extensions1}.
  }\label{fig:candidate}
\end{figure}

By \cref{clm:extensions1}, the gadget $\CJ$ has $f_i\deff \ext_\CJ(\rho_i)\ge 1$ if and only if $i\in \{0,1\}$. For a yet to be determined value $x$, let $(J_1,\{u_1\}),\ldots, (J_x,\{u_x\})$ be $x$ copies of the gadget $\CJ$. From this we obtain a graph $Z$ by identifying the vertices $u_1,\ldots, u_x$ to a single vertex $u$, which is then adjacent to the portal $p$ of $\CZ=(Z,\{p\})$.
See \cref{fig:candidate:first} for an illustration.

Now, consider extensions of the states $\rho_0$ and $\sigma_0$ of $\CZ$ depending on whether $\sigma$ and $\rho$ are finite or cofinite.
Since in both cases the vertex $u$ (i.e., the unique neighbor of $p$ in $Z$) is unselected,
each of the attached copies of $\CJ$ is in state $\rho_0$ or $\rho_1$.
Then, we obtain a solution whenever $u$
obtains a feasible number of selected neighbors
(from $p$ and its neighbors in the attached copies of $\CJ$).

We have
\begin{align*}
  \ext_\CZ(\rho_0)=&
  \begin{dcases}
    \sum_{r\in \rho}
        \binom{x}{r} f_1^r f_0^{x-r}
    &\text{if $\rho$ is finite.}
    \\
    (f_0+f_1)^x-\sum_{r\notin \rho}
        \binom{x}{r} f_1^r f_0^{x-r}
    &\text{if $\rho$ is cofinite.}
  \end{dcases} \\
  =&
  \begin{dcases}
    f_0^x \cdot\sum_{r\in \rho}
        \binom{x}{r} \left(\frac{f_1}{f_0}\right)^r
    &\text{if $\rho$ is finite.}
    \\
    (f_0+f_1)^x - f_0^x \cdot\sum_{r\notin \rho, r\ge 0}
        \binom{x}{r} \left(\frac{f_1}{f_0}\right)^r
    &\text{if $\rho$ is cofinite.}
  \end{dcases}
\end{align*}
and analogously
\begin{equation*}
  \ext_\CZ(\sigma_0)=
  \begin{dcases}
    f_0^x \cdot\sum_{r\in \rho, r\ge 1}
        \binom{x}{r-1} \left(\frac{f_1}{f_0}\right)^{r-1}
    &\text{if $\rho$ is finite.}
    \\
    (f_0+f_1)^x - f_0^x \cdot\sum_{r\notin \rho,r\ge 1}
        \binom{x}{r-1} \left(\frac{f_1}{f_0}\right)^{r-1}
    &\text{if $\rho$ is cofinite.}
  \end{dcases}
\end{equation*}

\begin{claim}\label{clm:ZisCandidate}
  There is an integer $x_0$ such that, for all $x\ge x_0$, $\CZ$ is a candidate.
\end{claim}
\begin{claimproof}
  We have to show properties
  \ref{item:candidate:rho0sig0ge1}, \ref{item:candidate:item:rho0NEQsig0},
  and \ref{item:candidate:item:rhoiEQsigiEQ0} from \cref{def:candidates}.

  First, consider the case where $\rho$ is finite.
  Then, $\ext_\CZ(\rho_0)/f_0^x$ is a polynomial in $x$ of degree $\rhoMax$, whereas $\ext_\CZ(\sigma_0)/f_0^x$ is a polynomial of degree $\rhoMax-1$.
  Since $\rhoMax\ge 1$ by the fact that $(\sigma, \rho)$ is non-trivial, it follows that these polynomials are not constant $0$.
  Hence, there is a value of $x$ that depends only on $\rho$ for which $\ext_\CZ(\rho_0)\neq \ext_\CZ(\sigma_0)$ and $\ext_\CZ(\rho_0), \ext_\CZ(\sigma_0)\ge 1$, as required.

  Second, suppose that $\rho$ is cofinite.
  Then, $\sum_{r\notin \rho} \binom{x}{r} (f_1/f_0)^r$ is a polynomial in $x$.
  Its degree is $\rhoMax-1$,
  i.e., the largest non-negative integer that is missing from $\rho$.
  Since $\rho\neq \ZZ_{\ge0}$, we have $\rhoMax\ge 1$
  and the polynomial is not constant $0$
  (but possibly a constant if $\rhoMax=1$).
  Similarly, $\sum_{r\notin \rho,r\ge 1} \binom{x}{r-1} (f_1/f_0)^{r-1}$
  is a polynomial in $x$ with degree $\rhoMax-2$.
  In this case, the sum might be empty (and thus, equal to $0$),
  but in any case the expressions for $\ext_\CZ(\rho_0)$
  and $\ext_\CZ(\sigma_0)$ are positive and distinct
  for all sufficiently large $x$.

  Finally, since $p$ has only a single neighbor in $Z$, there are no extensions in $\CZ$ for states $\sigma_i$ or $\rho_i$ if $i\ge 2$.
\end{claimproof}

Now, we are in one of three cases.
\begin{description}
  \item[Case 1:] $\ext_\CZ(\rho_1)=\ext_\CZ(\sigma_1)=0$.
  \item[Case 2:] $\ext_\CZ(\rho_1)=0$ and $\ext_\CZ(\sigma_1)\ge 1$.
  \item[Case 3:] $\ext_\CZ(\rho_1)\ge 1$.
\end{description}

In Case 1, the candidate $\CZ$ is a winner, the statement of \cref{lem:candidate} holds, and we are done.
In Case 2, $\CZ$ is a candidate, but not a strong candidate.
In this case, we define another gadget $\CZ^*$ (a modification of $\CZ$), and show that
$\CZ^*$ satisfies the constraint from Case~3. As a consequence, we then have to deal only with Case 3.
In this third case, the candidate $\CZ$ has the property from \cref{item:strongagg0}, and it remains to show that $\ext_\CZ(\rho_0)\neq \ext_\CZ(\sigma_0)+\ext_\CZ(\sigma_1)$ in order for $\CZ$ to be a strong candidate.
We start by defining $\CZ^*$, and subsequently show how to handle Case 3.

Recall that $u$ is the unique neighbor of the portal $p$ in $\CZ$.
Let $\CZ'=(Z',\{p'\})$ be another copy of $\CZ$.
Then, we attach, to the vertex $u$ in $\CZ$, the gadget $\CZ'$
by identifying $u$ and $p'$.
This forms the new gadget $\CZ^*=(Z^*, \{p\})$.
See \cref{fig:candidate:second} for an illustration of this construction.

\begin{claim}\label{clm:CZstar}
  Suppose that $\ext_\CZ(\rho_1)=0$ and $\ext_\CZ(\sigma_1)\ge 1$. Then, $\CZ^*$ is a candidate with $\ext_{\CZ^*}(\rho_1)\ge 1$. Moreover,
  $\ext_{\CZ^*}(\rho_0)=\ext_{\CZ}(\rho_0)^2$ and $\ext_{\CZ^*}(\sigma_0)=\ext_{\CZ}(\sigma_0)\cdot \ext_{\CZ}(\rho_0)$.
\end{claim}
\begin{claimproof}
  As $\CZ$ is a candidate with $\ext_\CZ(\rho_1)=0$, we have $\ext_\CZ(\rho_i)=0$ for $i\ge 1$.
  Note that $\ext_{\CZ^*}(\rho_0)=\ext_{\CZ}(\rho_0)^2$ because if $u$ is unselected, then $\CZ'$ has to be in state $\rho_0$ as $\CZ'$ is a copy of $\CZ$ and $\ext_\CZ(\rho_i)=0$ for $i\ge 1$. This does not affect the number of selected neighbors of $u$, so every partial solution of $\CZ$ that witnesses $\rho_0$ can be extended by every such partial solution of the copy $\CZ'$ to a partial solution of $\CZ^*$ that witnesses $\rho_0$.
  Analogously, $\ext_{\CZ^*}(\sigma_0)=\ext_{\CZ}(\sigma_0)\cdot \ext_{\CZ}(\rho_0)$.

  Thus, we still have $\ext_{\CZ^*}(\rho_0)\neq \ext_{\CZ^*}(\sigma_0)$ and $\ext_{\CZ^*}(\rho_0), \ext_{\CZ^*}(\sigma_0)\ge 1$.
  Also, as before, $p$ has only a single neighbor in $Z^*$, and consequently, there are no extensions for states $\sigma_i$ or $\rho_i$ if $i\ge 2$.

  We now verify that $\ext_{\CZ^*}(\rho_1)\ge 1$.
  Since $\ext_\CZ(\sigma_1)\ge 1$, there is a partial solution $S$ of $\CZ$
  that witnesses $\sigma_1$.
  Let $S'$ be a copy of this set for the gadget $\CZ'$.
  Then, $S^*\deff S'\cup S\setminus\{p\}$ is a partial solution of $\CZ^*$
  that witnesses $\rho_1$.
  The key observation is that $u$ has the same number of selected neighbors
  as in the partial solution $S$
  (where now instead of $p$, it has a single selected neighbor in $\CZ'$).
\end{claimproof}

In order to complete the proof of \cref{lem:candidate} for Cases 2 and 3, we set
\[
\CG\deff
\begin{cases}
	\CZ^* & \text{if $\ext_\CZ(\rho_1)=0$ and $\ext_\CZ(\sigma_1)\ge 1$.}\\
  	\CZ &\text{if $\ext_\CZ(\rho_1)\ge 1$.}
\end{cases}
\]
As a consequence, we have that $\CG$ is a candidate
with the properties from Case~3, i.e., $\ext_{\CG}(\rho_1)\ge 1$.
It remains to show that
$\ext_\CG(\rho_0)\neq \ext_\CG(\sigma_0)+\ext_\CG(\sigma_1)$
(\cref{item:strongagg}) for $\CG$ to be a strong candidate.

To this end, we consider two cases
depending on whether $\rho$ is cofinite or finite.
\begin{description}
\item [\boldmath $\rho$ is cofinite.]
In this case, we can directly obtain a strong candidate from $\CG$ as follows.
\begin{claim}
  If $\rho$ is cofinite, then there is an integer $x_0$ such that, for all $x \ge x_0$, $\CG$ is a strong candidate.
\end{claim}
\begin{claimproof}
  It remains to show that $\ext_\CG(\rho_0)\neq \ext_\CG(\sigma_0)+\ext_\CG(\sigma_1)$.
  For cofinite $\rho$, recall that,
  for the candidate $\CZ$, we have that
  \begin{align*}
    \ext_\CZ(\rho_0)= &(f_0+f_1)^x
      -f_0^x \cdot\sum_{r\notin \rho, r\ge 0}
        \binom{x}{r} \left(\frac{f_1}{f_0}\right)^r
      \\
    \intertext{and}
    \ext_\CZ(\sigma_0)= & (f_0+f_1)^x
      - f_0^x \cdot\sum_{r\notin \rho,r\ge 1}
        \binom{x}{r-1} \left(\frac{f_1}{f_0}\right)^{r-1}
      .
  \end{align*}
  This shows that $\ext_\CZ(\rho_0)< \ext_\CZ(\sigma_0)$ (using that $x$ is sufficiently large), and consequently, $\ext_\CZ(\rho_0)\neq \ext_\CZ(\sigma_0) +\ext_\CZ(\sigma_1)$.
  This concludes the proof if $\CG=\CZ$.

  For the case when $\CG = \CZ^*$,
  we get from \cref{clm:CZstar} that $\ext_{\CG}(\rho_0)=\ext_{\CZ}(\rho_0)^2$
  and $\ext_{\CG}(\sigma_0)=\ext_{\CZ}(\sigma_0)\cdot \ext_{\CZ}(\rho_0)$.
  Combined with the previous observations, this implies
  $\ext_{\CG}(\rho_0) < \ext_{\CG}(\sigma_0)$.
\end{claimproof}

\item[\boldmath $\rho$ is finite.]

Now, suppose that $\rho$ is finite. We fix $x$ in the definition of $\CG$ (i.e., in the definition of $\CZ$ or $\CZ^*$, respectively) to be sufficiently large for \cref{clm:ZisCandidate} to hold. We can assume that $\ext_\CG(\rho_0)= \ext_\CG(\sigma_0) +\ext_\CG(\sigma_1)$, as otherwise $\CG$ is a strong candidate and we are done.
We make one last modification and define $\CG'=(G',\{q\})$ as follows.

For a yet to be determined value $y$, let $(G_1,\{p_1\}),\ldots, (G_y,\{p_y\})$ be $y$ copies of the gadget $\CG$. We use the same trick as before. We obtain a graph $G'$ by identifying the vertices $p_1,\ldots, p_y$ to a single vertex $p$, which is then adjacent to the portal $q$ of $\CG'=(G',\{q\})$.
In order to complete the proof of \cref{lem:candidate}, we show the following claim.
\begin{claim}
  If $\rho$ is finite and $\ext_\CG(\rho_0)= \ext_\CG(\sigma_0) +\ext_\CG(\sigma_1)$, then there is an integer $y_0$ such that, for all $y\ge y_0$, $\CG'$ is a strong candidate.
\end{claim}
\begin{claimproof}
  Using that $\ext_{\CG}(\rho_0),\ext_{\CG}(\rho_1)\ge 1$, the proof that $\CG'$ is a candidate for sufficiently large $y$ is analogous to the proof that $\CZ$ is a candidate for sufficiently large $x$.

  Since $\CG$ is a candidate, we have $\ext_{\CG}(\sigma_0)\ge 1$ and $\ext_{\CG}(\rho_0)\neq \ext_{\CG}(\sigma_0)$. From the fact that $\ext_\CG(\rho_0)= \ext_\CG(\sigma_0) +\ext_\CG(\sigma_1)$, it then follows that $\ext_\CG(\sigma_1)\ge 1$.
  To shorten notation, let $g_0\deff \ext_{\CG}(\rho_0)$,
  $g_1\deff \ext_{\CG}(\rho_1)$,
  $h_0\deff \ext_{\CG}(\sigma_0)$, and
  $h_1\deff \ext_{\CG}(\sigma_1)$.
  From the properties of $\CG$, we have $g_0,g_1, h_0,h_1\ge 1$, $g_0\neq h_0$, and $g_0=h_0+h_1$.

  Then, as before, using the fact that $\rho$ is finite, we have
  \begin{align*}
    \ext_{\CG'}(\rho_0)&= g_0^y\cdot\sum_{r\in \rho}
      \binom{y}{r} \left(\frac{g_1}{g_0}\right)^r
    \quad \text{and} \\
    \ext_{\CG'}(\sigma_0)&= g_0^y \cdot \sum_{r \in \rho,r\ge 1}
      \binom{y}{r-1} \left(\frac{g_1}{g_0}\right)^{r-1}
    .
  \end{align*}

  First, suppose $\sigma$ is finite, then analogously
  \begin{align*}
    \ext_{\CG'}(\rho_1)&= h_0^y\cdot \sum_{s\in \sigma}
      \binom{y}{s} \left(\frac{h_1}{h_0}\right)^s
    \quad \text{and} \\
    \ext_{\CG'}(\sigma_1)&= h_0^y\cdot \sum_{s\in \sigma,s\ge 1}
      \binom{y}{s-1} \left(\frac{h_1}{h_0}\right)^{s-1}
    .
  \end{align*}
  So, $\ext_{\CG'}(\rho_1)\ge 1$ as $h_0,h_1\ge 1$, and the corresponding sum is not empty ($\CG'$ has property \cref{item:strongagg0}).
  Note that there are polynomials $p_1, p_2, p_3$ in $y$ such that
  $\ext_{\CG'}(\rho_0)      = g_0^y \cdot p_1(y)$,
  $\ext_{\CG'}(\sigma_0)    = g_0^y \cdot p_2(y)$,
  and $\ext_{\CG'}(\sigma_1)= h_0^y \cdot p_3(y)$,
  where $p_1(y) > p_2(y)+1$ (for sufficiently large $y$)
  since $\rho\neq \{0\}$ as $(\sigma, \rho)$ is non-trivial.
  Thus, using the fact that $g_0=h_0+h_1$ with $h_1\ge 1$,
  and consequently, $g_0>h_0$, for sufficiently large $y$ we have
  \[
    \frac{\ext_{\CG'}(\sigma_0)+\ext_{\CG'}(\sigma_1)}{\ext_{\CG'}(\rho_0)}
    =\frac{g_0^y\cdot p_2(y)+h_0^y\cdot p_3(y)}{g_0^y\cdot p_1(y)}
    =\frac{p_2(y)}{p_1(y)}+ \frac{h_0^y\cdot p_3(y)}{g_0^y\cdot p_1(y)}
    < 1.
  \]
  This shows that $\CG'$ is a strong candidate if $\sigma$ is finite.

  It remains to consider the case where $\sigma$ is cofinite. The expressions for $\ext_{\CG'}(\rho_0)$ and $\ext_{\CG'}(\sigma_0)$ are the same as in the previous case, but
  \begin{align*}
    \ext_{\CG'}(\rho_1) &= (h_0+h_1)^y
      - h_0^y\cdot \sum_{s\notin \sigma}
        \binom{y}{s} \left(\frac{h_1}{h_0}\right)^s
    \quad \text{and} \\
    \ext_{\CG'}(\sigma_1) &= (h_0+h_1)^y
      - h_0^y\cdot \sum_{s\notin \sigma,s\ge 1}
        \binom{y}{s-1} \left(\frac{h_1}{h_0}\right)^{s-1}
    .
  \end{align*}
  Clearly, $\ext_{\CG'}(\rho_1)\ge 1$ (since $y$ is sufficiently large).
  Moreover, $\ext_{\CG'}(\sigma_1) \le (h_0+h_1)^y=g_0^y$.
  Hence, for sufficiently large $y$,
  \[
    \frac{\ext_{\CG'}(\sigma_0)+\ext_{\CG'}(\sigma_1)}{\ext_{\CG'}(\rho_0)}
    \le\frac{g_0^y\cdot p_2(y)+g_0^y}{g_0^y\cdot p_1(y)}
    =\frac{p_2(y)+1}{p_1(y)}
    < 1.
  \]
  This shows that $\CG'$ is a strong candidate if $\sigma$ is cofinite.
\end{claimproof}
\end{description}
This concludes the construction of the gadget
which is either a winner or a strong candidate.
\end{proof}

\subsubsection*{Forcing a Single Unselected Vertex with no Selected Neighbors: Proof of \texorpdfstring{\cref{lem:forcesingle00}}{Lemma \ref{lem:forcesingle00}}}
\label{sec:forcesingle00}

We now prove \cref{lem:forcesingle00},
which we restate here for convenience.
\forcesinglezerozero*

\begin{proof}
Let $I=(G,\lambda)$ be an instance of \srCountDomSetRel[\CP+(\emptyset, \{0\})],
and let
$n=|V(G)|$. If $G$ contains no $(\emptyset, \{0\})$-vertex, then we can directly call the
\srCountDomSetRel[\CP]-oracle on $I$.
So, suppose that $p$ is a single $(\emptyset, \{0\})$-vertex
in $G$.

Note that $\CG=(G,\{p\}, \lambda)$
can be interpreted as a gadget for \srCountDomSetRel[\CP] with portal $p$.
Then, our goal is to compute $\ext_\CG(\rho_0)$,
as this corresponds to the number of partial solutions of $\CG$
for which $p$ is unselected and has no selected neighbors in $\CG$.
These are precisely the solutions of \srCountDomSetRel[\CP+(\emptyset, \{0\})]
on input $I$.

The proof outline is as follows. For an integer $x$, we use $x$ copies of the strong candidate $\CJ$ to define a gadget
$\CJ^{(x)}\deff(J^{(x)},\{p\})$.
The main idea is to combine the gadgets $\CG$ and $\CJ^{(x)}$
into an instance $I_x$ of \srCountDomSetRel[\CP].
Then, using the oracle for \srCountDomSetRel[\CP],
we can compute the number of solutions for $I_x$.
By our design of $\CJ^{(x)}$, we are then able to use polynomial interpolation over $x$
to recover the number of extensions for $\CG$ from these values.
In particular, we can recover the value of $\ext_{\CG}(\rho_0)$. As we have seen, this corresponds to the number of solutions for $I$.

\begin{figure}[t]
  \centering
  \includegraphics{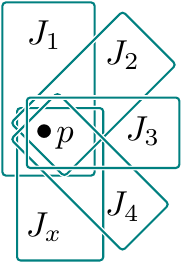}
  \caption{
    The construction of the gadget $\CJ^{(x)}$ for $x=5$
    from \cref{lem:forcesingle00}.
    This gadget is used to replace a single \((\emptyset,\{0\})\)-vertex $p$.
  }\label{fig:forcesingle00}
\end{figure}

\subparagraph*{\boldmath Definition of $I_x$.}
For a positive integer $x$,
we define the instance $I_x=(G_x,\lambda_x)$ of \srCountDomSetRel[\CP].
Intuitively, in $I_x$, the vertex $p$ is replaced by a $(\sigma, \rho)$-vertex
to which we attach $x$ copies of the strong candidate $\CJ$,
where $p$ acts as portal for each of them.
Note that since $\CJ$ is a gadget for \srCountDomSet,
it is also a gadget for \srCountDomSetRel[\CP]
that happens to use only $(\sigma, \rho)$-vertices.

More formally, for each $i\in \nset{x}$, let $(J_{i},\{v_{i}\})$ be a copy of $\CJ$.
The graph $G_x$ is obtained from $G$ by identifying all vertices in $\{p,v_1, \ldots,v_x\}$ to a single vertex.
Then, $\lambda_x$ is identical to $\lambda$ on the vertices in $V(G)\setminus \{p\}$, and maps the remaining vertices to $0$ (i.e., the vertices in $\bigcup_{i=1}^x V(J_i)$ are $(\sigma, \rho)$-vertices).
See \cref{fig:forcesingle00} for an illustration of the construction.

With this definition,
we let $J^{(x)}$ be the graph induced by the union of the $x$ copies of $\CJ$
that are attached to $p$ in $G_x$.
As mentioned above, $\CJ^{(x)} \deff (J^{(x)},\{p\})$
can be cast as a gadget for \srCountDomSetRel[\CP].
(Here, we dropped the $\lambda$-term since $\lambda_x$ is constant $0$ on $J^{(x)}$.)

We first analyze the pathwidth of $I_x$.
Since all copies of $\CJ$ are disjoint (except for the portal $p$),
we can, for each such copy of $\CJ$, duplicate one bag
from the original path decomposition of $G$,
and add the vertices of the corresponding copy of $\CJ$ to it.
If we create a fresh copy of the bag for each copy of $\CJ$,
then the pathwidth increases by at most the size of the strong candidate $\CJ$.
As the size of a strong candidate depends only on $\sigma$ and $\rho$,
it is bounded by a constant, and thus, the reduction is pathwidth-preserving.

\subparagraph*{Properties of the Gadgets and Graphs.}
To analyze the constructions and gadgets from above,
the usual definitions of $\rhoMax$ and $\sigMax$ are not convenient.
Instead, we use $r^*$ as the maximum element of $\rho$ if $\rho$ is finite,
and the maximum missing integer from $\rho$ if $\rho$ is cofinite.
Analogously, we define $s^*$ for $\sigma$.
In the following, we always assume that $x\ge \max\{s^*, r^*\}$.

Let $\#S(I_x)$ denote the number of solutions of \srCountDomSetRel[\CP]
on input $I_x$. We have
\begin{align}
  \#S(I_x) &= A_x + B_x
    \label{equ:SIx}
  \intertext{
  where $A_x$ denotes the number of solutions in $I_x$
  when $p$ is not selected, that is,
  }
  A_x &= \sum_{i+j \in \rho}
            \ext_\CG(\rho_i) \cdot \ext_{\CJ^{(x)}}(\rho_j)
  \intertext{
  and $B_x$ denotes the number of solutions in $I_x$ when $p$ is selected,
  that is,
  }
  B_x &= \sum_{i+j \in \sigma}
            \ext_\CG(\sigma_i) \cdot \ext_{\CJ^{(x)}}(\sigma_j)
  .
\end{align}

\subparagraph*{\boldmath Analyzing $A_x$.}
We first focus on the more detailed analysis of $A_x$,
which corresponds to the case when $p$ is not selected.
The results for $B_x$ then follow analogously.
By the definition of $\CJ^{(x)}$,
we get that $\ext_{\CJ^{(x)}}(\rho_r) = 0$
for all $r > x$.
By the definition of a strong candidate, we have $\ext_{\CJ}(\rho_j)\ge 1$ if and only if $j\in \{0,1\}$.
Thus, in order for $p$ to have state $\rho_r$ in $\CJ^{(x)}$, there are precisely $r$ copies of $\CJ$ in which $p$ has state $\rho_1$, and in the remaining $x-r$ copies it has state $\rho_0$.
Hence, for $r \in \fragment{0}{x}$,
\begin{align}
  \ext_{\CJ^{(x)}}(\rho_{r})
  &= \binom{x}{r} \ext_{\CJ}(\rho_{1})^r \ext_{\CJ}(\rho_{0})^{x-r}
  .\nonumber
  \intertext{
    For ease of notation, we rewrite this expression as
  }
  &= \ext_{\CJ}(\rho_{0})^{x-r^*} \cdot a_r(x)  \label{equ:extrho}
  \intertext{
    where
  }
  a_r(x) &\deff \binom{x}{r} \ext_{\CJ}(\rho_{1})^r \ext_{\CJ}(\rho_{0})^{r^*-r}
  \label{equ:aSUBrOFx}
\end{align}
Observe that $a_r(x)$ is a polynomial of degree $r$ in $x$
since $\ext_\CJ(\rho_0)$ and $\ext_\CJ(\rho_1)$ are positive integers
by our assumptions about the gadget $\CJ$.
Moreover, if $x\ge r$ and $r^*\ge r$, then $a_r(x)$ is a positive integer.

In the following, it is convenient to consider a state of $p$ in $\CG$, and then analyze the number of corresponding extensions to $\CJ^{(x)}$ that give a feasible selection for $I_x$ overall. Essentially, given that $p$ has precisely $i$ selected neighbors in $\CG$, the extensions to $\CJ^{(x)}$ have to complement this state of $p$ in $\CG$
such that $p$ has a feasible number of neighbors overall.
With
\begin{align}
  \alpha &\deff \ext_{\CJ}(\rho_1)+\ext_{\CJ}(\rho_0)
  \intertext{
    and the fact that
  }
  \alpha^x &= (\ext_{\CJ}(\rho_1)+\ext_{\CJ}(\rho_0))^x
            = \sum_{r=0}^x \ext_{\CJ^{(x)}}(\rho_{r}),
    \nonumber
  \intertext{
    we define
  }
  f_i(x) &\deff
    \begin{dcases}
      \sum_{i+r\in \rho} \ext_{\CJ^{(x)}}(\rho_{r})
        &\text{if $\rho$ is finite}\\
      \alpha^x -\sum_{i+r\notin \rho} \ext_{\CJ^{(x)}}(\rho_{r})
        &\text{if $\rho$ is cofinite}
    \end{dcases}
    \nonumber
  \intertext{
    Note that $f_i(x)$ is the number of extensions of $\CJ^{(x)}$ to a solution of $I_x$, given that $p$ is unselected and already has $i$ selected neighbors in $\CG$ (i.e., it has state $\rho_i$ in $\CG$).
    Using \cref{equ:extrho}, this equals
  }
  &=  \begin{dcases}
        \sum_{i+r\in \rho} \ext_{\CJ}(\rho_{0})^{x-r^*} \cdot a_r(x)
          &\text{if $\rho$ is finite}\\
        \alpha^x
          -\sum_{i+r\notin \rho} \ext_{\CJ}(\rho_{0})^{x-r^*} \cdot a_r(x)
          &\text{if $\rho$ is cofinite}
      \end{dcases}
      \nonumber
  \intertext{
    which we can simplify to
  }
  &=\begin{dcases}
      \ext_{\CJ}(\rho_{0})^{x-r^*} \cdot \sum_{i+r\in \rho}  a_r(x)
        &\text{if $\rho$ is finite}\\
      \alpha^x
        - \ext_{\CJ}(\rho_{0})^{x-r^*} \cdot \sum_{i+r\notin \rho} a_r(x)
        &\text{if $\rho$ is cofinite.}
    \end{dcases}
    \label{equ:extfix}
  \intertext{
    Now we revisit the definition of $A_x$
    and combine the previous observations.
    Let $k$ be the overall number of neighbors of $p$ in $G$.
    From the above arguments and the definition of $f_i(x)$, we get
  }
  A_x &= \sum_{i+j \in \rho}
            \ext_\CG(\rho_i) \cdot \ext_{\CJ^{(x)}}(\rho_j)
       = \sum_{i=0}^k \ext_\CG(\rho_i) \cdot f_i(x)
       .
  \intertext{
    Replacing $f_i(x)$ by \cref{equ:extfix} yields
  }
  &= \ext_{\CJ}(\rho_{0})^{x-r^*}\cdot
      \underbrace{
        \sum_{i=0}^k \ext_\CG(\rho_i)
          \cdot \sum_{i+r\in \rho} a_r(x)
      }_{a(x)}
      && \text{if $\rho$ is finite and}
      \label{equ:Axfin}
    \\
  &=
    \alpha^x\cdot \sum_{i=0}^k \ext_\CG(\rho_i) - \ext_{\CJ}(\rho_{0})^{x-r^*}\cdot
    \underbrace{
      \sum_{i=0}^k \ext_\CG(\rho_i)
      \cdot \sum_{i+r\notin \rho} a_r(x)
    }_{a(x)}
    && \text{if $\rho$ is cofinite.}
    \label{equ:Axcof}
\end{align}
In \cref{equ:Axfin} together with \cref{equ:Axcof},
we define an expression $a(x)$ depending on
whether $\rho$ is finite or cofinite.
We later use the crucial fact that $a(x)$ is a polynomial in $x$
because, for each $r$, $a_r(x)$ is a polynomial in $x$.

\subparagraph*{\boldmath Analyzing $B_x$.}
For the case when the vertex $p$ is selected,
we get the following results analogously to the previous observations.
Similarly to before, we have $\ext_{\CJ^{(x)}}(\sigma_\ell) = 0$
for all $\ell > x$, and, for $s \in \fragment{0}{x}$,
\begin{align}
  \ext_{\CJ^{(x)}}(\sigma_s)
  &= \ext_{\CJ}(\sigma_0)^{x-s^*} \cdot b_s(x)  \label{equ:extsig}
  \intertext{
    where
  }
  b_s(x) &\deff \binom{x}{s} \ext_{\CJ}(\sigma_1)^s \ext_{\CJ}(\sigma_0)^{s^*-s}
\end{align}
In contrast to $a_r(x)$,
$b_s(x)$ is not necessarily a positive integer
as the definition of the strong candidate $\CJ$ does not guarantee that $\ext_\CJ(\sigma_1)$ is non-zero.
However, $b_s(x)$ is still a polynomial in $x$.
Analogously, for
\begin{align}
  \beta &\deff \ext_{\CJ}(\sigma_1)+\ext_{\CJ}(\sigma_0),
  \intertext{
  we can rewrite $B_x$ as
  }
  &= \ext_{\CJ}(\sigma_{0})^{x-s^*}\cdot
    \underbrace{
      \sum_{i=0}^k \ext_\CG(\sigma_i) \cdot \sum_{i+s\in \sigma} b_s(x)
    }_{b(x)}
    && \text{if $\sigma$ is finite and}
    \label{equ:Bxfin}
    \\
  &= \beta^x\cdot \sum_{i=0}^k \ext_\CG(\sigma_i)
    - \ext_{\CJ}(\sigma_{0})^{x-s^*}\cdot
      \underbrace{
        \sum_{i=0}^k \ext_\CG(\sigma_i) \cdot \sum_{i+s\notin \sigma} b_s(x)
      }_{b(x)}
    && \text{if $\sigma$ is cofinite.}
    \label{equ:Bxcof}
\end{align}
Again, $b(x)$ is a polynomial in $x$
because each $b_s(x)$ is a polynomial in $x$.

\subparagraph*{Polynomial Interpolation.}
At this point,
let us recall that our goal is to compute $\ext_\CG(\rho_0)$ which equals
the number of solutions of \srCountDomSetRel[\CP+(\emptyset, \{0\})]
on input $I$.
We aim to use polynomial interpolation by using $A_x$
to obtain a polynomial in $x$ with $\ext_\CG(\rho_0)$ as a coefficient.
Given the different expressions for $A_x$ depending on whether $\sigma$ and $\rho$ are finite or cofinite, we split the proof into two cases at this point.
We start with the substantially easier case where both sets are finite.

\begin{description}
\item[\boldmath Both $\sigma$ and $\rho$ are finite.]
In order to do polynomial interpolation on $A_x$, we first show how to isolate the term $A_x$ from the value of $\#S(I_x)$, which we can compute by using an oracle call.
\begin{claim}\label{clm:recoverBx}
  For sufficiently large $x\in \Oh(n)$ and given $\#S(I_x)$, the term $A_x$ can be computed (in time polynomial in $n$).
\end{claim}
\begin{claimproof}
  From the definition of a candidate, we know that $\ext_{\CJ}(\sigma_{0})$ and $\ext_{\CJ}(\rho_{0})$ are distinct positive values. We show how to recover $A_x$ from $\#S(I_x)$ if $\ext_{\CJ}(\rho_{0})>\ext_{\CJ}(\sigma_{0})$. In the case $\ext_{\CJ}(\rho_{0})<\ext_{\CJ}(\sigma_{0})$, the term $B_x$ can be recovered analogously, and then $A_x= \#S(I_x)-B_x$.

  Plugging \cref{equ:Axfin,equ:Bxfin} into \cref{equ:SIx}, we obtain

  \begin{equation*}
    \#S(I_x)=
    \underbrace{
      \ext_{\CJ}(\rho_{0})^{x-r^*}\cdot
      \sum_{i=0}^k \ext_\CG(\rho_i)
      \cdot
      % \Bigl[
      \sum_{i+r\in \rho} a_r(x)
      % \Bigr]
    }_{= A_x}
    +
    \underbrace{
      \ext_{\CJ}(\sigma_{0})^{x-s^*}\cdot
      \sum_{i=0}^k \ext_\CG(\sigma_i)
      \cdot
      % \Bigl[
      \sum_{i+s\in \sigma} b_s(x)
      % \Bigr]
    }_{= B_x}
  \end{equation*}
  Note that the terms $b_s(x)$ are polynomials in $x$,
  and that each of the terms $\ext_\CG(\sigma_i)$ is upper bounded by $2^n$.
  Hence, the outer sum of $B_x$ can be bounded by $q(x) \cdot 2^n$,
  for some polynomial $q(x)$ in $x$.
  Thus, using the fact that $\ext_{\CJ}(\rho_{0})>\ext_{\CJ}(\sigma_{0})$, we can choose $x\in \Oh(n)$
  such that
  \[
    \frac{B_x}{\ext_{\CJ}(\rho_{0})^{x-r^*}}
    \le
    \frac{\ext_{\CJ}(\sigma_{0})^{x-s^*} q(x) \cdot 2^n}{\ext_{\CJ}(\rho_{0})^{x-r^*}}
    =
    \left(\frac{\ext_{\CJ}(\sigma_{0})}{\ext_{\CJ}(\rho_{0})}\right)^x\cdot \frac{\ext_{\CJ}(\rho_{0})^{r^*} q(x) \cdot 2^n}{\ext_{\CJ}(\sigma_{0})^{s^*}}
    < 1
  \]
  (using the fact that $\ext_{\CJ}(\rho_{0})>\ext_{\CJ}(\sigma_{0})$).
  By the definition of $A_x$,
  we observe that $A_x/\ext_{\CJ}(\rho_{0})^{x-r^*}$ is an integer.

  Therefore, $\floor{\#S(I_x)/\ext_{\CJ}(\rho_{0})^{x-r^*}}=A_x/\ext_{\CJ}(\rho_{0})^{x-r^*}$,
  and thus, $A_x$ can be recovered by computing $A_x= \floor{\#S(I_x)/\ext_{\CJ}(\rho_{0})^{x-r^*}}\cdot \ext_{\CJ}(\rho_{0})^{x-r^*}$.
\end{claimproof}

Recall that our ultimate goal is to compute $\ext_\CG(\rho_0)$
using calls to a \srCountDomSetRel[\CP]-oracle.
Using \cref{equ:Axfin} together with \cref{equ:aSUBrOFx}, we have
\[
  p(x) \deff \frac{A_x}{\ext_{\CJ}(\rho_{0})^{x-r^*}} =
      \sum_{i=0}^k \ext_\CG(\rho_i)
        \cdot \sum_{i+r\in \rho}
          \binom{x}{r} \ext_{\CJ}(\rho_{1})^r \ext_{\CJ}(\rho_{0})^{r^*-r}.
\]
Note that $p(x)$ is a polynomial in $x$ whose highest degree term is
$\ext_\CG(\rho_0) \ext_{\CJ}(\rho_{1})^{r^*} \cdot x^{r^*}$.
Using polynomial interpolation, we can recover the coefficients of $p(x)$ by evaluating the polynomial for $r^*+1$ distinct values of $x$.
This can be done since, according to \cref{clm:recoverBx},
we can compute $A_x$ using a \srCountDomSetRel[\CP]-oracle call
as long as $x$ is sufficiently large in $\Oh(n)$;
and we can compute $\ext_{\CJ}(\rho_{0})^{x-r^*}$
as $\ext_{\CJ}(\rho_{0})$ is a positive integer that does not depend on $n$.
This can be done in time polynomial in $x\in \Oh(n)$, i.e., in time polynomial in $n$.
Using the fact that $\ext_{\CJ}(\rho_{1})$ is non-zero by assumption of the lemma, we can recover $\ext_\CG(\rho_0)$ from the coefficient of the highest degree term $x^{r^*}$.

Summarizing, we have shown that, for finite $\sigma$ and $\rho$,
we can solve \srCountDomSetRel[\CP+(\emptyset, \{0\})] on instance $I$
using $r^*$ oracle calls to \srCountDomSetRel[\CP]
on instances of the form $I_x$.

\item[\boldmath One of $\sigma$ or $\rho$ is cofinite.]
In the cofinite case, there is an additional obstacle to computing $\ext_\CG(\rho_0)$ from $A_x$ by polynomial interpolation. If $\rho$ is cofinite, then $A_x$ contains the unwanted exponential expression $\alpha^x$, see \cref{equ:Axcof}.
We eliminate the leading exponential terms by considering
\begin{align}
  \psi(x)\deff& \#S(I_{x+2}) -(\alpha+\beta)\#S(I_{x+1}) + \alpha\beta\#S(I_{x})
  \nonumber \\
  =&
  \underbrace{
    A_{x+2}-(\alpha+\beta)A_{x+1}+ \alpha\beta A_x
  }_{\eqqcolon A'_x}
  +
  \underbrace{
    B_{x+2}-(\alpha+\beta)B_{x+1}+ \alpha\beta B_x
  }_{\eqqcolon B'_x}.
  \label{equ:SIxmodified}
\end{align}

So, we aim to recover $\ext_\CG(\rho_0)$ from $A'_x$, rather than from $A_x$ directly. The following claim establishes that we can isolate the term $A'_x$ from the value of $\psi(x)$, which we can compute using oracle calls.

\begin{claim}\label{clm:recoverAprimex}
  For sufficiently large $x\in \Oh(n)$ and given $\#S(I_{x+2})$, $\#S(I_{x+1})$, and $\#S(I_{x})$, the term $A'_x$ can be computed (in time polynomial in $n$).
\end{claim}

Before we proceed with the proof of the claim, to shorten notation, we set
\begin{align}
  a'(x)\deff&
    \ext_{\CJ}(\rho_{0})^{x+2-r^*}\cdot a(x+2)
    - (\alpha+\beta) \ext_{\CJ}(\rho_{0})^{x+1-r^*} \cdot a(x+1)
    + \alpha\beta \ext_{\CJ}(\rho_{0})^{x-r^*} \cdot a(x)
    \nonumber \\
  =& \ext_{\CJ}(\rho_{0})^{x-r^*} \cdot(
    \ext_{\CJ}(\rho_{0})^2 \cdot a(x+2)
    - (\alpha+\beta) \ext_{\CJ}(\rho_{0}) \cdot a(x+1)
    + a(x)
  )
  \label{equ:aprimex}
\end{align}
We use \cref{equ:Axfin,equ:Axcof} to expand the expression $A'_x$.
If $\rho$ is finite, then
\begin{align}
  A'_x &= a'(x).
  \intertext{
    If $\rho$ is cofinite, then
  }
  A'_x&= \left(\alpha^{x+2}-(\alpha+\beta)\alpha^{x+1}
  +\alpha\beta\alpha^{x}\right) \cdot \left(\sum_{i=0}^k \ext_\CG(\rho_i)\right)
  - a'(x) \nonumber\\
  &=-a'(x).
\end{align}

\begin{claimproof}[\ref{clm:recoverAprimex}]
  By the fact that $\CJ$ is a candidate,
  we know that $\ext_{\CJ}(\sigma_{0})$ and $\ext_{\CJ}(\rho_{0})$
  are distinct positive values.
  We show how to recover $A'_x$ from $\psi(x)$
  if $\ext_{\CJ}(\rho_{0})>\ext_{\CJ}(\sigma_{0})$.
  In the case $\ext_{\CJ}(\rho_{0})<\ext_{\CJ}(\sigma_{0})$,
  the term $B'_x$ can be recovered analogously, and then $A'_x= \psi(x)-B'_x$.

  Recall that $A'(x)=\pm a'(x)$,
  where the sign depends on whether $\rho$ is finite or cofinite.
  If, analogously to \cref{equ:aprimex}, we define
  \[
    b'(x) \deff \ext_{\CJ}(\sigma_{0})^{x+2-s^*}\cdot b(x+2)
      -(\alpha+\beta) \ext_{\CJ}(\sigma_{0})^{x+1-r^*} \cdot b(x+1)
      +\alpha\beta \ext_{\CJ}(\sigma_{0})^{x-r^*} \cdot b(x),
  \]
  then again $B'(x)=\pm b'(x)$,
  where the sign depends on whether $\sigma$ is finite or cofinite.
  Then, we recall that $b(x)$ is a polynomial in $x$, and that
  $\ext_\CG(\sigma_0)$, which contributes to the coefficients of $b(x)$,
  is upper bounded by $2^n$.
  Thus, we can choose $x\in \Oh(n)$ such that
  $B'_x/\ext_{\CJ}(\rho_{0})^{x-r^*}<1$, where we use the crucial fact
  that $\ext_{\CJ}(\rho_{0})>\ext_{\CJ}(\sigma_{0})$.
  Then, we observe that $A'_x/\ext_{\CJ}(\rho_{0})^{x-r^*}$
  is an integer as $a_r(x)$ is an integer for $r\le r^*$,
  and consequently, $a(x)$ is an integer.

  Therefore, $\floor{\psi(x)/\ext_{\CJ}(\rho_{0})^{x-r^*}}=A'_x/\ext_{\CJ}(\rho_{0})^{x-r^*}$,
  and $A'_x$ can be recovered by computing
  $A'_x= \floor{\psi(x)/\ext_{\CJ}(\rho_{0})^{x-r^*}}\cdot \ext_{\CJ}(\rho_{0})^{x-r^*}$.
\end{claimproof}

Recall that our ultimate goal is to compute $\ext_\CG(\rho_0)$
using calls to a \srCountDomSetRel[\CP]-oracle.
As noted previously, the terms $a(x)$ are polynomials in $x$.
Therefore, it is straightforward to verify that
$p(x)\deff A'_x/\ext_{\CJ}(\rho_{0})^{x-r^*}$ is also a polynomial in $x$
(whether $\rho$ is finite or cofinite).
With \cref{clm:recoverAprimex} in hand, we can compute $A'_x$,
and then compute $\ext_\CG(\rho_0)$ using standard interpolation.

\begin{claim}\label{clm:pxcoefficients}
  $\ext_\CG(\rho_0)$ can be computed from the coefficients of $p(x)$.
\end{claim}
\begin{claimproof}[\ref{clm:pxcoefficients}]
  We are interested in the highest degree monomial of $p(x)$.
  By \cref{equ:aprimex}, we have
  \begin{align}
    p(x) = \frac{A'_x}{\ext_{\CJ}(\rho_{0})^{x-r^*}}
    = \pm \ext_{\CJ}(\rho_{0})^2 \cdot a(x+2)
      \mp (\alpha+\beta) \ext_{\CJ}(\rho_{0}) \cdot a(x+1)
      \pm a(x)
      \label{equ:pOFxExpanded}
  \end{align}
  Let us first investigate the polynomial $a(x)$. Recall that $a_r(x)$ is a polynomial of degree $r$. Using this fact, we observe that the highest degree monomial of $a(x)$ is
  $
    \ext_{\CG}(\rho_{0}) \ext_{\CJ}(\rho_{1})^{r^*} \cdot x^{r^*},
  $
  where $r^*$ has a different meaning depending on whether $\rho$ is finite or cofinite.

  Therefore, using \cref{equ:pOFxExpanded},
  the highest degree monomial of $p(x)$ is the same
  as the highest degree monomial of
  \begin{equation*}
    \ext_{\CG}(\rho_{0}) \ext_{\CJ}(\rho_{1})^{r^*}\cdot
    \bigl[
    \ext_{\CJ}(\rho_{0})^2(x+2)^{r^*}-
    (\alpha+\beta)\ext_{\CJ}(\rho_{0})(x+1)^{r^*}+
    \alpha\beta x^{r^*}
    \bigr],
  \end{equation*}
  which is
  \begin{equation*}
    \ext_{\CG}(\rho_{0}) \ext_{\CJ}(\rho_{1})^{r^*}\cdot
    \bigl[
    \ext_{\CJ}(\rho_{0})^2-
    (\alpha+\beta)\ext_{\CJ}(\rho_{0})+
    \alpha\beta
    \bigr]\cdot x^{r^*}.
  \end{equation*}

  So, the coefficient of $x^{r^*}$ in $p(x)$ is $c=\ext_{\CG}(\rho_{0}) \ext_{\CJ}(\rho_{1})^{r^*} \cdot c'$, where $c'= (\ext_{\CJ}(\rho_{0})-\alpha)(\ext_{\CJ}(\rho_{0})-\beta)$.
  Recall that $\CJ$ is a strong candidate, and therefore, we have
  \begin{itemize}
    \item $\ext_{\CJ}(\rho_{1})\ge 1$, and consequently
    \item $\ext_{\CJ}(\rho_{0})\neq \ext_{\CJ}(\rho_{0})+\ext_{\CJ}(\rho_{1})=\alpha$, and also
    \item $\ext_{\CJ}(\rho_{0})\neq
    \ext_{\CJ}(\sigma_{0})+\ext_{\CJ}(\sigma_{1})=\beta$
    by \cref{item:strongagg}.
  \end{itemize}
  This shows that $c'\neq 0$ and $\ext_{\CJ}(\rho_{1})^{r^*}\neq 0$.
  Consequently, by computing the coefficient $c$,
  the sought-after value $\ext_{\CG}(\rho_{0})$ can be computed as well.
  This finishes the proof of \cref{clm:pxcoefficients}.
\end{claimproof}

Finally, the coefficients of $p(x)$ can be computed by polynomial interpolation
by evaluating $p(x)$ for $r^*+1$ distinct values of $x$.
This can be done since, according to \cref{clm:recoverAprimex},
we can compute $A'_x$ using three \srCountDomSetRel[\CP]-oracle calls
as long as $x$ is sufficiently large in $\Oh(n)$;
and we can efficiently compute $\ext_{\CJ}(\rho_{0})^{x-r^*}$
as $\ext_{\CJ}(\rho_{0})$ does not depend on $n$.
So, computing the coefficients can be done in time polynomial in $x\in \Oh(n)$,
and from the coefficients, one can compute $\ext_\CG(\rho_0)$
according to \cref{clm:pxcoefficients}.

Summarizing, we have shown that we can solve
\srCountDomSetRel[\CP+(\emptyset, \{0\})] on instance $I$
using at most $3(r^*+1)$ (three per $A'_x$) oracle calls
to \srCountDomSetRel[\CP] on instances of the form $I_x$.
This finishes the proof of \cref{lem:forcesingle00}.
\qedhere
\end{description}
\end{proof}

\subsection{\boldmath The Case
\texorpdfstring{$\rho=\NN$}{Rho is Unrestricted} and \texorpdfstring{$\sigma$}{Sigma} is Finite}
\label{sec:rhoiseverything}

The goal of this section is to prove \cref{lem:rhoeverything},
which formally states that we can realize $\HWset{1}$ relations
if $\rho$ does not enforce any constraints, i.e., $\rho=\NN$,
but the set $\sigma$ is finite.
Recall that \cref{lem:IIIrelversion} shows a reduction
from \srCountDomSetRel[\HWsetGen{=1}] to \srCountDomSetRel[\HWsetGen{\ge1}].
Hence, it suffices to show the reduction
from the latter problem to \srCountDomSet.
We give this reduction in three steps,
\begin{itemize}
  \item
  starting with a reduction from \srCountDomSetRel[\HWsetGen{\ge1}]
  to \srCountDomSetRel[(\{0\},\NN)] which is given in \cref{lem:IV},
  \item
  then reducing to \srCountDomSetRel[(\sigma, \emptyset)]
  in \cref{lem:sigfin2},
  \item
  and finally the reduction to \srCountDomSet
  in \cref{lem:sigfin1}.
\end{itemize}

We start with the first step of this chain.

\begin{lemma}\label{lem:IV}
  Let $(\sigma,\rho)\in \allSetsne^2$ be non-trivial
  with $\rho=\NN$.
  Then, \srCountDomSetRel[\HWge{1}]~$\pwred$
  \srCountDomSetRel[(\{0\},\ZZ_{\ge0})].
\end{lemma}
\begin{proof}
  Let $I=(V,E,\CC)$ be an instance of \srCountDomSetRel[\HWge{1}].
  Let $\CZ=\{\scope(C) \mid C\in\CC\}$.
  $I'=(V,E)$ can be cast as an instance of \srCountDomSetRel[(\{0\},\ZZ_{\ge0})].
  Then the number of solutions of $I$
  is identical to the number of those solutions of $I'$
  that select at least one vertex from each $Z\in \CZ$.

  For a positive integer $x$,
  let $I'_x$ be the instance of \srCountDomSetRel[(\{0\},\ZZ_{\ge0})]
  obtained from $I'$ by attaching, to each set $Z\in \CZ$,
  a total of $x$ $(\{0\},\ZZ_{\ge0})$-vertices $v^Z_1,\ldots, v^Z_x$,
  each of which is completely connected to $Z$.

  Let $S$ be some selection of vertices from $V$. Note that, in order for $S$ to be extended to a solution of $I'_x$, it is required that all vertices in $V$ that are \emph{not} in some $Z\in \CZ$ already have a feasible  number of selected neighbors ($(\sigma,\rho)$-constraint) in $S$. We say that subsets $S$ with this property are \emph{good}.
  Suppose that in some good $S$, a set $Z\in \CZ$ is entirely unselected.
  Then, there are $2^x$ feasible extensions
  to the attached vertices $v^Z_1,\ldots, v^Z_x$
  (each of them can be selected or not).
  For this, we crucially exploit that $\rho=\NN$.
  On the contrary, if at least one vertex of $Z$ is selected,
  then, in a feasible extension of $S$ to a solution of $I'_x$,
  all of the vertices $v^Z_1,\ldots, v^Z_x$ have to be unselected
  because they are $(\{0\},\NN)$-vertices.

  Let $a_i$ be the number of good subsets of $V$ in which precisely $i$ of the sets in $\CZ$ are entirely unselected. Let $\#S(I'_x)$ denote the number of solutions of $I'_x$. Then,
  \begin{align*}
    \#S(I'_x)
    &=\sum_{i=0}^{\abs{\CZ}} a_i\cdot 2^{xi}.
  \end{align*}
  Let $n$ be the number of vertices in $I$.
  Then, for each $i$, we have $a_i\le 2^n$.
  Choosing $x=n+1$ gives us that $a_0< 2^x$.
  Hence, $a_0= (\#S(I'_x) \mod 2^x)$ can be computed
  by a single \srCountDomSetRel[(\{0\},\ZZ_{\ge0})]-oracle call.
  Note that $a_0$ is the number of good selections $S$
  such that no set in $\CZ$ is entirely unselected,
  that is, each set $Z\in\CZ$ contains at least one vertex from $S$.
  In this case, all of the attached vertices
  ($v^Z_1,\ldots, v^Z_x$ for each $Z\in \CZ$) have to be unselected,
  which implies that all vertices in $Z$
  obtain a feasible number of selected neighbors already from $S$.
  Therefore, the sets $S$ that contribute to $a_0$
  are precisely the $(\sigma,\rho)$-sets of $I$
  with the additional property that each set $Z$ in $\CZ$
  contains at least one selected vertex.
  This shows that $a_0$ is precisely the number of solutions of $I$.

  Observe that the size of the new instance might increase quadratically,
  as, for each \HWge{1}-relation (and there might be $\Oh(n)$ many),
  we add $\Oh(n)$ vertices to the instance.

  It remains to argue that the pathwidth does not change too much.
  For each clique $Z$, there is a bag $B$ containing $Z$.
  We duplicate this bag $B$ a total of $x\in \Oh(n)$ times,
  and each vertex $v^Z_i$ is added to exactly one of these copies.
  Hence, the pathwidth of $I'_x$ is at most that of $I$
  plus an additive constant.
\end{proof}

For the second step in this chain of reductions,
we require that $\sigma$ is finite.
We add a simple gadget to the vertices we want to remove
to transform the output of the previous lemma
into an instance of \srCountDomSetRel[(\sigma,\emptyset)],
that is, there are vertices which are forced to be selected.

\begin{lemma}\label{lem:sigfin2}
  \label{lem:sigfin3}
  Let $(\sigma,\rho)\in \allSetsne^2$ be non-trivial
  with $\rho= \ZZ_{\ge0}$ and finite $\sigma$.
  Then, \srCountDomSetRel[\{(\{0\},\NN)\}]~$\pwred$
  \srCountDomSetRel[\{(\sigma,\emptyset)\}].
\end{lemma}
\begin{proof}
  Observe that the case $\sigMax=0$ is trivial,
  as then $\sigma=\{0\}$, and consequently, $(\{0\},\NN)$-vertices are also just $(\sigma, \rho)$-vertices.
  For the other case when $\sigMax \ge 1$, we first construct
  a $\{\sigma_{\sigMax},\rho_{\sigMax}\}$-realizer $\CJ$
  for \srCountDomSetRel[\{(\sigma,\emptyset)\}].
  (Recall the definition of a parsimonious realizer from \cref{def:realization}.)
  \begin{claim}
    There is a parsimonious $\{\sigma_{\sigMax},\rho_{\sigMax}\}$-realizer $\CJ$
    for \srCountDomSetRel[\{(\sigma,\emptyset)\}] with a single portal $p$.
  \end{claim}
  \begin{claimproof}
    We distinguish two cases depending on whether $\sigMax-1 \in \sigma$ or not.
    \begin{description}
      \item[\boldmath $\sigMax-1 \notin \sigma$.]
      $\CJ$ has a single portal $p$, which is a $(\sigma,\rho)$-vertex.
      The portal $p$ is fully adjacent to an $\sigMax$-clique $C$
      of $(\sigma,\emptyset)$-vertices.
      There is an additional vertex $v$
      which is also adjacent to all vertices in $C$, but not to $p$.
      Note that there are precisely two partial solutions for this gadget.
      The vertices of $C$ are $(\sigma,\emptyset)$-vertices,
      and therefore, they are always selected.
      If $p$ is selected,
      then the vertices in $C$ cannot have another selected neighbor
      (as $\sigma$ is finite),
      and therefore, $v$ is unselected (which is fine as $\rho=\NN$).
      If $p$ is unselected,
      then the vertices in $C$ have $\sigMax-1$ selected neighbors within $C$.
      Since $\sigMax-1\notin\sigma$, the vertex $v$ must be selected
      (which is fine as it now also has $\sigMax$ selected neighbors).
      In either case, $p$ has $\sigMax$ selected neighbors within $\CJ$.

      \item[\boldmath $\sigMax-1 \in \sigma$.]
      It would suffice to construct a parsimonious
      $\{\sigma_1,\rho_1\}$-realizer $\CJ'$
      for \srCountDomSetRel[\{(\sigma,\emptyset)\}]
      with a single portal $u$.
      Indeed, to obtain $\CJ$, we can use $\sigMax$ copies of $\CJ'$
      and identify all copies of the portal $u$ with the portal $p$ of $\CJ$.

      The realizer $\CJ'$ consists of an $\sigMax$-clique $C$
      of $(\sigma,\emptyset)$-vertices, precisely one of which
      is adjacent to $u$.
      Note that all vertices in $C$ have degree $\sigMax-1$ in $C$.
      Hence, selecting all vertices of $C$ gives a partial solution
      independently of the selection status of $u$
      (because $\sigMax-1,\sigMax\in\sigma$ and $\rho=\NN$).
      In either case, $u$ obtains one selected neighbor from $\CJ'$.
      \qedhere
    \end{description}
  \end{claimproof}

  Using the parsimonious $\{\sigma_{\sigMax},\rho_{\sigMax}\}$-realizer $\CJ$
  from the claim, we can give the reduction.
  Let $I=(G,\lambda)$ be an instance of \srCountDomSetRel[(\{0\},\NN)],
  and let $U$ be the set of $(\{0\},\NN)$-vertices of $I$.
  From $I$, we define an instance $I'$
  of \srCountDomSetRel[\{(\sigma,\emptyset)\}]
  by turning the vertices in $U$ into $(\sigma, \rho)$-vertices,
  and attaching, to each vertex $u$ in $U$, a new copy of the gadget $\CJ$,
  where $u$ is identified with the portal $p$ of $\CJ$.
  This way, in each solution,
  $u$ obtains $\sigMax$ selected neighbors from its copy of $\CJ$.
  Hence, $u$ can be treated as a $(\sigma-\sigMax, \rho-\sigMax)$-vertex,
  which is precisely a $(\{0\},\NN)$-vertex.
  Thus, the number of solutions of instance $I$
  for \srCountDomSetRel[(\{0\},\NN)]
  is the same as the number of solutions of $I'$
  for \srCountDomSetRel[\{(\sigma,\emptyset)\}].

  It remains to bound the pathwidth of the new instance.
  For each vertex $u$ in $U$,
  there is a bag in the path decomposition of $G$ containing $u$.
  We create a copy of this bag for the path decomposition of $I'$,
  and add the vertices of the copy of $\CJ$ which are attached to $u$ to it.
  As we can create a separate copy for each vertex in $U$, the
  pathwidth increases by at most the size of $\CJ$,
  which depends only on $\sigma$ and $\rho$.
  Hence, the pathwidth increases by at most a constant.
\end{proof}

It remains to remove the $(\sigma,\emptyset)$-vertices
which are introduced by the previous reduction.

\begin{lemma}\label{lem:sigfin1}
  Let $(\sigma,\rho)\in \allSetsne^2$ be non-trivial with $\rho= \ZZ_{\ge0}$.
  If $\sigma$ is finite,
  then \srCountDomSetRel[\{(\sigma,\emptyset)\}]~$\pwred$ \srCountDomSet.
\end{lemma}
\begin{proof}
  We first construct and analyze a gadget which we use in the reduction.
  Let $\CJ^{(x)}=(J_x,\{p\})$ be a gadget for \srCountDomSet
  that consists of $p$ together with $x$ cliques, each of order $\sigMin+1$.
  The portal $p$ is adjacent to precisely one vertex from each clique.

  We analyze the number of extensions for this gadget
  assuming the portal $p$ already has $i$ neighbors from the outside.
  First, consider the case when $p$ is not selected.
  Then, each of the cliques in $\CJ^{(x)}$ can be only entirely selected
  or entirely unselected (as selecting only a few vertices
  would give them less than $\sigMin$ selected neighbors).
  As $p$ is unselected,
  each combination of entirely selected or unselected cliques is feasible
  since $\rho=\ZZ_{\ge 0}$.
  (This is not true if $p$ is selected.)
  This means that the gadget $\CJ^{(x)}$ has $2^x$ feasible extensions
  for the selection $S_G$.

  If the portal $p$ is selected, the situation is more complicated.
  Assuming that $p$ already has $i$ neighbors, there are $f_i(x)$ feasible extensions to $\CJ^{(x)}$,
  where we have
  \[
    f_i(x) \deff \sum_{i+s\in \sigma}\binom{x}{s}
    .
  \]
  Since $\sigma$ is finite, $f_i(x)$ is a polynomial of degree at most $\sigMax$.
  Also, observe that the constant of $f_i(x)$ is 1 if $i\in \sigma$,
  and otherwise it is 0.

  Now, let $I=(G,\lambda)$ be an instance of \srCountDomSetRel[(\sigma,\emptyset)],
  and let $U$ be the set of $(\sigma,\emptyset)$-vertices of $I$.
  For a positive integer $x$, we define an instance $I_x$ of \srCountDomSet
  by making the vertices in $U$ $(\sigma, \rho)$-vertices,
  and attaching, to each vertex $u\in U$, a copy of $\CJ^{(x)}$,
  where $u$ is identified with the portal $p$ of $\CJ^{(x)}$.

  We first argue why this modification does not change the pathwidth too much.
  Recall that each vertex $u$ in $U$ is adjacent
  to one vertex of $x$ $\sigMin+1$-cliques.
  Since each clique is of constant order,
  we can duplicate an arbitrary bag containing $u$ $x$ times,
  and add, to each copy of that bag, the vertices of precisely one clique.
  As a consequence, the pathwidth increases by at most a constant.

  Let $\#S(I_x)$ be the number of solutions of \srCountDomSet
  for the instance $I_x$.
  Let $p(x) \deff \#S(I_x) \mod 2^x$.
  In the following claims,
  we show that $p(x)$ is a polynomial whose degree depends on $\sigma$,
  and whose constant term is the number of solutions of $I$
  for which the vertices in $U$ are selected, which is precisely what we want.
  We recover the constant term by interpolation.

  \begin{claim}
    $p(x)$ is a polynomial in $x$ with maximum degree $\sigMax\cdot\abs{U}$.
  \end{claim}
  \begin{claimproof}
    Observe that $\#S(I_x)$ is a sum over all possible solutions for $I$
    multiplied by the number of corresponding extensions to the copies of $\CJ^{(x)}$.

    Consider a solution $S$ of $I_x$.
    If there is one vertex in $U \setminus S$,
    that is, a former $(\sigma, \emptyset)$-vertex that is not selected,
    then there are $2^x$ extensions to the attached copy of $\CJ^{(x)}$.
    Hence, this solution does not contribute to $p(x)$.

    However, if all vertices from $U$ are selected, that is $U \subseteq S$,
    then, for each vertex $u\in U$, and for the number $i$ of selected neighbors of $u$ in $S\cap V(G)$, there are $f_i(x)$ extensions to the copy of $\CJ^{(x)}$ that is attached to $u$.
    By the above argument, $f_i(x)$ is a polynomial in $x$
    with maximum degree $\sigMax$.
  \end{claimproof}

  By the arguments from the claim,
  it suffices to analyze $\#S'(I_x)$,
  which we define as the number of solutions of $I_x$
  for which all vertices in $U$ are selected.
  It follows directly that $p(x)=\#S'(I_x) \mod 2^x$.

  \begin{claim}
    The constant term of $p(x)$ is equal to the number of solutions of $I$,
    that is, $p(0) = \#S(I)$.
  \end{claim}
  \begin{claimproof}
    Recall that $f_i(x)$ is a polynomial of degree at most $\sigMax$,
    and that the constant of $f_i(x)$ is 1 if $i\in \sigma$,
    and otherwise it is 0.

    Hence, a selection $S_G \supseteq U$ of vertices from $V(G)$
    contributes 1 to the constant term of $p(x)$ if and only if,
    for each vertex in $U$, the corresponding number $i$ of selected neighbors is in $\sigma$, and
    otherwise it contributes 0 to the constant.
    So, it contributes 1 if and only if $S_G$ is a solution of $I$
    in which all vertices of $U$ are selected.

    We conclude that the constant of $p(x)$ is the sought-after number
    of solutions of \srCountDomSetRel[(\sigma,\emptyset)] on instance $I$.
  \end{claimproof}

  As the degree of $p(x)$ in $x$ is bounded by $\sigMax \cdot \abs{U}$,
  we can recover $p(x)$ by polynomial interpolation
  using at most $\sigMax\cdot |U|+1$ different sufficiently large values of $x$
  and corresponding oracle calls to \srCountDomSet on instance $I_x$.
  The constant term of $p(x)$ is then the sought-after number of solutions of $I$.
\end{proof}

Now, we have everything ready to prove the main result of this section
which we restate here for convenience.
\lemrhoeverything*
\begin{proof}
  From \cref{lem:IIIrelversion} together with \cref{lem:IV},
  it follows that
  \[
    \srCountDomSetRel[\HWsetGen{=1}]
    \pwred \srCountDomSetRel[\HWsetGen{\ge1}]
    \pwred \srCountDomSetRel[\{(\{0\},\NN)\}].
  \]
  Then, by \cref{lem:sigfin1,lem:sigfin2}, we have
  \[
    \srCountDomSetRel[\{(\{0\},\NN)\}]
    \pwred \srCountDomSetRel[(\sigma,\emptyset)]
    \pwred \srCountDomSet,
  \]
  which concludes the proof.
  The combined reduction is pathwidth-preserving
  by the transitivity of pathwidth-preserving reductions
  (recall \cref{obs:pwredtransitivity}).
\end{proof}

\subsection{\boldmath The Case
  \texorpdfstring{$\rho=\NN$}{Rho is Unrestricted} and \texorpdfstring{$\sigma$}{Sigma} is Cofinite}
\label{sec:rhoiseverything2}

As outlined at the beginning of \cref{sec:realizingrelationsforcounting}, the proof for the case when $\rho=\NN$ and $\sigma$ is cofinite deviates from the previous proofs.
Instead of continuing with the definition of \srCountDomSetRel,
we introduce a variant of this problem
which we denote by \srCountDomSetRel[\altRel].
Intuitively, this problem variant is different in the sense that all vertices that are in the scope of some relation no longer act as $(\sigma, \rho)$-vertices, but as $(\NN,\NN)$-vertices instead.
Before introducing the required definitions, let us outline the chain of reductions
that is used for the proof of \cref{lem:rhoeverything2}. Afterward, we formally introduce the different problem variants.
Consult \cref{fig:count:removingRelations} for a visualization of the procedure.
\begin{enumerate}
  \item
  In the first step, we reduce from \srCountDomSetRel
  to a relation-weighted version of
  \srCountDomSetRel[\{(\ZZ_{\ge1},\NN)\},\altRel].
  This is made formal in \cref{lem:RelfromweightedRelS}
  and exploits that we are dealing with the counting version of the problem.
  \item
  We remove the weights from the relations
  analogously to the results we have seen in \cref{sec:high-level:counting}.
  We state these reductions formally in
  \cref{lem:vertexWeightedRelS,lem:relWeightedRelS}.
  \item
  We show that it suffices to model $\alt{\HWge{1}}$ relations.
  For this, we first follow the ideas from \cref{sec:RelToHW1}
  to replace arbitrary relations
  by $\alt{\HWeq{1}}$ relations (in \cref{lem:RelSfromHWS}).
  In a second step, which is analogous to \cref{lem:IIIrelversion},
  we replace these relations by $\alt{\HWge{1}}$ relations (in \cref{lem:HWeqS}).
  \item
  We remove the $\alt{\HWge{1}}$ relations by using $(\ZZ_{\ge1},\NN)$-vertices together with a single $(\NN,\emptyset)$-vertex, see \cref{lem:HWgeS}.
  \item
  In \cref{lem:sigcof2}, we replace the $(\ZZ_{\ge1},\NN)$-vertices
  by a constant number of $(\NN,\emptyset)$-vertices.
  \item
  Finally, in \cref{lem:sigcof1}, we model $(\NN,\emptyset)$-vertices using only $(\sigma, \rho)$-vertices.
  Here, we crucially require that there are not too many of these vertices (i.e., just a constant number).
\end{enumerate}

\subparagraph*{Notation and Problem Definitions.}
Let $(\sigma,\rho)\in \allSetsne^2$ and let $G = (V,E,\CC)$ be a graph with relations as defined in \cref{def:lower:graphWithRelations}.
Recall that, in this section, we define a new concept of $(\sigma,\rho)$-sets with relations.
To this end, let $U$ be the set of vertices that are in the scope of at least one relation.

Then, a \emph{$\alt{(\sigma,\rho)}$-set} of $G$
is a selection $X \subseteq V$ such that:
  \begin{itemize}
    \item For each $v\in X\setminus U$, we have $\abs{N(v)\cap X}\in \sigma$.
    \item For each $v\in V\setminus (X\cup U)$, we have $\abs{N(v)\cap X}\in \rho$.
    \item For each $C\in CC$, it holds that $X\cap \scope(C) \in \rel(C)$.
  \end{itemize}
Intuitively, the vertices in $V\setminus U$ are $(\sigma, \rho)$-vertices, but the vertices in $U$ act as $(\NN,\NN)$-vertices, and are not required to fulfill any $\sigma$-constraints or $\rho$-constraints. They are saturated in the sense that any number of selected neighbors is feasible for them to have.

As in the previous sections,
we work with a generalization of $\alt{(\sigma,\rho)}$-sets
to the setting where we can use a set of additional pairs
$\CP=\{(\sigma^{(i)}, \rho^{(i)}) \in \allSets^2\mid i\in \nset{k}\}$
(for some non-negative integer $k$). We set $(\sigma^{(0)}, \rho^{(0)})\coloneqq (\sigma, \rho)$.
Then, given a mapping $\lambda \from V \to \fragment{0}{k}$,
a \emph{$\alt{\lambda}$-set} of $G$ is a subset $X\subseteq V$
with the following properties:
\begin{itemize}
  \item For each $v\in X\setminus U$, we have $\abs{N(v)\cap X}\in \sigma^{(\lambda(v))}$.
  \item For each $v\in V\setminus (X\cup U)$, we have $\abs{N(v)\cap X}\in \rho^{(\lambda(v))}$.
  \item For each $C\in CC$, it holds that $X\cap \scope(C) \in \rel(C)$.
\end{itemize}

Then, \srCountDomSetRel[\CP,\alt{\Rel}] is the problem
that takes as input some graph with relations $G$
and a function $\lambda \from V(G) \to \fragment{0}{k}$,
and asks for the number of $\alt{\lambda}$-sets of $G$.
As usual, if $\CP$ is empty, we drop the superscript $\CP$,
and we also drop $\lambda$ as part of the input
to obtain the problem \srCountDomSetRel[{\altRel}]
that takes as input a graph with relations $G$
and asks for the number of $\alt{(\sigma,\rho)}$-sets of $G$.

\subparagraph*{Weighted Versions.}
We need to generalize our problem even further.
Analogously to \cref{sec:high-level:counting},
for a non-negative integer $q$,
we also introduce the problem \emph{$q$-relation-weighted}
\srCountDomSetRel[{\CP,\altRel}]
that takes as input a graph with weighted relations $G = (V,E,\CC)$
such that there are at most $q$ different weights, i.e., $\abs{\{ f_C(U) \mid C\in \CC, U \subseteq \scope(C) \} } \le q$.
Each solution $X \subseteq V$ of \srCountDomSetRel[\CP,\altRel]
on input $G$ obtains a weight $f(X)\deff \prod_{C\in \CC} f_C(X\cap \scope(C))$.
The task is then to compute the sum of the values $f(X)$
summed over all $\alt{\lambda}$-sets $X$ of $G$.
Similarly, in the \emph{$q$-vertex-weighted} \srCountDomSetRel[\CP,\altRel],
the input is a graph with relations (however, no relation weights)
together with a weight function $\wt\from V \to \SetQ$
such that $\abs{\{ \wt(v) \mid v \in V \} } \le q$,
for some non-negative integer $q$,
and the task is to compute the sum of the values
$g(X)\deff \prod_{v\in X} \wt(v)$ summed over all $\alt{\lambda}$-sets $X$ of $G$.

We are ready to state the first result of this section.

\begin{lemma}\label{lem:RelfromweightedRelS}
  Let $(\sigma,\rho)\in \allSetsne^2$ be non-trivial
  with $\rho= \ZZ_{\ge0}$ and cofinite $\sigma$.
  Then, there is a positive integer $q$ (depending only on $(\sigma,\rho)$)
  such that there is a \pwar-reduction from \srCountDomSetRel
  to $q$-relation-weighted \srCountDomSetRel[\{(\ZZ_{\ge1},\NN)\},\altRel].
\end{lemma}
\begin{proof}
  Let $G = (V,E,\CC)$ be an input to \srCountDomSetRel.
  Let $\CZ=\{\scope(C) \mid C\in \CC\}$,
  and let $U$ be the union of the sets in $\CZ$.
  We could cast $G$ as an instance of \srCountDomSetRel[{\altRel}].
  However, in solutions of \srCountDomSetRel[{\altRel}] on input $G$,
  all vertices in $U$ behave like $(\NN,\NN)$-vertices,
  whereas in solutions of \srCountDomSetRel on input $G$,
  they behave like $(\sigma, \rho)$-vertices.

  \subparagraph*{Idea.}
  As a first step, consider the modified instance $G'$
  which is obtained from $G$ by introducing,
  for each $Z\in \CZ$ and each $z\in Z$, a vertex $z_Z$.
  Let $U'=\{z_Z \mid Z\in \CZ, z\in Z\}$ and let $Z'=\{z_Z \mid z\in Z\}$.
  We know that, for each $Z\in \CZ$,
  there is a relational constraint in $G$ with scope $Z$.
  In $G'$, the scope of this constraint is $Z'$ rather than $Z$
  and the corresponding relation uses $z_Z\in Z'$
  whenever it previously used $z\in Z$.
  The vertices of $G'$ are $V\cup U'$.

  Observe that the solutions of \srCountDomSetRel on instance $G$
  are in a one-to-one correspondence with those solutions of \srCountDomSetRel[{\altRel}] on $G'$
  in which, for each $Z\in \CZ$,
  a vertex $z\in Z$ is selected if and only if $z_Z\in Z'$ is selected.
  Intuitively, this means that the vertices of $Z$
  behave as $(\sigma, \rho)$-vertices,
  but also mirror the selection status of the vertices in $Z'$,
  which means that they satisfy the corresponding relation.
  We now introduce a gadget $\CJ$, which we use to ensure this property.

  \begin{figure}[t]
    \centering
    \includegraphics[scale=1.9]{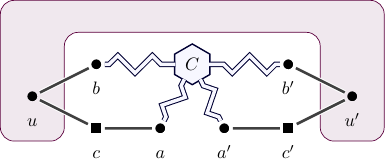}
    \caption{
      The construction of the gadget $\CJ$ with portals $\port$ and $\port'$
      from \cref{lem:RelfromweightedRelS}.
      The vertices $c$ and $c'$ are $(\ZZ_{\ge1},\NN)$-vertices
      and all other vertices are $(\sigma,\rho)$-vertices.
      The relation $C$ is such that $\rel(C) = \{a,b,a',b',aa',ab',ba',bb'\}$.
      \\
      Later, for each $Z \in \CZ$ and $z \in Z$,
      we create a copy of the gadget $\CJ$
      and identify $\port$ with $z$ and $\port'$ with $z_Z$.
      }\label{fig:eq-fake}
  \end{figure}

  \subparagraph*{\boldmath Definition of Auxiliary Gadget $\CJ$.}
  \newcommand{\A}{\allStates}
  The gadget $\CJ=(J, \{\port,\port'\}, \lambda_{\CJ})$ consists
  of a graph $J$ with (weighted) relations and two portals $\port$ and $\port'$.
  The vertex set of $J$ contains $\port$ and $\port'$, and additionally
  six new vertices $a,b,c,a',b',c'$ that are connected as follows:
  $\port$ is adjacent to $c$ and $b$,
  and $c$ is adjacent to $a$.
  Analogously, $\port'$ is adjacent to $c'$ and $b'$,
  and $c'$ is adjacent to $a'$.
  All of these vertices are $(\sigma, \rho)$-vertices, with the exception of $c$ and $c'$, which are $(\ZZ_{\ge1},\NN)$-vertices.
  Additionally, $J$ contains a new weighted relation $C$ with scope
  $\{a,b,a',b'\}$.
  It is more convenient
  to present this weighted relation as a relation $R\subseteq 2^{\{a,b,a',b'\}}$
  together with a weight function $f_C\from 2^{\{a,b,a',b'\}} \to \SetQ$
  that takes non-zero values only on elements from $R$.
  By this notation,
  the corresponding relation $R$ is the set of all subsets of $\{a,b,a',b'\}$
  that contain at most one of $a$ or $b$, and at most one of $a'$ or $b'$.
  We show how to choose the weight function $f_C$ in \cref{clm:tensor}.
  See \cref{fig:eq-fake} for an illustration of the gadget.

  We distinguish three different ``states'' of the vertex $\port$
  with respect to $\CJ$.
  \begin{itemize}
    \item
    State $\sigma_0$ means that $\port$ is selected
    but has no selected neighbor in $J$.
    \item
    State $\sigma_{\ge1}$ means that $\port$ is selected
    and has (at least) one selected neighbor in $J$.
    \item
    State $\rho$ means that $\port$ is unselected
    (with any number of selected neighbors in $J$).
  \end{itemize}
  We distinguish the same three states for the vertex $\port'$.
  Let $\A=\{\sigma_0,\sigma_{\ge1}, \rho\}$.
  Then, given $\tau, \tau' \in \A$,
  we say that an extension of $(\tau, \tau')$ to $\CJ$
  is a feasible selection of vertices from $J$
  such that $\port$ is in state $\tau$ and $\port'$ is in state $\tau'$.
  More precisely, this is a set $S\subseteq \{\port,\port', a, \dots, c\}$
  such that the following hold:
  \begin{itemize}
    \item
    The vertices $\port$ and $\port'$ have states $\tau$ and $\tau'$
    with regard to $S$, respectively.
    \item
    If $c \in S$, then it has at least one neighbor in $S$
    (it is a $(\ZZ_{\ge1},\NN)$-vertex).
    The same holds for $c'$.
    \item
    The set $S$ contains at most one of $a$ and $b$,
    and at most one of $a'$ and $b'$,
    that is, the set $S$ respects the relation $R$.
  \end{itemize}
  Let $\CS(\tau, \tau')$ be the set of extensions of $(\tau, \tau')$ to $\CJ$.
  Then, $\ext_{\CJ}(\tau, \tau') \deff \sum_{S\in \CS(\tau, \tau')}
  f_C(S\cap N(v))$ is the sum of extensions of $(\tau, \tau')$ to $\CJ$
  weighted according to $f_C$.
  We finish the construction of the gadget $\CJ$
  by defining the weights for the relation $R$, that is, defining $f_C$.

  \begin{claim}\label{clm:tensor}
    The weights $f_C \from \scope(C) \to \ZZ$ can be chosen
    such that $\ext_{\CJ}(\tau, \tau')=1$
    if $(\tau, \tau')\in \{(\sigma_0,\sigma_0), (\rho, \rho)\}$,
    and otherwise $\ext_{\CJ}(\tau, \tau')=0$.
  \end{claim}
  \begin{claimproof}
    For some feasible selection $S\subseteq \{a,b,c,a',b',c'\}$
    of vertices of $J$, by the definition of $R$,
    $S$ cannot contain both $a$ and $b$.
    More formally, $S\cap \{a,b\}$ can be in three possible ``states''
    (it can be $\emptyset$, $\{a\}$, or $\{b\}$).
    An analogous statement holds for the set $\{a',b'\}$.
    Let $\B=\{\emptyset, \{a\}, \{b\}\}$
    and $\B'=\{\emptyset, \{a'\}, \{b'\}\}$.
    For $\tau \in \A$ and $\phi\in \B$,
    let $m_{\tau,\phi}$ be the number of possible selection statuses of $c$
    that can be extended to a feasible selection $S$ of $J$
    in which $\port$ is in state $\tau$ and $S\cap\{a,b\}=\phi$.
    For $\tau' \in \A$ and $\phi'\in \B'$,
    we define $m_{\tau',\phi'}$ analogously.

    By checking for which selection status
    the constraint of the $(\ZZ_{\ge1},\NN)$-vertex $c$ can be satisfied,
    we obtain
    \[
    M\deff(m_{\tau,\phi})_{\tau\in \A, \phi\in \B}=
    \begin{tikzpicture}[scale=.7,baseline=(a.base)]
      \node[anchor=base,scale=.8] at (1,-.1) {\(\emptyset\)};
      \node[anchor=base,scale=.8] at (2,-.1) {\(\{a\}\)};
      \node[anchor=base,scale=.8] at (3,-.1) {\(\{b\}\)};

      \node[anchor=base,scale=.8] at (4.3,-1) {\(\sigma_0\)};
      \node[anchor=base] at (1,-1) {\(1\)};
      \node[anchor=base] at (2,-1) {\(1\)};
      \node[anchor=base] at (3,-1) {\(0\)};

      \node[anchor=base,scale=.8] at (4.3,-2) {\(\sigma_{\ge1}\)};
      \node[anchor=base] at (.5,-2)
      {\(\left(\vphantom{\rule{0pt}{6ex}}\right.\)};
      \node[anchor=base] at (1,-2) {\(1\)};
      \node[anchor=base] (a) at (2,-2) {\(1\)};
      \node[anchor=base] at (3,-2) {\(2\)};
      \node[anchor=base] at (3.5,-2)
      {\(\left.\vphantom{\rule{0pt}{6ex}}\right)\)};

      \node[anchor=base,scale=.8] at (4.3,-3) {\(\rho\)};
      \node[anchor=base] at (1,-3) {\(1\)};
      \node[anchor=base] at (2,-3) {\(2\)};
      \node[anchor=base] at (3,-3) {\(1\)};
    \end{tikzpicture}
    \]
    By the structure of $J$, it follows that
    $M'=(m_{\tau',\phi'})_{\tau'\in \A, \phi'\in \B}=M$
    if we use the same order of rows and columns for $M'$.

    For $\tau, \tau'\in \A$, $\phi\in B$, and $\phi'\in \B'$,
    let $m_{\tau, \tau',\phi, \phi'}$ be the number of partial solutions
    $S\in \CS(\tau, \tau')$ with $S\cap \{a,a',b,b'\}=\phi\cup \phi'$.
    Then, we have
    \[
      \ext_{\CJ}(\tau, \tau') =
        \sum_{\phi \in \B, \phi'\in \B'}
          m_{\tau, \tau',\phi, \phi'} \cdot f_C(\phi\cup \phi').
    \]
    Moreover, it holds that
    $m_{\tau, \tau',\phi, \phi'}= m_{\tau,\phi} \cdot m_{\tau',\phi'}$.
    Using this fact,
    for some corresponding order of tuples $(\tau, \tau')\in \A^2$
    and $(\phi,\phi')\in \B\times \B'$
    (lexicographic order w.r.t.~the orders $\sigma_0<\sigma_{\ge1}<\rho$
    and $\emptyset <\{a\},\{a'\} < \{b\},\{b'\}$),
    we have
    \[
      \left(\ext_{\CJ}(\tau, \tau')\right)_{(\tau, \tau')\in \A^2} = (M\otimes M) \cdot
      \left(f_C(\phi\cup \phi')\right)_{(\phi,\phi')\in \B\times \B'},
    \]
    where $\otimes$ denotes the Kronecker product of matrices.
    Now, observe that $M$ is invertible, and consequently, $M\otimes M$ is invertible.
    So, we can set
    \begin{equation}\label{eq:tensor}
      \left(f_C(\phi\cup \phi')\right)_{(\phi,\phi')\in \B\times \B'} =(M\otimes M)^{-1} \mathbf{e}
    \end{equation}
    where we can choose $\mathbf{e}$ as the vector of our choice,
    that is, as $\left(\ext_{\CJ}(\tau, \tau')\right)_{(\tau, \tau')\in \A^2}$.
    Note that $f_C$ can be efficiently computed.
  \end{claimproof}

  \subparagraph*{Constructing the Target Instance.}
  As mentioned above, we now use the gadget $\CJ$
  to ensure that all pairs $z$ and $z_Z$ have the same selection status in $G'$.
  Given $G'$, we define an instance $I^*=(G^*,\lambda)$ of \emph{$q$-relation-weighted}
  \srCountDomSetRel[\{(\ZZ_{\ge1},\NN)\},\altRel],
  where $G^*$ is a graph with weighted relations,
  and $\lambda$ is the mapping that determines
  which vertices are $(\sigma, \rho)$-vertices,
  and which are $(\ZZ_{\ge1},\NN)$-vertices.
  We choose $q$ later. 
  For each pair $(Z,z)$ in $G'$ with $Z\in \CZ$ and $z\in Z$,
  we create a copy $J_{(Z,z)}$ of the gadget $\CJ=(J, \{\port,\port'\})$
  from above 
  and identify $\port$ with $z$ and $\port'$ with $z_Z$.

  Then, we obtain the graph with relations $G^*$ from $I^*$ as the union of $G'$
  and $\{J_{(Z,z)} \mid Z\in \CZ, z\in Z\}$.
  We choose $\lambda$ such that it coincides with $\lambda_{\CJ}$ on each copy of $\CJ$, and such that all vertices of $G'$
  (including $U$) are $(\sigma,\rho)$-vertices (that is, only the copies of $c$ and $c'$ in $\CJ$ are $(\ZZ_{\ge 1}, \NN)$-vertices).
  The relations transfer directly and are not changed.

  \subparagraph*{Correctness.}
  A \emph{configuration} is a function $\gamma$ that maps each pair $(Z,z)$
  with $Z\in \CZ$ and $z\in Z$ to a pair of states
  $(\gamma_1(Z, z), \gamma_2(Z, z)) \in \A^2$.
  Let $\Gamma$ be the set of configurations.
  We say that a solution of ($q$-relation-weighted)
  \srCountDomSetRel[\{(\ZZ_{\ge1},\NN)\},\altRel] on input $I^*$
  is a \emph{$\gamma$-selection} if, for each
  pair $(Z,z)$, the vertex $z$ is in state $\gamma_1(Z,z)$
  (with respect to $J_{(Z,z)}$),
  and $z_Z$ is in state $\gamma_2(Z,z)$.

  Recall that each solution $S$ of $I^*$ has a weight $f(S)=\prod_{C\in \CC} f_C(S\cap \scope(C))$, where the product is over all constraints of $I^*$.
  For each $\gamma \in \Gamma$,
  let $\#S_\gamma(I^*)$ be the sum of weights $f(S)$
  over all $\gamma$-selections of $I^*$.
  Then, $\sum_{\gamma\in \Gamma} \#S_\gamma(I^*)$ is the output of
  \srCountDomSetRel[\{(\ZZ_{\ge1},\NN)\},\altRel] on input $I^*$.

  A selection of vertices from $G'$
  is \emph{compatible} with a configuration $\gamma$
  if it can be extended to a $\gamma$-selection of $I^*$.
  Let $\#S_\gamma(G')$ be the number of selections of $G'$
  that are compatible with $\gamma$.
  (Here, we can use the number of solutions
  as the weights of all relations in $G'$ are $1$.)
  \begin{claim}
  	We have
  	\[
  		\#S_\gamma(I^*) = \#S_\gamma(G') \cdot \prod_{Z\in \CZ, z\in Z}
  	\ext_{\CJ}(\gamma_1(Z,z),\gamma_2(Z,z)).
  	\]
  \end{claim}
  \begin{claimproof}
  	We partition the set of $\gamma$-selections
    according to their restriction to $V(G')$.
    Hence, the set of restrictions that appear in the partition
    are precisely those sets of vertices from $G'$
    that are compatible with $\gamma$.
    Now, let $S'$ be some fixed set of vertices from $G'$
    that is compatible with $\gamma$.
    We analyze in how many ways $S'$ can be extended
    to a $\gamma$-selection of $I^*$.
    For each $Z \in \CZ$ and for each $z\in Z$,
  	we can choose any extension to $J_{(Z,z)}$
    in which $z$ has state $\gamma_1(Z,z)$ and $z_Z$ has state $\gamma_2(Z,z)$.
    Note that these extensions are independent
    of the selected vertices in other copies of $J$.
    Here, we use the fact that we defined extensions to $\CJ$
    in a way that precisely captures the constraints
    that are imposed by a $\alt{\lambda}$-set.
    In particular, we do \emph{not} require
    the copies of $a$, $a'$, $b$, and $b'$ to act as $(\sigma,\rho)$-vertices
    since they are in the scope of the relation $R_v$,
    and therefore, act as $(\NN,\NN)$-vertices.
    So, if we take the weights into account,
    the set $S'$ contributes $\prod_{Z\in \CZ, z\in Z}
  	\ext_{\CJ}(\gamma_1(Z,z),\gamma_2(Z,z))$
    (the value of this term is actually independent of $S'$)
    to the weighted sum $\#S_\gamma(I^*)$.
  \end{claimproof}

  It follows from \cref{clm:tensor}
  that $\#S_\gamma(I^*)$ can be non-zero only
  if, for each $(Z,z)$, the pair $(\gamma_1(Z,z), \gamma_2(Z,z))$
  is one of $(\sigma_0,\sigma_0)$ or $(\rho, \rho)$.
  Moreover, in this case, we have $\#S_\gamma(I^*)=\#S_\gamma(G')$.
  We say that configurations with this property are \emph{good},
  and denote the set of good configurations by $\Gamma_{\!\textsf{good}}$.
  Consequently, the output on input $I^*$ is
  \[
    \sum_{\gamma\in \Gamma} \#S_\gamma(I^*)
    = \sum_{\gamma\in \Gamma_{\!\textsf{good}} } \#S_\gamma(G')
    .
  \]
  Note that the solutions that contribute to
  $\sum_{\gamma\in \Gamma_{\!\textsf{good}} } \#S_\gamma(G')$
  are in a one-to-one correspondence
  to the solutions of the original problem \srCountDomSetRel on input $G$.
  To see this, note that, for both allowed pairs of states $(\sigma_0,\sigma_0)$ or $(\rho, \rho)$, the corresponding vertices $z$ and $z_Z$ have the same selection status.
  Moreover, if $z$ is selected, then it does not obtain additional selected neighbors from $J_{(Z,z)}$, and therefore it has some $s\in \sigma$ selected neighbors in $V$.
  If $z$ is unselected, then it may obtain additional selected neighbors from $J_{(Z,z)}$, but this does not matter as $\rho=\NN$, and therefore, $z$ still has $r\in \rho$ selected neighbors in $V$.
  So, the solutions that contribute to
  $\sum_{\gamma\in \Gamma_{\!\textsf{good}} } \#S_\gamma(G')$
  are precisely those $(\sigma,\rho)$-sets of the original graph $G$
  in which each vertex $z\in Z$ is selected
  if and only if $z_Z\in Z'$ is selected.
  We have observed earlier that these are precisely
  the solutions of \srCountDomSetRel on input $G$.

  \subparagraph*{Analyzing the Reduction.}
  Finally, it remains to show that the reduction is a \pwar-reduction.
  First, note that all of the weights are either $0$, $1$,
  or solutions of \cref{eq:tensor}, in which $M$ is a fixed matrix
  and $\mathbf{e}$ is a vector in $\{0,1\}^9$.
  Hence there is some constant upper bound $q$
  on the number of different weights.
  Second, the reduction is arity-preserving as all relations in $I^*$
  are the copy of some relation from $G$,
  or otherwise they are contained in $\CJ$,
  in which case they have arity $4$.
  Third, the reduction is pathwidth-preserving.
  To see this, recall that in the path decomposition associated with $G$,
  for each set $Z\in \CZ$, there is some bag
  that contains all vertices of $\CZ$.
  For each $z\in Z$, we can add a copy of this bag
  that also contains the (constant number of) vertices
  of the copy $J_{Z,z}$ of $\CJ$.
\end{proof}

In the next two steps, we show how to model relation weights. Instead of carrying through
the $(\ZZ_{\ge1},\NN)$-vertices used in \cref{lem:RelfromweightedRelS}, we state the result for arbitrary sets of pairs $\CP$.
We return to work with $(\ZZ_{\ge1},\NN)$-vertices in \cref{lem:HWgeS}.
For now, we first replace the weighted relations by weighted vertices,
and then completely remove them from the instance.
The proofs are analogous to the results presented in
\cref{sec:high-level:counting}.
To make this fact easily verifiable,
we give the slightly modified constructions.

\begin{lemma}
  \label{lem:vertexWeightedRelS}
  Let $(\sigma,\rho)\in \allSetsne^2$ be non-trivial
  with $\rho= \NN$ and cofinite $\sigma$.
  For all $\CP \subseteq \allSets^2$
  and for all constants $q$,
  there is a \pwar-reduction from
  $q$-relation-weighted \srCountDomSetRel[\CP,\altRel]
  to $q$-vertex-weighted \srCountDomSetRel[\CP, \altRel].
\end{lemma}
\begin{proof}
  We follow the ideas behind the proof of
  \cref{lem:count:relationWeightedToVertexWeighted}.
  The only difference to that construction
  is that here we do not need the $\{\sigma_s,\rho_r\}$-providers that were used in \cref{lem:count:relationWeightedToVertexWeighted}.
  The purpose of these providers was to ensure that the newly introduced vertices
  that are not adjacent to other vertices,
  that is, $v_1,\dots,v_q$ in the following,
  get a feasible number of neighbors depending on their selection status.
  Omitting the aforementioned providers is not an issue
  as these vertices are in the scope of the relation $R$
  and, by the definition of \srCountDomSetRel[\CP, \altRel],
  such vertices can now be treated like $(\NN,\NN)$-vertices.
  Hence, the previous proof goes through without any additional changes.

  We give the modified construction.
  For a fixed $q$,
  let $w_1,\dots,w_{q'}$ be the $q' \le q$ different weights used
  by the relations of a $q$-relation-weighted
  \srCountDomSetRel[\CP,\altRel] instance $G$.
  Formally, we have
  $\{ w_1, \dots, w_{q'} \} = \{ f_C(U) \mid C \in \CC, U \subseteq \scope(C) \}$.
  For ease of notation,
  we assume without loss of generality that $q'=q$ in the following.
  We apply the following modification for all constraints of $G$ to form an instance $G'$ of $q$-vertex-weighted \srCountDomSetRel[\CP, \altRel].

  For each $C\in \CC$, we introduce $q$ new vertices $v_1,\dots,v_q$
  which form an independent set.
  We replace the weighted relational constraint $C$ by an unweighted relational constraint $C'$ with $\scope(C')\coloneqq\scope(C)\cup \{v_1, \ldots, v_q\}$. The new constraint $C'$ behaves (i.e., accepts or rejects)
  in exactly the same way as $C$ on the scope of $C$,
  but whenever $C$ accepts with weight $w_i$,
  $C'$ also forces the vertex $v_i$ to be selected whereas all the vertices $v_j$ with $i\neq j$ are forced to be unselected.
  The weight function $\wt$ assigns weight 1
  to all vertices of the original instance,
  and for each $i$ it assigns $\wt(v_i)=w_i$.

  The reduction increases pathwidth by at most $q$,
  as for each relation, there is a bag containing its scope,
  and we can copy this bag and add the additional vertices $v_1, \ldots, v_q$ to this bag.
  Moreover, the arity of the instance increases by at most $q$.
\end{proof}

As a next step, we remove the vertex weights from the vertices
to obtain an unweighted instance of the problem.
We follow the proof of \cref{lem:count:vertexWeightedToUnweighted} for this.

\begin{lemma}
  \label{lem:relWeightedRelS}
  Let $(\sigma,\rho)\in \allSetsne^2$ be non-trivial
  with $\rho= \NN$ and cofinite $\sigma$.
  % \ppp{Do we need the latter assumption?}
  For all $\CP \subseteq \allSets^2$
  and for all constants $q$,
  there is a \pwar-reduction from
  $q$-vertex-weighted \srCountDomSetRel[\CP, \altRel]
  to unweighted \srCountDomSetRel[\CP,\altRel].
\end{lemma}
\begin{proof}
  As for the previous proof, the only difference is
  that the $\{\sigma_s,\rho_r\}$-providers are omitted from the construction.
  As mentioned before, this does not incur any issues
  in this setting as the portals of these providers
  are always in the scope of a relation.
  As a consequence, they behave as $(\NN, \NN)$-vertices.

  We use known interpolation techniques
  to remove the vertex weights.
  Let $G$ be the graph with relations,
  and let $\wt \from V(G) \to \SetQ$ be the vertex weight function
  of a given $q$-vertex-weighted \srCountDomSetRel[\CP, \altRel] instance $I$.
  Let $w_1, \dots, w_q$ be the distinct values of the vertex weights,
  that is, $\{w_1, \dots, w_q\} = \{ \wt(v) \mid v \in V(G) \}$.

  We replace each weight $w_i$ by a variable $x_i$
  and treat the sought-after weighted sum of solutions of $I$ as a polynomial $P$ in the $q$ variables $x_1,\dots,x_q$.
  Observe that there are at most $n$ vertices,
  and hence, the total degree is at most $n$ for each variable.
  Hence, if we can realize all combinations of $n+1$ different weights for each variable $x_i$,
  then we can use \cref{fct:count:recoveringPolynomial}
  to recover the coefficients of $P$ in time $\Oh((n+1)^{3q})$,
  which is polynomial in $n$ as $q$ is a constant.
  Then, we can output $P(w_1, \dots, w_q)$ to recover the original solution.

  It remains to realize $n+1$ different weights for each variable.
  For this, it suffices
  to realize weights of the form $2^\ell$.

  Let $v$ be the vertex for which we want to realize weight $\wt(v)=2^\ell$.
  For this, we introduce $\ell$ new vertices $v_1,\dots,v_\ell$.
  Moreover, we add a relational constraint $C_j$ with scope $\{v, v_j\}$ that ensures that if $v$ is unselected,
  then $v_j$ must also be unselected, but if $v$ is selected, then $v_j$ can be either selected or unselected.
  Whenever $v$ is selected,
  this construction contributes a factor of $2^\ell$ to the solution,
  whereas if $v$ is unselected, then it contributes only a factor of $1$.

  Observe that the arity of the resulting graph increases by at most $2$
  as $C$ has arity $2$ and no other relations are introduced.

  It is straightforward to see
  that this modification does not change the pathwidth too much.
  Indeed, for each $v$, we can pick an arbitrary bag containing $v$,
  duplicate this bag for each $j \in \numb{\ell}$,
  and then add $v_j$ to this copy of the bag.
  The claim then directly follows.
\end{proof}

In the following, we show how to replace arbitrary relations
by gadgets involving only $\alt{\HWsetGen{=1}}$ relations. Technically, such an instance of \srCountDomSetRel[{\altRel}] contains only $\HWsetGen{=1}$ relations.
However, we still refer to them as $\alt{\HWsetGen{=1}}$ relations (and also $\alt{\EQset}$ relations) to highlight that these relations are interpreted differently in the context of \srCountDomSetRel[{\altRel}].
Correspondingly, it is convenient for the following statements
to use \srCountDomSetRel[\alt{\HWsetGen{=1}}]
to denote the problem \srCountDomSetRel[{\altRel}] restricted to instances
in which all of the relations are \HWeq[\ell]{1} relations
of any arity $\ell$.
Analogously, we define \srCountDomSetRel[\alt{\HWsetGen{\ge 1}}].
The proofs are very similar to the results from \cref{sec:RelToHW1}.

\begin{lemma} \label{lem:RelSfromHWS}
  Let $(\sigma,\rho)\in \allSetsne^2$ be non-trivial with $\rho= \ZZ_{\ge0}$ and cofinite $\sigma$.
  Let $\CP$ be some set of pairs from $\allSets^2$.
  For all constants $d$, there is a pathwidth-preserving reduction
  from \srCountDomSetRel[\CP,\altRel] with arity at most $d$
  to \srCountDomSetRel[\CP,\alt{\HWsetGen{=1}}].
\end{lemma}
\begin{proof}
  The proof consists of two parts.
  First, we replace arbitrary relations by gadgets
  which use only $\alt{\EQset}$ and $\alt{\HWsetGen{=1}}$ as relations.
  In the second part, we then also remove the $\alt{\EQset}$ relations,
  and replace them by gadgets using only $\alt{\HWsetGen{=1}}$ relations.
  The first part is analogous to the proofs of
  \cref{lem:realizing:arbitraryRelations,%
  lem:realizing:parsiArbitraryToHWoneAndEQ},
  and the second part is based on the proofs of
  \cref{lem:realizing:eqConditional,lem:realizing:parsiHWoneAndEQToHWone}.

  \begin{description}
    \item[Removing arbitrary relations.]
    As for the previous two lemmas,
    the only difference to the earlier results from \cref{sec:RelToHW1}
    is the fact that the $\{\sigma_s,\rho_r\}$-providers
    are omitted from the construction.
    For completeness, we give the construction of the gadget
    which replaces each arbitrary relation in the following.

    Given a relation $R=\{r_1,\dots,r_{\abs{R}}\}$,
    where $r_i \subseteq \numb{d'}$ for $d' \le d$, we define a graph with portals $H_R=(G,\Port)$
    with $\Port=\{\port_1,\dots,\port_{d'}\}$.
    For each set \( r_i \),
    we add \( \abs{\numb{d'} \setminus r_i} \)
    independent vertices \( s^{(i,j)} \) with \(j \in \numb{d'} \setminus r_i\),
    and we also add an extra vertex \( t^{(i)} \) to \(G\).
    Next, we introduce an $\alt{\EQset}$ relation of arity $\abs{\numb{d} \setminus r_i} + 1$
    that observes the vertices \(s^{(i,\star)}\) and \(t^{(i)}\).
    Additionally, we add an $\abs{R}$-ary $\alt{\HWsetGen{=1}}$ relation that observes the vertices
    \(t^{(\star)}\).
    Finally, for each \(j \in \numb{d}\),
    we introduce a $\alt{\HWsetGen{=1}}$ relation that observes \(\port_j\)
    and all vertices \(s^{(\star,j)}\).

    For a given \srCountDomSetRel[\CP,\altRel] instance $H$,
    we replace all relations $R$ that are neither $\alt{\EQset}$ nor $\alt{\HWsetGen{=1}}$ relations by the corresponding graph $H_R$,
    where we identify the portals of $H_R$
    with the respective vertices in the scope of $R$.
    Let $H'$ be the new instance of
    \srCountDomSetRel[\CP,\alt{\EQset},\alt{\HWsetGen{=1}}].
    It follows from the proofs of \cref{lem:realizing:arbitraryRelations,%
    lem:realizing:parsiArbitraryToHWoneAndEQ}
    that the number of solutions of $H$ and $H'$ are the same.
    Moreover, as the arity in the source problem is constant,
    the modification does not change the pathwidth
    by more than an additive constant (see \cref{lem:realizing:parsiArbitraryToHWoneAndEQ} for the details).

    \item[Removing equality relations.]
    Let $H$ be an instance of
    \srCountDomSetRel[\CP,\alt{\EQset},\alt{\HWsetGen{=1}}].
    We replace each relation \EQ{k} by the graph $H_k$,
    which is defined as follows,
    to obtain a \srCountDomSetRel[\CP,\alt{\HWsetGen{=1}}] instance $H'$.

    The graph with portals $H_k=(G_k, \Port_k)$,
    where we have \(\Port_k = \{\port_1,\dots,\port_k\}\),
    contains a single vertex $v$ and a $\alt{\HWsetGen{=1}}$ relation between
    $v$ and each of $\port_1,\dots,\port_k$.

    It directly follows from the construction and the proofs
    of \cref{lem:realizing:eqConditional,lem:realizing:parsiHWoneAndEQToHWone}
    that the number of solutions for $H$ and $H'$ is the same,
    the pathwidth increases only by an additive constant,
    and the arity increases by at most $2$.
  \end{description}
  By first applying the first part of the proof and then the second part,
  the statement follows.
\end{proof}

As a next step, we show how to replace the $\alt{\HWsetGen{=1}}$ relations
by $\alt{\HWsetGen{\ge1}}$ relations.

\begin{lemma}\label{lem:HWeqS}
  Let $(\sigma,\rho)\in \allSetsne^2$ be non-trivial with $\rho= \ZZ_{\ge0}$ and cofinite $\sigma$.
  Let $\CP$ be some set of pairs from $\allSets^2$.
  Then, \srCountDomSetRel[\CP,\alt{\HWsetGen{=1}}] $\pwred$
  \srCountDomSetRel[\CP,\alt{\HWsetGen{\ge1}}].
\end{lemma}
\begin{proof}
  The proof is almost identical to that of \cref{lem:IIIrelversion}.
  Let $I=(V,E,\CC)$ be an instance of \srCountDomSetRel[\CP,\alt{\HWsetGen{=1}}].
  Let $\CZ=\{\scope(C) \mid C\in \CC\}$, and let $U$ be the union of the sets in $\CZ$.

  Let $G=(V,E)$.
  For a positive integer $x$, let $I_x$ be the instance of
  \srCountDomSetRel[\alt{\CP,\HWsetGen{\ge1}}] obtained from $G$
  by attaching, for each $Z\in \CZ$, a relational constraint
  with relation $\alt{\HWge{1}}$ and scope $Z$.
  In addition, we introduce, for each subset $Z'$ of $Z$ with $\abs{Z'}=\abs{Z}-1$,
  a total of $x$ $(\sigma,\rho)$-vertices
  $u^{(Z,Z')}_1,\dots,u^{(Z,Z')}_x$. For each $i\in \numb{x}$, we introduce a relation $\alt{\HWge{1}}$ with scope $Z'\cup \{u^{(Z,Z')}_i\}$.

  As in previous proofs, a selection of vertices from $V$ is \emph{good} if all vertices in $V\setminus U$ have a feasible number of selected neighbors.
  If, in a good selection $S$, some $Z'$ has at least one selected vertex,
  then there are $2^x$ feasible extensions
  to the vertices $u^{(Z,Z')}_1, \dots, u^{(Z,Z')}_x$.
  If no vertex of $Z'$ is selected, then there is only a single such extension (where all of $u^{(Z,Z')}_1, \ldots, u^{(Z,Z')}_x$ are selected).
  Continuing from this, we say that if a set $Z\in \CZ$ contains at least 2 selected vertices (from $S$), then it is \emph{undesired}, and if it contains precisely 1 selected vertex, then it is \emph{desired}.

  Observe that each $Z \in \CZ$ has to contain at least one selected vertex,
  that is, $Z$ is either desired or undesired.
  For an undesired set $Z\in \CZ$, there is at least one vertex selected in each of the $\abs{Z}$ corresponding subsets $Z'$ (this gives a total of $2^{x\abs{Z}}$ extensions), whereas, for a desired set $Z$, there is exactly one subset $Z'$ that is entirely unselected (which gives a total of $2^{x(\abs{Z}-1)}$ extensions).

  For $\CY\subseteq \CZ$, let $a_{\CZ'}$ be the number of good selections in which precisely the sets in $\CY$ are undesired and the sets in $\CZ\setminus \CY$ are desired. Similarly, let $a_i$ be the number of good selections in which precisely $i$ of the sets in $\CZ$ are undesired. Let $\#S(I_x)$ denote the number of solutions of $I_x$. Then,
  \begin{align*}
    \#S(I_x)
    &=\sum_{\CY\subseteq \CZ} a_{\CY}\cdot
      \left( \prod_{Z\in \CY} 2^{x\abs{Z}}
        \cdot \prod_{Z\notin \CY} 2^{x(\abs{Z}-1)} \right)\\
    &= 2^{x\sum_{Z\in \CZ}(\abs{Z}-1)}
    \cdot \sum_{\CY\subseteq \CZ} a_{\CZ'}\cdot  2^{x\abs{\CY}}\\
    &= F \cdot \sum_{i=0}^{\abs{\CZ}} a_{i}\cdot  2^{xi}
    \intertext{where we set}
    F &\deff 2^{x\sum_{Z\in \CZ}(\abs{Z}-1)}
    .
  \end{align*}
  Note that $\#S(I_x)/F$ is a polynomial in $2^x$ of degree $\abs{\CZ}$,
  and by choosing $\abs{\CZ}$ different values of $x$,
  we can use polynomial interpolation to recover the coefficients $a_i$.
  In particular, $a_{0}$ is the number of good selections of vertices from $V$
  in which no set in $\CZ$ is undesired,
  that is, every set $Z\in \CZ$ is desired (i.e., contains exactly one selected vertex). So, $a_0$ is precisely the sought-after number of solutions of instance $I$
  of \srCountDomSetRel[\CP,\alt{\HWsetGen{=1}}].

  Recall that the scope of each relation of $I$
  is considered to be a clique in the definition of the pathwidth of $I$.
  Using this, for $Z\in \CZ$, we can add,
  for each $Z'\subseteq Z$ with $\abs{Z'}=\abs{Z}-1$,
  the vertices $v^{(Z,Z')}_1,\dots, v^{(Z,Z')}_x$
  one after another to a new copy of the original bag containing the clique $Z$.
  Hence, the pathwidth of $I_x$ is at most that of $I$
  plus an additive constant.
\end{proof}

As the next step, we remove the relations from the instance.
We do this by making use of $(\NN, \emptyset)$-vertices,
that is, vertices which are forced to be selected,
as well as $(\ZZ_{\ge1},\NN)$-vertices, which we already used in \cref{lem:RelfromweightedRelS}.

\begin{lemma}\label{lem:HWgeS}
  Let $(\sigma,\rho)\in \allSetsne^2$ be non-trivial
  with $\rho= \ZZ_{\ge0}$ and cofinite $\sigma$.
  Let $\CP$ be some set of pairs from $\allSets^2$.
  Then, \srCountDomSetRel[\CP, \alt{\HWsetGen{\ge1}}]
  $\pwred$ \srCountDomSetRel[\CP + (\ZZ_{\ge1},\NN) + (\NN,\emptyset)],
  even if $(\NN,\emptyset)$ is $1$-bounded in the target problem.
\end{lemma}
\begin{figure}[t]
  \centering
  \includegraphics[scale=1.9]{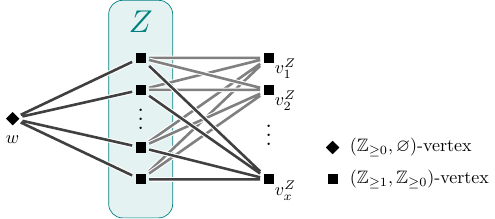}
  \caption{
    The construction of the instance $I_x$
    from \cref{lem:HWgeS}.
    The vertex $w$ is a $(\NN,\emptyset)$-vertex
    and the other vertices are $(\ZZ_{\ge 1},\NN)$-vertices.
  }\label{fig:HWgeS}
\end{figure}
\begin{proof}
  The proof uses a simple interpolation.
  Let $I=(G,\lambda)$ with $G = (V,E,\CC)$ be an instance of
  \srCountDomSetRel[\CP,\alt{\HWsetGen{\ge 1}}].
  Let $\CZ=\{\scope(C)\mid C\in \CC\}$, and let $U$ be the union of sets in $\CZ$.
  As all relational constraints use $\HWge{1}$-relations,
  the $\alt{\lambda}$-sets of $G$ are precisely the selections
  of vertices from $V$ for which
  (1) every vertex in $V\setminus U$ has a feasible number
  of selected neighbors according to the constraints from $\lambda$,
  and (2) at least one vertex from each set $Z\in \CZ$ is selected.

  For a positive integer $x$, let $I_x=(G',\lambda')$ be the instance
  of \srCountDomSetRel[\CP+(\ZZ_{\ge1},\NN) + (\NN,\emptyset)]
  obtained from the graph $(V,E)$ (without relational constraints) by
  adding one $(\NN,\emptyset)$-vertex $w$,
  and adding, for each set $Z\in \CZ$, a total of $x$ new
  $(\ZZ_{\ge1},\NN)$-vertices $v_1^Z, \ldots, v_x^Z$.
  The vertices in $Z$ are completely connected to all vertices $v_i^Z$
  and to the vertex $w$.
  Moreover, $\lambda'$ is defined such that each vertex in $U$ is a $(\ZZ_{\ge1},\NN)$-vertex. On the vertices in $V\setminus U$, $\lambda'$ is identical to $\lambda$.
  See \cref{fig:HWgeS} for an illustration of this construction.

  A selection of vertices from $V$ is \emph{good} if all vertices in $V\setminus U$ have a feasible number of selected neighbors (according to $\lambda$).
  Now, let us consider the number of extensions of a good selection $S$
  to solutions of $I_x$.
  First, in every solution of $I_x$, the $(\NN,\emptyset)$-vertex $w$
  has to be selected, which implies that, for each $Z\in \CZ$,
  all vertices in $Z$ (they are $(\ZZ_{\ge1},\NN)$-vertices)
  have a feasible number of selected neighbors
  independently of how many neighbors they had in $S$.
  If in a good selection some $Z\in \CZ$ contains at least one selected vertex, then there are $2^x$ extensions to the vertices $v_1^Z, \ldots, v_x^Z$, all of which can either be selected or not.
  If no vertex in $Z$ is selected, then there is only a single extension, in which all of $v_1^Z, \ldots, v_x^Z$ are unselected.
  For $i\in \fragment{0}{\abs{\CZ}}$, let $a_{i}$ be the number of good selections in which exactly $i$ of the sets in $\CZ$ contain at least one selected vertex.
  Then, the number of solutions $\#S(I_x)$ of $I_x$ is
  \[
    \#S(I_x)=\sum_{i=0}^{\abs{\CZ}} a_i \cdot 2^{xi}.
  \]
  This is a polynomial in $2^x$ of degree $\abs{\CZ}$,
  and by choosing $\abs{\CZ}$ different values of $x$,
  we can use polynomial interpolation to recover the coefficients $a_i$.
  In particular, $a_{\abs{\CZ}}$ is the number of good selections in which every set $Z\in \CZ$
  contains at least one selected vertex,
  which is precisely the sought-after number of solutions
  of instance $I$ of \srCountDomSetRel[\CP,\alt{\HWsetGen{\ge1}}].

  Since each set $Z\in \CZ$ is considered as a clique in the definition of the pathwidth of $I$,
  we can add all vertices $v^Z_1,\ldots, v^Z_x$
  one after another to a new copy of the original bag containing the clique $Z$.
  Moreover, we can add the vertex $w$ to all bags of the decomposition.
  Hence, the pathwidth of $I_x$ is at most that of $I$
  plus an additive constant.
\end{proof}

Recall that, for a positive integer $s$ and a set $\tau$ of integers,
we use the short form $\tau-s$ to denote the set $\{t - s \mid t \in \tau, t-s \ge 0\}$.
In the next step, we show how a total of $s$ many $(\NN,\emptyset)$-vertices can be used to ``shift'' some pair $(\sigma', \rho')$ down to $(\sigma'-s, \rho'-s)$.
The shifting works for any $(\sigma', \rho')$.
In the end,
we apply this result to $(\sigma',\rho')=(\sigma,\rho)$ with $s=\sigMax-1$
to model the pair $(\sigma-s,\rho-s)=(\ZZ_{\ge1}, \NN)$
(where we use that $\rho=\NN$).

\begin{lemma}\label{lem:sigcof2}
  Let $(\sigma,\rho)\in \allSetsne^2$ be non-trivial
  with $\rho= \ZZ_{\ge0}$ and cofinite $\sigma$.
  Let $(\sigma',\rho') \in \allSetsne^2$.
  Then, for all constants $c,s\ge 0$,
  \srCountDomSetRel[(\NN,\emptyset)+(\sigma'-s,\rho'-s)] $\pwred$
  \srCountDomSetRel[(\NN,\emptyset)+(\sigma',\rho')],
  where $(\NN,\emptyset)$ is $c$-bounded in the source problem,
  and $(\NN,\emptyset)$ is $c+s$-bounded in the target problem.
\end{lemma}
\begin{proof}
  Let $I$ be an instance of \srCountDomSetRel[(\NN,\emptyset),(\sigma'-s,\rho'-s)] with at most $c$ $(\NN,\emptyset)$-vertices.
  Let $U$ be the set of $(\sigma'-s,\rho'-s)$-vertices.
  From $I$, we define an instance $I'=(G',\lambda')$ of
  \srCountDomSetRel[(\NN,\emptyset)+(\sigma',\rho')]
  by augmenting the graph $G$ by $s$ $(\NN,\emptyset)$-vertices $v_1,\dots,v_s$,
  which are all adjacent to all vertices in $U$. This means that $I'$ contains a total of at most $c+s$ $(\NN,\emptyset)$-vertices ($c$ from $I$, and $s$ new ones).
  We define $\lambda'$ such that all vertices in $U$ are $(\sigma',\rho')$-vertices in $I'$.
  Note that the solutions of $I'$ are in a one-to-one correspondence
  to the solutions of $I$ for \srCountDomSetRel[(\NN,\emptyset),(\sigma'-s,\rho'-s)]
  as $v_1,\dots,v_s$ are always selected, and therefore, each vertex in $U$ obtains precisely $s$ additional selected neighbors.

  Observe that we can add all vertices $v_1,\dots,v_s$
  to every bag of a given path decomposition of $I$.
  Thus, the pathwidth increases by at most the constant $s$.
\end{proof}

The final step in the chain of reductions
is to remove the $(\NN,\emptyset)$-vertices from the instance,
that is, the vertices that are forced to be selected.

\begin{lemma}\label{lem:sigcof1}
  Let $(\sigma,\rho)\in \allSetsne^2$ be non-trivial
  with $\rho= \NN$ and cofinite $\sigma$.
  For every constant $c$,
  \srCountDomSetRel[(\NN,\emptyset)] $\pwred$ \srCountDomSet,
  where $(\NN,\emptyset)$ is $c$-bounded in the source problem.
\end{lemma}
\begin{proof}
  First, observe that it suffices to give the pathwidth-preserving reduction
  from \srCountDomSetRel[(\NN,\emptyset)]
  where $(\NN,\emptyset)$ is $c+1$-bounded
  to \srCountDomSetRel[(\NN,\emptyset)]
  where $(\NN,\emptyset)$ is $c$-bounded.
  This means that we have to give a pathwidth-preserving reduction
  that removes a single $(\NN, \emptyset)$-vertex
  from the \srCountDomSetRel[\CP+(\NN,\emptyset)] instance.
  We can apply this reduction $c$ times
  so that all $(\NN,\emptyset)$-vertices are removed.
  As $c$ is a constant, this chain of reductions has constant length
  and is still pathwidth-preserving in total.

  Let $I=(G,\lambda)$ be an instance of \srCountDomSetRel[(\NN,\emptyset)],
  and let $u$ be some (fixed) $(\NN,\emptyset)$-vertex of $I$.
  (If no such vertex exists, then we are done
  as we can simply cast $I$ as an instance of \srCountDomSet.)

  The proof splits into two cases.
  In the first case, suppose that $\sigMin>0$.
  For some yet to be determined integer $x$,
  from $I$ we define an instance $I_x$ of \srCountDomSetRel[]
  by turning $u$ into a $(\sigma,\rho)$-vertex
  and attaching $x$ copies of the following gadget $\CJ$,
  where $u$ acts as the portal to each of the copies.
  $\CJ$ has the single portal $u$ with two adjacent vertices $v_1$ and $v_2$
  that are part of an $(\sigMin+1)$-clique $C$
  from which the edge between $v_1$ and $v_2$ is removed.
  All the vertices of $\CJ$ are $(\sigma,\rho)$-vertices.
  This means that $v_1$ and $v_2$ have $\sigMin-1$ neighbors in $C$,
  whereas the remaining vertices of $C$ have $\sigMin$ neighbors.
  If $u$ is unselected, then each of $v_1$ and $v_2$ can have at most $\sigMin-1$ selected neighbors. Hence, they cannot be selected in this case. Then, if $v_1$ and $v_2$ are unselected, the rest of $C$ has to be unselected as well. Therefore, there is exactly $1$ partial solution of $\CJ$ in which $u$ is unselected: all vertices of $C$ unselected. (This is indeed a partial solution as $\rho=\NN$.) Similarly, there are exactly $2$ partial solutions if $u$ is selected: all vertices of $C$ unselected, in which case $u$ obtains no selected neighbors --- or all vertices of $C$ selected, in which case $u$ obtains two selected neighbors.

  A subset $W$ of the vertices of $G$ is called \emph{good}
  if all vertices in $V(G) \setminus \{u\}$ have a feasible number of neighbors
  in $W$ according to $\lambda$.
  Consider a good subset $W$ of the vertices of $G$
  that does \emph{not} contain $u$.
  Since $\rho=\ZZ_{\ge0}$, there is exactly one way
  to extend $W$ to a solution $S$ of $I_x$,
  and that is by not selecting any vertices of the $x$ attached cliques.
  Now, consider a good subset $W'$ of the vertices of $G$
  that \emph{does} contain $u$ along with $i$ neighbors of $u$ in $G$. Such a set $W'$ can be extended to a solution of $I_x$ by selecting $s$ of the $x$ cliques in such a way that the total number of selected neighbors of $u$, that is, $i+2s$, is in $\sigma$.
  So, there are
  \begin{align*}
    \sum_{\substack{s\in \fragment{0}{x}:\\i+2s \in \sigma}} \binom{x}{s}
    &= 2^x - \sum_{\substack{s\in\fragment{0}{x}:\\i+2s\notin \sigma}} \binom{x}{s}
    = 2^x - p_i(x)
  \intertext{
    extensions for $W'$ to a solution of $I_x$,
    where
  }
    p_i(x)&\deff \sum_{i+2s\notin \sigma} \binom{x}{s}
  \end{align*}
  is a polynomial in $x$.

  Let $g$ be the number of good sets $W\subseteq V(G)$ with $u\notin W$. Let $d$ denote the degree of $u$ in $G$.
  For all $i \in \fragment{0}{d}$,
  let $f_i$ be the number of good sets $W'\subseteq V(G)$
  with $u\in W'$ and $\abs{N(u)\cap W'}=i$.
  Hence, $\sum_{i=0}^d f_i$ is precisely the sought-after number of solutions of $I$.
  Let $\#S(I_x)$ be the number of solutions of $I_x$ for \srCountDomSetRel[].
  Then,
  \begin{align*}
    \#S(I_x)
    &=g +\sum_{i=0}^d f_i\cdot(2^x-p_i(x)).
  \end{align*}
 	Let $n$ be the number of vertices of $G$, and let $L(x)\deff \min_{i\in\fragment{0}{d}} (2^x-p_i(x))$.
  We claim that, for sufficiently large $x\in \Oh(n)$, we can compute the value $\sum_{i=0}^d f_i$
  from $\#S(I_x)$ by computing $\floor{\#S(I_x)/L(x)}$.
  For given $x$, this value is easily computed using one oracle call
  to \srCountDomSetRel[(\NN,\emptyset)] where $(\NN,\emptyset)$ is $c$-bounded
  (that is, \srCountDomSet if $c=0$)
  as the polynomials $p_i(x)$ depend only on $\sigma$. The intuition is that $g$ and each of the $f_i$'s are upper-bounded by $2^n$, whereas $L(x)$ is exponential in $x$. Therefore, dividing through by $L(x)$ with sufficiently large $x$, eliminates the influence of $g$ and ensures that we differ from $\sum_{i=0}^d f_i$ by less than 1.

  \begin{claim}
    There is an $x \in \Oh(n)$ such that
    $\floor{\#S(I_x)/L(x)} = \sum_{i=0}^d f_i$.
  \end{claim}
  \begin{claimproof}
    We directly get that
    \[
      \frac{\#S(I_x)}{L(x)}
      \ge \sum_{i=0}^d f_i \cdot
        \frac{(2^x-p_i(x))}{\min_{i\in\fragment{0}{d}} (2^x-p_i(x))}
      \ge \sum_{i=0}^d f_i
    \]
    by our choice of $L(x)$.

    To prove the claim, it remains to show that,
    for sufficiently large $x \in \Oh(n)$,
    \begin{align}
      \frac{g}{L(x)}
      &< \frac{1}{d+2}
      \label{eqn:cofinite:boundForG}
    \intertext{and}
      \frac{2^x - p_i(x)}{L(x)}
      &< 1 + \frac{1}{(d+2) f_i}
      .\label{eqn:cofinite:boundForF}
    \end{align}
    With these two bounds and the definition of $\#S(I_x)$, we get
    \begin{align*}
      \frac{\#S(I_x)}{L(x)}
      &\le \frac{g}{L(x)}
        + \sum_{i=0}^d f_i \cdot \frac{2^x - p_i(x)}{L(x)} \\
      &< \frac{1}{d+2}
        + \sum_{i=0}^d \left( f_i + \frac{1}{d+2} \right) \\
      &\le \sum_{i=0}^d f_i +1,
    \end{align*}
    and therefore, we have $\floor{\#S(I_x)/L(x)} = \sum_{i=0}^d f_i$.

    For ease of notation, we denote by $j$ the index from $\fragment{0}{d}$
    such that $p_j(x) = \max_i p_i(x)$.
    To see that \cref{eqn:cofinite:boundForG} holds,
    observe that $g \le 2^n$.
    We can choose $x$ sufficiently large in $\Oh(n)$
    such that $(d+2)2^n<L(x)$.

    To prove \cref{eqn:cofinite:boundForF}, we first observe that
    \[
      \frac{2^x - p_i(x)}{L(x)}
      \le \frac{2^x}{2^x - p_j(x)}
      = 1 + \frac{p_j(x)}{2^x - p_j(x)}
      .
    \]
    Since $f_i\le 2^n$ and $p_j(x)$ is a polynomial in $x$,
    we can choose $x$ sufficiently large in $\Oh(n)$
    such that  $\frac{p_j(x)}{2^x - p_j(x)}<  \frac{1}{(d+2)2^n}$.
  \end{claimproof}

  Now, we consider the second case where $\sigMin=0$. We use a similar approach as in the previous case, but with a different gadget. Let $t$ be the minimum integer missing from $\sigma$. So, $t\ge 1$ since $\sigMin=0$.
  Again, for some yet to be determined integer $x$,
  from $I$ we define an instance $I_x$ of \srCountDomSetRel[]
  by making $u$ a $(\sigma,\rho)$-vertex
  and attaching $x$ copies of a gadget $\CJ$,
  where $u$ acts as the portal to each of the copies.
  Here, $\CJ$ has the single portal $u$ that is fully connected to a $t$-clique $C$. All the vertices of $\CJ$ are $(\sigma,\rho)$-vertices.
  Note that there are $2^t$ partial solutions of $\CJ$ in which $u$ is unselected (all vertices of $C$ can either be selected or unselected, which is fine as $\rho=\NN$ and $\sigma$ contains all values $1,\ldots,t-1$) -- and there are $2^t-1$ partial solutions of $\CJ$ in which $u$ is selected (all the previous options with the exception of selecting the entire clique $C$ as $t\notin\sigma$).

  We first analyze the number of extensions assuming that $u$
  is selected and already has $i$ selected neighbors.
  Then, the number of extensions is given by
  \begin{align*}
    \ext_i(x)
    &\deff \sum_{i+s\in \sigma}
          \sum_{\substack{~i_1+\dots+i_x=s;~\\i_{\ast} \le t-1}}
          \prod_{j=1}^x \binom{t}{i_j} \\
    &= (2^t-1)^x
      - \underbrace{\sum_{i+s\notin\sigma} \sum_{\substack{~i_1+\dots+i_x=s;~\\i_{\ast}\le t-1}}
          \prod_{j=1}^x \binom{t}{i_j}}_{p_i(x) \deff}
   \end{align*}
	Intuitively, $p_i(x)$ is the number of possibilities to select $s$ vertices from the $x$ cliques such that $i+s\notin \sigma$ and no clique is entirely selected.
	Hence, $p_i(x)$ is upper-bounded by $q_i(x)\deff \sum_{i+s \notin \sigma} \binom{xt}{s}$, which is the number of possibilities to select $s$ vertices from the $x$ cliques (each of size $t$).
    In particular, $q_i(x)$ is a polynomial in $x$ with $q_i(x) \ge p_i(x)$.
  If we let $g$ and $f_i$ be defined as in the previous case,
  we obtain
  \[
    \#S(I_x) = g\cdot 2^{tx}
        + \sum_{i=0}^d f_i \cdot ((2^t-1)^x-p_i(x)).
  \]
  We claim that we can compute $\sum_{i=0}^k f_i$ from $\#S(I_x)$
  by computing $\floor{ (\#S(I_x)\mod 2^{tx} ) /L(x)}$,
  where $L(x) \deff \min_{i\in\fragment{0}{d}} \ext_i(x)$.
  \begin{claim}
    For sufficiently large $x\in \Oh(n)$, we have
    $\sum_{i=0}^d f_i = \floor{ (\#S(I_x)\mod 2^{tx} ) /L(x)}$.
  \end{claim}
  \begin{claimproof}
    We let $\psi(x) \deff \#S(I_x)\mod 2^{tx}$ to simplify notation.
    The first step is to show that
    $\psi(x) = \sum_{i=0}^d f_i\cdot \ext_i(x)$. 
    By the earlier observations and the fact that $G$ has $n$ vertices,
    we have
    \[
      \sum_{i=0}^d f_i\cdot \ext_i(x)
      = \sum_{i=0}^d f_i \cdot ((2^t-1)^x-p_i(x))
      < (2^t-1)^x \cdot \sum_{i=0}^d f_i
      \le (2^t-1)^x \cdot 2^n.
    \]
    Then, for sufficiently large $x\in \Oh(n)$,
   we have $\left(\frac{2^t-1}{2^t}\right)^x<2^n$, and consequently, $(2^t-1)^x \cdot 2^n<2^{tx}$.

    By our choice of $L(x)$ as the minimum over all $\ext_i(x)$,
    we directly obtain that $\psi(x) \ge \sum_{i=0}^d f_i$.
    For the upper bound, we follow the ideas from the previous claim.
    Again, let $j$ be the index such that $L(x) = (2^t-1)^x-p_j(x)$.
    Thus, it suffices to show that, for all $i$,
    \[
              \frac{\ext_i(x)}       {L(x)}
      =       \frac{(2^t-1)^x-p_i(x)}{L(x)}
      \le     \frac{(2^t-1)^x}       {(2^t-1)^x-p_j(x)}
      =   1 + \frac{p_j(x)}          {(2^t-1)^x-p_j(x)}
      \le 1 + \frac{q_j(x)}          {(2^t-1)^x - q_j(x)}
    \]
    is strictly smaller than $1 + {1}/{(d+1) f_i}$.
    Since $q_j(x)$ is a polynomial,
    the exponential term $(2^t-1)^x$ again dominates the quotient.
    Moreover, since $f_i \le 2^n$, we can choose $x \in \Oh(n)$
    to show the bound.
  \end{claimproof}

  In both cases, the reduction is pathwidth-preserving
  as we can add each copy of the gadget $\CJ$ one after the other
  to a new copy of a bag containing $u$ from the original path decomposition.
  Thus, the pathwidth increases by at most the size of $\CJ$,
  which is bounded by a constant (which might depend on the fixed $\sigma$).
\end{proof}

Now, we have everything ready to prove the main result of this section
which we restate in the following.
\lemrhoeverythingtwo*
\begin{proof}
  The proof chains a number of previous reductions as follows. Also, see the orange (rightmost) box in \cref{fig:count:removingRelations}
  for a visualization.
  \begin{enumerate}
  		\item
      By chaining \cref{lem:RelfromweightedRelS}
      together with \cref{lem:vertexWeightedRelS,lem:relWeightedRelS},
      where we set $\CP=\{(\ZZ_{\ge 1},\NN)\}$ in the latter two results,
      we obtain a $\pwar$-reduction from \srCountDomSetRel
      to \srCountDomSetRel[\{(\ZZ_{\ge 1},\NN)\},\altRel].
      Since these reductions are arity-preserving,
      in order to solve \srCountDomSetRel on instances of arity at most $d$,
      it suffices to solve \srCountDomSetRel[\{(\ZZ_{\ge 1},\NN)\},\altRel]
      on instances of arity at most $d'$,
      where $d'$ is $d$ plus some constant.
  		\item
      \Cref{lem:RelSfromHWS} with $\CP=\{(\ZZ_{\ge 1},\NN)\}$
      gives a pathwidth-preserving reduction
      from \srCountDomSetRel[\{(\ZZ_{\ge 1},\NN)\},\altRel]
      on instances of arity at most $d'$ to
      \srCountDomSetRel[\{(\ZZ_{\ge 1},\NN)\},\alt{\HWsetGen{=1}}].
  		\item
      \Cref{lem:HWeqS,lem:HWgeS}, again with $\CP=\{(\ZZ_{\ge 1},\NN)\}$,
      give a pathwidth-preserving reduction
      from \srCountDomSetRel[\{(\ZZ_{\ge 1},\NN)\},\alt{\HWsetGen{=1}}]
      to \srCountDomSetRel[(\ZZ_{\ge1},\NN) + (\NN,\emptyset)],
      even if $(\NN,\emptyset)$ is $1$-bounded in the target problem.
  		\item
      We set $c=1$, $s=\sigMax-1$,
      and $(\sigma',\rho')=(\sigma, \rho)$ in \cref{lem:sigcof2}
      to obtain a pathwidth-preserving reduction
      from \srCountDomSetRel[(\ZZ_{\ge1},\NN) + (\NN,\emptyset)]
      with $1$-bounded $(\NN,\emptyset)$ to \srCountDomSetRel[(\NN,\emptyset)]
      with $s'$-bounded $(\NN,\emptyset)$, where $s'=s+1=\sigMax$.
      (Pro forma, note that \srCountDomSetRel[(\sigma,\rho)+(\NN,\emptyset)]
      is the same problem as \srCountDomSetRel[(\NN,\emptyset)].)
  		\item
      \cref{lem:sigcof1} with $c=s'=\sigMax$
      gives a pathwidth-preserving reduction from
      \srCountDomSetRel[(\NN,\emptyset)] with $s'$-bounded $(\NN,\emptyset)$
      to \srCountDomSet, which finishes the proof.\qedhere
  \end{enumerate}
\end{proof}

\clearpage
\small
\bibliographystyle{plainnat}
\bibliography{gen-dom-set}

\end{document}